\documentclass[a4paper,11pt]{article}
\usepackage{jheppub}
\usepackage{xcolor}
\newif\ifshowchanges
\showchangesfalse
\newcommand{\rev}[1]{{\ifshowchanges\color{blue}\else\color{black}\fi #1}}
\newenvironment{revision}{\begingroup\ifshowchanges\color{blue}\else\color{black}\fi}{\endgroup}

\makeatletter
\gdef\@fpheader{}
\makeatother
\usepackage[T1]{fontenc}
\usepackage[utf8]{inputenc}
\usepackage{lmodern}
\usepackage{amsthm,mathtools}
\usepackage{mathrsfs}
\usepackage{bm}
\usepackage{booktabs}
\usepackage{array}
\usepackage{enumitem}
\usepackage{microtype}
\usepackage{tikz}
\usepackage[nameinlink,noabbrev]{cleveref}
\graphicspath{{figures/}}
\hypersetup{
  hypertexnames=false,
  pdftitle={When does a state-dependent proto-area define a bulk geometry?},
  pdfauthor={Ruiliang Li},
  pdfsubject={Nonlinear geometrizability in holographic quantum error correction}
}
\usetikzlibrary{arrows.meta,calc,positioning,decorations.pathreplacing}
\newtheorem{theorem}{Theorem}[section]
\newtheorem{proposition}[theorem]{Proposition}
\newtheorem{corollary}[theorem]{Corollary}
\newtheorem{lemma}[theorem]{Lemma}
\theoremstyle{definition}
\newtheorem{definition}[theorem]{Definition}

\newcommand{\Rank}{\operatorname{rank}}
\newcommand{\Dist}{\operatorname{dist}}

\newcommand{\Supp}{\operatorname{supp}}
\newcommand{\Rgeq}{\mathbb{R}_{\geq 0}}
\newcommand{\Rpos}{\mathbb{R}_{>0}}
\newcommand{\one}{\mathbf{1}}
\newcommand{\cA}{\mathscr{C}}
\newcommand{\calI}{\mathcal{I}}
\newcommand{\bfa}{\boldsymbol{a}}
\newcommand{\bfq}{\boldsymbol{q}}
\newcommand{\bfw}{\boldsymbol{w}}

\newcommand{\bfe}{\boldsymbol{e}}
\newcommand{\cK}{\mathcal{K}}
\newcommand{\cP}{\mathcal{P}}
\newcommand{\cZ}{\mathcal{Z}}

\newcommand{\ip}[2]{\left\langle #1,#2\right\rangle}
\newcommand{\norm}[1]{\left\lVert #1\right\rVert}
\newcommand{\abs}[1]{\left\lvert #1\right\rvert}
\newcommand{\Ran}{\operatorname{Ran}}

\newcommand{\Id}{\operatorname{Id}}
\newcommand{\tr}{\operatorname{tr}}

\title{\boldmath When does a state-dependent proto-area define a bulk geometry?}

\author[a]{Ruiliang Li}
\affiliation[a]{Tsinghua University, Beijing 100084, China}
\emailAdd{lirl23@mails.tsinghua.edu.cn}

\abstract{\rev{A holographic code assigns a state-dependent entropy remainder to every boundary region, but such remainders need not be the extremal lengths of any one bulk metric. We study this common-source question for the proto-area of a calibrated recovery, the entropy of a boundary region minus that of the matter recovered from it. At zeroth order the reference values must match the lengths of one reference geometry in a fixed renormalization scheme, and this matching also determines the entropy-to-length factor. On a graph with a fixed minimum-cut chamber the data derivatives must lie in the image of the cut-incidence matrix, whose left null vectors are witnesses. On the hyperbolic disk the first variation of renormalized length is the rank-two geodesic X-ray transform, with closed range in an adapted data space, and the geodesic displacement, fixed by a forced Jacobi equation, enters at second order. The resulting normal Hessian gives a gauge-invariant criterion, necessary and sufficient, for a regular pair of response coefficients to come from one metric two-jet. For the constant-curvature family generated by a boundary reparametrization the Hessian has a nonzero rank-four mode, so the image of the length map is curved. For a fixed recovery the proto-area Hessian is minus the Bogoliubov--Kubo--Mori information lost in recovery, which the Petz map decomposes orthogonally at a product reference, and which, on a finite-dimensional space of logical directions with positive minimum retained fraction, is at most a sharp multiple of the retained information. Geometric realization requires this loss to reproduce the normal length Hessian on every interval. Two codes show that exact regional area operators can fail to sew into one graph geometry, and that a matched reference with vanishing first response can fail at second order. All criteria refer to one recovery branch and one extremal-surface branch.}}

\begin{document}
\maketitle
\flushbottom

\section{Introduction}
\label{sec:introduction}

\begin{revision}
Holographic entropy is organized by a single bulk spacetime.  The
Ryu--Takayanagi formula assigns to each boundary region the area of an
extremal surface, and one metric determines all of these areas at
once~\cite{RyuTakayanagi2006,HubenyRangamaniTakayanagi2007,FaulknerLewkowyczMaldacena2013,EngelhardtWall2015}.
Quantum error correction explains why this is compatible with the
redundancy of bulk reconstruction, and it identifies the area term with
the entropy of degrees of freedom that a region shares with its
complement across the
surface~\cite{AlmheiriDongHarlow2015,PastawskiEtAl2015,DongHarlowWall2016,Harlow2017,Cotler:2017erl,AlmheiriDongSwingle2017,AkersPenington2022}.
In a finite code this identification produces an area-like number for
each region on its own.  Nothing in the regional construction says that
these numbers are the areas of one geometry, and this paper asks what
does.
\end{revision}

\begin{revision}
The object we study is the proto-area of a region.  Let $V$ be an
isometric encoding of a logical state $\rho$ into boundary degrees of
freedom, let $\omega_A(\rho)=\tr_{\bar A}(V\rho V^\dagger)$ be its
marginal on a boundary region $A$, and let $\tau_A^{\mathrm{cal}}(\rho)$
be the matter state returned by a recovery map that is calibrated once
on the code and then held fixed while $\rho$ varies.  The entropy-valued
proto-area is
\begin{equation}
 S_{\mathrm{PA}}(V,\rho;A)
 =S(\omega_A(\rho))-S(\tau_A^{\mathrm{cal}}(\rho)),
 \label{eq:introduction-proto-area-definition}
\end{equation}
with $S$ the von Neumann entropy in natural logarithms.  It is the part
of the boundary entropy that the recovery does not account for as matter
entropy.  Cao, Cheng, Karthikeyan, Li, and Preskill introduced this
quantity, choosing the two local recoveries on $A$ and $\bar A$ so as to
maximize the coherent information of their joint
output~\cite{CaoEtAl2026}.  Section~\ref{subsec:proto-area-optimal-recovery}
reviews that calibration, including the possibility that the optimum is
not unique.  Our criteria use only two features of the construction,
that the recovery is fixed before the state is varied and that the
output channels are linear in $\rho$.  They therefore apply to the
entropy remainder of any fixed recovery with a stated operational
calibration, which we call a branch proto-area.  The code examples below
use exact algebraic recovery or a recovery minimizing the worst-case
channel error, and we do not identify either calibration with
coherent-information maximization.
\end{revision}

\begin{revision}
Exact recovery provides the benchmark against which state dependence is
measured.  If the logical algebra carried by the entanglement wedge of
$A$ is a full matrix algebra and complementary recovery is exact, with
$A$ recovering its wedge algebra and $\bar A$ the complementary one,
local decoding puts the boundary marginal in the form
\begin{equation}
 \omega_A(\rho)\simeq\rho_{a_A}\otimes\chi_{A_2},
 \qquad \tau_A^{\mathrm{cal}}(\rho)=\rho_{a_A},
 \qquad S_{\mathrm{PA}}=S(\chi_{A_2}),
 \label{eq:introduction-exact-benchmark}
\end{equation}
where $\simeq$ denotes a fixed local isometric identification and
$\chi_{A_2}$ does not depend on $\rho$.  The proto-area is then a
constant, which is the code version of a state-independent area.  If
the reconstructible algebra has a center with minimal projectors
$P_\alpha$, again with exact complementary recovery, the remainder
becomes $\sum_\alpha p_\alpha S(\chi_{A_2,\alpha})$
with $p_\alpha=\tr(P_\alpha\rho)$ the sector
probabilities~\cite{Harlow2017}.  It is the expectation value of a
central operator, affine in $\rho$, and constant along every family of
states with fixed sector probabilities.  Witten's distinction between an
area function and an area operator asks whether a remainder is affine in
this sense~\cite{Witten2026}, and approximate recovery can produce
remainders that are not.  State independent, central affine, and
nonlinear are the three behaviors that a single region can show.  The
question of this paper is a different one.  Even when every region has
a well-defined area operator, the family of regional remainders has to
satisfy relations among regions before it can be the length data of one
geometry, and those relations are what we derive.
\end{revision}

\begin{revision}
The relations are already nontrivial on a time-reflection-symmetric
slice of $\mathrm{AdS}_3$, where the extremal surfaces are geodesics of
the hyperbolic disk and the data are renormalized geodesic lengths, one
for each pair of boundary points.  A symmetric tensor field on the disk
changes all of these lengths together, through its integrals along the
geodesics, so the lengths of a one-parameter family of metrics trace out
a very special curve in the space of interval functions.  We ask whether
the regional proto-areas of a state-dependent code trace out such a
curve.  Reconstruction of the metric from such a curve, by boundary
rigidity or by tensor
tomography~\cite{PorratiRabadan2004,GrahamEtAl2019,Lefeuvre2020,CaoQiSwingleTang2020},
presupposes that the data lie in the image of a metric at all, and that
is what has to be decided first.
\end{revision}

\begin{revision}
\subsection{Reference data and the common-source condition}
\label{subsec:introduction-common-source}
\end{revision}

\begin{revision}
Fix a reference logical state $\rho_0$, and let the reference geometry be
the hyperbolic disk $(M,g_{\mathbb H})$ of unit curvature radius, with a
fixed boundary angular coordinate $\theta$ and the fixed conformal
representative $d\theta^2$ at infinity.  A positive constant
$4G_{\mathrm{eff}}$ converts entropy into length, all lengths being
measured in units of the curvature radius.
Section~\ref{subsec:reference-normalization} explains what fixes this
constant independently of the response under test.  Write
\begin{equation}
 a_\rho(A)=4G_{\mathrm{eff}}S_{\mathrm{PA}}(V,\rho;A),
 \qquad \mathcal A_{\rho_0}(\rho)=a_\rho-a_{\rho_0}.
 \label{eq:introduction-proto-area-map}
\end{equation}
On the geometric side, let $g$ be a simple asymptotically hyperbolic
metric near $g_{\mathbb H}$, where simple means that two distinct ideal
boundary points are joined by exactly one geodesic.  A boundary defining
function $x$ is a function vanishing to first order exactly on the
boundary, normalized so that $x^2g_{\mathbb H}$ restricts to $d\theta^2$
there.  The renormalized length $\mathcal L_g^x(u,v)$ of the geodesic
with ideal endpoints $u,v$ is the finite part of its length once the
divergent piece $-2\log\varepsilon$ at the cutoff $x=\varepsilon$ has
been removed.  The geometric response map is
\begin{equation}
 \mathcal B(g)=\mathcal L_g^x-\mathcal L_{g_{\mathbb H}}^x,
 \label{eq:introduction-boundary-length-map}
\end{equation}
and metrics related by a diffeomorphism that fixes the boundary points
and the boundary representative give the same data.
\end{revision}

\begin{revision}
The two sides of \eqref{eq:introduction-proto-area-map} and
\eqref{eq:introduction-boundary-length-map} have to be compared at
zeroth order before any derivative is taken.  For the round
representative,
\begin{equation}
 \mathcal L_{g_{\mathbb H}}^x(u,v)
 =2\log\!\left[2\left|\sin\frac{u-v}{2}\right|\right],
 \label{eq:introduction-reference-length}
\end{equation}
and an absolute geometric interpretation of the code requires the
reference proto-areas $a_{\rho_0}$ to agree with these lengths in one
fixed renormalization scheme, up to endpoint counterterms that are
state independent and prescribed by the scheme.  Subtracting
$a_{\rho_0}$ makes the reference response vanish whether or not this
holds, so the derivative tests of this paper have two readings.  With a
matched reference they test a geometric interpretation of the code
data.  Without one they test whether the changes in the data can be
generated by perturbing the stated reference geometry.
\end{revision}

\begin{revision}
The common-source condition is the factorization
\begin{equation}
 \mathcal A_{\rho_0}=\mathcal B\circ\mathcal G,
 \qquad \mathcal G(\rho_0)=g_{\mathbb H},
 \label{eq:introduction-factorization}
\end{equation}
through a single metric-valued map $\mathcal G$ on a manifold $\Sigma$
of logical states on which the boundary and recovered states have fixed
supports and are positive definite on them.  We decide
\eqref{eq:introduction-factorization} order by order in derivatives.
The \emph{two-jet} of a function at a point consists of its value and
its first two derivatives there.  After background subtraction the value
is zero, and along a state path $\rho(t)$ we write
\begin{equation}
 \mathcal A_{\rho_0}(\rho(t))
 =t a_1+\tfrac12t^2 a_2+o(t^2),
 \label{eq:introduction-data-two-jet}
\end{equation}
so that $a_1$ and $a_2$ are derivatives, with the factor one half kept
explicit.  The same convention is used for metric and entropy expansions
throughout.
\end{revision}

\begin{revision}
\subsection{Linear sewing and the quadratic condition}
\label{subsec:introduction-main-results}
\end{revision}

\begin{revision}
It is instructive to begin with a weighted graph, where the geometry
consists of edge weights and the regional data are minimum-cut weights.
A \emph{chamber} is an open set of weights on which the minimizing cut
of every monitored region is unchanged.  Label the chosen cuts by
$\sigma$ and let $M_\sigma$ be the matrix whose row for region $A$
indicates the edges in that cut.  Inside the chamber the data depend
linearly on the weights, and a response $q$ comes from edge variations
$\delta w$ exactly when
\begin{equation}
 q=M_\sigma\delta w
 \quad\Longleftrightarrow\quad
 y^{\mathsf T}q=0\quad\hbox{for all }y\in\ker M_\sigma^{\mathsf T}.
 \label{eq:introduction-discrete-criterion}
\end{equation}
A \emph{witness} is a linear combination of regional data that vanishes
on every response the model allows and is nonzero on the data under
test.  Here it is a left null vector $y$ of $M_\sigma$.  The verdict of
a witness is tied to the graph and the chamber, so it excludes a local
realization on that background and says nothing about other bulk
topologies.  Absolute data must in addition be realized by positive
weights satisfying the inequalities that select the chosen cuts, and for
central area operators all of these statements hold coefficientwise in
each central sector.
\end{revision}

\begin{revision}
In the continuum the source is a symmetric rank-two tensor $h$, the
first-order change of the metric, and the role of $M_\sigma$ is played
by the geodesic \emph{X-ray transform}
\begin{equation}
 (I_2h)(u,v)=\int_{\gamma_{u,v}}h(T,T)\,ds,
 \qquad D\mathcal B_{g_{\mathbb H}}[h]=\tfrac12 I_2h,
 \label{eq:introduction-linearized-boundary-map}
\end{equation}
with $s$ the hyperbolic arclength and $T$ the unit tangent of the
reference geodesic $\gamma_{u,v}$.  The factor one half comes from
differentiating the line element.  The perturbation also displaces the
geodesic at first order, but the length of the reference geodesic is
stationary under variations of the curve, so this displacement drops
out of the first variation of the length.  Infinitesimal
diffeomorphisms that fix the boundary lie in the kernel of $I_2$, which
is the gauge ambiguity of the inverse problem.  The range of $I_2$ on
the Poincar\'e disk is characterized by the tensor tomography of
Eptaminitakis, Monard, and
Zou~\cite{EptaminitakisMonardZou2025,EptaminitakisMonardZou2026}.  In a
Hilbert space $\mathcal Y^\tau$ of interval functions with the spectral
decay needed for stable inversion, defined in
Section~\ref{sec:linearized-geometry}, that range is a closed subspace
$\mathcal Y_{\mathrm{geo}}$ and
\begin{equation}
 \mathcal Y^\tau=\mathcal Y_{\mathrm{geo}}
 \widehat\oplus\mathcal Y_{\mathrm{ng}},
 \label{eq:introduction-continuum-splitting}
\end{equation}
with $\mathcal Y_{\mathrm{ng}}$ the orthogonal complement and
$\Pi_{\mathrm{ng}}$ the projection onto it.  The first-order data of a
metric perturbation occupy only the scalar and the rank-two tensor
sectors of this space, while the sectors of tensor order four and higher
are out of reach of any metric at first order.  The parameter $\tau$ weights the tail over
tensor order and does not change the norm inside a fixed sector.  The
first-order condition is $\Pi_{\mathrm{ng}}a_1=0$.  When it holds, the
reconstruction operator $\mathscr R_2$ inverts $\tfrac12I_2$ on its
range in the fixed gauge of Section~\ref{subsec:gauge-reduction}, giving
the first metric coefficient $h_1=\mathscr R_2a_1$.  When it fails, the
normalized residual $\Pi_{\mathrm{ng}}a_1/\|\Pi_{\mathrm{ng}}a_1\|$ is
the optimal witness in the data norm.
\end{revision}

\begin{revision}
The second-order condition has a different origin.  The first-order
displacement of the geodesic, which stationarity removes from the first
variation of the length, enters its second variation.  The normal
displacement $J_{h_1}$ of the geodesic solves the forced Jacobi
equation
\begin{equation}
 \mathcal J_\gamma J_{h_1}=\mathcal F_\gamma(h_1),
 \qquad \mathcal J_\gamma=-\partial_s^2+1,
 \label{eq:introduction-forced-jacobi}
\end{equation}
where $\mathcal F_\gamma(h_1)$ is the normal force exerted on the
geodesic by the change of the Levi-Civita connection, and
$\mathcal J_\gamma$ is the Jacobi operator of a geodesic in curvature
$-1$, which is positive and invertible.  Writing $h_{TT}=h(T,T)$, the
Hessian of the length map is
\begin{equation}
 D^2\mathcal B_{g_{\mathbb H}}[h,k]
 =-\frac14\int_\gamma h_{TT}k_{TT}\,ds
 -\left\langle\mathcal F_\gamma(h),
     \mathcal J_\gamma^{-1}\mathcal F_\gamma(k)\right\rangle_{L^2(ds)}.
 \label{eq:introduction-length-hessian}
\end{equation}
The first term is the second variation of the line element on the fixed
reference curve.  The second is the energy of the displacement, and it
is nonlocal along the geodesic because the Green function of
$\mathcal J_\gamma$ is $\tfrac12e^{-|s-r|}$.  For a metric path
$g(t)=g_{\mathbb H}+t h_1+\tfrac12t^2h_2+o(t^2)$ the second data
derivative is $\tfrac12 I_2h_2+D^2\mathcal B[h_1,h_1]$.  The term in
$h_2$ lies in $\mathcal Y_{\mathrm{geo}}$ and can be adjusted at will,
but it has no normal component.  What remains is the condition
\begin{equation}
 \mathfrak O_2(a_1,a_2)
 :=\Pi_{\mathrm{ng}}\left[a_2-
    D^2\mathcal B_{g_{\mathbb H}}[h_1,h_1]\right]=0,
 \qquad h_1=\mathscr R_2a_1,
 \label{eq:introduction-quadratic-obstruction}
\end{equation}
which defines the quadratic obstruction $\mathfrak O_2$ once the
first-order condition holds.  The complement
$\mathcal Y_{\mathrm{ng}}$ is normal to the \emph{linear} range only.
Second derivatives of geometric data generally do have a component in
$\mathcal Y_{\mathrm{ng}}$, but that component is fixed by $h_1$ through
\eqref{eq:introduction-length-hessian}.  The obstruction is the residual
left after this forced component has been removed.
Table~\ref{tab:introduction-theorem-map} collects the notation.
\end{revision}

\begin{revision}
\begin{table}[t]
\ifshowchanges\color{blue}\else\color{black}\fi
 \centering\small
 \setlength{\tabcolsep}{4pt}\renewcommand{\arraystretch}{1.12}
 \begin{tabular}{p{0.21\textwidth}p{0.71\textwidth}}
 \toprule
 notation & meaning and role \\
 \midrule
 $a_1,a_2$ & first and second derivatives of the background-subtracted,
 length-normalized regional data along a state path \\
 $M_\sigma$, $\ker M_\sigma^{\mathsf T}$ & cut-incidence matrix of a
 graph chamber and its space of witnesses, the left null vectors \\
 $I_2$, $A=\tfrac12I_2$ & geodesic X-ray transform of a symmetric
 rank-two tensor and the first derivative of the length map \\
 $\mathcal Y^\tau$ & Hilbert space of interval data with the spectral
 regularity of Section~\ref{sec:linearized-geometry}, with $\tau$
 weighting the higher tensor sectors \\
 $\mathcal Y_{\mathrm{geo}}$ & range of $A$ in that space, spanned by the
 scalar and rank-two sectors \\
 $\mathcal Y_{\mathrm{ng}},\Pi_{\mathrm{ng}}$ & orthogonal complement of
 the range, the sectors of tensor order four and higher, and the
 projection onto it \\
 $\mathscr R_2$ & inverse of $A$ on its range in the fixed gauge
 of Section~\ref{subsec:gauge-reduction} \\
 $\mathcal J_\gamma,\mathcal F_\gamma$ & Jacobi operator and
 metric-induced normal force on a reference geodesic \\
 $D^2\mathcal B$ & Hessian of the length map, the sum of the direct and
 the bending contributions in \eqref{eq:introduction-length-hessian} \\
 $\mathfrak O_2$ & normal residual of $a_2$ after subtracting the Hessian
 forced by the reconstructed first metric coefficient \\
 $\boldsymbol\Delta(X,Y)$ & regional BKM information discarded by the
 recovery for logical directions $X,Y$ \\
 \bottomrule
 \end{tabular}
 \caption{Notation for the sewing tests.  The first-order range
 determines the admissible velocities, and the normal Hessian determines
 the compatible accelerations.}
 \label{tab:introduction-theorem-map}
\end{table}
\end{revision}

\begin{revision}
Theorems~\ref{thm:linear-geometrizability} and
\ref{thm:two-jet-geometrizability} show that
\begin{equation}
 \Pi_{\mathrm{ng}}a_1=0,\qquad \mathfrak O_2(a_1,a_2)=0
 \label{eq:introduction-two-jet-criterion}
\end{equation}
are necessary and sufficient for a regular pair of coefficients
$(a_1,a_2)$ to be the two-jet of the length data of a metric two-jet.
Both metric coefficients are then reconstructed, and the condition has a
bilinear form for independent state directions.  Strong Taylor
remainders in the data norm need the differentiability stated in those
theorems.  In the language of the factorization
\eqref{eq:introduction-factorization}, the two conditions separate two
failures that are often conflated,
\begin{align}
  \text{state dependence}
  &\not\Rightarrow \text{metric geometrizability},
  \label{eq:introduction-first-nonimplication}\\
  \text{first-order geometrizability}
  &\not\Rightarrow \text{second-order geometrizability}.
  \label{eq:introduction-second-nonimplication}
\end{align}
The first failure is a cross-region sewing obstruction.  The second is
the curvature of the image of $\mathcal B$ inside interval-data space.
The criterion refers to one smooth geodesic branch.  A recovery branch
and an extremal-surface branch are separate choices, and a change of the
dominant Ryu--Takayanagi surface can make the minimal length nonsmooth
even though the geometry is smooth.
Section~\ref{subsec:extremal-branches} discusses this situation.
Table~\ref{tab:introduction-criteria} lists the tests with their
hypotheses and outputs, and Figure~\ref{fig:overview} summarizes the
logic.
\end{revision}

\begin{revision}
\begin{table}[t]
\ifshowchanges\color{blue}\else\color{black}\fi
  \centering
  \small
  \setlength{\tabcolsep}{4pt}
  \renewcommand{\arraystretch}{1.16}
  \begin{tabular}{>{\raggedright\arraybackslash}p{0.19\textwidth}p{0.42\textwidth}p{0.31\textwidth}}
    \toprule
    setting & test and output & hypotheses \\
    \midrule
    finite graph
      & $q\in\operatorname{im}M_\sigma$, equivalently $y^{\mathsf T}q=0$
        on $\ker M_\sigma^{\mathsf T}$, with output the edge
        reconstruction $M_\sigma^{+}q$ or an exact left-null witness
      & fixed graph and selected chamber, and for absolute data a
        matched positive reference inside the chamber \\
    disk, first order
      & $\Pi_{\mathrm{ng}}a_1=0$, with output the gauge-fixed metric
        coefficient $h_1=\mathscr R_2a_1$ or the optimal witness
        $\Pi_{\mathrm{ng}}a_1/\|\Pi_{\mathrm{ng}}a_1\|$
      & hyperbolic disk, data in $\mathcal Y^\tau$, regular
        reconstructed tensor \\
    disk, second order
      & $\Pi_{\mathrm{ng}}a_1=0$ and $\mathfrak O_2(a_1,a_2)=0$, with
        output the reconstruction of $(h_1,h_2)$ or the quadratic
        witness $\mathfrak O_2$
      & regular coefficient two-jet, and pathwise differentiability for
        a strong $\mathcal Y^\tau$ remainder \\
    quantum data
      & BKM--Jacobi matching
        \eqref{eq:introduction-bkm-jacobi-matching} for all logical
        directions
      & fixed recovery and positive definite outputs, and a product
        reference for the Petz form of the loss \\
    \bottomrule
  \end{tabular}
  \caption{Geometrizability criteria.  The finite chamber test and the
  regular coefficientwise two-jet test are necessary and sufficient in
  their stated classes, and each failure comes with an explicit
  witness.}
  \label{tab:introduction-criteria}
\end{table}
\end{revision}

\begin{revision}
\begin{figure}[t]
\ifshowchanges\color{blue}\else\color{black}\fi
\centering
\ifshowchanges\colorlet{revisionink}{blue}\else\colorlet{revisionink}{black}\fi
\resizebox{\textwidth}{!}{%
\begin{tikzpicture}[font=\small,>=Latex,x=1cm,y=1cm,draw=revisionink,text=revisionink]
  \node[draw,rounded corners=2pt,minimum width=3.25cm,minimum height=2.35cm]
    (states) at (1.7,1.65) {};
  \node at (1.7,2.52) {logical states $\Sigma$};
  \node at (1.7,0.78) {fixed recovery};
  \fill (0.95,1.25) circle (1.7pt) node[left=1pt] {$\rho_0$};
  \fill (2.45,1.95) circle (1.7pt) node[right=1pt] {$\rho$};
  \draw[->,line width=0.8pt] (1.02,1.30).. controls (1.45,1.85) and (1.95,1.50)..(2.38,1.90);

  \draw[->,line width=0.9pt] (3.42,1.65)--(5.03,1.65)
    node[midway,above] {$\mathcal A_{\rho_0}$};

  \node[draw,rounded corners=2pt,minimum width=5.55cm,minimum height=3.45cm]
    (data) at (7.95,1.65) {};
  \node[anchor=north west] at (5.34,3.20) {interval-data space $\mathcal Y^\tau$};
  \draw[dashed] (5.68,0.82)--(10.18,0.82)
    node[below right] {$\mathcal Y_{\mathrm{geo}}$};
  \draw[line width=1pt]
    plot[smooth] coordinates {(5.70,0.82) (6.65,0.91) (7.62,1.18) (8.60,1.65) (9.60,2.37) (10.18,2.96)};
  \node[anchor=west] at (9.12,2.77) {$\operatorname{Im}\mathcal B$};
  \fill (6.05,0.85) circle (1.7pt);
  \draw[->,line width=0.8pt] (6.05,0.85).. controls (6.75,0.92) and (7.30,1.05)..(7.82,1.27);
  \fill (7.82,1.27) circle (1.7pt);
  \draw[->,dotted,line width=0.9pt] (6.05,0.85).. controls (6.95,0.85) and (7.80,0.84)..(8.60,0.84);
  \draw[->,line width=0.8pt] (8.60,0.84)--(8.60,2.45);
  \fill (8.60,2.52) circle (1.7pt);
  \node[align=right,anchor=east] at (8.45,1.98) {nongeometric\\normal component};

  \draw[->,line width=0.9pt] (10.78,2.30)--(12.05,2.78);
  \node[draw,rounded corners=2pt,minimum width=3.15cm,minimum height=1.15cm,align=center]
    at (13.60,2.82) {metric coefficients\\$h_1,h_2$ modulo gauge};
  \draw[->,line width=0.9pt] (10.78,1.05)--(12.05,0.57);
  \node[draw,rounded corners=2pt,minimum width=3.15cm,minimum height=1.15cm,align=center]
    at (13.60,0.53) {witness\\$\Pi_{\mathrm{ng}}a_1$ or $\mathfrak O_2$};
\end{tikzpicture}}
\caption{The common-source condition.  Proto-area data define a path in
interval-data space.  Metrics near $g_{\mathbb H}$ generate the curved
image of $\mathcal B$, whose tangent space at the background is the
first-order range $\mathcal Y_{\mathrm{geo}}$.  Data on the image
reconstruct the metric coefficients modulo gauge, and an incompatible
normal component gives a boundary-accessible witness.  The curved locus
is schematic and stands for the regular two-jet image at the hyperbolic
background.}
\label{fig:overview}
\end{figure}
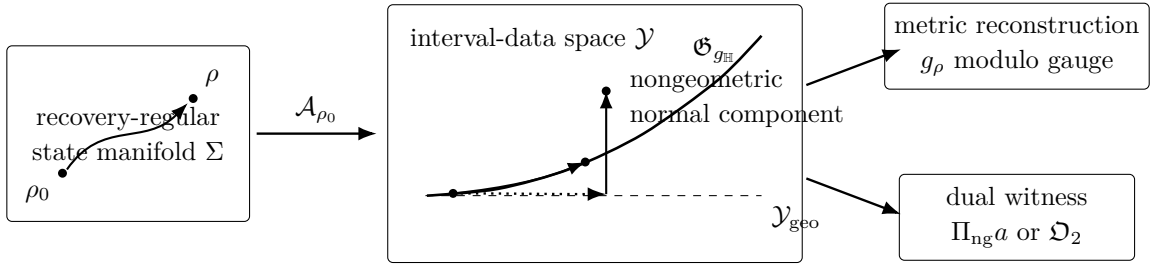
\end{revision}

\begin{revision}
\subsection{Information loss and geometric response}
\label{subsec:introduction-relation}
\end{revision}

\begin{revision}
On the quantum side the second derivative has an equally definite form.
Along an affine path of logical states the entropy Hessian is minus the
Bogoliubov--Kubo--Mori (BKM) metric
$\mathfrak g_\sigma(U,W)=\tr[U\,D\log_\sigma(W)]$ of the state
variation, where $D\log_\sigma$ is the derivative of the matrix
logarithm at a faithful, that is positive definite, state $\sigma$.  For
logical tangent directions $X,Y$ let $\boldsymbol\Delta(X,Y)$ be the
interval function whose value on $A$ is the BKM form retained by the
boundary output minus the BKM form retained by the recovered matter.
The proto-area Hessian is $-4G_{\mathrm{eff}}\boldsymbol\Delta(X,Y)$.
When the recovered matter is obtained from the decoded boundary output
by a quantum channel, monotonicity of the BKM metric makes
$\Delta_A(X,X)$ nonnegative, so it is an information loss, and at a
product reference with partial-trace decoding the Petz recovery map
$Z\mapsto Z\otimes\chi$ decomposes this loss orthogonally as a squared
norm.  The decomposition also yields a sharp quantitative statement,
Theorem~\ref{thm:section2-sharp-sufficiency-angle}.  On a
finite-dimensional space of logical directions, the information
discarded by the recovery is at most $(1-\alpha)/\alpha$ times the
information retained, where $\alpha>0$ is the minimum fraction of the
BKM norm that the recovery retains on that space, and no bound of this kind
holds once directions invisible to the recovery are admitted.  This
separates the ensemble-typical suppression found by Witten for random
deformations of the encoding from the worst-case statement that a fixed
code requires.  Theorem~\ref{thm:bkm-jacobi-matching} combines the two
Hessians into the condition
\begin{equation}
 \Pi_{\mathrm{ng}}\left[-4G_{\mathrm{eff}}\boldsymbol\Delta(X,Y)
   -D^2\mathcal B_{g_{\mathbb H}}[h_X,h_Y]\right]=0,
 \label{eq:introduction-bkm-jacobi-matching}
\end{equation}
where $h_X,h_Y$ are the reconstructed first metric responses.  The
information discarded by recovery has to be distributed over the
intervals in exactly the pattern of the normal acceleration forced by one
metric.  This BKM--Jacobi matching is a signed identity across regions,
stronger than an inequality for each region.  Because it compares a term
linear in $4G_{\mathrm{eff}}$ with one quadratic in it, the
entropy-to-length normalization enters the test physically.
\end{revision}

\begin{revision}
The second-order condition has content independent of any inverse
problem.  Section~\ref{subsec:explicit-normal-hessian} treats the family
of constant-curvature metrics obtained by pulling $g_{\mathbb H}$ back
under a diffeomorphism that acts on the boundary as the reparametrization
$\theta\mapsto\theta+t\sin2\theta$.  Its length data are known in closed
form, their first derivative lies in the rank-two sector, and their
second derivative carries a nonzero rank-four component that is fixed by
the first.  A data path with the same first derivative and vanishing
second derivative therefore passes every linear test and fails
\eqref{eq:introduction-quadratic-obstruction}.
\end{revision}

\begin{revision}
Two finite code constructions show the corresponding failures in
information-theoretic terms.  In the first, a classical center controls
bond rotations in a tensor network.  Every region then has an exact
central area operator, and a controlled coupling between bonds in
different cells makes these operators violate a cut-sewing relation.  In
the second, a single global encoding with a copied center is probed
along an affine path with fixed central probabilities.  Its reference
proto-areas coincide with the minimum cuts of a positive-weight
four-terminal star for every bipartition, every first derivative
vanishes, and the second derivative violates a left-null relation of
that star.  Reference matching and linear sewing therefore do not
exhaust the geometric conditions, even for a fixed encoding and a fixed
recovery.  The passage from finite codes to the continuum is a separate
matter.  Under source density, convergence of the sampled data, and
consistency of the first and second derivatives of the discrete length
map, the finite projections, minimum-norm reconstructions, and corrected
quadratic witnesses converge to their continuum counterparts
(Theorems~\ref{thm:finite-code-continuum-limit}
and~\ref{thm:second-order-discrete-continuum}).  A fixed graph has no
continuum interpretation unless these hypotheses are imposed.
\end{revision}

\begin{revision}
The question addressed here comes before reconstruction.  Hole-ography
and kinematic space recover bulk curves and distances from interval
data~\cite{Balasubramanian:2013lsa,CzechLamprou2014,Czech:2015qta,CzechEtAl2016},
boundary rigidity and tensor tomography identify a metric from geometric
data up to gauge~\cite{PorratiRabadan2004,GrahamEtAl2019,Lefeuvre2020},
and Cao, Qi, Swingle, and Tang use the tensor Radon transform to test
the linear geometricity of near-AdS entanglement
data~\cite{CaoQiSwingleTang2020}.  Entanglement and emergent locality
motivate tensor-network and continuum descriptions of this
kind~\cite{VanRaamsdonk2010,Swingle2012}.  We apply the range question
to recovery-defined remainders and extend the membership test to second
derivatives, where the first metric response fixes an additional normal
component.  The extremal-surface variations of
Refs.~\cite{GhoshMishra2016,GhoshMishra2018,Mosk2018,BaoCaoFischettiKeeler2019,JokelaLiimatainenSarkkinenTzou2025}
supply the geometric ingredients of the second variation, and projecting
that variation onto the complement of the first-order range turns it
into a membership condition for arbitrary regular data.  Graph models
and holographic entropy cones provide the finite
setting~\cite{Headrick2014,BaoEtAl2015,HaydenEtAl2016}.  The witnesses
used here are covectors on data space that test a fixed locality
structure, a different question from the flow duality of bit
threads~\cite{FreedmanHeadrick2017,HeadrickHubeny2018}.  Modular Berry
transport sews relative edge-mode
frames~\cite{CzechLamprouMcCandlishSully2018,CzechDeBoerGeLamprou2019},
whereas the conditions derived here sew scalar remainders into one
metric response.
\end{revision}

\begin{revision}
Everything below is kinematic.  No Einstein equation is imposed, and
quantum extremality and bulk dynamics add further conditions, discussed
in Section~\ref{sec:physical-interpretation}.  The continuum results are
local to a time-reflection-symmetric disk with fixed conformal boundary
data and one simple single-interval geodesic branch.  Covariant HRT
data, higher dimensions, and competing branches require additional
range and regularity theory.
\end{revision}

\begin{revision}
Section~\ref{sec:proto-area} sets up the calibrated entropy data, the
reference matching, and their BKM derivatives.
Section~\ref{sec:discrete} derives the finite sewing conditions on
graphs.  Sections~\ref{sec:linearized-geometry}
and~\ref{sec:nonlinear-consistency} derive the linear range condition and
the Jacobi correction, including the explicit hyperbolic family.
Section~\ref{sec:skewed-codes-revised} constructs the codes, and
Section~\ref{sec:physical-interpretation} discusses extremal branches and
what a dynamical interpretation would require beyond kinematics.  The
appendices contain the recovery regularity theory, the
discrete-to-continuum estimates, the analytic properties of the length
map, and the finite-dimensional code calculations.
\end{revision}

\section{Proto-area data and the geometrizability problem}
\label{sec:proto-area}

\begin{revision}
The proto-area is defined through a recovery map, and the first question
is which recovery.  A recovery reoptimized for every logical state would
introduce derivatives of the recovery itself into the remainder and
could change branch nonsmoothly where the preferred recovery switches.
We therefore calibrate the recovery once, on the code, and hold it
fixed while the state varies, so that the derivatives below depend only
on the state.  This section specifies the calibration, identifies the
state sectors on which the remainder is differentiable, and derives its
two-jet.  The Hessian turns out to be an information-loss form, and at
a product reference the Petz recovery map decomposes this loss
orthogonally and gives it a sharp constant.  The zeroth-order question
of reference matching is taken up in Section~\ref{subsec:reference-normalization},
before the derivatives, and the section ends with the hierarchy of
geometrizability levels that the later sections test.  The optimizer
and entropy regularity estimates are collected in
Appendix~\ref{app:recovery-entropy}.
\end{revision}

\subsection{Optimal recovery and the proto-area function}
\label{subsec:proto-area-optimal-recovery}

\begin{revision}
Let $\mathcal H_L$ be a finite-dimensional logical Hilbert space and let
\begin{equation}
  V:\mathcal H_L\longrightarrow\mathcal H_{\partial}
  \label{eq:section2-encoding-isometry}
\end{equation}
be an isometric encoding. 
For a boundary bipartition
$\mathcal H_{\partial}=\mathcal H_A\otimes\mathcal H_{\bar A}$,
first consider the subsystem setting with a logical factorization
$\mathcal H_L=\mathcal H_{a_A}\otimes\mathcal H_{\bar a_A}$.
The case of a recoverable algebra with a center uses the direct-sum
convention specified below.  The encoded state and its $A$ marginal are
\end{revision}

\begin{revision}
\begin{equation}
  \omega(\rho)=V\rho V^\dagger,
  \qquad
  \omega_A(\rho)=\tr_{\bar A}\omega(\rho).
  \label{eq:section2-encoded-state}
\end{equation}
\end{revision}

\begin{revision}
An admissible recovery pair $\mathbf R=(R_A,R_{\bar A})$ consists of local
Stinespring isometries, that is, isometries into an enlarged output space
of which part is subsequently discarded, or equivalently local unitaries
after adjoining fixed pure ancillas.
All ancilla and output dimensions are fixed before the calibration is
optimized.  Each local isometry is represented by a rectangular matrix
$R$ satisfying $R^\dagger R=I$, and the space of such matrices is a complex
Stiefel manifold.  Their compact product is the admissible parameter
space.  The recoveries identify output factorizations
\end{revision}

\begin{revision}
\begin{equation}
 \begin{aligned}
 R_A&:\mathcal H_A\longrightarrow\mathcal H_{A_1}\otimes\mathcal H_{A_2},\\
 R_{\bar A}&:\mathcal H_{\bar A}\longrightarrow
 \mathcal H_{\bar A_1}\otimes\mathcal H_{\bar A_2},
 \end{aligned}
 \label{eq:section2-recovery-factorizations}
\end{equation}
\end{revision}

\begin{revision}
In the subsystem setting,
$\mathcal H_{A_1}\simeq\mathcal H_{a_A}$ and
$\mathcal H_{\bar A_1}\simeq\mathcal H_{\bar a_A}$.
For a recoverable algebra with a center, the retained regional outputs
instead represent finite direct-sum algebras by orthogonal sector records
and the assigned within-sector matter factors.  Both complementary outputs
may retain the same classical record.  The subsystem tensor identification
is not imposed on these algebraic outputs.  Recoveries preserve the minimal
central projectors, and their fixed-branch certificates concern the stated
algebra.  The recovered channel is
\end{revision}

\begin{revision}
\begin{equation}
\begin{split}
  \mathcal Q_{A,\mathbf R}(\rho)
  =\tr_{A_2\bar A_2}\Bigl[
  &(R_A\otimes R_{\bar A})V\rho V^\dagger
  (R_A\otimes R_{\bar A})^\dagger
  \Bigr].
  \label{eq:section2-recovered-channel}
\end{split}
\end{equation}
\end{revision}

\begin{revision}
Let $d_L=\dim\mathcal H_L$ and
\begin{equation}
  |\Phi\rangle_{LR}
  =\frac{1}{\sqrt{d_L}}
  \sum_{j=1}^{d_L}|j\rangle_L|j\rangle_R,
  \qquad
  \Phi_{LR}=|\Phi\rangle\langle\Phi|.
  \label{eq:section2-maximally-entangled-calibration}
\end{equation}
The channel coherent information calibrated on the maximally mixed logical input is~\cite{SchumacherNielsen1996}
\begin{equation}
  I_c(\mathcal Q_{A,\mathbf R})
  =S\!\left(\mathcal Q_{A,\mathbf R}(\pi_L)\right)
  -S\!\left[
    (\mathcal Q_{A,\mathbf R}\otimes\Id_R)(\Phi_{LR})
  \right],
  \qquad
  \pi_L=\frac{I_L}{d_L}.
  \label{eq:section2-coherent-information}
\end{equation}
We use the joint two-sided channel calibration of Section~3 of Ref.~\cite{CaoEtAl2026}, maximizing this coherent information. Its recovered-matter entropies and its proto-area entropy are Definitions~3.1 and~3.2 of that paper, in the numbering of arXiv:2603.13475v2.  The optimal recoveries form the compact nonempty set
\begin{equation}
  \mathsf{Opt}_A(V)
  =\operatorname*{arg\,max}_{\mathbf R}
  I_c(\mathcal Q_{A,\mathbf R}).
  \label{eq:section2-optimal-recovery-set}
\end{equation}
The maximization depends on the code and the bipartition, not on the
logical state subsequently evaluated.  A state-dependent recalibration
would introduce derivatives of the recovery itself and could develop
singularities when the preferred recovery changes.
\end{revision}

For $\mathbf R\in\mathsf{Opt}_A(V)$, define
\begin{equation}
  \tau_{A,\mathbf R}(\rho)
  =\tr_{\bar A_1}\mathcal Q_{A,\mathbf R}(\rho).
  \label{eq:section2-recovered-matter-state}
\end{equation}
Following the proto-area prescription of Cao, Cheng, Karthikeyan, Li, and
Preskill~\cite{CaoEtAl2026}, but retaining the possible nonuniqueness of an
optimal recovery, the operational prescription is the set
\begin{equation}
  \mathfrak P_A(V,\rho)
  =\left\{
  S\!\left(\omega_A(\rho)\right)
  -S\!\left(\tau_{A,\mathbf R}(\rho)\right)
  :\mathbf R\in\mathsf{Opt}_A(V)
  \right\}.
  \label{eq:section2-set-valued-proto-area}
\end{equation}

\begin{definition}[Recovery-regular state sector]
\label{def:section2-recovery-regular}
Let $\Sigma\subset\mathcal D(\mathcal H_L)$ be a connected smooth manifold of logical states and let $\calI$ be a finite family of boundary regions. The pair $(V,\Sigma)$ is recovery-regular on $\calI$ if, for every $A\in\calI$, the set in Equation~\eqref{eq:section2-set-valued-proto-area} is a singleton for every $\rho\in\Sigma$, and a fixed $\mathbf R_A^\star\in\mathsf{Opt}_A(V)$ can be chosen such that $\omega_A(\rho)$ and $\tau_{A,\mathbf R_A^\star}(\rho)$ have constant support on $\Sigma$. The sector is faithfully recovery-regular if both states are positive definite on those supports.
\end{definition}

The singleton condition identifies recoveries that give the same recovered-entropy function\rev{, and it} does not require a unique Stinespring representation.

\begin{revision}
\begin{definition}[Fixed calibrated recovery branch]
\label{def:section2-certified-recovery-branch}
A fixed calibrated recovery branch is a state-independent choice
$\mathbf R_A^\sharp$ for every $A\in\calI$, together with a stated
operational certificate for that choice, such as an optimality gap or a
bound on the diamond-norm distance between the recovered channel and the
ideal one, the diamond norm being the worst-case distinguishability of
two channels over all inputs, including entangled reference systems. It
is faithfully
branch-regular on $\Sigma$ if $\omega_A(\rho)$ and
$\tau_{A,\mathbf R_A^\sharp}(\rho)$ have constant supports and are
positive definite on those supports throughout $\Sigma$.  The corresponding branch proto-area is
\begin{equation}
 S_{\mathrm{PA}}^\sharp(V,\rho;A)
 =S(\omega_A(\rho))-S(\tau_{A,\mathbf R_A^\sharp}(\rho)).
 \label{eq:section2-branch-proto-area}
\end{equation}
\end{definition}
\end{revision}

The prescription of Ref.~\cite{CaoEtAl2026} is recovered when
$\mathbf R_A^\sharp\in\mathsf{Opt}_A(V)$ and the set in
Equation~\eqref{eq:section2-set-valued-proto-area} is a singleton.  A
diamond-minimax branch is used in the fixed-fiber example of
Section~\ref{subsec:fixed-center-erasure-code}.  Entropy differentiation
and every geometrizability criterion below are branchwise statements
and require only the fixed linear output channels and support
regularity.  Optimality with respect to a different recovery objective
is not inferred from those calculations.

In formulas common to both cases, write $\mathbf R_A^{\mathrm{cal}}$ for
the selected branch and
$\tau_A^{\mathrm{cal}}:=\tau_{A,\mathbf R_A^{\mathrm{cal}}}$.

\begin{proposition}[Unitary equivalence of optimal recovered marginals]
\label{prop:section2-unitary-equivalent-recoveries}
Suppose that for every $\mathbf R\in\mathsf{Opt}_A(V)$ there is a unitary $U_{A,\mathbf R}$ on $A_1$, independent of $\rho$, such that
\begin{equation}
  \tau_{A,\mathbf R}(\rho)
  =U_{A,\mathbf R}
  \tau_{A,\mathbf R_A^\star}(\rho)
  U_{A,\mathbf R}^\dagger
  \label{eq:section2-unitary-equivalent-recovered-states}
\end{equation}
for every $\rho$ in the affine span of $\Sigma$. Then the entropy in Equation~\eqref{eq:section2-set-valued-proto-area} is independent of the optimal recovery. If, in addition, $\omega_A(\rho)$ and the reference recovered branch have constant support on $\Sigma$, then recovery regularity holds.
\end{proposition}

\begin{revision}
\begin{proposition}[Persistence of a nondegenerate optimal recovery]
\label{prop:section2-optimal-recovery-persistence}
Let $s\mapsto V_s$ be a $C^k$ family of encodings, $k\geq2$,
and let $\mathfrak R_A$ be the compact parameter space of admissible
recovery pairs.  Quotient by the output-unitary action
\begin{equation}
 R_A\longmapsto(U_{A_1}\otimes U_{A_2})R_A,
 \qquad
 R_{\bar A}\longmapsto
 (U_{\bar A_1}\otimes U_{\bar A_2})R_{\bar A},
 \label{eq:recovery-full-output-gauge}
\end{equation}
with the unitaries restricted to preserve the assigned sector algebras
when these are present.  Assume that a saturated neighborhood of the
maximizing orbit $[\mathbf R_0]$ has a smooth quotient and a smooth
local section.  Suppose that the states entering
Equation~\eqref{eq:section2-coherent-information} have constant rank
there, that $[\mathbf R_0]$ is the unique global maximizing orbit at
$s=0$, and that the Hessian of the induced objective on the quotient
is negative definite at $[\mathbf R_0]$.  Then there are
$\varepsilon>0$ and a unique $C^{k-1}$ maximizing branch modulo this
output-unitary action,
\begin{equation}
 (-\varepsilon,\varepsilon)\ni s
 \longmapsto[\mathbf R_A^\star(s)].
 \label{eq:section2-smooth-optimal-recovery-branch}
\end{equation}
A local representative can be chosen $C^{k-1}$.  If the boundary and
recovered marginals have constant rank near $(0,\Sigma)$, their
proto-area is $C^{k-1}$ in $(s,\rho)$.  Retained-output unitaries
conjugate the recovered marginal and hence leave its entropy unchanged,
while discarded-output unitaries leave the marginal itself unchanged.
\end{proposition}
\end{revision}

The proof combines the implicit function theorem with a uniform gap excluding remote maxima.  The resulting branch is calibrated once and held fixed under logical-state differentiation.  If the code itself varies, the envelope theorem applies to the optimized coherent information, not automatically to the distinct proto-area function\rev{, which} generally retains an optimizer-velocity term.

On a recovery-regular sector or a fixed calibrated branch, the
entropy-valued proto-area is
\begin{equation}
  S_{\mathrm{PA}}(V,\rho;A)
  =S\!\left(\omega_A(\rho)\right)
  -S\!\left(\tau_A^{\mathrm{cal}}(\rho)\right).
  \label{eq:section2-proto-area-entropy}
\end{equation}
Natural logarithms and $\hbar=1$ are used throughout. A fixed factor $4G_{\mathrm{eff}}>0$ puts the quantity in length units,
\begin{equation}
  a_\rho(A)=4G_{\mathrm{eff}}S_{\mathrm{PA}}(V,\rho;A).
  \label{eq:section2-length-normalized-proto-area}
\end{equation}
For a reference state $\rho_0\in\Sigma$, define
\begin{equation}
\begin{split}
  \mathcal A_{\rho_0}:\Sigma&\longrightarrow\mathbb R^{\calI},\\
  \mathcal A_{\rho_0}(\rho)(A)
  &=a_\rho(A)-a_{\rho_0}(A).
  \label{eq:section2-background-subtracted-map}
\end{split}
\end{equation}
The code, recovery calibration, region assignment, and conversion factor are held fixed.

Exact complementary recovery has the standard subsystem normal form underlying the code derivation of the RT/FLM formula~\cite{FaulknerLewkowyczMaldacena2013,AlmheiriDongHarlow2015,Harlow2017,CaoEtAl2026}. In the normalization above, this form immediately fixes the proto-area remainder.

\begin{proposition}[State independence under exact complementary recovery]
\label{prop:section2-exact-code-state-independence}
Suppose that an optimal recovery pair and a fixed pure state $\chi_{A_2\bar A_2}$ satisfy
\begin{equation}
  (R_A\otimes R_{\bar A})V\rho V^\dagger
  (R_A\otimes R_{\bar A})^\dagger
  =\rho_{A_1\bar A_1}\otimes\chi_{A_2\bar A_2}
  \label{eq:section2-exact-recovery-factorization}
\end{equation}
for every logical state, with $\rho_{A_1}=\rho_{a_A}$. Then
\begin{equation}
  S_{\mathrm{PA}}(V,\rho;A)=S(\chi_{A_2})
  \label{eq:section2-exact-code-proto-area}
\end{equation}
for every $\rho$, and every derivative of the background-subtracted map vanishes.
\end{proposition}

The proof, together with the corresponding fixed-support differentiation
statements, is given in Appendix~\ref{appA:proto-area-differentials}.

The proposition applies within a single sector. For an algebra with a center, Harlow's finite-dimensional algebraic RT decomposition~\cite{Harlow2017}, built on operator-algebra quantum error correction~\cite{sec6BenyKempfKribs2007}, gives the following central affine benchmark.

\begin{proposition}[Central area terms in exact operator-algebra codes]
\label{prop:section2-central-area-operator}
Let the exact reconstructible algebra for the wedge $a$ have the finite-dimensional form
\begin{equation}
  \mathcal M_a
  =\bigoplus_\alpha
  \left(\mathcal B(\mathcal H_{a_\alpha})\otimes I_{\bar a_\alpha}\right),
  \qquad
  \mathcal H_L
  =\bigoplus_\alpha
  \left(\mathcal H_{a_\alpha}\otimes\mathcal H_{\bar a_\alpha}\right),
  \label{eq:section2-oa-algebra-decomposition}
\end{equation}
with minimal central projectors $P_\alpha$. Suppose exact complementary recovery gives, for region $A$,
\begin{equation}
  U_A\omega_A(\rho)U_A^\dagger
  =\bigoplus_\alpha
   p_\alpha(\rho)\,
   \rho_{a_\alpha}(\rho)\otimes\chi_{A_2,\alpha},
  \qquad
  p_\alpha(\rho)=\tr(P_\alpha\rho),
  \label{eq:section2-central-block-state}
\end{equation}
where $\rho_{a_\alpha}$ is normalized when $p_\alpha>0$ and the states $\chi_{A_2,\alpha}$ are independent of $\rho$ within each central block. Let
\begin{equation}
  S(\rho;\mathcal M_a)
  =H(p)+\sum_\alpha p_\alpha S(\rho_{a_\alpha}),
  \qquad
  H(p)=-\sum_\alpha p_\alpha\log p_\alpha,
  \label{eq:section2-algebraic-bulk-entropy}
\end{equation}
be the entropy of the state restricted to $\mathcal M_a$. Then
\begin{equation}
  S(\omega_A(\rho))-S(\rho;\mathcal M_a)
  =\sum_\alpha p_\alpha(\rho)\ell_{A,\alpha},
  \qquad
  \ell_{A,\alpha}=S(\chi_{A_2,\alpha}).
  \label{eq:section2-central-area-affine}
\end{equation}
Equivalently, the area contribution is the expectation of the central operator
\begin{equation}
  L_A=\bigoplus_\alpha \ell_{A,\alpha}P_\alpha
  \in Z(\mathcal M_a),
  \qquad
  \tr(\rho L_A)=\sum_\alpha p_\alpha(\rho)\ell_{A,\alpha}.
  \label{eq:section2-central-area-observable}
\end{equation}
It is affine as a function of the code state and is constant on every submanifold with fixed central probabilities $p_\alpha$.
\end{proposition}

\begin{proof}
The entropy of a block-diagonal density matrix is the Shannon entropy of the block probabilities plus the probability-weighted entropies of the normalized blocks. Applying this to Equation~\eqref{eq:section2-central-block-state} gives
\begin{equation}
  S(\omega_A(\rho))
  =H(p)+\sum_\alpha p_\alpha S(\rho_{a_\alpha})
   +\sum_\alpha p_\alpha S(\chi_{A_2,\alpha}).
\end{equation}
Subtracting Equation~\eqref{eq:section2-algebraic-bulk-entropy} gives Equation~\eqref{eq:section2-central-area-affine}. The operator in Equation~\eqref{eq:section2-central-area-observable} is central because it is scalar on each central block, and its expectation is the same affine function. The reconstructible wedge algebra is block diagonal with respect to the projectors $P_\alpha$, so its unitaries preserve the central probabilities. If two sectors have the same value of $\ell_{A,\alpha}$, commutation with $L_A$ alone is weaker than preservation of the individual central projectors\rev{, and} the central projectors are the invariant data used here.
\end{proof}

On the block-diagonal algebraic state space, this benchmark fits the
operational branch formalism by retaining the sector record and the
$a_\alpha$ subsystem in the recovered output.  Its recovered entropy is
$H(p)+\sum_\alpha p_\alpha S(\rho_{a_\alpha})$, so the operational
entropy difference equals
Equation~\eqref{eq:section2-central-area-affine}.  More generally, a
fixed center-preserving approximate branch sends block-diagonal states
to
\begin{revision}
\begin{equation}
 \tau_A(\rho)=\bigoplus_\alpha
 p_\alpha\tau_{A,\alpha}(\rho_\alpha),
 \qquad
 \rho_\alpha=\frac{P_\alpha\rho P_\alpha}{p_\alpha},
 \label{eq:section2-approximate-algebraic-output}
\end{equation}
\end{revision}%
\rev{Here $\tau_{A,\alpha}$ is a fixed channel applied to the
normalized logical block $\rho_\alpha$, and in the approximate case
this channel need not factor through the marginal on $a_\alpha$.  The
ordinary von Neumann entropy of $\tau_A(\rho)$ is the recovered
algebraic entropy
$H(p)+\sum_\alpha p_\alpha S(\tau_{A,\alpha}(\rho_\alpha))$.}  The entropy-jet
theorems therefore apply blockwise on every fixed central fiber.  This
is the convention used in
Section~\ref{subsec:fixed-center-erasure-code}.  This blockwise statement does
not extend by itself to coherent superpositions of distinct central blocks,
which require a separate correctable-algebra analysis.

Exact recovery therefore supplies the central affine benchmark for the finite constructions in Section~\ref{sec:skewed-codes-revised}. A nontrivial cross-region condition remains even when every $L_A$ exists. The regional operators must arise from one family of local edge operators, or in the continuum from one operator-valued metric perturbation. The controlled-sector model violates the finite-chamber condition. Corollary~\ref{cor:continuum-central-sewing} gives the coefficientwise continuum criterion, and a continuum limit of that code additionally requires the transfer hypotheses of Appendix~\ref{app:finite-geometry}. The separate fixed-fiber model has a \rev{nonzero} noncentral Hessian and therefore requires a controlled departure from exact complementary recovery. In continuum QFT, Ref.~\cite{Witten2026} raises the additional obstruction to defining a central codimension-two area operator without a cutoff.

\begin{proposition}[Area functions and area operators]
\label{prop:section2-area-function-operator}
Let $U$ be a convex open subset of the affine space of density matrices and let $F:U\to\mathbb R$ be $C^2$. There are a Hermitian operator $O$ and a constant $c$ such that
\begin{equation}
  F(\rho)=\tr(O\rho)+c
  \label{eq:section2-area-operator-representation}
\end{equation}
throughout $U$ if and only if $F$ is affine. A nonzero Hessian therefore excludes a state-independent operator representation.
\end{proposition}

Proposition~\ref{prop:section2-area-function-operator} gives a local differential form of the area-function/area-operator distinction emphasized in Ref.~\cite{Witten2026}. A central area operator has an affine expectation value and is constant on a fixed central fiber. A nonzero fixed-fiber Hessian therefore detects a noncentral area function. This is logically independent of the operator-valued sewing test\rev{, since} a family of affine regional functions may still fail to arise from one local metric. No assumption on the Newton-constant scaling of a recovery error enters either geometrizability criterion.

On a disk, complementary boundary arcs have the same unordered endpoints and hence the same bulk geodesic. Introduce
\begin{equation}
  (\mathsf C f)(A)=f(\bar A),
  \qquad
  f_\pm=\frac12(f\pm\mathsf C f).
  \label{eq:section2-complement-decomposition}
\end{equation}

\begin{proposition}[Complement parity]
\label{prop:section2-complement-parity}
Every interval-length data set produced by a single-valued static disk geometry satisfies $f_-=0$. Hence
\begin{equation}
  \mathcal A_{\rho_0,-}(\rho)(A)
  =\frac12\left[
  \mathcal A_{\rho_0}(\rho)(A)
  -\mathcal A_{\rho_0}(\rho)(\bar A)
  \right]
  \label{eq:section2-complement-odd-obstruction}
\end{equation}
is an immediate nongeometric obstruction whenever it is nonzero.
\end{proposition}

\begin{proof}
The geodesic anchored at the endpoints of $A$ is the same unoriented curve as the geodesic anchored at the endpoints of $\bar A$. Its renormalized length is independent of the chosen boundary arc, and the statement survives background subtraction.
\end{proof}

The recovery optimization need not enforce this parity. All tomographic tests below are therefore applied to the complement-even sector, while Equation~\eqref{eq:section2-complement-odd-obstruction} is retained as a separate obstruction.

\begin{revision}
\subsection{Reference compatibility and the entropy-to-length normalization}
\label{subsec:reference-normalization}
\end{revision}

\begin{revision}
Background subtraction isolates the state-dependent response, but it
says nothing about whether the reference proto-areas themselves describe
the chosen geometry.  These are two different questions.  A matched
reference means that the unsubtracted, renormalized proto-areas at
$\rho_0$ equal the boundary lengths of one reference metric.  A response
test at a prescribed reference asks whether the centered family can be
generated by perturbing that metric, without asserting the zeroth-order
equality.  The first- and second-order criteria of
Sections~\ref{sec:linearized-geometry} and~\ref{sec:nonlinear-consistency}
answer the second question.  Their interpretation as a geometry of the
original code requires the first.
\end{revision}

\begin{revision}
It is useful to restore the AdS radius $\ell$ for a moment.  In three
bulk dimensions $G_{\mathrm{eff}}$ has units of length, and the
dimensionless entropy-to-length factor is
\begin{equation}
 \kappa=\frac{4G_{\mathrm{eff}}}{\ell}.
 \label{eq:dimensionless-area-normalization}
\end{equation}
Elsewhere we set $\ell=1$, so that $a_\rho$ is the length in these units
and $\kappa=4G_{\mathrm{eff}}$, with physical lengths obtained by
multiplying by $\ell$.  The radius, the boundary labels, the boundary
metric representative, and $\kappa$ are all fixed before any state
derivative is taken.
\end{revision}

\begin{revision}
For the round hyperbolic reference,
\begin{equation}
 g_{\mathbb H}=\frac{4|dz|^2}{(1-|z|^2)^2},
 \qquad x=\frac{1-|z|^2}{1+|z|^2},
 \label{eq:zero-order-hyperbolic-reference}
\end{equation}
the defining function $x$ induces $d\theta^2$ at conformal infinity, the
boundary circle at $|z|=1$ together with its conformal class of metrics.
For ideal endpoints $u\ne v$ on this circle, the renormalized length is
\begin{equation}
 L_0(u,v):=\mathcal L_{g_{\mathbb H}}^x(u,v)
 =2\log\!\left(2\left|\sin\frac{u-v}{2}\right|\right).
 \label{eq:zero-order-hyperbolic-length}
\end{equation}
The normalization can be read off from the disk distance formula.  Let
$z_{u,\epsilon}$ and $z_{v,\epsilon}$ be the points where the geodesic
crosses $\{x=\epsilon\}$.  Then
\begin{equation}
 \cosh d_{\mathbb H}(z_{u,\epsilon},z_{v,\epsilon})
 =1+\frac{2|z_{u,\epsilon}-z_{v,\epsilon}|^2}
 {(1-|z_{u,\epsilon}|^2)(1-|z_{v,\epsilon}|^2)}.
 \label{eq:zero-order-disk-distance}
\end{equation}
At either end $1-|z_{\cdot,\epsilon}|^2=2\epsilon+O(\epsilon^2)$, while
the numerator tends to $8\sin^2((u-v)/2)$, so the truncated length is
$2\log(2|\sin((u-v)/2)|/\epsilon)+o(1)$.  Adding the two-end counterterm
$2\log\epsilon$ gives Equation~\eqref{eq:zero-order-hyperbolic-length}.
The finite part is negative for short intervals, which is one reason a
bare finite-dimensional entropy cannot simply be equated with a
renormalized length.
\end{revision}

\begin{revision}
Let $p_0(u,v)$ be the continuum reference proto-area in entropy units,
after its ultraviolet divergence has been removed in a specified scheme.
The existence of such renormalized reference data is an assumption on
the continuum limit of a family of finite codes.  If the entropy
regulator and the geometric defining function differ by a fixed endpoint
convention, reference matching takes the form
\begin{equation}
 \kappa p_0(u,v)=L_0(u,v)+\omega(u)+\omega(v).
 \label{eq:zero-order-reference-matching}
\end{equation}
The smooth function $\omega$ records that convention.  Once the
regulators and the boundary representative are matched it is fixed, and
in a common scheme it vanishes.  The counterterms must be local at the
endpoints and independent of the logical state.  Constants assigned
separately to individual regions are not a permitted calibration.
\end{revision}

\begin{revision}
There is a convenient test that eliminates the endpoint terms
altogether.  For four distinct boundary points define
\begin{equation}
 W_{ijkl}[f]=f(u_i,u_j)+f(u_k,u_l)
             -f(u_i,u_k)-f(u_j,u_l).
 \label{eq:reference-four-point-combination}
\end{equation}
Any function of the form $\omega(u)+\omega(v)$ is annihilated by
$W_{ijkl}$, so Equation~\eqref{eq:zero-order-reference-matching}
requires
\begin{equation}
 \kappa W_{ijkl}[p_0]
 =2\log\!\left|
 \frac{\sin((u_i-u_j)/2)\sin((u_k-u_l)/2)}
      {\sin((u_i-u_k)/2)\sin((u_j-u_l)/2)}
 \right|
 \label{eq:zero-order-four-point-test}
\end{equation}
for every quadruple.  The right-hand side is the logarithm of a cross
ratio of the four boundary points, and the test is the statement that
the reference remainders have the cross-ratio structure of hyperbolic
lengths.  Conversely, for smooth symmetric data the full collection of
these equalities determines $\kappa p_0-L_0$ up to one endpoint sum.
Fixing two auxiliary points expresses the difference as a function of
$u$ plus a function of $v$, and symmetry identifies the two functions up
to a constant.  The test therefore checks exactly the regional
dependence that a state-independent endpoint convention cannot alter,
and the endpoint function that remains has to agree with the chosen
renormalization scheme.
\end{revision}

\begin{revision}
The scale $\kappa$ must be fixed independently of the responses being
tested.  In a classical Einstein--AdS$_3$/CFT$_2$ realization, the
separately specified central charge and radius give the Brown--Henneaux
relation
\begin{equation}
 c=\frac{3\ell}{2G_N},\qquad
 \kappa=\frac{4G_N}{\ell}=\frac{6}{c},
 \label{eq:independent-brown-henneaux-calibration}
\end{equation}
when the proto-area remainder is identified with the classical area
term~\cite{BrownHenneaux1986,RyuTakayanagiAspects2006}.  The vacuum
provides a check of the whole scheme.  The area contribution to the
entropy of an arc on the circle has the regulated form
\begin{equation}
 p_{0,\mathrm{area}}^\epsilon(u,v)
 =\frac{c}{3}\log\!\left(
 \frac{2|\sin((u-v)/2)|}{\epsilon}\right),
 \label{eq:matched-vacuum-area-reference}
\end{equation}
and multiplying by $6/c$ and adding $2\log\epsilon$ reproduces
Equation~\eqref{eq:zero-order-hyperbolic-length} for every interval.
This is an explicit matched reference for the classical dictionary.
For a finite code, the identification of its recovered entropy with a
renormalized bulk entropy is additional input, and quantum corrections
would require a consistently renormalized area coefficient together with
an independently specified matter remainder.
\end{revision}

\begin{revision}
For a code without a gravitational dictionary, independently matched
reference data can still fix $\kappa$.  Any nonzero four-point
combination in Equation~\eqref{eq:zero-order-four-point-test} determines
a candidate value, and every other quadruple must then agree with the
same positive number.  This uses only the reference state and none of
its responses.  A code that offers neither reference data nor a
dictionary supplies no Newton constant by itself, and for it the
results below are conditional on an externally specified conversion
factor.  In particular, $G_{\mathrm{eff}}$ cannot be readjusted after
the quadratic obstruction has been evaluated.  Its nontrivial
second-order scaling is displayed in
Equation~\eqref{eq:scale-two-jet-obstruction}.
\end{revision}

\begin{revision}
If Equation~\eqref{eq:zero-order-reference-matching} holds with a fixed
endpoint convention, the endpoint terms cancel upon subtracting
$\rho_0$, and the centered criteria then test the next two orders of
the same absolute length family.  If it does not hold, passing a
centered criterion does not supply the missing reference geometry.  The
finite graph problem has the same structure.  It assumes that the
reference data are realized by the stated weights $w_0$, and this is
not implied by any condition on the centered derivatives.
\end{revision}

\subsection{Differential structure of proto-area}
\label{subsec:proto-area-differential-structure}

On a fixed faithful support, define
\begin{equation}
  \mathcal T_\sigma(X)
  :=D\log_\sigma[X]
  =\int_0^\infty
  (\sigma+sI)^{-1}X(\sigma+sI)^{-1}\,ds
  \label{eq:section2-frechet-logarithm}
\end{equation}
and the Bogoliubov--Kubo--Mori form~\cite{Petz1996}
\begin{equation}
  \mathfrak g_\sigma(X,Y)
  =\tr\!\left[X\mathcal T_\sigma(Y)\right].
  \label{eq:section2-bkm-form}
\end{equation}
In an eigenbasis of $\sigma$, the coefficient multiplying $X_{ij}$ is the divided difference
\begin{equation}
  \ell(p_i,p_j)=
  \begin{cases}
    \dfrac{\log p_i-\log p_j}{p_i-p_j},&p_i\neq p_j,\\[2mm]
    \dfrac1{p_i},&p_i=p_j,
  \end{cases}
  \label{eq:section2-divided-logarithm}
\end{equation}
which is positive and symmetric. The resolvent proof, spectral bounds, and perturbation estimates are given in Appendix~\ref{appA:fixed-support-entropy}.

The fixed-support entropy Hessian is standard~\cite{Petz1996,LesniewskiRuskai1999,Higham2008}. In the BKM normalization of Equation~\eqref{eq:section2-bkm-form}, it takes the following form.

\begin{theorem}[Entropy two-jet]
\label{thm:section2-entropy-two-jet}
Let $t\mapsto\sigma(t)$ be a $C^2$ curve of faithful density matrices on a fixed support. At $t=0$,
\begin{align}
  \left.\frac{d}{dt}\right|_{0}S(\sigma(t))
  &=-\tr(\dot\sigma\log\sigma),
  \label{eq:section2-entropy-first-derivative}\\
  \left.\frac{d^2}{dt^2}\right|_{0}S(\sigma(t))
  &=-\tr(\ddot\sigma\log\sigma)
  -\mathfrak g_\sigma(\dot\sigma,\dot\sigma).
  \label{eq:section2-entropy-second-derivative}
\end{align}
For a $C^2$ multiparameter family,
\begin{equation}
  \partial_i\partial_jS(\sigma)
  =-\tr(\sigma_{ij}\log\sigma)
  -\mathfrak g_\sigma(\sigma_i,\sigma_j).
  \label{eq:section2-entropy-mixed-hessian}
\end{equation}
\end{theorem}

\begin{theorem}[Proto-area two-jet]
\label{thm:section2-proto-area-two-jet}
Let $(V,\Sigma)$ be faithfully recovery-regular on $\calI$, or let a
fixed calibrated recovery be faithfully branch-regular there, and let
$\rho(t)$ be a $C^2$ state curve with $\rho(0)=\rho_0$. Set
\begin{equation}
  \omega_A(t)=\omega_A(\rho(t)),
  \qquad
  \tau_A(t)=\tau_A^{\mathrm{cal}}(\rho(t)).
  \label{eq:section2-output-state-curves}
\end{equation}
Then
\begin{equation}
  a_{\rho(t)}(A)-a_{\rho_0}(A)
  =t\,a_1(A)+\frac{t^2}{2}a_2(A)+o(t^2),
  \label{eq:section2-proto-area-jet-expansion}
\end{equation}
where
\begin{align}
  \frac{a_1(A)}{4G_{\mathrm{eff}}}
  ={}&-\tr(\dot\omega_A\log\omega_A)
  +\tr(\dot\tau_A\log\tau_A),
  \label{eq:section2-proto-area-first-jet}\\
  \frac{a_2(A)}{4G_{\mathrm{eff}}}
  ={}&-\tr(\ddot\omega_A\log\omega_A)
  -\mathfrak g_{\omega_A}(\dot\omega_A,\dot\omega_A)
  \notag\\
  &+\tr(\ddot\tau_A\log\tau_A)
  +\mathfrak g_{\tau_A}(\dot\tau_A,\dot\tau_A).
  \label{eq:section2-proto-area-second-jet}
\end{align}
For an affine path $\rho(t)=\rho_0+tX$,
\begin{equation}
  \frac{a_2(A)}{4G_{\mathrm{eff}}}
  =\mathfrak g_{\tau_A}
  \!\left(\mathcal M_A^{\mathrm{cal}} X,\mathcal M_A^{\mathrm{cal}} X\right)
  -\mathfrak g_{\omega_A}
  \!\left(\mathcal N_A X,\mathcal N_A X\right),
  \label{eq:section2-affine-proto-area-curvature}
\end{equation}
with
\begin{equation}
  \mathcal N_A(X)=\tr_{\bar A}(VXV^\dagger),
  \qquad
  \mathcal M_A^{\mathrm{cal}}(X)
  =\tr_{\bar A_1}\mathcal Q_{A,\mathbf R_A^{\mathrm{cal}}}(X).
  \label{eq:section2-linear-output-maps}
\end{equation}
\end{theorem}

\begin{revision}
The recovered regional channel factors through the boundary marginal.
Tracing the complementary recovery outputs in
Equation~\eqref{eq:section2-recovered-channel} gives
\begin{equation}
 \tau_A^{\mathrm{cal}}(\rho)
 =\operatorname{tr}_{A_2}\sigma_A(\rho),
 \qquad
 \sigma_A(\rho)
 =R_A^{\mathrm{cal}}\omega_A(\rho)
 (R_A^{\mathrm{cal}})^\dagger.
 \label{eq:regional-recovery-factorization}
\end{equation}
Since the recovery is isometric before the partial trace, the regional
remainder is a conditional entropy,
\begin{equation}
 S_{\mathrm{PA}}(V,\rho;A)
 =S(\sigma_A(\rho))-S(\operatorname{tr}_{A_2}\sigma_A(\rho))
 =S(A_2|A_1)_{\sigma_A(\rho)}.
 \label{eq:proto-area-conditional-entropy}
\end{equation}
This identity also holds for a fixed calibrated branch.
\end{revision}

\begin{revision}
\begin{proposition}[Concavity at fixed recovery]
\label{prop:fixed-recovery-concavity}
On a convex family of logical states with a fixed recovery, the
regional proto-area is concave.  Wherever its affine-path Hessian
exists, $a_2(A)\leq0$.
\end{proposition}
\end{revision}

\begin{revision}
\begin{proof}
Put $T=\operatorname{tr}_{A_2}$.  Along an affine path,
$\sigma_A(t)=\sigma_A(0)+tY$, the entropy Hessian is
\begin{equation}
 \left.\frac{d^2}{dt^2}\right|_0S_{\mathrm{PA}}(V,\rho_t;A)
 =-\mathfrak g_{\sigma_A}(Y,Y)
  +\mathfrak g_{T\sigma_A}(TY,TY)\leq0.
 \label{eq:general-branch-bkm-contraction}
\end{equation}
The inequality is monotonicity of the BKM metric under the channel
$T$~\cite{Petz1996,LesniewskiRuskai1999}.  On the relative interior
of any line segment of density matrices the support of the mixture,
and the support of its channel image, are constant.  Integration of
this nonpositive second derivative gives concavity on the open
segment.  Continuity of entropy extends it to the endpoints.
\end{proof}
\end{revision}

\begin{revision}
The sign in Equation~\eqref{eq:general-branch-bkm-contraction} does not
require a product reference.  Such a reference is needed for the
explicit Petz-orthogonal decomposition below.  It makes the lost
BKM information a squared norm, with a direct condition for equality.
For ball-shaped regions around the holographic vacuum, the Hessian of
relative entropy was identified with bulk canonical energy in
Ref.~\cite{LashkariVanRaamsdonk2016}.  Here the comparison is between
information lost under regional recovery and the normal Hessian of
the complete boundary-length map, and the Einstein equation is not
assumed.
\end{revision}

Let $T=\tr_{A_2}$ and fix faithful states $\tau_0$ and $\chi$ on
$A_1$ and $A_2$.  Write
\begin{equation}
 \sigma_0=\tau_0\otimes\chi,
 \qquad
 R_\chi(Z)=Z\otimes\chi,
 \qquad
 P_\chi=R_\chi T,
 \qquad
 Q_\chi=I-P_\chi.
 \label{eq:section2-petz-projection-definitions}
\end{equation}
The map $R_\chi$ is the Petz recovery map for the pair
$(T,\sigma_0)$. Channel recovery in this form belongs to Petz's
sufficiency theory~\cite{Petz1986,Petz1988}, and $P_\chi$ retains the
tangent component sufficient for the partial trace. Preservation of
monotone metrics under sufficient channels is part of the same
framework~\cite{Petz1996,GaoLiMarvianRouze2023}, while universal recovery
channels provide a complementary route to entanglement-wedge
reconstruction~\cite{Cotler:2017erl}. At a product reference, the Petz map
gives the explicit orthogonal decomposition below.

\begin{revision}
\begin{theorem}[BKM--Petz orthogonal decomposition]
\label{thm:section2-bkm-petz-decomposition}
For Hermitian tangent operators $Y$ at $\sigma_0$ and $X,Z$ at $\tau_0$,
\begin{align}
 \mathfrak g_{\sigma_0}(R_\chi X,Y)
 &=\mathfrak g_{\tau_0}(X,TY),
 \label{eq:section2-bkm-adjointness}\\
 \mathfrak g_{\sigma_0}(R_\chi X,R_\chi Z)
 &=\mathfrak g_{\tau_0}(X,Z).
 \label{eq:section2-bkm-isometry}
\end{align}
Consequently $P_\chi$ is the BKM-orthogonal projection onto
$\operatorname{Ran}R_\chi$, and the second-order data-processing defect
has the exact Pythagorean form
\begin{equation}
 \begin{split}
 \Delta_T^{(2)}(Y)
 &:={\mathfrak g}_{\sigma_0}(Y,Y)
   -{\mathfrak g}_{\tau_0}(TY,TY)\\
 &=\mathfrak g_{\sigma_0}(Q_\chi Y,Q_\chi Y).
 \end{split}
 \label{eq:section2-bkm-pythagorean-defect}
\end{equation}
More generally, polarization gives
\begin{equation}
 \mathfrak g_{\sigma_0}(Y,Z)
 -\mathfrak g_{\tau_0}(TY,TZ)
 =\mathfrak g_{\sigma_0}(Q_\chi Y,Q_\chi Z).
 \label{eq:section2-bkm-polarized-defect}
\end{equation}
In particular,
\begin{equation}
 \Delta_T^{(2)}(Y)=0
 \quad\Longleftrightarrow\quad
 Y=(TY)\otimes\chi.
 \label{eq:section2-infinitesimal-sufficiency}
\end{equation}
\end{theorem}
\end{revision}

\begin{proof}
For a differentiable curve $\tau(s)$ with $\dot\tau(0)=X$,
functional calculus gives
\begin{equation}
 \left.\frac{d}{ds}\right|_0
 \log(\tau(s)\otimes\chi)
 =D\log_{\tau_0}[X]\otimes I_{A_2}.
 \label{eq:section2-product-log-derivative}
\end{equation}
The BKM form is symmetric.  Hence
\begin{equation}
 \begin{split}
 \mathfrak g_{\sigma_0}(R_\chi X,Y)
 &=\tr\!\left[
 Y\bigl(D\log_{\tau_0}[X]\otimes I_{A_2}\bigr)
 \right]\\
 &=\tr\!\left[(TY)D\log_{\tau_0}[X]\right]
 =\mathfrak g_{\tau_0}(X,TY),
 \end{split}
\end{equation}
which proves Equation~\eqref{eq:section2-bkm-adjointness}.
Putting $Y=R_\chi Z$ and using $TR_\chi=I$ proves
Equation~\eqref{eq:section2-bkm-isometry}.  Moreover,
\begin{equation}
 \mathfrak g_{\sigma_0}
 \bigl(P_\chi Y,Y-P_\chi Y\bigr)
 =\mathfrak g_{\tau_0}
 \bigl(TY,T(Y-P_\chi Y)\bigr)=0.
\end{equation}
Thus $P_\chi$ is an orthogonal projection.  Pythagoras and the
isometry of $R_\chi$ give Equation~\eqref{eq:section2-bkm-pythagorean-defect}.
Applying the same orthogonal decomposition to two tangents gives
Equation~\eqref{eq:section2-bkm-polarized-defect}.
Positive definiteness of the BKM form gives
Equation~\eqref{eq:section2-infinitesimal-sufficiency}.
\end{proof}

For an affine decoded-state path $\sigma(t)=\sigma_0+tY$, the entropy
two-jet gives
\begin{equation}
 \left.\frac{d^2}{dt^2}\right|_0
 \bigl[S(\sigma(t))-S(T\sigma(t))\bigr]
 =-\norm{Q_\chi Y}_{\mathrm{BKM},\sigma_0}^{2}.
 \label{eq:section2-pa-hessian-petz-norm}
\end{equation}
When $\chi=I_{A_2}/d_2$, the finite-displacement form of Theorem 4.1 in
Ref.~\cite{CaoEtAl2026}, as recast in Witten's analysis~\cite{Witten2026},
takes the form
\begin{equation}
 \begin{split}
 &S(\sigma(t))-S(T\sigma(t))-\log d_2\\
 &\qquad=-\left[
 D(\sigma(t)\Vert\sigma_0)
 -D(T\sigma(t)\Vert\tau_0)
 \right].
 \end{split}
 \label{eq:section2-witten-relative-entropy-identity}
\end{equation}
For a non-maximally mixed $\chi$, Equation~\eqref{eq:section2-pa-hessian-petz-norm}
still holds, whereas the finite identity acquires the corresponding
$\chi$-modular-energy term.

\begin{theorem}[Sharp tangent-observability bound and its obstruction]
\label{thm:section2-sharp-sufficiency-angle}
Let $E\neq\{0\}$ be a finite-dimensional linear space of Hermitian trace-zero tangents at
$\sigma_0$ and define its minimum retained BKM fraction by
\begin{equation}
 \alpha_E=\inf_{\substack{0\neq Y\in E}}
 \frac{\mathfrak g_{\tau_0}(TY,TY)}
      {\mathfrak g_{\sigma_0}(Y,Y)}
 =\inf_{\substack{0\neq Y\in E}}
 \frac{\norm{P_\chi Y}_{\mathrm{BKM},\sigma_0}^{2}}
      {\norm{Y}_{\mathrm{BKM},\sigma_0}^{2}}.
 \label{eq:section2-sufficiency-angle}
\end{equation}
Equivalently, $\alpha_E=\cos^2\theta_E$, where $\theta_E$ is the
largest BKM principal angle from $E$ to
$\operatorname{Ran}R_\chi$, \rev{the largest angle, measured in the BKM
inner product, between a vector of $E$ and its projection onto that
subspace}.
The optimal constant in
\begin{equation}
 \Delta_T^{(2)}(Y)
 \leq C_E\,\mathfrak g_{\tau_0}(TY,TY),
 \qquad Y\in E,
 \label{eq:section2-angle-stability-bound}
\end{equation}
is
\begin{equation}
 C_E=
 \begin{cases}
 \dfrac{1-\alpha_E}{\alpha_E},&\alpha_E>0,\\[2mm]
 +\infty,&\alpha_E=0.
 \end{cases}
 \label{eq:section2-optimal-angle-constant}
\end{equation}
This constant is finite \rev{if and only if}
$E\cap\ker T=\{0\}$.  When finite, it can equivalently be written
$C_E=\tan^2\theta_E$.  Without a
restriction on the allowed tangent space there is no finite
dimension-independent, or even fixed-dimension, constant in
Equation~\eqref{eq:section2-angle-stability-bound} whenever
$\dim A_2>1$.
\end{theorem}

\begin{proof}
Theorem~\ref{thm:section2-bkm-petz-decomposition} gives
\begin{equation}
 \frac{\Delta_T^{(2)}(Y)}
      {\mathfrak g_{\tau_0}(TY,TY)}
 =\frac{\mathfrak g_{\sigma_0}(Y,Y)}
       {\mathfrak g_{\tau_0}(TY,TY)}-1
\end{equation}
whenever $TY\neq0$.  Taking the supremum over $E$ proves
Equation~\eqref{eq:section2-optimal-angle-constant}.  In finite
dimension the infimum in Equation~\eqref{eq:section2-sufficiency-angle}
is attained on the BKM unit sphere, so it vanishes exactly when a
nonzero vector of $E$ lies in $\ker T$.  Finally, choose nonzero Hermitian operators $Z$ and $B$ with
$\tr B=0$ and put $Y=Z\otimes B$.  Then $TY=0$ while
$\Delta_T^{(2)}(Y)=\mathfrak g_{\sigma_0}(Y,Y)>0$.  Since $\sigma_0$
is faithful, $\sigma_0+tY$ remains a state for sufficiently small
$|t|$ after imposing $\tr Y=0$.  No finite unrestricted constant can
therefore exist.
\end{proof}

The theorem is a quantitative tangent-recovery statement.  On $E$,
the Petz lift reconstructs the sufficient component and
\begin{equation}
 \norm{Y-P_\chi Y}_{\mathrm{BKM},\sigma_0}
 \leq\sqrt{C_E}\,
 \norm{P_\chi Y}_{\mathrm{BKM},\sigma_0}.
 \label{eq:section2-tangent-petz-stability}
\end{equation}
It does not identify reduced distinguishability with a channel
reconstruction error.  Bounds for an operational recovery error
require a specified encoding channel, recovery map, and spectral
control\rev{. P}reservation of regular quantum Fisher metrics and its
relation to Petz recoverability have been studied in
Ref.~\cite{GaoLiMarvianRouze2023}.

\begin{corollary}[Fixed-central-fiber decomposition]
\label{cor:section2-fixed-center-bkm}
Let
\begin{equation}
 \sigma_0=\bigoplus_\alpha p_\alpha
 (\tau_{0,\alpha}\otimes\chi_\alpha)
 \label{eq:section2-central-product-reference}
\end{equation}
be faithful on its block support.  For fixed-central-probability
tangents $Y=\bigoplus_\alpha p_\alpha Y_\alpha$, define $T$, $R$, $P$
and $Q$ blockwise.  Then
\begin{equation}
 \Delta_T^{(2)}(Y)
 =\sum_\alpha p_\alpha
 \norm{Q_\alpha Y_\alpha}_{\mathrm{BKM},
 \tau_{0,\alpha}\otimes\chi_\alpha}^{2},
 \label{eq:section2-central-bkm-sum}
\end{equation}
and Theorem~\ref{thm:section2-sharp-sufficiency-angle} applies with the
weighted direct-sum norm.  Thus central affine response and
within-block sufficiency defect are orthogonal mechanisms already at
second order.
\end{corollary}

\begin{proof}
For a block $p_\alpha\sigma_\alpha$ and tangent
$p_\alpha Y_\alpha$,
$D\log_{p_\alpha\sigma_\alpha}[p_\alpha Y_\alpha]
=D\log_{\sigma_\alpha}[Y_\alpha]$.  The BKM form is therefore the
$p_\alpha$-weighted sum of the block forms.  Apply
Theorem~\ref{thm:section2-bkm-petz-decomposition} in each block.
\end{proof}

Witten's random encoding-skew calculation probes a particular
encoding parameter at fixed bulk input.  In the stated \rev{Gaussian
unitary ensemble (GUE)} and large-dimension regime it gives a typical retained-to-total relative
entropy ratio of order $d_2^{-2}$~\cite{Witten2026}.  In the language
of Equation~\eqref{eq:section2-sufficiency-angle}, this is a typical
Rayleigh quotient of the sampled random tangent, not the infimum over
all logical-state directions.  It explains a typical ratio of order
$d_2^2$ between defect and retained distinguishability in that model,
but does not by itself furnish a model-independent bound in terms of
an optimized reconstruction error.

A nonzero affine-path curvature quantifies the failure of a state-independent area-operator representation for that region. Locality requires additional relations among the curvatures and first derivatives for the full family of regions.

For a torsion-free connection on a multidimensional state manifold, normal coordinates at $\rho_0$ give
\begin{equation}
\begin{split}
  \frac{(\nabla d\mathcal A_{\rho_0})(X,Y)(A)}{4G_{\mathrm{eff}}}
  ={}&-\tr\!\left(\omega_{A;XY}\log\omega_A\right)
  -\mathfrak g_{\omega_A}(\omega_{A;X},\omega_{A;Y})\\
  &+\tr\!\left(\tau_{A;XY}\log\tau_A\right)
  +\mathfrak g_{\tau_A}(\tau_{A;X},\tau_{A;Y}).
  \label{eq:section2-proto-area-mixed-hessian}
\end{split}
\end{equation}
Equation~\eqref{eq:section2-proto-area-mixed-hessian} supplies the quantum-information input to the state-space obstruction tensor in Section~\ref{subsec:state-space-factorization}.

\begin{revision}
Constant-rank curves can be reduced smoothly to fixed support. At a rank change, an eigenvalue emerging at order $t^2$ produces a $t^2\log t$ term in its entropy. The proto-area can still have a second jet if the singular terms in the boundary and recovered entropies cancel. Constant support is a sufficient condition that avoids this additional cancellation problem. Precise entropy asymptotics appear in Appendix~\ref{appA:rank-changing}.
\end{revision}

For later robustness estimates, define
\begin{align}
  F_d(\varepsilon)
  &=\begin{cases}
  0,&d=1,\\
  \varepsilon\log(d-1)+h_2(\varepsilon),
  &d\geq2,\ 0\leq\varepsilon\leq1-d^{-1},\\
  \log d,&d\geq2,\ 1-d^{-1}<\varepsilon\leq1,
  \end{cases}
  \label{eq:section2-audenaert-function}\\
  h_2(\varepsilon)
  &=-\varepsilon\log\varepsilon
    -(1-\varepsilon)\log(1-\varepsilon),
  \qquad 0\log0:=0.
  \label{eq:section2-binary-entropy}
\end{align}

\begin{corollary}[Finite-dimensional continuity]
\label{cor:section2-proto-area-continuity}
If $\rho,\sigma$ lie in a recovery-regular sector and
$\varepsilon=\tfrac12\norm{\rho-\sigma}_1$, then
\begin{equation}
  \abs{a_\rho(A)-a_\sigma(A)}
  \leq4G_{\mathrm{eff}}
  \left[F_{d_A}(\varepsilon)+F_{d_{a_A}}(\varepsilon)\right].
  \label{eq:section2-proto-area-continuity-bound}
\end{equation}
\end{corollary}

Appendix~\ref{appA:recovery-errors} gives the channel-error estimate used in the robustness analysis of Section~\ref{subsec:sec6-robustness}.

\subsection{Sampled data and continuum limits}
\label{subsec:proto-area-continuum-limit}

A finite boundary code produces a sampled vector, not a smooth metric. Let the $N$th code have cyclic boundary sites $B_N$, a complement-closed family $\calI_N$ of connected intervals, and reference state $\rho_{N,0}$. Define
\begin{equation}
  \bfq_N(\rho_N)
  =\left(
  a_{N,\rho_N}(A)-a_{N,\rho_{N,0}}(A)
  \right)_{A\in\calI_N}.
  \label{eq:section2-sampled-proto-area-vector}
\end{equation}
At fixed $N$, geometrizability is the finite polyhedral problem solved in Section~\ref{sec:discrete}.

Let
\begin{equation}
  (\mathsf C_Nq)(A)=q(\bar A),
  \qquad
  P_{N,\pm}=\frac12(I\pm\mathsf C_N).
  \label{eq:section2-discrete-complement-projections}
\end{equation}
The even data are interpolated into the real Hilbert space
$\mathcal Y^\tau_{\mathbb R}$ of Section~\ref{sec:linearized-geometry}\rev{, while} the odd data are retained in
\begin{equation}
  L^2_-(\mathcal G_{\mathbb H})
  =\ker(\operatorname{Id}+S_A^*)
  \subset L^2(\mathcal G_{\mathbb H},d\beta\,db).
  \label{eq:section2-odd-geodesic-data-space}
\end{equation}
Choose bounded maps
\begin{equation}
  J_N^+:P_{N,+}\mathbb R^{\calI_N}
  \longrightarrow\mathcal Y^\tau_{\mathbb R},
  \qquad
  J_N^-:P_{N,-}\mathbb R^{\calI_N}
  \longrightarrow L^2_-(\mathcal G_{\mathbb H}).
  \label{eq:section2-data-interpolation-map}
\end{equation}
For a parameter family $\rho_N(\lambda)$, set
\begin{align}
  \alpha_N^+(\lambda)
  &=J_N^+P_{N,+}\bfq_N(\rho_N(\lambda)),
  \label{eq:section2-interpolated-even-data}\\
  \alpha_N^-(\lambda)
  &=J_N^-P_{N,-}\bfq_N(\rho_N(\lambda)).
  \label{eq:section2-interpolated-odd-data}
\end{align}

\begin{definition}[Continuum proto-area family]
\label{def:section2-continuum-proto-area}
The code family has a $C^k$ continuum proto-area on a compact parameter set $K_0$ if there are maps
\begin{equation}
  \alpha^+\in C^k(K_0;\mathcal Y^\tau_{\mathbb R}),
  \qquad
  \alpha^-\in C^k(K_0;L^2_-(\mathcal G_{\mathbb H}))
  \label{eq:section2-continuum-even-odd-data}
\end{equation}
for which
\begin{equation}
  \norm{\alpha_N^+-\alpha^+}_{C^k(K_0;\mathcal Y^\tau)}
  +\norm{\alpha_N^--\alpha^-}_{C^k(K_0;L^2_-)}
  \longrightarrow0.
  \label{eq:section2-Ck-continuum-convergence}
\end{equation}
A disk-geometric continuum family must satisfy $\alpha^-=0$\rev{, and} then $\alpha:=\alpha^+$ is its continuum proto-area map.
\end{definition}

\begin{revision}
Data convergence and convergence of the interior realization operators are distinct requirements. The consequences of the density and consistency hypotheses that connect finite cut maps to the tensor X-ray transform are established in Theorem~\ref{thm:finite-code-continuum-limit} and Appendix~\ref{appB:continuum-approximation}.
\end{revision}

\begin{proposition}[Compatibility of jets with the continuum limit]
\label{prop:section2-jets-continuum-limit}
If Equation~\eqref{eq:section2-Ck-continuum-convergence} holds with $k\geq2$, then for $j=0,1,2$ the even and odd derivatives converge in their respective data spaces at every interior parameter value. A nonzero odd limit obstructs disk geometrizability before the X-ray range test.
\end{proposition}

\begin{proof}
The derivative convergence is part of the defining $C^k$ norm. The obstruction is Proposition~\ref{prop:section2-complement-parity}.
\end{proof}

\subsection{The geometrizability hierarchy}
\label{subsec:geometrizability-hierarchy}

\begin{revision}
The geometric target is a neighborhood of the hyperbolic disk with the
boundary points and conformal representative fixed. Let
$\operatorname{Diff}_0^x(M;\partial M)$ denote the identity component of
the diffeomorphisms that preserve the chosen regularity class and satisfy
\begin{equation}
 \Psi|_{\partial M}=\operatorname{Id},\qquad
 \left.\frac{x\circ\Psi}{x}\right|_{\partial M}=1.
 \label{eq:section2-fixed-representative-gauge}
\end{equation}
The second condition excludes a change of the boundary length unit. Set
\begin{equation}
  \mathfrak M_{\mathrm{AH}}
  =\frac{\left\{
  \text{regular simple AH metrics with fixed conformal representative}
  \right\}}
  {\operatorname{Diff}_0^x(M;\partial M)}.
  \label{eq:section2-AH-metric-moduli}
\end{equation}
Regular metric coefficients have the positive boundary order specified
in Equation~\eqref{eq:quadratic-decay-class}. The subtracted length is
well defined off the endpoint diagonal. Its membership in the data
Hilbert space is a further condition, so we use the restricted set
\begin{equation}
 \mathfrak M_{\mathrm{AH}}^\tau
 =\{[g]\in\mathfrak M_{\mathrm{AH}}:
       \mathcal L_g^x-\mathcal L_{g_{\mathbb H}}^x
       \in\mathcal Y^\tau_{\mathbb R}\}.
 \label{eq:section2-regular-length-domain}
\end{equation}
The boundary-length map is
\begin{equation}
\begin{split}
  \mathcal B:\mathfrak M_{\mathrm{AH}}^\tau
  &\longrightarrow\mathcal Y^\tau_{\mathbb R},\\
  \mathcal B([g])
  &=\mathcal L_g^x-\mathcal L_{g_{\mathbb H}}^x.
  \label{eq:section2-boundary-length-map}
\end{split}
\end{equation}
Differentiability into $\mathcal Y^\tau$ is imposed on the parameter
families under consideration, since pointwise smoothness off the
diagonal alone does not imply it.
Geometrizability asks whether
\begin{equation}
  \mathcal A_{\rho_0}=\mathcal B\circ\mathcal G
  \label{eq:section2-central-factorization}
\end{equation}
holds for a metric-valued map
$\mathcal G:\Sigma\to\mathfrak M_{\mathrm{AH}}^\tau$ with
$\mathcal G(\rho_0)=[g_{\mathbb H}]$.
\end{revision}

\begin{revision}
\begin{definition}[Levels of geometrizability]
\label{def:section2-geometrizability-levels}
A $C^2$ continuum proto-area map is locally geometrizable near $\rho_0$ if Equation~\eqref{eq:section2-central-factorization} holds on a neighborhood for a regular metric family whose subtracted lengths are jointly $C^2$ in $\mathcal Y^\tau$. It is first-order geometrizable if a linear map
\begin{equation*}
  K:T_{\rho_0}\Sigma
  \longrightarrow T_{[g_{\mathbb H}]}\mathfrak M_{\mathrm{AH}}
\end{equation*}
satisfies
\begin{equation}
  d\mathcal A_{\rho_0}
  =D\mathcal B_{g_{\mathbb H}}\circ K.
  \label{eq:section2-first-order-factorization}
\end{equation}
It is coefficientwise second-order geometrizable if there is a metric-valued coefficient two-jet $j^2_{\rho_0}\mathcal G$ such that
\begin{equation}
  j^2_{\rho_0}\mathcal A_{\rho_0}
  =j^2_{\rho_0}(\mathcal B\circ\mathcal G).
  \label{eq:section2-second-order-factorization}
\end{equation}
A strong $\mathcal Y^\tau$-valued two-jet additionally requires the
pathwise differentiability condition in
Theorem~\ref{thm:two-jet-geometrizability}.
A sequence of finite code families is asymptotically geometrizable to
order $k\geq1$ if its interpolated data converge in $C^k$ to a locally
geometrizable map and its linear realization operators satisfy the
source-density and operator-consistency hypotheses of
Theorem~\ref{thm:finite-code-continuum-limit}. For $k\geq2$, we additionally require all the hypotheses of
Theorem~\ref{thm:second-order-discrete-continuum} in one common source
and data norm pair, including convergence of the first two data
coefficients and both geometric derivatives in those norms. These norms
must control both reconstruction and the Hessian. In particular, convergence of
linear cut maps alone is a first-order statement.

A geometrizable map is dynamically admissible relative to a specified bulk matter model if the realizing metrics can be completed to initial data satisfying the gravitational constraints. On a time-reflection-symmetric two-dimensional slice this includes
\begin{equation}
  R(g)=2\Lambda+16\pi G\,T_{nn}.
  \label{eq:section2-Hamiltonian-constraint}
\end{equation}
\end{definition}
\end{revision}

\begin{proposition}[Implication hierarchy]
\label{prop:section2-geometrizability-implications}
Exact local geometrizability implies second-order geometrizability, which implies first-order geometrizability. Dynamical admissibility implies exact local geometrizability after forgetting the matter data. Asymptotic geometrizability to order $k\geq2$ implies convergence of the finite first- and second-order obstruction data to their continuum counterparts.
\end{proposition}

\begin{revision}
\begin{proof}
Take successive jets of Equation~\eqref{eq:section2-central-factorization}. A dynamically admissible family already contains realizing metrics. For the final statement, $C^k$ data convergence gives convergence of
the first two data coefficients. Theorem~\ref{thm:finite-code-continuum-limit}
transfers the first-order projection and reconstruction. The additional
second-derivative consistency hypothesis gives convergence of the
quadratic residual by
Theorem~\ref{thm:second-order-discrete-continuum}.
\end{proof}
\end{revision}

At fixed resolution, failure of the common-source condition is an exact
feasibility problem for the cut-incidence map\rev{, and the next section
solves it}.

\section{Discrete geometry and exact criteria for finite codes}
\label{sec:discrete}

A finite holographic code supplies only finitely many entropy or proto-area values, one for each boundary region under consideration. This finiteness does not make the geometric consistency problem vacuous. With the normalization introduced in Section~\ref{sec:proto-area}, a state change $\rho_0\to\rho$ produces the sampled variation
\begin{equation}
  q_\rho=\bigl(a_\rho(A)-a_{\rho_0}(A)\bigr)_{A\in\calI}.
  \label{eq:sampled-proto-area-variation}
\end{equation}
Once the graph topology is held fixed, all components of this vector must be generated by the same collection of local edge weights. The resulting compatibility conditions are nonlocal in the boundary labels even though their unknowns are local in the graph. They are the finite-dimensional precursor of the range conditions for the tensor X-ray transform developed in Section~\ref{sec:linearized-geometry}. The four-terminal star in Section~\ref{subsec:star-example} also supports the Hamiltonian-skewed examples of Section~\ref{sec:skewed-codes-revised}, allowing one witness to be traced from a finite cut relation to the nonlinear code constructions.

Graph models for static holographic entropy encode boundary entropies as minimum-cut weights and connect tensor networks, holographic quantum error correction, random tensor networks, and the holographic entropy cone~\cite{RyuTakayanagi2006,Headrick2014,PastawskiEtAl2015,BaoEtAl2015,HaydenEtAl2016,Harlow2017}. Their continuum max-flow/min-cut dual is the bit-thread formulation~\cite{FreedmanHeadrick2017,HeadrickHubeny2018}. The dual witnesses used below are covectors in sampled interval-data space, not bulk flows. They test whether a prescribed state-dependent variation belongs to the image of the cut-incidence map for a fixed graph and chamber. The entropy-cone problem instead varies the graph and asks which entropy vectors can occur. A vector may satisfy every universal holographic entropy inequality and still fail this fixed-background realization test.

The fixed-topology data locus is a finite union of polyhedral cones. At a regular background, where every relevant minimum cut is unique, the local inverse problem is linear and its obstructions are the left-null vectors of a cut-incidence matrix. Finite perturbations within a prescribed chamber obey Farkas inequalities that enforce the data while preventing a cut transition\rev{. These are the dual inequalities by which the Farkas lemma of linear programming certifies that a system of linear equalities and inequalities is feasible}. These two tests give the discrete common-source criterion needed below.

\subsection{Fixed graph geometry and minimum-cut chambers}
\label{subsec:fixed-graph}

Let $G=(V,E)$ be a finite connected undirected graph. A distinguished set $\partial V\subset V$ consists of boundary vertices, including a purifier whenever one is required\rev{, that is, an extra boundary vertex carrying the purification of a mixed boundary state}. Edge weights are collected in a vector
\begin{equation}
  w=(w_e)_{e\in E}\in \Rgeq^{E}.
  \label{eq:edge-weights}
\end{equation}
For a nonempty proper boundary region $\varnothing\neq A\subsetneq \partial V$, an $A$-cut is a vertex set $W\subset V$ satisfying $W\cap\partial V=A$. Its cut set is
\begin{equation}
  \delta W=\bigl\{\{x,y\}\in E:x\in W,\ y\notin W\bigr\}.
\end{equation}
We identify the cut with its incidence vector $c\in\{0,1\}^{E}$, where $c_e=1$ if and only if $e\in\delta W$. Repeated incidence vectors are discarded, and the resulting finite set is denoted by $\cA_A$. For a fixed nonempty finite family $\calI$ of boundary regions, the discrete area map is
\begin{equation}
  a_A^G(w)=\min_{c\in\cA_A} c^{\mathsf T}w,
  \qquad
  \bfa_G(w)=\bigl(a_A^G(w)\bigr)_{A\in\calI}\in\Rgeq^{\calI}.
  \label{eq:discrete-area-map}
\end{equation}
The normalization that converts cut weight into entropy or proto-area can be absorbed into $w$, which is shared by all regions. In an exact tensor-network model, $w_e$ is proportional to the logarithm of the bond dimension. For a skewed code, a small state-dependent change in $w_e$ represents the component of proto-area response assigned to a local geometric degree of freedom.

A \emph{cut pattern} is a choice
\begin{equation}
  \sigma=(c_A^\sigma)_{A\in\calI}\in\prod_{A\in\calI}\cA_A.
\end{equation}
Its cut-incidence matrix $M_\sigma\in\{0,1\}^{\calI\times E}$ has rows
\begin{equation}
  (M_\sigma)_{A,e}=(c_A^\sigma)_e.
  \label{eq:cut-matrix}
\end{equation}
The closed chamber in which the selected cuts are all minimizing is
\begin{equation}
  \cK_\sigma=
  \left\{
  w\in\Rgeq^{E}:
  (c_A^\sigma-c)^{\mathsf T}w\leq0
  \text{ for every }A\in\calI\text{ and }c\in\cA_A
  \right\}.
  \label{eq:closed-chamber}
\end{equation}
Its regular part is
\begin{equation}
  \Omega_\sigma=
  \left\{
  w\in\Rpos^{E}:
  (c-c_A^\sigma)^{\mathsf T}w>0
  \text{ for every }A\in\calI\text{ and }c\in\cA_A\setminus\{c_A^\sigma\}
  \right\}.
  \label{eq:open-chamber}
\end{equation}
Thus $\Omega_\sigma$ consists of positive weights for which every selected cut is the unique minimizer. The regular set $\Omega_\sigma$ may be empty. The closed cone $\cK_\sigma$ always contains the origin and may be lower-dimensional.

\begin{theorem}[Fixed-topology geometric locus]
\label{thm:fixed-topology-locus}
The map $\bfa_G:\Rgeq^{E}\to\Rgeq^{\calI}$ is continuous, coordinatewise nondecreasing, positively homogeneous, concave, and piecewise linear. More specifically,
\begin{equation}
  \bfa_G(w)=M_\sigma w
  \qquad\text{for every }w\in\cK_\sigma.
  \label{eq:area-on-chamber}
\end{equation}
The exact data locus of the fixed graph is
\begin{equation}
  \mathfrak G_G
  :=\bfa_G(\Rgeq^{E})
  =\bigcup_{\sigma}M_\sigma\cK_\sigma.
  \label{eq:fixed-topology-locus}
\end{equation}
It is a finite union of closed polyhedral cones. Moreover,
\begin{equation}
  \overline{\bfa_G(\Rpos^{E})}=\mathfrak G_G.
  \label{eq:positive-density}
\end{equation}
\end{theorem}

\begin{proof}
For fixed $A$, Equation~\eqref{eq:discrete-area-map} is the minimum of finitely many linear functionals. It is therefore continuous, positively homogeneous, concave, and piecewise linear. Since every cut-incidence vector has nonnegative entries, it is also coordinatewise nondecreasing. These properties hold componentwise for $\bfa_G$.

Every $w\in\Rgeq^{E}$ admits at least one minimizing cut for each $A$. Choosing one produces a pattern $\sigma$ with $w\in\cK_\sigma$. Conversely, if $w\in\cK_\sigma$, then $c_A^\sigma$ is minimizing for every region, and hence
\begin{equation}
  a_A^G(w)=(c_A^\sigma)^{\mathsf T}w.
\end{equation}
This proves Equation~\eqref{eq:area-on-chamber} and the union in Equation~\eqref{eq:fixed-topology-locus}. Each $\cK_\sigma$ is the intersection of finitely many homogeneous closed half-spaces, so it is a closed polyhedral cone. A linear image of a polyhedral cone is again a closed polyhedral cone in finite dimensions. There are finitely many patterns, which proves the asserted structure of $\mathfrak G_G$.

Finally, for $w\in\Rgeq^{E}$ and $\varepsilon>0$, the vector $w+\varepsilon\one$ has strictly positive entries and converges to $w$. Continuity gives $\bfa_G(w+\varepsilon\one)\to\bfa_G(w)$, so $\mathfrak G_G$ is contained in the closure on the left-hand side of Equation~\eqref{eq:positive-density}. The reverse inclusion follows from the closedness of $\mathfrak G_G$.
\end{proof}

For a fixed background $w_0$, the exact set of finite variations obtainable on the same graph, while allowing arbitrary cut transitions, is therefore
\begin{equation}
  \mathfrak V_G(w_0)
  :=\{\bfa_G(w)-\bfa_G(w_0):w\in\Rgeq^E\}
  =\bigcup_{\sigma}
  \bigl(M_\sigma\cK_\sigma-\bfa_G(w_0)\bigr).
  \label{eq:global-variation-locus}
\end{equation}
Thus $q_\rho$ is realizable on the fixed topology if and only if $\bfa_G(w_0)+q_\rho\in\mathfrak G_G$. This global criterion is disjunctive because the endpoint may lie in any cut chamber.

Theorem~\ref{thm:fixed-topology-locus} isolates the only source of nonlinearity in a fixed graph. Within a chamber, all region data depend linearly on the local weights. Nonlinearity occurs when a competing cut becomes equally light and the system crosses a chamber wall, the discrete analogue of an extremal-surface phase transition.

The distance to the nearest wall controls the validity of a linearized geometric interpretation. Let $\norm{\cdot}_X$ be any norm on $\mathbb R^E$, and let $\norm{\cdot}_{X,*}$ denote its dual norm. Write $\bfe_e$ for the coordinate covector associated with edge $e$.

\begin{proposition}[Exact chamber radius]
\label{prop:chamber-radius}
Suppose $w_0\in\Omega_\sigma$. Define
\begin{equation}
\begin{split}
  \rho_X(w_0;\sigma)=\min\Bigg\{
  &\min_{e\in E}\frac{(w_0)_e}{\norm{\bfe_e}_{X,*}},\\
  &\min_{\substack{A\in\calI\\c\in\cA_A\setminus\{c_A^\sigma\}}}
  \frac{(c-c_A^\sigma)^{\mathsf T}w_0}
       {\norm{c-c_A^\sigma}_{X,*}}
  \Bigg\}.
  \label{eq:chamber-radius}
\end{split}
\end{equation}
where a minimum over an empty competitor set is understood as $+\infty$. Then
\begin{equation}
  \rho_X(w_0;\sigma)=\Dist_X\bigl(w_0,\mathbb R^E\setminus\Omega_\sigma\bigr).
  \label{eq:radius-distance}
\end{equation}
In particular, $w_0+h\in\Omega_\sigma$ whenever $\norm{h}_X<\rho_X(w_0;\sigma)$, and no larger open $X$-ball centered at $w_0$ is contained in $\Omega_\sigma$.
\end{proposition}

\begin{proof}
The complement of $\Omega_\sigma$ is the union of the half-spaces
$\{w:w_e\leq0\}$ and
$\{w:(c-c_A^\sigma)^{\mathsf T}w\leq0\}$.  If
$u^{\mathsf T}w_0>0$, H\"older's inequality gives
\begin{equation}
 \Dist_X\bigl(w_0,\{w:u^{\mathsf T}w\leq0\}\bigr)
 =\frac{u^{\mathsf T}w_0}{\norm{u}_{X,*}}.
\end{equation}
The lower bound is immediate, and equality follows because the dual
norm is attained on the compact $X$-unit sphere.  The distance to a
finite union of closed sets is the minimum of the individual distances.
Substitution of $u=\bfe_e$ and $u=c-c_A^\sigma$ proves
Equations~\eqref{eq:chamber-radius} and~\eqref{eq:radius-distance}.
\end{proof}

\subsection{Regular backgrounds and exact local criteria}
\label{subsec:regular-backgrounds}

Fix a regular background $w_0\in\Omega_\sigma$ and abbreviate
\begin{equation}
  M=M_\sigma,
  \qquad
  p_0=\bfa_G(w_0),
  \qquad
  \rho=\rho_X(w_0;\sigma).
  \label{eq:regular-abbreviations}
\end{equation}
A boundary data variation $q\in\mathbb R^{\calI}$ is locally geometric at $w_0$ if there is an edge variation $h$ such that
\begin{equation}
  \bfa_G(w_0+h)-\bfa_G(w_0)=q
  \label{eq:local-realization}
\end{equation}
without leaving the regular chamber. Define the quotient norm induced by $M$ and $X$ by
\begin{equation}
  \norm{q}_{M,X}=
  \inf\left\{\norm{h}_X:Mh=q\right\},
  \label{eq:quotient-norm}
\end{equation}
with the convention that the infimum is $+\infty$ when $q\notin\operatorname{im}M$.

\begin{theorem}[Regular-background geometrizability]
\label{thm:regular-geometrizability}
Let $w_0\in\Omega_\sigma$. The following statements hold.

\begin{enumerate}[label=\textup{(\roman*)},leftmargin=8mm]
\item For every $h$ with $\norm{h}_X<\rho$,
\begin{equation}
  \bfa_G(w_0+h)-\bfa_G(w_0)=Mh.
  \label{eq:exact-local-linearity}
\end{equation}
Thus the first variation is exact throughout the maximal chamber ball, rather than merely asymptotic.

\item A data variation $q$ has a realization satisfying $\norm{h}_X<\rho$ if and only if
\begin{equation}
  q\in\operatorname{im}M
  \qquad\text{and}\qquad
  \norm{q}_{M,X}<\rho.
  \label{eq:local-realizability-quotient}
\end{equation}

\item The linear compatibility condition has the dual form
\begin{equation}
  q\in\operatorname{im}M
  \quad\Longleftrightarrow\quad
  y^{\mathsf T}q=0
  \text{ for every }y\in\ker M^{\mathsf T}.
  \label{eq:left-kernel-criterion}
\end{equation}
Every nonzero $y\in\ker M^{\mathsf T}$ therefore defines a boundary-accessible nongeometric witness.

\item In Euclidean norms, let $M^+$ be the Moore--Penrose pseudoinverse and $P_M=MM^+$ the orthogonal projector onto $\operatorname{im}M$. Then
\begin{equation}
  h_\star=M^+q
  \label{eq:pseudoinverse-reconstruction}
\end{equation}
is the unique minimum-norm least-squares reconstruction. If $q\in\operatorname{im}M$, it is the minimum-norm exact solution of $Mh=q$. For arbitrary $q$,
\begin{equation}
  D_G^{(1)}(q)
  :=\Dist_2(q,\operatorname{im}M)
  =\norm{q-Mh_\star}_2
  =\norm{(I-P_M)q}_2.
  \label{eq:first-discrete-defect}
\end{equation}
If $D_G^{(1)}(q)>0$, the normalized residual
\begin{equation}
  y_\star=\frac{(I-P_M)q}{\norm{(I-P_M)q}_2}
  \label{eq:optimal-linear-witness}
\end{equation}
obeys $y_\star\in\ker M^{\mathsf T}$, $\norm{y_\star}_2=1$, and
\begin{equation}
  y_\star^{\mathsf T}q=D_G^{(1)}(q)
  =\max_{\substack{y\in\ker M^{\mathsf T}\\\norm{y}_2\leq1}}y^{\mathsf T}q.
  \label{eq:optimal-witness-value}
\end{equation}

\item If $r=\Rank M$ and $\sigma_r(M)$ is the smallest nonzero singular value, then for every $q\in\operatorname{im}M$,
\begin{equation}
  \norm{M^+q}_2\leq\sigma_r(M)^{-1}\norm{q}_2.
  \label{eq:singular-stability}
\end{equation}
The dimensions of the invisible edge-deformation space and the independent witness space are
\begin{equation}
  \dim\ker M=|E|-r,
  \qquad
  \dim\ker M^{\mathsf T}=|\calI|-r.
  \label{eq:nullity-counts}
\end{equation}
\end{enumerate}
\end{theorem}

\begin{proof}
Part (i) follows immediately from Proposition~\ref{prop:chamber-radius} and Equation~\eqref{eq:area-on-chamber}. For part (ii), if $q=Mh$ with $\norm{h}_X<\rho$, then $q\in\operatorname{im}M$ and $\norm{q}_{M,X}\leq\norm{h}_X<\rho$. Conversely, suppose $q\in\operatorname{im}M$ and $\norm{q}_{M,X}<\rho$. Choose a minimizing sequence in the closed affine space $\{h:Mh=q\}$. Its norms are bounded, so a convergent subsequence exists in finite dimensions\rev{, and} the limit remains in the affine space and attains the infimum in Equation~\eqref{eq:quotient-norm}. Hence there exists $h$ with $Mh=q$ and $\norm{h}_X=\norm{q}_{M,X}<\rho$. Part (i) then gives Equation~\eqref{eq:local-realization}.

The fundamental theorem of linear algebra gives
\begin{equation}
  (\operatorname{im}M)^\perp=\ker M^{\mathsf T},
\end{equation}
which proves part (iii). For Euclidean norms, the standard pseudoinverse identities imply that $M^+q$ is the solution orthogonal to $\ker M$ and therefore the unique minimum-norm solution. They also give $P_M=MM^+$ and the orthogonal decomposition
\begin{equation}
  q=P_Mq+(I-P_M)q.
\end{equation}
This proves Equation~\eqref{eq:first-discrete-defect}. When the residual is nonzero, it lies in $(\operatorname{im}M)^\perp=\ker M^{\mathsf T}$. Cauchy--Schwarz shows that every unit witness satisfies
\begin{equation}
  y^{\mathsf T}q=y^{\mathsf T}(I-P_M)q
  \leq\norm{(I-P_M)q}_2,
\end{equation}
with equality at $y=y_\star$. This proves part (iv).

Finally, in a singular-value decomposition of $M$, the nonzero singular values of $M^+$ are $\sigma_j(M)^{-1}$. Its operator norm on $\operatorname{im}M$ is therefore $\sigma_r(M)^{-1}$, proving Equation~\eqref{eq:singular-stability}. The dimension statements follow from rank--nullity applied to $M$ and $M^{\mathsf T}$.
\end{proof}

To compare regional observables, fix embeddings of their centers into
one finite-dimensional commutative algebra $\mathcal Z$.  This common
identification is additional sewing data\rev{, since} without it sums of area
operators belonging to different regional centers are not defined.
Applying the scalar cut relations coefficientwise in this common center
gives the following sewing criterion.

\begin{theorem}[Operator-valued sewing of central areas]
\label{thm:operator-valued-central-sewing}
Let $\mathcal Z$ be a finite-dimensional commutative
$C^*$-algebra and let $L_A\in\mathcal Z_{\mathrm{sa}}$ be a central
area operator for every $A\in\calI$.  There are common local central
edge operators $W_e\in\mathcal Z_{\mathrm{sa}}$ satisfying
\begin{equation}
 L_A=\sum_{e\in E}M_{A,e}W_e,
 \qquad A\in\calI,
 \label{eq:operator-valued-edge-sewing}
\end{equation}
if and only if
\begin{equation}
 \sum_{A\in\calI}y_A L_A=0
 \quad\text{for every }y\in\ker M^{\mathsf T}.
 \label{eq:operator-valued-left-kernel}
\end{equation}
If the edge operators are required to be positive, they exist if and
only if
\begin{equation}
 \sum_{A\in\calI}z_A L_A\geq0\quad\text{in }\mathcal Z
 \quad\text{for every }z\in\mathbb R^{\calI}
 \text{ with }M^{\mathsf T}z\geq0.
 \label{eq:operator-valued-farkas}
\end{equation}
\end{theorem}

\begin{proof}
Write $\mathcal Z=\bigoplus_{\alpha=1}^m\mathbb CP_\alpha$ and
\begin{equation}
 L_A=\sum_\alpha\ell_A^{(\alpha)}P_\alpha,
 \qquad
 W_e=\sum_\alpha w_e^{(\alpha)}P_\alpha.
 \label{eq:central-coefficient-expansion}
\end{equation}
Equation~\eqref{eq:operator-valued-edge-sewing} is equivalent to the
$m$ scalar systems
$\ell^{(\alpha)}=Mw^{(\alpha)}$.  The fundamental theorem of linear
algebra gives solvability \rev{if and only if}
$y^{\mathsf T}\ell^{(\alpha)}=0$ for every
$y\in\ker M^{\mathsf T}$ and every $\alpha$, which is exactly
Equation~\eqref{eq:operator-valued-left-kernel}.  Positivity of $W_e$
is equivalent to $w_e^{(\alpha)}\geq0$ for all $e,\alpha$.
Finite-dimensional Farkas duality says
\begin{equation}
 \ell^{(\alpha)}\in M\mathbb R_{\geq0}^{E}
 \quad\Longleftrightarrow\quad
 z^{\mathsf T}\ell^{(\alpha)}\geq0
 \text{ whenever }M^{\mathsf T}z\geq0.
\end{equation}
Holding for every $\alpha$, these scalar inequalities are equivalent
to the operator inequality in
Equation~\eqref{eq:operator-valued-farkas}.
\end{proof}

\begin{revision}
The positivity criterion above enforces positive edge operators.
To represent actual minimum-cut areas in the selected chamber, the
sectorwise edge coefficients must also belong to that chamber.  If
$W_e=\sum_\alpha w_e^{(\alpha)}P_\alpha$, this requires
\begin{equation}
 w^{(\alpha)}\in\mathcal K_\sigma,
 \quad\text{equivalently}\quad
 (c_A^\sigma-c)^{\mathsf T}w^{(\alpha)}\leq0
 \quad\text{for every }A,c,\alpha,
 \label{eq:operator-valued-chamber-feasibility}
\end{equation}
in addition to nonnegative weights.  Strict inequalities and positive
weights give the regular chamber.  These are the coefficientwise
constraints of Section~\ref{subsec:finite-and-degenerate}, and for small
variations around a matched regular background they follow from the
chamber-radius bound.  A nonzero left-kernel witness already excludes
sewing before these additional feasibility constraints are imposed.
\end{revision}

The theorem separates two logically independent requirements.
Exact operator-algebra recovery may supply each region with a valid
central area operator, but locality requires the entire operator-valued
family to factor through one set of edge operators.  A single
operator-valued left-kernel violation prevents such sewing even though
every regional algebraic RT formula is exact.

Two different ambiguities occur. A nonzero element of $\ker M$ is a local edge deformation that is invisible to the chosen boundary regions. The resulting ambiguity concerns identifiability of the metric representative. A nonzero element of $\ker M^{\mathsf T}$ is instead a linear relation among boundary variations that every local edge model must obey. A nonzero value excludes a metric realization in the stated class. Enlarging $\calI$ can remove the first ambiguity by increasing the rank of $M$, while at the same time creating more independent tests in the second space.

\subsection{Finite perturbations}
\label{subsec:finite-and-degenerate}

The linear criterion in Theorem~\ref{thm:regular-geometrizability} is complete only inside the chamber ball. A finite data variation can lie in $\operatorname{im}M$ and nevertheless require an edge vector that makes a competing cut lighter. The chamber constraints must then be imposed together with the data equations.

For a pattern $\sigma$ and a background $w_0\in\cK_\sigma$, collect all chamber inequalities into a matrix $B_\sigma$. Its rows are
\begin{equation}
  -\bfe_e^{\mathsf T}
  \quad(e\in E),
  \qquad
  (c_A^\sigma-c)^{\mathsf T}
  \quad(A\in\calI,\ c\in\cA_A\setminus\{c_A^\sigma\}).
  \label{eq:B-rows}
\end{equation}
Define the nonnegative slack vector
\begin{equation}
  d_\sigma(w_0)=-B_\sigma w_0\geq0.
  \label{eq:slack-vector}
\end{equation}
An edge perturbation remains in the selected closed chamber if and only if
\begin{equation}
  h\in\cP_\sigma(w_0)
  :=\{h\in\mathbb R^E:B_\sigma h\leq d_\sigma(w_0)\}.
  \label{eq:perturbation-polyhedron}
\end{equation}
The associated finite data-variation set is
\begin{equation}
  \cZ_\sigma(w_0)=M_\sigma\cP_\sigma(w_0).
  \label{eq:finite-data-polyhedron}
\end{equation}

The affine Farkas alternative and the Minkowski--Weyl theorem in the form used
below are standard\rev{, see for example} Refs.~\cite{Rockafellar1970,Ziegler1995}.

\begin{theorem}[Finite same-chamber criterion]
\label{thm:farkas-finite}
The set $\cZ_\sigma(w_0)$ is a closed convex polyhedron. A vector $q\in\mathbb R^{\calI}$ belongs to it if and only if
\begin{equation}
  y^{\mathsf T}q+\lambda^{\mathsf T}d_\sigma(w_0)\geq0
  \label{eq:farkas-inequality}
\end{equation}
for every pair $(y,\lambda)$ satisfying
\begin{equation}
  M_\sigma^{\mathsf T}y+B_\sigma^{\mathsf T}\lambda=0,
  \qquad
  \lambda\geq0.
  \label{eq:farkas-dual-feasible}
\end{equation}
If $q\notin\cZ_\sigma(w_0)$, there exists such a pair with strict inequality in the opposite direction,
\begin{equation}
  y^{\mathsf T}q+\lambda^{\mathsf T}d_\sigma(w_0)<0.
  \label{eq:farkas-strict-certificate}
\end{equation}
The cone of dual pairs in Equation~\eqref{eq:farkas-dual-feasible} is polyhedral. It therefore suffices to require equality in Equation~\eqref{eq:farkas-inequality} on a basis of its lineality space\rev{, the largest linear subspace contained in the cone,} and nonnegativity on lifts of the extreme rays of its pointed quotient.
\end{theorem}

\begin{proof}
The set $\cP_\sigma(w_0)$ is a nonempty polyhedron because it contains
$h=0$.  Its linear image $\cZ_\sigma(w_0)$ is therefore a closed
convex polyhedron.  If $q=M_\sigma h$ and
$B_\sigma h\leq d_\sigma(w_0)$, then every pair obeying
Equation~\eqref{eq:farkas-dual-feasible} satisfies
\begin{equation}
 \begin{split}
 y^{\mathsf T}q+\lambda^{\mathsf T}d_\sigma(w_0)
 &=-\lambda^{\mathsf T}B_\sigma h
   +\lambda^{\mathsf T}d_\sigma(w_0)\\
 &=\lambda^{\mathsf T}
   \bigl(d_\sigma(w_0)-B_\sigma h\bigr)\geq0.
 \end{split}
\end{equation}
Conversely, the affine Farkas alternative applied to
$M_\sigma h=q$, $B_\sigma h\leq d_\sigma(w_0)$ says that either this
system is feasible or there exists a dual-feasible $(y,\lambda)$ for
which Equation~\eqref{eq:farkas-strict-certificate} holds.  This proves
necessity, sufficiency, and strict separation.  Finally, the dual set is
a polyhedral cone.  Splitting off its lineality space and applying the
Minkowski--Weyl theorem to the pointed quotient gives the stated finite
test by a lineality basis and extreme rays.
\end{proof}

The lineality sector $\lambda=0$ recovers the left-kernel equalities of
Theorem~\ref{thm:regular-geometrizability}.  Nonzero chamber
multipliers record the finite margin available before a competing cut
becomes lighter.  Equation~\eqref{eq:farkas-strict-certificate} thereby
distinguishes failure of a common local source from an extremal-surface
transition and exhausts the chamber obstruction at the regular
backgrounds used in the code constructions.

\subsection{The four-terminal model and the continuum interpretation}
\label{subsec:star-example}

A four-terminal star already displays the chamber, reconstruction, and witness structures. Let the central vertex be $o$, let the boundary vertices be $1,2,3,4$, and let edge $e_i$ join $o$ to $i$. For every region $A$, the two possible cut weights are
\begin{equation}
  a_A^G(w)=\min\left\{\sum_{i\in A}w_i,\ \sum_{i\notin A}w_i\right\}.
  \label{eq:star-area}
\end{equation}
Take
\begin{equation}
  w_0=\left(1,\frac65,3,\frac{17}{5}\right)
  \label{eq:star-background}
\end{equation}
and the region family
\begin{equation}
  \calI=(1,2,3,4,12,23).
  \label{eq:star-regions}
\end{equation}
For these weights, the cut that isolates the listed region is uniquely minimal in every case. In the stated ordering,
\begin{equation}
  M=
  \begin{pmatrix}
  1&0&0&0\\
  0&1&0&0\\
  0&0&1&0\\
  0&0&0&1\\
  1&1&0&0\\
  0&1&1&0
  \end{pmatrix}.
  \label{eq:star-cut-matrix}
\end{equation}
The matrix has rank $4$ and singular values
\begin{equation}
  2,\quad\sqrt2,\quad1,\quad1.
  \label{eq:star-singular-values}
\end{equation}
Its left kernel is generated over $\mathbb Z$ by the primitive witnesses
\begin{equation}
\begin{split}
  y^{12}&=(-1,-1,0,0,1,0)^{\mathsf T},\\
  y^{23}&=(0,-1,-1,0,0,1)^{\mathsf T}.
  \label{eq:star-witnesses}
\end{split}
\end{equation}
Hence every geometric variation in this chamber obeys
\begin{equation}
  q_{12}=q_1+q_2,
  \qquad
  q_{23}=q_2+q_3.
  \label{eq:star-relations}
\end{equation}
Equation~\eqref{eq:star-relations} gives local geometric consistency relations for this fixed code background\rev{, and} universal entropy inequalities play no role in their derivation.

For the $\ell^\infty$ norm on edge space, Proposition~\ref{prop:chamber-radius} gives
\begin{equation}
  \rho_\infty(w_0;\sigma)=\frac1{20}.
  \label{eq:star-radius}
\end{equation}
The nearest wall is the exchange of the two cuts for region $23$. Its background gap is
\begin{equation}
  (w_1+w_4)-(w_2+w_3)=\frac15,
\end{equation}
and the corresponding normal has $\ell^1$ norm $4$, producing the distance $1/20$.

Consider first the local edge response
\begin{equation}
  h=\left(\frac1{50},-\frac1{100},\frac3{100},0\right).
  \label{eq:star-geometric-h}
\end{equation}
It lies strictly inside the chamber ball and produces
\begin{equation}
  q=Mh=
  \left(
  \frac1{50},-\frac1{100},\frac3{100},0,
  \frac1{100},\frac1{50}
  \right)^{\mathsf T}.
  \label{eq:star-geometric-q}
\end{equation}
Both relations in Equation~\eqref{eq:star-relations} are satisfied.

Now change only the $12$ component and define
\begin{equation}
  \widetilde q=
  \left(
  \frac1{50},-\frac1{100},\frac3{100},0,
  \frac3{100},\frac1{50}
  \right)^{\mathsf T}.
  \label{eq:star-nongeometric-q}
\end{equation}
Then $(y^{12})^{\mathsf T}\widetilde q=1/50$, so no local edge deformation can produce these data. The Euclidean least-squares reconstruction is
\begin{equation}
  \widehat h=M^+\widetilde q
  =\left(\frac{11}{400},-\frac1{200},\frac{11}{400},0\right)^{\mathsf T},
  \label{eq:star-best-h}
\end{equation}
which remains inside the chamber ball. The residual is
\begin{equation}
  r=\widetilde q-M\widehat h
  =\frac1{400}(-3,-2,1,0,3,-1)^{\mathsf T},
  \label{eq:star-residual}
\end{equation}
with
\begin{equation}
  D_G^{(1)}(\widetilde q)=\frac{\sqrt6}{200},
  \qquad
  y_\star=\frac1{2\sqrt6}(-3,-2,1,0,3,-1)^{\mathsf T}.
  \label{eq:star-distance-witness}
\end{equation}
The witness is an exact element of $\ker M^{\mathsf T}$ and attains $y_\star^{\mathsf T}\widetilde q=D_G^{(1)}(\widetilde q)$.

Figure~\ref{fig:star} shows the background graph and the selected cuts responsible for the two primitive relations, illustrating how interior locality enforces relations among separated boundary labels.

\begin{figure}[t]
\centering
\begin{tikzpicture}[font=\small,>=Latex]
  \def\sx{0.92}
  \begin{scope}[xshift=0cm]
    \node at (0,2.15) {\textbf{a}};
    \coordinate (o) at (0,0.9);
    \coordinate (v1) at (-1.15,1.75);
    \coordinate (v2) at (1.15,1.75);
    \coordinate (v3) at (1.15,0.05);
    \coordinate (v4) at (-1.15,0.05);
    \draw (o)--node[pos=0.55,above left,fill=white,inner sep=1pt] {$1$} (v1);
    \draw (o)--node[pos=0.55,above right,fill=white,inner sep=1pt] {$6/5$} (v2);
    \draw (o)--node[pos=0.55,below right,fill=white,inner sep=1pt] {$3$} (v3);
    \draw (o)--node[pos=0.55,below left,fill=white,inner sep=1pt] {$17/5$} (v4);
    \fill (o) circle (2pt) node[below=2pt] {$o$};
    \fill (v1) circle (2pt) node[above left=1pt] {$1$};
    \fill (v2) circle (2pt) node[above right=1pt] {$2$};
    \fill (v3) circle (2pt) node[below right=1pt] {$3$};
    \fill (v4) circle (2pt) node[below left=1pt] {$4$};
    \node at (0,-0.35) {background weights};
  \end{scope}
  \begin{scope}[xshift=4.3cm]
    \node at (0,2.15) {\textbf{b}};
    \coordinate (o) at (0,0.9);
    \coordinate (v1) at (-1.15,1.75);
    \coordinate (v2) at (1.15,1.75);
    \coordinate (v3) at (1.15,0.05);
    \coordinate (v4) at (-1.15,0.05);
    \draw[line width=1.25pt] (o)--(v1);
    \draw[line width=1.25pt] (o)--(v2);
    \draw (o)--(v3);
    \draw (o)--(v4);
    \fill (o) circle (2pt);
    \fill (v1) circle (2pt) node[above left=1pt] {$1$};
    \fill (v2) circle (2pt) node[above right=1pt] {$2$};
    \fill (v3) circle (2pt) node[below right=1pt] {$3$};
    \fill (v4) circle (2pt) node[below left=1pt] {$4$};
    \draw[dashed,rounded corners] (-1.45,2.03) rectangle (1.45,1.42);
    \node at (0,-0.35) {$A=12$};
  \end{scope}
  \begin{scope}[xshift=8.6cm]
    \node at (0,2.15) {\textbf{c}};
    \coordinate (o) at (0,0.9);
    \coordinate (v1) at (-1.15,1.75);
    \coordinate (v2) at (1.15,1.75);
    \coordinate (v3) at (1.15,0.05);
    \coordinate (v4) at (-1.15,0.05);
    \draw (o)--(v1);
    \draw[line width=1.25pt] (o)--(v2);
    \draw[line width=1.25pt] (o)--(v3);
    \draw (o)--(v4);
    \fill (o) circle (2pt);
    \fill (v1) circle (2pt) node[above left=1pt] {$1$};
    \fill (v2) circle (2pt) node[above right=1pt] {$2$};
    \fill (v3) circle (2pt) node[below right=1pt] {$3$};
    \fill (v4) circle (2pt) node[below left=1pt] {$4$};
    \draw[dashed,rounded corners] (0.83,2.03) rectangle (1.47,-0.23);
    \node at (0,-0.35) {$A=23$};
  \end{scope}
\end{tikzpicture}
\caption{Fixed four-terminal geometry and local consistency relations. \textbf{(a)} Background weights from Equation~\eqref{eq:star-background}. \textbf{(b,c)} The heavy edges are the selected cuts for regions $12$ and $23$, which imply $q_{12}=q_1+q_2$ and $q_{23}=q_2+q_3$. Violating either identity rules out a local edge-weight realization in this chamber.}
\label{fig:star}
\end{figure}
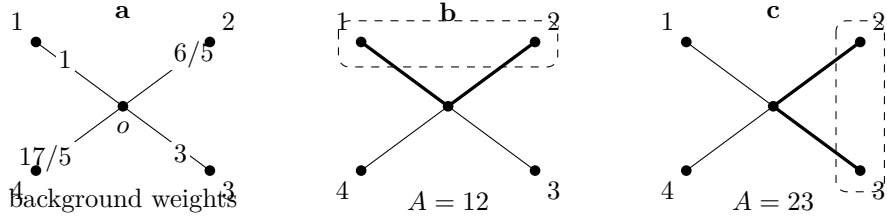

\section{Linearized geometry on an asymptotically \texorpdfstring{$\mathrm{AdS}_3$}{AdS3} slice}
\label{sec:linearized-geometry}

Section~\ref{sec:discrete} asks whether one set of edge weights explains every sampled interval. On the hyperbolic disk, the unknown is a symmetric two-tensor $k$, and each boundary interval probes it through a geodesic integral. The complete first-order response must therefore lie in the range of a single tensor X-ray transform\rev{, and} its orthogonal complement is the continuum counterpart of the finite left-null witness space.

We work on a time-reflection-symmetric slice of unit-radius $\mathrm{AdS}_3$. Its reference metric is the hyperbolic metric on the disk. This restriction isolates the kinematic question relevant to proto-area\rev{, namely} whether a first-order state-dependent response can be represented by a local metric perturbation. It does not impose the Hamiltonian constraint and therefore does not yet test whether the reconstructed perturbation is dynamically generated by a specified bulk stress tensor.

Tensor tomography on compact simple surfaces identifies potential tensors as the diffeomorphism gauge and proves injectivity on suitable representatives~\cite{PaternainSaloUhlmann2013,PaternainSaloUhlmann2023}. In the asymptotically hyperbolic setting, renormalized lengths and weighted function spaces control the ideal boundary~\cite{GrahamEtAl2019,Eptaminitakis2022}, while blow-up methods resolve endpoint singularities~\cite{MazzeoMonard2024}. Eptaminitakis, Monard, and Zou established the iterated transverse-traceless decomposition and the range characterization for even tensors on the Poincar\'e disk~\cite{EptaminitakisMonardZou2025}\rev{, and} the scalar range and mapping results appear in Ref.~\cite{EptaminitakisMonardZou2026}. Applied to proto-area, these theorems identify the closed geometric range and its orthogonal complement. The latter gives the nongeometric distance and optimal witnesses, while inversion on the former gives stable metric reconstruction. Symmetry-reduced HRT inversion instead uses a stationary radial ansatz and Abel transforms~\cite{Chae2026}\rev{, whereas} no radial or rotational symmetry is imposed here.

Boundary rigidity within a fixed bulk conformal class addresses the complementary
uniqueness problem for asymptotically hyperbolic metrics
\cite{HumbertMunozThon2026}.  The rank-two range theorem of
Ref.~\cite{EptaminitakisMonardZou2025}, together with the scalar mapping results
of Ref.~\cite{EptaminitakisMonardZou2026}, is needed to make the complete
Poincar\'e-disk membership condition and its reconstruction operator explicit\rev{, whereas}
the finite sewing theory and the first and second variation formulas for
boundary length do not depend on these results.

\subsection{Renormalized length and physical interval data}
\label{subsec:renormalized-data}

Let $\overline M=\overline{\mathbb D}$ and $M=\mathbb D$. In the complex coordinate $z=re^{i\theta}$, set
\begin{equation}
  g_{\mathbb H}=\frac{4\,|dz|^2}{(1-|z|^2)^2},
  \qquad
  x(z)=\frac{1-|z|^2}{1+|z|^2}.
  \label{eq:hyperbolic-metric-bdf}
\end{equation}
The function $x$ is a smooth boundary defining function, normalized by
\begin{equation}
  \left.x^2g_{\mathbb H}\right|_{T\partial M}=d\theta^2,
  \qquad
  \left.|dx|_{x^2g_{\mathbb H}}\right|_{\partial M}=1.
  \label{eq:round-conformal-infinity}
\end{equation}
It is the defining function used in the weighted tensor tomography below. The geodesic defining function associated with the same boundary representative is, in a collar of $\partial M$,
\begin{equation}
  \rho=2\frac{1-r}{1+r},
  \qquad
  g_{\mathbb H}=\rho^{-2}
  \left[d\rho^2+\left(1-\frac{\rho^2}{4}\right)^2d\theta^2\right].
  \label{eq:hyperbolic-geodesic-normal-form}
\end{equation}
For every sufficiently regular asymptotically hyperbolic metric $g$ close to $g_{\mathbb H}$ and with the same conformal representative, its associated geodesic defining function $\rho_g$ gives, after a boundary-fixing collar identification,
\begin{equation}
  g=\rho_g^{-2}\bigl(d\rho_g^2+h^g_{\rho_g}(\theta)d\theta^2\bigr),
  \qquad h^g_0(\theta)=1.
  \label{eq:ah-normal-form}
\end{equation}
The ratio of any two such defining functions is smooth and strictly positive up to the boundary, so the weighted spaces defined with $x$ are unchanged up to equivalent norms. The analysis is local near $g_{\mathbb H}$ and is restricted to simple metrics, so each ordered pair of distinct ideal endpoints determines a unique geodesic. Negative curvature and sufficiently small compactly supported perturbations provide a concrete open class with this property.

For $u,v\in S^1$ with $u\neq v$, let $\gamma^g_{u,v}$ be the corresponding unoriented $g$-geodesic. The renormalized boundary length relative to $x$ is
\begin{equation}
  \mathcal L_g^x(u,v)
  =\lim_{\varepsilon\downarrow0}
  \left[
    \operatorname{Length}_g\bigl(\gamma^g_{u,v}\cap\{x\geq\varepsilon\}\bigr)
    +2\log\varepsilon
  \right].
  \label{eq:renormalized-boundary-length}
\end{equation}
This limit is smooth on $S^1\times S^1\setminus\operatorname{diag}$ for simple asymptotically hyperbolic metrics~\cite{GrahamEtAl2019}. The subtraction fixes the background normalization required for comparing state-dependent variations. Each end of the geodesic contributes $-\log\varepsilon$, while the finite remainder carries the geometric information.

The defining function is a renormalization convention tied to the boundary metric. Let
\begin{equation}
  \mathscr E
  =\left\{(u,v)\longmapsto\omega(u)+\omega(v):
  \omega\in C^\infty(S^1)\right\}
  \label{eq:endpoint-gauge-space}
\end{equation}
be the endpoint subspace.

\begin{revision}
\begin{proposition}[Endpoint gauge and the background-subtracted map]
\label{prop:endpoint-gauge}
If $\widehat x=e^\omega x$ is another boundary defining function, with $\omega\in C^\infty(\overline M)$, then
\begin{equation}
  \mathcal L_g^{\widehat x}(u,v)-\mathcal L_g^x(u,v)
  =\omega(u)+\omega(v).
  \label{eq:bdf-length-shift}
\end{equation}
Consequently,
\begin{equation}
  \mathcal B(g)
  :=\mathcal L_g^x-\mathcal L_{g_{\mathbb H}}^x
  \in
  C^\infty_{\mathrm{sym}}
  \bigl(S^1\times S^1\setminus\operatorname{diag}\bigr)
  \label{eq:boundary-length-map}
\end{equation}
independent of the common choice of defining function. It is also invariant under boundary-fixing diffeomorphisms $\Psi$ that preserve the fixed conformal representative, with $(x\circ\Psi)/x\to1$ at $\partial M$.
\end{proposition}
\end{revision}

\begin{revision}
\begin{proof}
Equation~\eqref{eq:bdf-length-shift} is the transformation law for the renormalized length, with $\omega(u)$ and $\omega(v)$ denoting the boundary values of $\omega$. Near the incoming endpoint, replacing the cutoff $x=\varepsilon$ by $\widehat x=\varepsilon$ shifts the logarithmic radial coordinate by $\omega(u)$, and the outgoing endpoint contributes $\omega(v)$. A complete derivation using the meromorphic regularization of the geodesic integral is given in~\cite{GrahamEtAl2019}. Applying the same change of defining function to $g$ and $g_{\mathbb H}$ cancels the endpoint term in their difference. For the stated class of $\Psi$, the ratio $(x\circ\Psi^{-1})/x$ has boundary value one. The change of defining function therefore contributes no endpoint term. Pullback by $\Psi$ also preserves the ideal endpoint labels, proving invariance. Fixing the boundary points alone would not suffice, because a nonunit normal derivative can change the representative of conformal infinity.
\end{proof}
\end{revision}

The common endpoint term cancels after background subtraction.  We nevertheless keep the boundary representative and defining function fixed when the logical state varies.  Allowing a state-dependent change of defining function would add $\dot\omega(u)+\dot\omega(v)$ before background subtraction and would confuse a boundary Weyl redefinition with bulk backreaction.

\begin{revision}
We use the geodesic coordinates, orientation convention, and $L^2$
measure of Ref.~\cite{EptaminitakisMonardZou2025}, with its transverse
coordinate $a$ renamed $b$. Here and below, theorem and equation
numbers in that paper refer to arXiv:2510.04144v2. The geodesic
coordinates, the measure on the space of geodesics, and the antipodal
scattering relation are introduced in its Section 2.2, and the geodesic
parametrization is its Equation (4.1). The scalar mapping results are
cited in the numbering of
arXiv:2405.02521v4~\cite{EptaminitakisMonardZou2026}.
\end{revision}

The interval data can be identified with functions on the space of oriented hyperbolic geodesics
\begin{equation}
  \mathcal G_{\mathbb H}
  =\bigl(\mathbb R/2\pi\mathbb Z\bigr)_\beta\times\mathbb R_b.
  \label{eq:geodesic-space}
\end{equation}
A unit-speed geodesic is
\begin{equation}
  z_{\beta,b}(t)
  =e^{i\beta}
  \frac{(2+ib)\tanh(t/2)+ib}
       {ib\tanh(t/2)-2+ib},
  \qquad t\in\mathbb R,
  \label{eq:hyperbolic-geodesic-parametrization}
\end{equation}
with ordered endpoints
\begin{equation}
  u=\beta,
  \qquad
  v=\beta+\pi+2\arctan b
  \pmod{2\pi}.
  \label{eq:endpoint-geodesic-coordinates}
\end{equation}
Reversing the orientation of a geodesic is represented by the antipodal scattering involution
\begin{equation}
  S_A(\beta,b)
  =\bigl(\beta+\pi+2\arctan b,-b\bigr).
  \label{eq:antipodal-scattering}
\end{equation}
Thus a symmetric interval function $a(u,v)=a(v,u)$ pulls back to an even geodesic function
\begin{equation}
  (\mathcal Ua)(\beta,b)
  =a\bigl(\beta,\beta+\pi+2\arctan b\bigr),
  \qquad
  S_A^*\mathcal Ua=\mathcal Ua.
  \label{eq:interval-to-geodesic-data}
\end{equation}
The invariant data measure is $d\beta\,db$. We write
\begin{equation}
  L^2_+(\mathcal G_{\mathbb H})
  =\ker\bigl(\operatorname{Id}-S_A^*\bigr)
  \subset L^2(\mathcal G_{\mathbb H},d\beta\,db).
  \label{eq:even-data-space}
\end{equation}
Interval data and their pullbacks by $\mathcal U$ will be denoted by the same symbol. The complement-odd component
\begin{equation}
  a_{-}=\frac12\bigl(a-S_A^*a\bigr)
  \label{eq:complement-odd-data}
\end{equation}
can never be produced by an even-rank tensor and is therefore an immediate nongeometric obstruction. The complement-symmetric sector $a_-=0$ is appropriate to a disk geometry with one unoriented geodesic for each endpoint pair. An odd component contributes an independent first-order defect through its norm.

\subsection{The differential of the boundary-length map}
\label{subsec:first-variation}

For a symmetric two-tensor $k$, its geodesic X-ray transform on the reference disk is
\begin{equation}
  I_2k(\beta,b)
  =\int_{-\infty}^{\infty}
  k_{z_{\beta,b}(t)}\bigl(\dot z_{\beta,b}(t),\dot z_{\beta,b}(t)\bigr)\,dt.
  \label{eq:tensor-xray-transform}
\end{equation}
The integral is absolutely convergent for compactly supported tensors and for the positively weighted classes defined in Subsection~\ref{subsec:gauge-reduction}. The standard first variation of renormalized boundary length is recalled below with the normalization used here\rev{, and may be compared with} Section 4.3 of Ref.~\cite{GrahamEtAl2019}.

\begin{theorem}[Linearization of renormalized boundary length]
\label{thm:first-variation-boundary-length}
Let $g_t$ be a $C^2$ family of simple asymptotically hyperbolic metrics with $g_0=g_{\mathbb H}$ and fixed conformal representative $d\theta^2$. Suppose
\begin{equation}
  k=\left.\frac{d}{dt}g_t\right|_{t=0}
  \label{eq:metric-variation-k}
\end{equation}
is smooth and either compactly supported or satisfies, in the zero-cotangent frame near $x=0$ \rev{(the frame $dx/x$, $d\theta/x$ of unit $g_{\mathbb H}$-length, introduced in Section~\ref{subsec:gauge-reduction})},
\begin{equation}
  \nabla_{g_{\mathbb H}}^j k=O(x^{1+\eta}),
  \qquad 0\leq j\leq2,
  \qquad \eta>0.
  \label{eq:first-variation-decay}
\end{equation}
Then
\begin{equation}
  D\mathcal B_{g_{\mathbb H}}[k]
  =\frac12 I_2k.
  \label{eq:boundary-map-differential}
\end{equation}
The identity extends by continuity to the weighted Hilbert source space defined in Equation~\eqref{eq:source-quotient}.
\end{theorem}

\begin{proof}
Fix distinct endpoints $(u,v)$ and let $\gamma=\gamma^{g_{\mathbb H}}_{u,v}$. For sufficiently small $\varepsilon$, choose points $p_\varepsilon,q_\varepsilon\in\gamma$ approaching $u,v$ and satisfying
\begin{equation}
  x(p_\varepsilon)=x(q_\varepsilon)=\varepsilon.
  \label{eq:truncated-endpoints}
\end{equation}
The smooth extension property of the renormalized distance gives, uniformly for $t$ near zero,
\begin{equation}
  \mathcal L_{g_t}^x(u,v)
  =\lim_{\varepsilon\downarrow0}
  \left[d_{g_t}(p_\varepsilon,q_\varepsilon)+2\log\varepsilon\right].
  \label{eq:renorm-via-fixed-points}
\end{equation}
The approach of $p_\varepsilon$ and $q_\varepsilon$ to the boundary is arbitrary\rev{, and} it need not follow a $g_t$-geodesic. This freedom is useful because the finite endpoints in Equation~\eqref{eq:renorm-via-fixed-points} do not move with $t$.

Let $\eta_{t,\varepsilon}$ be the unique $g_t$-geodesic from $p_\varepsilon$ to $q_\varepsilon$. At $t=0$, it is the segment of $\gamma$ between these points. Parametrize $\eta_{0,\varepsilon}$ by $g_{\mathbb H}$-arclength $s$ and write $T=\partial_s\eta_{0,\varepsilon}$. Differentiating the length functional gives
\begin{equation}
\begin{split}
  \left.\frac{d}{dt}\right|_{t=0}
  d_{g_t}(p_\varepsilon,q_\varepsilon)
  ={}&\frac12\int_{\eta_{0,\varepsilon}}k(T,T)\,ds\\
  &+\left[g_{\mathbb H}(J,T)\right]_{p_\varepsilon}^{q_\varepsilon}
  -\int_{\eta_{0,\varepsilon}}
  g_{\mathbb H}(J,\nabla_TT)\,ds,
  \label{eq:finite-first-variation}
\end{split}
\end{equation}
where $J=\partial_t\eta_{t,\varepsilon}|_{t=0}$. The endpoint term vanishes because $p_\varepsilon$ and $q_\varepsilon$ are fixed, and the final integral vanishes because $\eta_{0,\varepsilon}$ is a geodesic. Hence
\begin{equation}
  \left.\frac{d}{dt}\right|_{t=0}
  d_{g_t}(p_\varepsilon,q_\varepsilon)
  =\frac12\int_{\eta_{0,\varepsilon}}k(T,T)\,ds.
  \label{eq:finite-first-variation-reduced}
\end{equation}

The counterterm in Equation~\eqref{eq:renorm-via-fixed-points} is independent of $t$. The renormalized distance extends smoothly to the boundary off the diagonal, so its $t$-derivative is obtained by differentiating before taking $\varepsilon\downarrow0$. Under Equation~\eqref{eq:first-variation-decay}, the integrand is $O(x^{1+\eta})$ in a unit zero-frame, while $ds$ is asymptotic to $|dx|/x$ at either end. It is therefore dominated by an integrable multiple of $x^\eta|dx|$. Dominated convergence in Equation~\eqref{eq:finite-first-variation-reduced} yields
\begin{equation}
  \left.\frac{d}{dt}\right|_{t=0}\mathcal L_{g_t}^x(u,v)
  =\frac12\int_\gamma k(T,T)\,ds
  =\frac12 I_2k(u,v).
  \label{eq:renorm-first-variation-result}
\end{equation}
The background term in $\mathcal B$ is fixed, proving Equation~\eqref{eq:boundary-map-differential} on the smooth class. The weighted extension follows from density and the boundedness of the transform on the source spaces in~\cite{Eptaminitakis2022,EptaminitakisMonardZou2025}.
\end{proof}

The theorem has a direct local interpretation. The value $k(T,T)$ is the infinitesimal change of squared line element in the direction of the reference RT geodesic. The curve itself need not be varied at first order because the background length is stationary. Its displacement enters at second order and produces the quadratic consistency condition of Section~\ref{sec:nonlinear-consistency}.

Equation~\eqref{eq:boundary-map-differential} is invariant under metric perturbations generated by boundary-fixing coordinate changes.

\begin{revision}
\begin{corollary}[Diffeomorphism gauge]
\label{cor:diffeomorphism-gauge}
Let $\omega$ be a smooth one-form of positive boundary order, so that $\omega(T)\to0$ at both ends of every complete reference geodesic. Let $d^s\omega$ denote its symmetrized covariant derivative,
\begin{equation}
  (d^s\omega)_{ij}
  =\frac12\bigl(\nabla_i\omega_j+\nabla_j\omega_i\bigr).
  \label{eq:symmetrized-derivative}
\end{equation}
Then
\begin{equation}
  I_2(d^s\omega)=0.
  \label{eq:potential-kernel}
\end{equation}
The identity extends by density and continuity to
$\omega\in x^\delta H_0^1(M;{}^0T^*M)$ for every $\delta>0$. In
particular, let $X$ generate a flow in
$\operatorname{Diff}_0^x(M;\partial M)$ and suppose that $X^\flat$ has
positive boundary order, with the weighted regularity just stated. Then
$k=\mathcal L_Xg_{\mathbb H}=2d^s(X^\flat)$ satisfies
$D\mathcal B_{g_{\mathbb H}}[k]=0$. Vanishing of the ordinary coordinate
components of $X$ at the boundary is not by itself the required decay.
\end{corollary}
\end{revision}

\begin{proof}
Along a complete reference geodesic,
\begin{equation}
  (d^s\omega)(T,T)
  =(\nabla_T\omega)(T)
  =\frac{d}{ds}\bigl(\omega(T)\bigr),
  \label{eq:potential-total-derivative}
\end{equation}
where $\nabla_TT=0$ was used. Positive boundary decay makes the endpoint values vanish, so integration proves Equation~\eqref{eq:potential-kernel}. The weighted extension is Lemma~2.4 of \cite{EptaminitakisMonardZou2025}. It also follows by density from polyhomogeneous tensors of positive boundary order\rev{, that is, tensors with asymptotic expansions in powers of $x$ and $\log x$ at the boundary whose leading power is positive}.
\end{proof}

\subsection{Gauge reduction and the iterated transverse-traceless representative}
\label{subsec:gauge-reduction}

The correct analytic source space is expressed using the zero-cotangent bundle ${}^0T^*\overline M$, locally spanned by $dx/x$ and $d\theta/x$. For a vector bundle $E$ built from ${}^0T^*\overline M$, write
\begin{equation}
  x^\delta H_0^s(M;E)
  =\left\{f:x^{-\delta}f\in H_0^s(M;E)\right\}.
  \label{eq:weighted-zero-sobolev}
\end{equation}
Here $H_0^s$ is the Sobolev scale generated by zero-vector fields, and $H_0^0=L^2(M,dV_{g_{\mathbb H}})$. We use the distinguished weight $\delta=1/2$. Compactly supported smooth tensors lie in this space, as do normalizable perturbations that vanish sufficiently rapidly at conformal infinity.

Theorem 1 of Ref.~\cite{EptaminitakisMonardZou2025}, specialized to
$m=2$, $s=0$, and $\delta=1/2$, gives the following decomposition.

\begin{revision}
\begin{theorem}[Iterated transverse-traceless decomposition]
\label{thm:itt-decomposition}
Every
\begin{equation}
  k\in x^{1/2}L^2
  \bigl(M;S^2({}^0T^*M)\bigr)
  \label{eq:rank-two-weighted-source}
\end{equation}
can be written uniquely as
\begin{equation}
  k=d^s\omega+\phi g_{\mathbb H}+k^{\mathrm{tt}},
  \label{eq:rank-two-itt-decomposition}
\end{equation}
where
\begin{equation}
\begin{split}
  &\omega\in x^{1/2}H_0^1(M;{}^0T^*M),
  \qquad
  \phi\in x^{1/2}L^2(M),\\
  &k^{\mathrm{tt}}
  \in x^{1/2}L^2
  \bigl(M;S^2({}^0T^*M)\bigr),
  \qquad
  \delta_{g_{\mathbb H}}k^{\mathrm{tt}}=0,
  \qquad
  \operatorname{tr}_{g_{\mathbb H}}k^{\mathrm{tt}}=0.
  \label{eq:tt-conditions}
\end{split}
\end{equation}
Moreover, there is a constant $C_{\mathrm{itt}}$ such that
\begin{equation}
  \norm{\omega}_{x^{1/2}H_0^1}
  +\norm{\phi}_{x^{1/2}L^2}
  +\norm{k^{\mathrm{tt}}}_{x^{1/2}L^2}
  \leq C_{\mathrm{itt}}
  \norm{k}_{x^{1/2}L^2}.
  \label{eq:itt-decomposition-estimate}
\end{equation}
A polyhomogeneous tensor has an asymptotic expansion in powers of
$x$ and $\log x$, with smooth coefficients along the boundary. If $k$
has such an expansion of positive order $r$, the decomposition is again
polyhomogeneous, with each summand of order $r'$ for every
$0<r'<\min\{r,2\}$, as follows from Theorem 1 of
Ref.~\cite{EptaminitakisMonardZou2025}. The loss at the indicial threshold\rev{, the boundary order at which
the decomposition operator has a formal solution with a logarithm,}
allows logarithmic terms and does not assert preservation of an
arbitrarily large boundary order.
\end{theorem}
\end{revision}

The first term in Equation~\eqref{eq:rank-two-itt-decomposition} is pure gauge. The remaining tensor
\begin{equation}
  k^{\mathrm{itt}}
  :=\phi g_{\mathbb H}+k^{\mathrm{tt}}
  \label{eq:itt-representative}
\end{equation}
will be called the iterated transverse-traceless \rev{(ITT)} representative. We denote the Hilbert space of these representatives by
\begin{equation}
  \mathcal X_{\mathrm{itt}}
  =\left\{\phi g_{\mathbb H}+k^{\mathrm{tt}}:
  \phi\in x^{1/2}L^2(M),\;
  k^{\mathrm{tt}}\in x^{1/2}L^2(S^2({}^0T^*M)),\;
  \delta_{g_{\mathbb H}} k^{\mathrm{tt}}
  =\operatorname{tr}_{g_{\mathbb H}}k^{\mathrm{tt}}=0\right\}.
  \label{eq:itt-source-space}
\end{equation}
In rank two, its two summands have a direct geometric meaning. The scalar $\phi$ changes the local line element isotropically, whereas $k^{\mathrm{tt}}$ is a divergence-free shear. Their X-ray data occupy orthogonal sectors on the Poincar\'e disk. This orthogonality is the principal advantage of the iterated transverse-traceless gauge over an arbitrary representative.

Define the quotient source space
\begin{equation}
  \mathcal X
  =\frac{x^{1/2}L^2
  \bigl(M;S^2({}^0T^*M)\bigr)}
  {d^s\bigl(x^{1/2}H_0^1(M;{}^0T^*M)\bigr)}
  \label{eq:source-quotient}
\end{equation}
with norm
\begin{equation}
  \norm{[k]}_{\mathcal X}^2
  =\norm{\phi}_{x^{1/2}L^2}^2
  +\norm{k^{\mathrm{tt}}}_{x^{1/2}L^2}^2,
  \label{eq:source-quotient-norm}
\end{equation}
where $k^{\mathrm{itt}}=\phi g_{\mathbb H}+k^{\mathrm{tt}}$ is the unique representative from Theorem~\ref{thm:itt-decomposition}. The estimate in Equation~\eqref{eq:itt-decomposition-estimate} makes this norm equivalent to the ordinary quotient norm and identifies $\mathcal X$ isomorphically with $\mathcal X_{\mathrm{itt}}$. Corollary~\ref{cor:diffeomorphism-gauge} gives
\begin{equation}
  I_2k=I_2k^{\mathrm{itt}}
  =I_0\phi+I_2k^{\mathrm{tt}},
  \label{eq:transform-itt-splitting}
\end{equation}
where $I_0\phi=\int_\gamma\phi\,ds$ and the identity $I_2(\phi g_{\mathbb H})=I_0\phi$ follows from $g_{\mathbb H}(T,T)=1$.

\subsection{Range conditions and stable reconstruction}
\label{subsec:range-reconstruction}

\begin{revision}
The range of $I_2$ is easiest to describe using a distinguished orthogonal decomposition of $L^2_+(\mathcal G_{\mathbb H})$. Following the data-space notation of Section 2.2 of Ref.~\cite{EptaminitakisMonardZou2025}, in particular its Theorem 2, let
\begin{equation}
  I_{p,q}=I_q\bigl[z^p(dz)^q\bigr],
  \qquad p\geq0,
  \qquad q\geq1,
  \label{eq:distinguished-tensor-data}
\end{equation}
where $(dz)^q$ denotes the symmetric $q$-fold product. Set
\begin{equation}
  \widehat I_{p,q}
  =\frac{I_{p,q}}{\norm{I_{p,q}}_{L^2(\mathcal G_{\mathbb H})}}.
  \label{eq:normalized-tensor-data}
\end{equation}
For $q\geq1$, define the closed sector
\begin{equation}
  \mathcal E_{2q}
  =\overline{\operatorname{span}}^{L^2}
  \left\{
    \widehat I_{p,2q},
    \overline{\widehat I_{p,2q}}:p\geq0
  \right\}.
  \label{eq:even-tensor-sector}
\end{equation}
The same data-space theorem gives the orthogonal decomposition
\begin{equation}
  L^2_+(\mathcal G_{\mathbb H})
  =\overline{I_0\bigl(x^{1/2}L^2(M)\bigr)}^{L^2}
  \mathbin{\widehat{\oplus}}
  \mathop{\widehat{\bigoplus}}_{q=1}^{\infty}\mathcal E_{2q}.
  \label{eq:full-even-data-decomposition}
\end{equation}
Let $\Pi_0$ and $\Pi_{2q}$ be the associated $L^2$-orthogonal projections.
\end{revision}

The closure in the scalar term is important. The actual transform range is generally not closed in the raw $L^2$ topology\rev{, because} spectral decay is part of the inverse problem. A stable metric reconstruction therefore requires a source-adapted topology rather than an unqualified $L^2$ projection.

\begin{revision}
For an orthonormal family $(e_p)_{p\geq0}$, define
\begin{equation}
  h^{1/2}(e_p)
  =\left\{
  \sum_{p=0}^{\infty}c_pe_p:
  \sum_{p=0}^{\infty}(p+1)|c_p|^2<\infty
  \right\}.
  \label{eq:h-half-space}
\end{equation}
Set
\begin{equation}
\begin{split}
  \mathcal H_0
  &:=I_0\bigl(x^{1/2}L^2(M)\bigr),\\
  \mathcal H_{2q}
  &:=h^{1/2}\bigl(\widehat I_{p,2q}:p\geq0\bigr)
  \mathbin{\widehat{\oplus}}
  h^{1/2}\bigl(\overline{\widehat I_{p,2q}}:p\geq0\bigr),
  \qquad q\geq1.
  \label{eq:sector-hilbert-spaces}
\end{split}
\end{equation}
The scalar space is equipped with the intrinsic norm
\begin{equation}
  \norm{I_0\phi}_{\mathcal H_0}
  :=\norm{\phi}_{x^{1/2}L^2(M)},
  \label{eq:scalar-intrinsic-norm}
\end{equation}
which is well defined by scalar injectivity. As a set, equivalently,
\begin{equation}
  \mathcal H_0
  =\overline{I_0\bigl(x^{1/2}L^2(M)\bigr)}^{L^2}
  \cap H_{T,D,+}^{1/2}(\mathcal G_{\mathbb H}),
  \label{eq:scalar-intrinsic-range}
\end{equation}
where $H_{T,D,+}^{1/2}$ is the Dirichlet anisotropic Sobolev space of~\cite{EptaminitakisMonardZou2026}. The scalar range identification is Theorem 2.25 of that reference, with Sobolev order $s=0$. Equation~\eqref{eq:scalar-intrinsic-range} separates the scalar moment conditions, encoded by the first factor, from the necessary spectral decay.
\end{revision}

\begin{revision}
Theorem 3 and Remark 2.14 of Ref.~\cite{EptaminitakisMonardZou2025},
specialized to tensor order $m=2$ and weight $\delta=1/2$, together
with Theorem 2.25 of Ref.~\cite{EptaminitakisMonardZou2026}, give the
complete rank-two range characterization in the normalization used here.
\end{revision}

\begin{lemma}[Rank-two range criterion of Eptaminitakis, Monard, and Zou]
\label{thm:rank-two-range}
A function $a\in L^2_+(\mathcal G_{\mathbb H})$ belongs to
\begin{equation}
  I_2\left(x^{1/2}L^2
  \bigl(M;S^2({}^0T^*M)\bigr)\right)
  \label{eq:rank-two-transform-range}
\end{equation}
if and only if all of the following conditions hold.
\begin{enumerate}[label=\textnormal{(\roman*)},leftmargin=2.4em]
  \item The higher tensor moments vanish,
  \begin{equation}
    \Pi_{2q}a=0
    \qquad\text{for every }q\geq2.
    \label{eq:higher-moment-vanishing}
  \end{equation}
  \item The rank-two sector has the spectral decay
  \begin{equation}
    \Pi_2a\in\mathcal H_2.
    \label{eq:rank-two-spectral-decay}
  \end{equation}
  \item The scalar sector obeys
  \begin{equation}
    \Pi_0a\in\mathcal H_0.
    \label{eq:scalar-spectral-decay}
  \end{equation}
\end{enumerate}
When these conditions hold, there is a unique iterated transverse-traceless representative $k^{\mathrm{itt}}=\phi g_{\mathbb H}+k^{\mathrm{tt}}$ satisfying $I_2k^{\mathrm{itt}}=a$.
\end{lemma}

\begin{proof}
The cited range characterization places $I_0\phi$ in the scalar sector and $I_2k^{\mathrm{tt}}$ in $\mathcal E_2$, with orthogonal images. The scalar mapping theory identifies the actual scalar range with $\mathcal H_0$~\cite{EptaminitakisMonardZou2026}, while the transverse-traceless transform is a homeomorphism onto $\mathcal H_2$~\cite{EptaminitakisMonardZou2025}. No component in $\mathcal E_{2q}$ with $q\geq2$ can be produced by a two-tensor. Conversely, conditions~\textnormal{(ii)} and~\textnormal{(iii)} reconstruct $k^{\mathrm{tt}}$ and $\phi$, respectively, and their sum gives the desired representative. Uniqueness follows from injectivity on the iterated transverse-traceless subspace.
\end{proof}

Condition~\eqref{eq:higher-moment-vanishing} admits a fully boundary-based form. Define
\begin{equation}
\begin{split}
  W_{p,q}^{+}(a)
  &=\ip{a}{\widehat I_{p,2q}}_{L^2(\mathcal G_{\mathbb H})},\\
  W_{p,q}^{-}(a)
  &=\ip{a}{\overline{\widehat I_{p,2q}}}_{L^2(\mathcal G_{\mathbb H})},
  \qquad p\geq0,\quad q\geq2.
  \label{eq:moment-witnesses}
\end{split}
\end{equation}
Then
\begin{equation}
  \Pi_{2q}a=0\ \text{for all }q\geq2
  \quad\Longleftrightarrow\quad
  W_{p,q}^{\pm}(a)=0\ \text{for all }p\geq0,\ q\geq2.
  \label{eq:moment-witness-equivalence}
\end{equation}
These are the continuum counterparts of the left-null relations $y^{\mathsf T}q=0$ in Section~\ref{sec:discrete}. Each witness combines intervals distributed around the boundary and annihilates every local rank-two metric response.

\begin{revision}
To obtain a closed geometric subspace, fix $\tau>1/2$ and define the regular data Hilbert space
\begin{equation}
  \mathcal Y^\tau
  =\mathcal H_0
  \mathbin{\widehat{\oplus}}
  \mathop{\widehat{\bigoplus}}_{q=1}^{\infty}
  \mathcal H_{2q},
  \label{eq:adapted-data-space}
\end{equation}
with norm
\begin{equation}
  \norm{a}_{\mathcal Y^\tau}^2
  =\norm{\Pi_0a}_{\mathcal H_0}^2
  +\sum_{q=1}^{\infty}(1+q)^{2\tau}
  \norm{\Pi_{2q}a}_{\mathcal H_{2q}}^2.
  \label{eq:adapted-data-norm}
\end{equation}
The real subspace, denoted $\mathcal Y^\tau_{\mathbb R}$, consists of real $S_A$-even data. Every metric response generated by $\mathcal X$ belongs to $\mathcal H_0\widehat\oplus\mathcal H_2$ and hence to $\mathcal Y^\tau$ for every finite $\tau$. For general code data, membership in $\mathcal Y^\tau$ is a substantive continuum regularity hypothesis, since it controls spectral decay inside each tensor sector and the tail over tensor order. No additional bulk field equation enters this kinematic criterion. The linear splitting itself only requires $\tau\geq0$, and we keep a fixed
$\tau>1/2$ throughout. Increasing $\tau$ strengthens the tail over tensor
order, while leaving the norm within each fixed sector unchanged.
Pointwise sampling and nonlinear Hessian estimates therefore require
the additional hypotheses stated below.
\end{revision}

The restriction to $\mathcal Y^\tau$ retains the low-sector conditions in Lemma~\ref{thm:rank-two-range}. A datum in raw $L^2_+$ that fails Equation~\eqref{eq:rank-two-spectral-decay} or Equation~\eqref{eq:scalar-spectral-decay} is already nongeometric, but the failure has no positive distance in the raw $L^2$ norm because the relevant transform ranges are not closed there. The adapted space selects the regular class on which the inverse problem is well posed\rev{, and} within that class the remaining obstruction is the closed higher-sector component.

Define
\begin{equation}
  \mathcal Y_{\mathrm{geo}}
  =\mathcal H_0\mathbin{\widehat{\oplus}}\mathcal H_2,
  \qquad
  \mathcal Y_{\mathrm{ng}}
  =\mathop{\widehat{\bigoplus}}_{q=2}^{\infty}\mathcal H_{2q}.
  \label{eq:geometric-nongeometric-splitting}
\end{equation}
They are closed orthogonal subspaces of $\mathcal Y^\tau$, with projections
\begin{equation}
  \Pi_{\mathrm{geo}}=\Pi_0+\Pi_2,
  \qquad
  \Pi_{\mathrm{ng}}=\sum_{q=2}^{\infty}\Pi_{2q}.
  \label{eq:geometric-data-projectors}
\end{equation}

\begin{revision}
\begin{theorem}[Linear proto-area geometrizability]
\label{thm:linear-geometrizability}
Let $a_1\in\mathcal Y^\tau_{\mathbb R}$ be the first derivative of the background-subtracted, length-normalized proto-area for all boundary intervals. The following statements are equivalent.
\begin{enumerate}[label=\textnormal{(\alph*)},leftmargin=2.4em]
  \item There exists a metric deformation class $[k]\in\mathcal X$ such that
  \begin{equation}
    a_1=D\mathcal B_{g_{\mathbb H}}[k].
    \label{eq:linear-geometric-realization}
  \end{equation}
  \item The nongeometric projection vanishes,
  \begin{equation}
    \Pi_{\mathrm{ng}}a_1=0.
    \label{eq:linear-projector-condition}
  \end{equation}
  \item Every higher moment witness vanishes,
  \begin{equation}
    W_{p,q}^{\pm}(a_1)=0
    \qquad\text{for all }p\geq0,\ q\geq2.
    \label{eq:all-witnesses-vanish}
  \end{equation}
\end{enumerate}
If these conditions hold, the unique iterated transverse-traceless representative is
\begin{equation}
  k_1^{\mathrm{itt}}
  =\mathscr R_2a_1,
  \qquad
  \mathscr R_2
  :=2\left(I_2\big|_{\mathcal X_{\mathrm{itt}}}\right)^{-1}
  \Pi_{\mathrm{geo}},
  \label{eq:metric-reconstruction-operator}
\end{equation}
where $\mathcal X_{\mathrm{itt}}$ is identified with the representatives
in Equation~\eqref{eq:itt-representative}. Suppose in addition that the
reconstructed tensor $k_1^{\mathrm{itt}}$ has smooth or polyhomogeneous
zero-cotensor coefficients satisfying
Equation~\eqref{eq:quadratic-decay-class}. Then
\begin{equation}
  g_t=g_{\mathbb H}+t k_1^{\mathrm{itt}}
  \label{eq:integrated-linear-metric-family}
\end{equation}
is, for sufficiently small $|t|$, a smooth or polyhomogeneous simple asymptotically hyperbolic metric family with conformal infinity $d\theta^2$.
\end{theorem}
\end{revision}

\begin{revision}
\begin{proof}
Theorem~\ref{thm:first-variation-boundary-length} identifies the differential of the boundary-length map with $I_2/2$. Lemma~\ref{thm:rank-two-range} then gives the equivalence of \textnormal{(a)} and \textnormal{(b)} on the regular data space. Equation~\eqref{eq:moment-witness-equivalence} gives the equivalence of \textnormal{(b)} and \textnormal{(c)}. The inverse in Equation~\eqref{eq:metric-reconstruction-operator} exists and is unique on the iterated transverse-traceless representative by Lemma~\ref{thm:rank-two-range}. The factor $2$ compensates for the factor $1/2$ in Equation~\eqref{eq:boundary-map-differential}. The additional regularity assumption is imposed on this reconstructed
tensor, rather than inferred from membership of $a_1$ in
$\mathcal Y^\tau$. It gives vanishing tangential compactification at
$x=0$, so Equation~\eqref{eq:integrated-linear-metric-family} preserves
the representative in Equation~\eqref{eq:round-conformal-infinity}. The compactified endomorphism $g_{\mathbb H}^{-1}k_1^{\mathrm{itt}}$ is bounded, so positivity follows for sufficiently small $|t|$. Simplicity persists after shrinking the parameter interval because the perturbation is small in the required $C^2$ topology.
\end{proof}
\end{revision}

\begin{corollary}[Continuum sewing of central area operators]
\label{cor:continuum-central-sewing}
Let $\mathcal Z$ be a finite-dimensional commutative
$C^*$-algebra and let
$L\in\mathcal Y^\tau_{\mathbb R}\otimes\mathcal Z_{\mathrm{sa}}$
be an operator-valued interval datum.  There is a common
operator-valued local metric perturbation
$K\in\mathcal X\otimes\mathcal Z_{\mathrm{sa}}$ such that
\begin{equation}
 L=\frac12(I_2\otimes\operatorname{Id}_{\mathcal Z})K
 \label{eq:continuum-operator-valued-xray}
\end{equation}
if and only if
\begin{equation}
 (\Pi_{\mathrm{ng}}\otimes
 \operatorname{Id}_{\mathcal Z})L=0.
 \label{eq:continuum-operator-valued-range}
\end{equation}
Equivalently, every central-sector coefficient of $L$ obeys all scalar
rank-two range conditions.  When the condition holds, the unique
iterated transverse-traceless representative is
$(\mathscr R_2\otimes\operatorname{Id}_{\mathcal Z})L$.
\end{corollary}

\begin{proof}
Expand $L$ in the minimal central projections of $\mathcal Z$ and
apply Theorem~\ref{thm:linear-geometrizability} coefficientwise.
\end{proof}

The inverse is quantitatively stable in the adapted topology.

\begin{theorem}[Stable reconstruction and best geometric approximation]
\label{thm:stable-continuum-reconstruction}
There are constants $0<c\leq C<\infty$ such that, for every $[k]\in\mathcal X$,
\begin{equation}
  c\norm{[k]}_{\mathcal X}
  \leq
  \norm{\frac12I_2k}_{\mathcal Y_{\mathrm{geo}}}
  \leq
  C\norm{[k]}_{\mathcal X}.
  \label{eq:continuum-stability-two-sided}
\end{equation}
Consequently,
\begin{equation}
  \norm{\mathscr R_2a-\mathscr R_2\widetilde a}_{\mathcal X}
  \leq c^{-1}
  \norm{\Pi_{\mathrm{geo}}(a-\widetilde a)}_{\mathcal Y^\tau}.
  \label{eq:continuum-reconstruction-stability}
\end{equation}
For arbitrary $a\in\mathcal Y^\tau$, the datum $\Pi_{\mathrm{geo}}a$ is the unique nearest geometric datum, and $\mathscr R_2a$ is its unique iterated transverse-traceless metric reconstruction.
\end{theorem}

\begin{proof}
The scalar transform is a homeomorphism from its weighted source space to $\mathcal H_0$~\cite{EptaminitakisMonardZou2026}\rev{, and} the transverse-traceless transform is a homeomorphism from $x^{1/2}L^2(S^2_{\mathrm{tt}})$ to $\mathcal H_2$~\cite{EptaminitakisMonardZou2025}. Their images are orthogonal, and the source norm in Equation~\eqref{eq:source-quotient-norm} is the corresponding direct-sum norm. The open mapping theorem gives Equation~\eqref{eq:continuum-stability-two-sided}. Applying the inverse to $\Pi_{\mathrm{geo}}(a-\widetilde a)$ gives Equation~\eqref{eq:continuum-reconstruction-stability}. Since Equation~\eqref{eq:geometric-nongeometric-splitting} is an orthogonal Hilbert decomposition, $\Pi_{\mathrm{geo}}a$ is the unique metric projection of $a$ onto $\mathcal Y_{\mathrm{geo}}$.
\end{proof}

\subsection{Nongeometric distance, finite witnesses, and code limits}
\label{subsec:continuum-witnesses}

The closed splitting in Equation~\eqref{eq:geometric-nongeometric-splitting} defines the linear nongeometric defect
\begin{equation}
  D_{\mathrm{geo}}^{(1)}(a)
  :=\operatorname{dist}_{\mathcal Y^\tau}
  \bigl(a,\mathcal Y_{\mathrm{geo}}\bigr)
  =\norm{\Pi_{\mathrm{ng}}a}_{\mathcal Y^\tau}.
  \label{eq:linear-nongeometric-defect}
\end{equation}
Writing the moment coefficients as in Equation~\eqref{eq:moment-witnesses}, one obtains the exact expansion
\begin{equation}
\begin{split}
  \bigl(D_{\mathrm{geo}}^{(1)}(a)\bigr)^2
  =\sum_{q=2}^{\infty}(1+q)^{2\tau}
  \sum_{p=0}^{\infty}(p+1)
  \left(
  \abs{W_{p,q}^{+}(a)}^2
  +\abs{W_{p,q}^{-}(a)}^2
  \right).
  \label{eq:defect-moment-expansion}
\end{split}
\end{equation}
The defect $D_{\mathrm{geo}}^{(1)}$ is the norm of the complete collection of forbidden tensor moments, rather than a test of one selected relation.

\begin{proposition}[Optimal witness and robustness]
\label{prop:optimal-continuum-witness}
Let $a\in\mathcal Y^\tau_{\mathbb R}$ and suppose $D_{\mathrm{geo}}^{(1)}(a)>0$. Then
\begin{equation}
  Y_\star(a)
  =\frac{\Pi_{\mathrm{ng}}a}
  {\norm{\Pi_{\mathrm{ng}}a}_{\mathcal Y^\tau}}
  \label{eq:optimal-continuum-witness-vector}
\end{equation}
annihilates every geometric datum and satisfies
\begin{equation}
  \ip{Y_\star(a)}{a}_{\mathcal Y^\tau}
  =D_{\mathrm{geo}}^{(1)}(a).
  \label{eq:optimal-continuum-witness-value}
\end{equation}
It is the unique unit vector in $\mathcal Y_{\mathrm{ng}}\cap\mathcal Y^\tau_{\mathbb R}$ maximizing the real pairing with $a$. Moreover, for every $\widetilde a\in\mathcal Y^\tau_{\mathbb R}$,
\begin{equation}
  \abs{D_{\mathrm{geo}}^{(1)}(a)
  -D_{\mathrm{geo}}^{(1)}(\widetilde a)}
  \leq\norm{a-\widetilde a}_{\mathcal Y^\tau}.
  \label{eq:defect-lipschitz}
\end{equation}
If measured data $a_{\mathrm{obs}}=a_{\mathrm{true}}+\eta$ obey
$\norm{\eta}_{\mathcal Y^\tau}\leq\varepsilon$ and
\begin{equation}
  D_{\mathrm{geo}}^{(1)}(a_{\mathrm{obs}})>\varepsilon,
  \label{eq:noise-certification-threshold}
\end{equation}
then $a_{\mathrm{true}}$ is not linearly geometrizable.
\end{proposition}

\begin{proof}
The vector in Equation~\eqref{eq:optimal-continuum-witness-vector} belongs to $\mathcal Y_{\mathrm{ng}}$, so it is orthogonal to $\mathcal Y_{\mathrm{geo}}$. Cauchy--Schwarz shows that no unit vector in $\mathcal Y_{\mathrm{ng}}$ has a larger pairing with $a$, and equality fixes the direction of the residual. The distance to a nonempty closed subset of a normed space is one-Lipschitz, giving Equation~\eqref{eq:defect-lipschitz}. Applying it to $a_{\mathrm{obs}}$ and $a_{\mathrm{true}}$ yields
\begin{equation}
  D_{\mathrm{geo}}^{(1)}(a_{\mathrm{true}})
  \geq D_{\mathrm{geo}}^{(1)}(a_{\mathrm{obs}})-\varepsilon>0.
\end{equation}
\end{proof}

Only finitely many moments can be resolved in a numerical or experimental calculation. For integers $P\geq0$ and $Q\geq2$, define
\begin{equation}
\begin{split}
  \bigl(D_{\mathrm{geo}}^{(1)}[P,Q](a)\bigr)^2
  ={}&\sum_{q=2}^{Q}(1+q)^{2\tau}
  \sum_{p=0}^{P}(p+1)\\
  &\times\left(
  \abs{W_{p,q}^{+}(a)}^2
  +\abs{W_{p,q}^{-}(a)}^2
  \right).
  \label{eq:truncated-continuum-defect}
\end{split}
\end{equation}
These defects increase monotonically as the resolution window grows and satisfy
\begin{equation}
  \lim_{P,Q\to\infty}
  D_{\mathrm{geo}}^{(1)}[P,Q](a)
  =D_{\mathrm{geo}}^{(1)}(a).
  \label{eq:truncated-defect-convergence}
\end{equation}
A single nonzero truncated witness already excludes a local metric realization\rev{, and} increasing $(P,Q)$ can only strengthen the certificate.

A continuum limit of finite codes must approximate the geometric operator, not merely a list of interval values. Let
\begin{equation}
  A:=D\mathcal B_{g_{\mathbb H}}
  =\frac12I_2:\mathcal X\longrightarrow\mathcal Y_{\mathrm{geo}}.
  \label{eq:continuum-geometric-operator}
\end{equation}
For a sequence of regular minimum-cut chambers, let $M_N:\mathbb R^{E_N}\to\mathbb R^{\calI_N}$ be the chamber incidence maps of Section~\ref{sec:discrete}. Equip the finite source and data spaces with positive-definite inner products and pass to the visible source space
\begin{equation}
  \mathcal X_N=(\ker M_N)^\perp,
  \qquad \overline M_N=M_N|_{\mathcal X_N}.
  \label{eq:finite-visible-source}
\end{equation}
Assume these inner products are chosen together with isometric source
embeddings and bounded data interpolation maps
\begin{equation}
  E_N:\mathcal X_N\longrightarrow\mathcal X,
  \qquad
  J_N:\mathbb R^{\calI_N}\longrightarrow\mathcal Y^\tau.
  \label{eq:finite-continuum-embeddings}
\end{equation}
The map $E_N$ interprets visible edge deformations as tensor-field approximants, while $J_N$ incorporates interpolation and quadrature on interval space.

\begin{theorem}[Stable graph-to-continuum limit]
\label{thm:finite-code-continuum-limit}
Let $P_N=E_NE_N^*$ be the orthogonal projection onto $E_N\mathcal X_N$. Assume
\begin{equation}
  P_N\xrightarrow[N\to\infty]{\mathrm{s}}\operatorname{Id}_{\mathcal X}
  \label{eq:source-density-assumption}
\end{equation}
and the operator consistency estimate
\begin{equation}
  \varepsilon_N
  :=\norm{J_N\overline M_N-AE_N}_{\mathcal X_N\to\mathcal Y^\tau}
  \longrightarrow0.
  \label{eq:operator-consistency-assumption}
\end{equation}
Set $B_N=J_N\overline M_N$, let $\mathcal R_N=\Ran B_N$, and let $\Pi_N$ be the orthogonal projection of $\mathcal Y^\tau$ onto $\mathcal R_N$. Then the following conclusions hold.
\begin{enumerate}[label=\textnormal{(\roman*)},leftmargin=2.4em]
  \item For all sufficiently large $N$,
  \begin{equation}
    \norm{B_Nx}_{\mathcal Y^\tau}
    \geq(c-\varepsilon_N)\norm{x}_{\mathcal X_N}
    \geq\frac c2\norm{x}_{\mathcal X_N},
    \label{eq:uniform-discrete-inf-sup}
  \end{equation}
  where $c$ is the lower stability constant in Equation~\eqref{eq:continuum-stability-two-sided}.
  \item The finite geometric projectors converge strongly,
  \begin{equation}
    \Pi_Na\longrightarrow\Pi_{\mathrm{geo}}a
    \qquad\text{for every }a\in\mathcal Y^\tau.
    \label{eq:finite-projector-convergence}
  \end{equation}
  More specifically, the leakage of a normalized finite geometric response into the continuum nongeometric sector satisfies
  \begin{equation}
    \sup_{0\neq z\in\mathcal R_N}
    \frac{\norm{\Pi_{\mathrm{ng}}z}_{\mathcal Y^\tau}}
         {\norm{z}_{\mathcal Y^\tau}}
    \leq\frac{\varepsilon_N}{c-\varepsilon_N}.
    \label{eq:finite-range-leakage}
  \end{equation}
  \item If sampled proto-area data $q_N\in\mathbb R^{\calI_N}$ obey
  \begin{equation}
    J_Nq_N\longrightarrow a_1
    \qquad\text{in }\mathcal Y^\tau,
    \label{eq:code-data-convergence}
  \end{equation}
  then
  \begin{align}
    \operatorname{dist}_{\mathcal Y^\tau}(J_Nq_N,\mathcal R_N)
    &\longrightarrow D_{\mathrm{geo}}^{(1)}(a_1),
    \label{eq:defect-continuum-convergence}\\
    J_Nq_N-\Pi_NJ_Nq_N
    &\longrightarrow\Pi_{\mathrm{ng}}a_1,
    \label{eq:residual-continuum-convergence}\\
    E_NB_N^\dagger J_Nq_N
    &\longrightarrow\mathscr R_2a_1
    \qquad\text{in }\mathcal X.
    \label{eq:reconstruction-continuum-convergence}
  \end{align}
  Here $B_N^\dagger$ is the Moore--Penrose inverse.
\end{enumerate}
In particular, a nonzero continuum defect is eventually detected by the finite chamber problems, while geometrizable data have convergent minimum-norm reconstructions.
\end{theorem}

\begin{revision}
\begin{proof}[Proof sketch]
The lower bound for $A$ and the operator error in Equation~\eqref{eq:operator-consistency-assumption} give the uniform discrete inf--sup estimate. They also bound the angle between $\mathcal R_N$ and $\mathcal Y_{\mathrm{geo}}$. Source density then approximates every geometric response, while the angle bound suppresses leakage from the nongeometric sector. These two estimates imply strong convergence of the projectors, defects, and pseudoinverse reconstructions. The full operator proof and finite-band sufficient conditions for $E_N$ and $J_N$ are given in Appendix~\ref{appB:continuum-approximation}.
\end{proof}
\end{revision}

The hypotheses of Theorem~\ref{thm:finite-code-continuum-limit} separate two logically distinct requirements. Strong density in Equation~\eqref{eq:source-density-assumption} says that visible edge modes resolve every continuum metric deformation. Operator consistency in Equation~\eqref{eq:operator-consistency-assumption} says that their cut responses converge uniformly to geodesic integration. Pointwise agreement on a sparse list of intervals implies neither condition. Under these hypotheses, however, the finite left-null residual converges to the complete continuum obstruction rather than to a discretization artifact.

Two exact examples display the separation. Let $\varphi\in C_c^\infty(M;\mathbb R)$ and set
\begin{equation}
  k_{\mathrm{conf}}=2\varphi g_{\mathbb H}.
  \label{eq:conformal-geometric-mode}
\end{equation}
Then
\begin{equation}
  D\mathcal B_{g_{\mathbb H}}[k_{\mathrm{conf}}]
  =I_0\varphi,
  \qquad
  \Pi_{\mathrm{ng}}I_0\varphi=0.
  \label{eq:conformal-mode-data}
\end{equation}
This is \rev{the response of a} local metric\rev{ perturbation}. By contrast, define the real normalized fourth-order datum
\begin{equation}
  a_{\mathrm{ng}}
  =\frac{3^{-\tau}}{\sqrt2}
  \left(
  \widehat I_{0,4}+\overline{\widehat I_{0,4}}
  \right).
  \label{eq:explicit-nongeometric-mode}
\end{equation}
It belongs entirely to $\mathcal H_4$, has
\begin{equation}
  \norm{a_{\mathrm{ng}}}_{\mathcal Y^\tau}=1,
  \qquad
  D_{\mathrm{geo}}^{(1)}(a_{\mathrm{ng}})=1,
  \label{eq:explicit-nongeometric-defect}
\end{equation}
 and cannot be the first variation of any local metric on the stated slice. The obstruction is a rank mismatch in the complete boundary moment hierarchy and is independent of positivity and gauge choices.

\begin{figure}[t]
\centering
\begin{tikzpicture}[font=\small,>=Latex]
  \begin{scope}[xshift=0cm]
    \node at (-1.35,2.55) {\textbf{a}};
    \draw[line width=0.7pt] (0,1.25) circle (1.15);
    \fill[black!12] (0.15,1.15) circle (0.42);
    \draw[line width=0.8pt] (-0.93,1.93) .. controls (-0.28,0.75) and (0.44,0.62) .. (0.98,1.84);
    \draw[line width=0.8pt] (-1.08,1.02) .. controls (-0.28,1.58) and (0.42,1.63) .. (1.08,1.02);
    \fill (-0.93,1.93) circle (1.5pt);
    \fill (0.98,1.84) circle (1.5pt);
    \fill (-1.08,1.02) circle (1.5pt);
    \fill (1.08,1.02) circle (1.5pt);
    \node[align=center] at (0,-0.28) {one local tensor field\\probed by all intervals};
  \end{scope}

  \begin{scope}[xshift=4.45cm]
    \node at (-1.55,2.55) {\textbf{b}};
    \draw[rounded corners,fill=black!5] (-1.35,0.35) rectangle (1.35,2.25);
    \draw (-1.35,1.45)--(1.35,1.45);
    \draw (-1.35,0.92)--(1.35,0.92);
    \node at (0,1.86) {$\mathcal H_0$};
    \node at (0,1.18) {$\mathcal H_2$};
    \node at (0,0.62) {$\mathcal H_4\oplus\mathcal H_6\oplus\cdots$};
    \draw[decorate,decoration={brace,amplitude=4pt},line width=0.7pt]
      (1.47,2.23)--(1.47,0.94)
      node[midway,right=5pt,align=left] {$\mathcal Y_{\rm geo}$};
    \draw[decorate,decoration={brace,amplitude=4pt},line width=0.7pt]
      (1.47,0.90)--(1.47,0.37)
      node[midway,right=5pt,align=left] {$\mathcal Y_{\rm ng}$};
    \node[align=center] at (0,-0.28) {rank-two geometry occupies\\only the first two sectors};
  \end{scope}

  \begin{scope}[xshift=9.45cm]
    \node at (-1.45,2.55) {\textbf{c}};
    \draw[->] (-1.25,0.38)--(1.45,0.38) node[right] {$p$};
    \draw[->] (-1.25,0.38)--(-1.25,2.30) node[above] {$q$};
    \foreach \x in {-0.9,-0.55,-0.2,0.15,0.5,0.85,1.2}{
      \foreach \y in {0.65,0.95,1.25,1.55,1.85,2.15}{
        \fill (\x,\y) circle (1.2pt);
      }
    }
    \draw[dashed,rounded corners] (-1.05,1.08) rectangle (0.67,1.70);
    \node[align=center] at (0,-0.28) {finite witness window\\$0\leq p\leq P$, $2\leq q\leq Q$};
  \end{scope}
\end{tikzpicture}
\caption{Linear geometrizability on the hyperbolic disk. \textbf{(a)} Every boundary interval probes the same local metric perturbation along its reference geodesic. \textbf{(b)} The scalar and rank-two transverse-traceless sectors form the closed geometric range in the source-adapted data topology\rev{, while} sectors of order four and above are nongeometric for a metric perturbation. \textbf{(c)} A finite calculation measures a rectangular set of moment witnesses. Enlarging the window converges monotonically to the exact nongeometric defect.}
\label{fig:linear-geometrizability}
\end{figure}
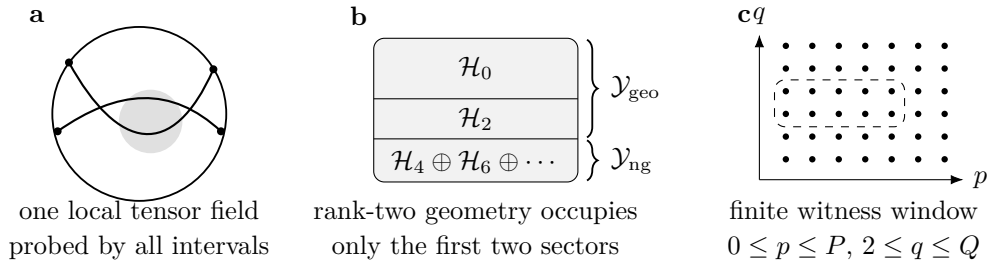

Figure~\ref{fig:linear-geometrizability} summarizes the projection, reconstruction, and exclusion branches of the linear criterion.

For a given proto-area datum, the linear theory yields three outputs. The projection $\Pi_{\mathrm{geo}}a_1$ gives the closest local metric response, $\mathscr R_2a_1$ reconstructs its unique gauge-fixed tensor field, and $\Pi_{\mathrm{ng}}a_1$ measures the part that no local metric can reproduce. These are complete first-order outputs. Their continuation to a metric two-jet is controlled by the normal Hessian of the boundary-length map.

\section{Nonlinear consistency and the quadratic obstruction}
\label{sec:nonlinear-consistency}

The tensor X-ray range theorem determines the linearized space of admissible boundary-length data. Extremality removes geodesic displacement from the first variation, but not from the second. The forced Jacobi equation fixes this contribution and hence the normal acceleration of the length data. Its normal projection is defined directly on regular metric two-jets and does not require a full nonlinear Banach-manifold description of the boundary-length image. A proto-area path may therefore have a first derivative in $\mathcal Y_{\mathrm{geo}}$ and still fail to arise from any metric path. The resulting coefficientwise two-jet condition is independent of any bulk equation of motion.

Second-variation formulas for geodesic length and Jacobi fields are classical, as is the use of the linearized boundary-distance map in rigidity theory~\cite{Besse1987,PaternainSaloUhlmann2023}. Nonlinear boundary rigidity for negatively curved asymptotically hyperbolic surfaces establishes that the complete renormalized boundary distance determines the metric up to a boundary-fixing diffeomorphism~\cite{GrahamEtAl2019,Lefeuvre2020}. Those results address uniqueness once geometric data are supplied. Here the input is an arbitrary state-dependent proto-area two-jet. Membership in the metric two-jet image is decided by the normal Hessian, yielding gauge-invariant witnesses and a necessary and sufficient criterion for regular coefficient two-jets.

The analysis is kinematic\rev{, in that} it determines whether interval data can be represented by a local Riemannian metric on the time-reflection-symmetric slice. No Einstein constraint is imposed. The conformal representative at infinity and the defining function from Section~\ref{subsec:renormalized-data} are held fixed. Metric variations are real and either compactly supported in $M$ or satisfy, for some $\eta>0$,
\begin{equation}
  \nabla^j h=O(x^{1+\eta}),
  \qquad 0\leq j\leq4,
  \label{eq:quadratic-decay-class}
\end{equation}
in a unit zero-frame. Every tensor entering a multilinear variation is assumed to satisfy this bound. It supplies the integrability needed for endpoint limits, integrations by parts, and differentiation under the renormalized geodesic integral. The formulas are first proved for compact support and then extended to Equation~\eqref{eq:quadratic-decay-class} by dominated convergence.

\begin{revision}
We follow a fixed smooth extremal-geodesic branch with an invertible
normal Jacobi operator. On the simple disk used here, a pair of distinct
ideal endpoints has one geodesic, so this choice is unambiguous. In a
setting with competing extremal surfaces, the same calculation applies
within a region of parameter space where the chosen branch remains
smooth and is the minimizing branch. A crossing of extremal lengths
can make the minimum nondifferentiable, and the two-jet test does not
span such a crossing. This is the continuum counterpart of fixing a
minimum-cut chamber in the graph problem.  Competing branches are
discussed in Subsection~\ref{subsec:extremal-branches}.
\end{revision}

\subsection{Metric forcing and the normal Jacobi equation}
\label{subsec:forced-jacobi}

\begin{revision}
The geometric origin of the second-order condition is the following.
The length of the reference geodesic is stationary under variations of
the curve, so when the metric is perturbed the curve may be held fixed
at first order and only the line element changes.  This is why the
first variation is the X-ray transform of the perturbation.  At second
order the situation is different.  The perturbed metric has a different
Levi-Civita connection, the reference curve is no longer a geodesic of
it, and the true geodesic with the same endpoints is displaced by an
amount of first order in the perturbation.  The length of the displaced
curve differs from the length of the reference curve at second order,
by an amount that is quadratic in the displacement and hence in the
perturbation.  The displacement is determined by a linear equation, the
Jacobi equation with a source given by the change of the connection,
and we now derive this equation and the resulting contribution to the
Hessian.
\end{revision}

Fix an oriented reference geodesic $\gamma=\gamma_{\beta,b}$ and parametrize it by hyperbolic arclength $s\in\mathbb R$. Write
\begin{equation}
  T=\dot\gamma,
  \qquad
  \nabla_TT=0,
  \qquad
  g_{\mathbb H}(T,T)=1,
  \label{eq:reference-geodesic-frame}
\end{equation}
and let $N$ be the parallel unit normal chosen so that $(T,N)$ is positively oriented. We use the curvature convention
\begin{equation}
  R(V,T)T=K V
  \qquad\text{for }V\perp T,
  \label{eq:curvature-convention}
\end{equation}
so that $R(V,T)T=-V$ on $(M,g_{\mathbb H})$.

For a symmetric two-tensor $h$, let $C_h$ denote the variation of the Levi-Civita connection in the affine metric direction $h$,
\begin{equation}
  C_h(X,Y)
  =\left.\frac{d}{dt}\right|_{t=0}
  \nabla_X^{g_{\mathbb H}+th}Y.
  \label{eq:connection-variation-definition}
\end{equation}
The Koszul formula gives
\begin{equation}
\begin{split}
  2g_{\mathbb H}\bigl(C_h(X,Y),Z\bigr)
  ={}&(\nabla_Xh)(Y,Z)+(\nabla_Yh)(X,Z)\\
  &-(\nabla_Zh)(X,Y).
  \label{eq:connection-variation-formula}
\end{split}
\end{equation}
Only the normal component of $C_h(T,T)$ changes the unparametrized geodesic. Define
\begin{equation}
  \mathcal F_\gamma(h)
  :=\bigl[C_h(T,T)\bigr]^\perp
  =f_hN,
  \label{eq:metric-geodesic-forcing}
\end{equation}
where
\begin{equation}
  f_h
  =(\nabla_Th)(T,N)
  -\frac12(\nabla_Nh)(T,T).
  \label{eq:scalar-metric-forcing}
\end{equation}
The normal Jacobi operator is
\begin{equation}
  \mathcal J_\gamma
  =-D_s^2-R(\,\cdot\,,T)T.
  \label{eq:normal-jacobi-operator}
\end{equation}
On the hyperbolic disk,
\begin{equation}
  \mathcal J_\gamma(jN)=(-j''+j)N.
  \label{eq:hyperbolic-jacobi-operator}
\end{equation}

\begin{lemma}[Mixed metric--curve derivative]
\label{lem:mixed-metric-curve-derivative}
Let $\eta_r$ be a smooth variation of $\gamma$ through curves with fixed ideal endpoints, and let
\begin{equation}
  V=\left.\partial_r\eta_r\right|_{r=0}
  \label{eq:curve-variation-field}
\end{equation}
be normal to $\gamma$ and decay at both ends. Then
\begin{equation}
  \left.\frac{d}{dr}\right|_{r=0}
  \left[
  \frac12\int_{\eta_r}h(T_r,T_r)\,ds_r
  \right]
  =-\int_\gamma
  g_{\mathbb H}\bigl(\mathcal F_\gamma(h),V\bigr)\,ds.
  \label{eq:mixed-metric-curve-identity}
\end{equation}
\end{lemma}

\begin{proof}
Differentiate the integrand using the hyperbolic connection. At $r=0$ one has
\begin{equation}
\begin{split}
  \left.\frac{d}{dr}\right|_{0}
  \left[h(T_r,T_r)\,ds_r\right]
  ={}&(\nabla_Vh)(T,T)\,ds
  +2h(D_sV,T)\,ds\\
  &-h(T,T)g_{\mathbb H}(D_sV,T)\,ds.
  \label{eq:metric-curve-differentiation}
\end{split}
\end{equation}
Because $V$ is normal and $\gamma$ is geodesic,
\begin{equation}
  g_{\mathbb H}(D_sV,T)
  =-g_{\mathbb H}(V,D_sT)=0.
\end{equation}
The decay of $V$ and $h$ permits integration by parts, giving
\begin{equation}
  \int_\gamma h(D_sV,T)\,ds
  =-\int_\gamma(\nabla_Th)(V,T)\,ds.
\end{equation}
Therefore the left-hand side of Equation~\eqref{eq:mixed-metric-curve-identity} equals
\begin{equation}
  \int_\gamma
  \left[
  \frac12(\nabla_Vh)(T,T)
  -(\nabla_Th)(V,T)
  \right]ds.
\end{equation}
Equations~\eqref{eq:connection-variation-formula} and~\eqref{eq:metric-geodesic-forcing} identify the bracket with
$-g_{\mathbb H}(\mathcal F_\gamma(h),V)$.
\end{proof}

Let $g_t=g_{\mathbb H}+th+O(t^2)$ and let $\gamma_t$ be the $g_t$-geodesic with the same ordered ideal endpoints as $\gamma$. Reparametrize $\gamma_t$ so that its variation field
\begin{equation}
  J_h=\left.\partial_t\gamma_t\right|_{t=0}
  \label{eq:metric-induced-geodesic-variation}
\end{equation}
is normal. This removes only a tangential reparametrization and does not change the image of the curve or its length.

The metric-forced equation below is the geodesic specialization of the
inhomogeneous Jacobi equation studied in
Refs.~\cite{GhoshMishra2016,GhoshMishra2018,Mosk2018}, with the
curvature convention and ideal-endpoint decay fixed above.

\begin{theorem}[Forced Jacobi equation]
\label{thm:forced-jacobi-equation}
For every $h$ in the class of Equation~\eqref{eq:quadratic-decay-class}, the normal displacement $J_h$ is the unique field decaying at both ideal endpoints that solves
\begin{equation}
  \mathcal J_\gamma J_h=\mathcal F_\gamma(h).
  \label{eq:forced-jacobi-equation}
\end{equation}
On the hyperbolic disk it is given explicitly by
\begin{equation}
  J_h(s)
  =\frac12N(s)
  \int_{-\infty}^{\infty}
  e^{-|s-r|}f_h(r)\,dr.
  \label{eq:forced-jacobi-green-solution}
\end{equation}
In particular,
\begin{equation}
  \mathcal J_\gamma^{-1}(s,r)
  =\frac12e^{-|s-r|}
  \label{eq:hyperbolic-jacobi-green-kernel}
\end{equation}
on the normal decaying sector.
\end{theorem}

\begin{proof}
Differentiate the $g_t$-geodesic equation. If $T_t$ denotes the tangent of $\gamma_t$, its normal component vanishes,
\begin{equation}
  \bigl[\nabla^{g_t}_{T_t}T_t\bigr]^\perp_{g_t}=0.
  \label{eq:normal-geodesic-equation}
\end{equation}
At $t=0$, the covariant derivative of the curve variation satisfies the standard commutation identity
\begin{equation}
  \left.\frac{D}{dt}\right|_0
  \nabla_{T_t}T_t
  =D_s^2J_h+R(J_h,T)T.
  \label{eq:geodesic-variation-commutator}
\end{equation}
The variation of the connection contributes $C_h(T,T)$. Taking the normal component of the derivative of Equation~\eqref{eq:normal-geodesic-equation} gives
\begin{equation}
  D_s^2J_h+R(J_h,T)T+\mathcal F_\gamma(h)=0.
  \label{eq:forced-jacobi-before-sign}
\end{equation}
With the sign convention in Equation~\eqref{eq:normal-jacobi-operator}, this is Equation~\eqref{eq:forced-jacobi-equation}.

On $\mathbb H^2$, write $J_h=j_hN$. Equation~\eqref{eq:hyperbolic-jacobi-operator} reduces the problem to
\begin{equation}
  (-\partial_s^2+1)j_h=f_h.
  \label{eq:scalar-forced-jacobi}
\end{equation}
The distributional identity
\begin{equation}
  (-\partial_s^2+1)\left(\frac12e^{-|s-r|}\right)
  =\delta(s-r)
  \label{eq:green-kernel-identity}
\end{equation}
proves Equation~\eqref{eq:forced-jacobi-green-solution}. The decay assumption makes the integral absolutely convergent and gives exponential decay of $j_h$ at both ends. The difference of two decaying solutions solves $-j''+j=0$ and is a linear combination of $e^s$ and $e^{-s}$\rev{, and} decay at both ends forces both coefficients to vanish.
\end{proof}

The Green representation displays the nonlocal bending effect relevant in the present setting. A local metric change at one point of a reference RT geodesic displaces the extremal curve along its entire length. The propagation is controlled by the inverse stability operator of the extremal surface, and on the hyperbolic slice it decays exponentially in geodesic distance.

\subsection{The renormalized Hessian}
\label{subsec:renormalized-hessian}

The bilinear second derivative of $\mathcal B$ is obtained by differentiating the affine two-parameter metric family
\begin{equation}
  g_{t,u}=g_{\mathbb H}+th+uk
  \label{eq:affine-two-parameter-metric}
\end{equation}
at $(t,u)=(0,0)$. The absence of a mixed $tu$ term in Equation~\eqref{eq:affine-two-parameter-metric} ensures that the result is the bilinear Hessian rather than a path-dependent acceleration.

Second-order extremal-surface expansions under ambient metric
perturbations were developed in
Refs.~\cite{GhoshMishra2018,Mosk2018}. The theorem below gives the
hyperbolic-geodesic and renormalized ideal-boundary specialization
needed for the normal-range projection.

\begin{theorem}[Second variation of renormalized boundary length]
\label{thm:renormalized-second-variation}
Let $h$ and $k$ satisfy Equation~\eqref{eq:quadratic-decay-class}. For every complete reference geodesic $\gamma$,
\begin{equation}
\begin{split}
  D^2\mathcal B_{g_{\mathbb H}}[h,k](\gamma)
  ={}&-\frac14\int_\gamma
  h(T,T)k(T,T)\,ds\\
  &-\int_\gamma
  g_{\mathbb H}\left(
  \mathcal F_\gamma(h),
  \mathcal J_\gamma^{-1}\mathcal F_\gamma(k)
  \right)ds.
  \label{eq:second-variation-operator-form}
\end{split}
\end{equation}
Equivalently, on the hyperbolic disk,
\begin{equation}
\begin{split}
  D^2\mathcal B_{g_{\mathbb H}}[h,k](\gamma)
  ={}&-\frac14\int_{-\infty}^{\infty}
  h_{TT}(s)k_{TT}(s)\,ds\\
  &-\frac12\int_{-\infty}^{\infty}
  \int_{-\infty}^{\infty}
  e^{-|s-r|}f_h(s)f_k(r)\,dr\,ds.
  \label{eq:second-variation-green-form}
\end{split}
\end{equation}
The expression is symmetric in $h$ and $k$. For an affine one-parameter family $g_t=g_{\mathbb H}+th$,
\begin{equation}
  D^2\mathcal B_{g_{\mathbb H}}[h,h](\gamma)
  \leq0.
  \label{eq:affine-length-concavity}
\end{equation}
\end{theorem}

\begin{proof}
Fix two finite points $p,q\in\gamma$. Let $\eta_{t,u}$ be the unique $g_{t,u}$-geodesic from $p$ to $q$ and let
\begin{equation}
  \ell(t,u)=\operatorname{Length}_{g_{t,u}}(\eta_{t,u}).
\end{equation}
The first derivative in the $h$ direction is
\begin{equation}
  \partial_t\ell(0,0)
  =\frac12\int_{\gamma_{pq}}h(T,T)\,ds.
  \label{eq:finite-first-derivative-h}
\end{equation}
The $u$ derivative of the right-hand side has two contributions. If the curve is held fixed, the mixed metric derivative of the square-root integrand is
\begin{equation}
  \left.\partial_t\partial_u\right|_{0}
  \sqrt{(g_{\mathbb H}+th+uk)(T,T)}
  =-\frac14h(T,T)k(T,T).
  \label{eq:fixed-curve-metric-hessian}
\end{equation}
The curve variation is $J_k=\partial_u\eta_{0,u}|_{u=0}$. After removing its tangential component, Lemma~\ref{lem:mixed-metric-curve-derivative} gives
\begin{equation}
  \partial_u\partial_t\ell(0,0)
  =-\frac14\int_{\gamma_{pq}}h_{TT}k_{TT}\,ds
  -\int_{\gamma_{pq}}
  g_{\mathbb H}(\mathcal F_\gamma(h),J_k)\,ds.
  \label{eq:finite-reduced-hessian}
\end{equation}

For fixed endpoints, $J_k$ vanishes at $p$ and $q$. Differentiating the geodesic equation, as in Theorem~\ref{thm:forced-jacobi-equation}, shows that
\begin{equation}
  \mathcal J_{\gamma,pq}J_k=\mathcal F_\gamma(k),
  \qquad J_k(p)=J_k(q)=0,
  \label{eq:finite-dirichlet-jacobi}
\end{equation}
where $\mathcal J_{\gamma,pq}$ is the normal Jacobi operator with Dirichlet boundary conditions. Negative curvature makes this operator positive and invertible. Substitution into Equation~\eqref{eq:finite-reduced-hessian} gives
\begin{equation}
\begin{split}
  \partial_u\partial_t\ell(0,0)
  ={}&-\frac14\int_{\gamma_{pq}}h_{TT}k_{TT}\,ds\\
  &-\int_{\gamma_{pq}}
  g_{\mathbb H}\left(
  \mathcal F_\gamma(h),
  \mathcal J_{\gamma,pq}^{-1}\mathcal F_\gamma(k)
  \right)ds.
  \label{eq:finite-hessian-operator}
\end{split}
\end{equation}
This expression is symmetric because $\mathcal J_{\gamma,pq}^{-1}$ is self-adjoint.

Choose $p=p_\varepsilon$ and $q=q_\varepsilon$ as in Equation~\eqref{eq:truncated-endpoints}. The counterterm $2\log\varepsilon$ is independent of $(t,u)$ because the defining function and conformal representative are fixed. For compactly supported $h,k$, the integrations in Equation~\eqref{eq:finite-hessian-operator} are eventually supported in a fixed bounded $s$ interval. The Dirichlet Green kernel on $[s_-,s_+]$ is
\begin{equation}
  G_{s_-,s_+}(s,r)
  =\frac{
  \sinh(s_<-s_-)\sinh(s_+-s_>)
  }{
  \sinh(s_+-s_-)
  },
  \label{eq:finite-jacobi-green-kernel}
\end{equation}
where $s_<=\min\{s,r\}$ and $s_>=\max\{s,r\}$. As $s_-\to-\infty$ and $s_+\to+\infty$,
\begin{equation}
  G_{s_-,s_+}(s,r)
  \longrightarrow\frac12e^{-|s-r|}
  \label{eq:green-kernel-limit}
\end{equation}
uniformly on compact subsets. Dominated convergence gives Equation~\eqref{eq:second-variation-operator-form} and then Equation~\eqref{eq:second-variation-green-form}. Under Equation~\eqref{eq:quadratic-decay-class}, $h_{TT}$, $k_{TT}$, $f_h$, and $f_k$ are integrable with a positive exponential margin in the geodesic coordinate near both ends. An integrable majorant gives the mixed estimate as well.

Finally, the two terms in Equation~\eqref{eq:second-variation-operator-form} are symmetric. On the diagonal, the first is nonpositive and the second is
\begin{equation}
  -\ip{\mathcal F_\gamma(h)}
  {\mathcal J_\gamma^{-1}\mathcal F_\gamma(h)}_{L^2(\gamma)},
\end{equation}
which is nonpositive because $\mathcal J_\gamma=-\partial_s^2+1$ is positive. This proves Equation~\eqref{eq:affine-length-concavity}.
\end{proof}

The first line of Equation~\eqref{eq:second-variation-operator-form} is the direct change of the line element. The second line is the bending contribution. It is nonlocal along each geodesic and is fixed by the stability operator of the extremal curve. Omitting it would give the Hessian of the length of the reference curve, not the Hessian of the extremal length.

For a general metric path
\begin{equation}
  g_t
  =g_{\mathbb H}+t h_1+\frac{t^2}{2}h_2+o(t^2),
  \label{eq:metric-two-jet}
\end{equation}
the chain rule gives
\begin{equation}
  \left.\frac{d^2}{dt^2}\right|_{0}\mathcal B(g_t)
  =Ah_2+D^2\mathcal B_{g_{\mathbb H}}[h_1,h_1].
  \label{eq:boundary-map-second-chain-rule}
\end{equation}
The first term can be adjusted by choosing the second metric coefficient. Its data always lie in $\mathcal Y_{\mathrm{geo}}$. The normal component of the Hessian cannot be adjusted in this way.

The direct contribution can also be written as a rank-four X-ray transform. Let $h\odot k$ be the completely symmetrized tensor product normalized by
\begin{equation}
  (h\odot k)(T,T,T,T)=h(T,T)k(T,T).
  \label{eq:symmetric-rank-four-product}
\end{equation}
Then
\begin{equation}
  D^2\mathcal B[h,k]
  =-\frac14 I_4(h\odot k)-\mathcal K(h,k),
  \label{eq:hessian-rank-four-decomposition}
\end{equation}
where
\begin{equation}
  \mathcal K(h,k)(\gamma)
  =\ip{\mathcal F_\gamma(h)}
  {\mathcal J_\gamma^{-1}\mathcal F_\gamma(k)}_{L^2(\gamma)}.
  \label{eq:bending-bilinear-operator}
\end{equation}
Equation~\eqref{eq:hessian-rank-four-decomposition} explains why quadratic metric data are not confined to the rank-two range. The product of two rank-two perturbations and the geodesic bending term can populate higher even tensor sectors. Their higher-sector component equals the normal acceleration of the geometric data locus.

The Hessian itself depends on a choice of coordinates on the space of metrics, but its projection normal to the linear range is gauge invariant.

\begin{revision}
\begin{proposition}[Gauge covariance of the Hessian]
\label{prop:hessian-gauge-covariance}
Let $X$ generate a flow in
$\operatorname{Diff}_0^x(M;\partial M)$, with $X^\flat$ and its
covariant derivatives of order $O(x^{1+\eta})$ for some $\eta>0$.
Assume that $h$, $\mathcal L_Xg_{\mathbb H}$, and $\mathcal L_Xh$
satisfy Equation~\eqref{eq:quadratic-decay-class}. Then
\begin{equation}
  D^2\mathcal B_{g_{\mathbb H}}
  [\mathcal L_Xg_{\mathbb H},h]
  =-A(\mathcal L_Xh).
  \label{eq:hessian-gauge-identity}
\end{equation}
Consequently,
\begin{equation}
  \Pi_{\mathrm{ng}}
  D^2\mathcal B_{g_{\mathbb H}}
  [\mathcal L_Xg_{\mathbb H},h]=0.
  \label{eq:normal-hessian-gauge-invariance}
\end{equation}
The bilinear form
\begin{equation}
  \mathrm{II}_{g_{\mathbb H}}([h],[k])
  :=\Pi_{\mathrm{ng}}
  D^2\mathcal B_{g_{\mathbb H}}[h,k]
  \label{eq:second-fundamental-form-definition}
\end{equation}
is therefore well defined on metric deformation classes modulo boundary-fixing infinitesimal diffeomorphisms.
\end{proposition}
\end{revision}

\begin{proof}
Let $\varphi_s$ be the flow of $X$. Boundary-fixing diffeomorphism invariance gives
\begin{equation}
  \mathcal B\bigl(\varphi_s^*(g_{\mathbb H}+t h)\bigr)
  =\mathcal B(g_{\mathbb H}+t h).
  \label{eq:boundary-map-flow-invariance}
\end{equation}
The first derivatives of the metric on the left at $(s,t)=(0,0)$ are $\mathcal L_Xg_{\mathbb H}$ and $h$, and its mixed derivative is $\mathcal L_Xh$. Differentiating Equation~\eqref{eq:boundary-map-flow-invariance} once in each parameter gives
\begin{equation}
  D^2\mathcal B[\mathcal L_Xg_{\mathbb H},h]
  +D\mathcal B[\mathcal L_Xh]=0,
\end{equation}
which is Equation~\eqref{eq:hessian-gauge-identity}. Since $A(\mathcal L_Xh)\in\mathcal Y_{\mathrm{geo}}$, applying $\Pi_{\mathrm{ng}}$ proves Equation~\eqref{eq:normal-hessian-gauge-invariance}. Bilinearity then shows that Equation~\eqref{eq:second-fundamental-form-definition} is unchanged when either argument is altered by an infinitesimal boundary-fixing diffeomorphism.
\end{proof}

\begin{revision}
On a finite-dimensional regular source slice where the boundary-length map is twice differentiable in the data norm, Equation~\eqref{eq:second-fundamental-form-definition} is the component of the slice's second fundamental form normal to the full rank-two range. In the full regular class it is best regarded as the gauge-invariant normal Hessian of the boundary-length map. The coefficientwise results below require only this Hessian, the range splitting, and regular reconstruction of the two metric coefficients, and they do not assume a full nonlinear chart for the image.
\end{revision}

\subsection{Quadratic geometrizability}
\label{subsec:quadratic-geometrizability}

\begin{revision}
Let a background-subtracted, length-normalized proto-area path have the expansion
\begin{equation}
  a(t)=t a_1+\frac{t^2}{2}a_2+o(t^2)
  \qquad\text{in }\mathcal Y^\tau_{\mathbb R}.
  \label{eq:proto-area-two-jet}
\end{equation}
If the first-order condition holds, its unique gauge-fixed first metric coefficient is
\begin{equation}
  h_1=\mathscr R_2a_1.
  \label{eq:first-metric-from-data}
\end{equation}
The Hessian term below is the second variation computed in Theorem~\ref{thm:renormalized-second-variation}.
Define the quadratic obstruction by
\begin{equation}
  \mathfrak O_2(a_1,a_2)
  :=\Pi_{\mathrm{ng}}
  \left[
  a_2-D^2\mathcal B_{g_{\mathbb H}}[h_1,h_1]
  \right].
  \label{eq:quadratic-obstruction-definition}
\end{equation}
\end{revision}

\begin{revision}
The definition is meaningful whenever $a_1$ is linearly geometrizable
and the Hessian datum belongs to $\mathcal Y^\tau$. In the theorem
below, regularity means that $h_1=\mathscr R_2a_1$ and
$\mathscr R_2[a_2-D^2\mathcal B(h_1,h_1)]$ have smooth or
polyhomogeneous zero-cotensor coefficients satisfying
Equation~\eqref{eq:quadratic-decay-class}, and that
$D^2\mathcal B(h_1,h_1)\in\mathcal Y^\tau$. The reconstruction
operator includes the geometric projection, so this condition can be
stated before the normal obstruction is tested. Data-space membership
alone does not supply this metric regularity. The same distinction
already occurs at linear order in
Theorem~\ref{thm:linear-geometrizability}.
\end{revision}

\begin{theorem}[Necessary and sufficient two-jet criterion]
\label{thm:two-jet-geometrizability}
Let $a(t)$ be as in Equation~\eqref{eq:proto-area-two-jet}, with regular coefficients $a_1$ and $a_2$. The following are equivalent.
\begin{enumerate}[label=\textnormal{(\alph*)},leftmargin=2.4em]
  \item There is a regular asymptotically hyperbolic metric two-jet
  with fixed conformal infinity,
  \begin{equation}
    (h_1,h_2),
    \label{eq:realizing-metric-two-jet}
  \end{equation}
  whose boundary-length coefficients obey
  \begin{equation}
    a_1=Ah_1,
    \qquad
    a_2=Ah_2+D^2\mathcal B_{g_{\mathbb H}}[h_1,h_1].
    \label{eq:two-jet-realization}
  \end{equation}
  \item The linear and quadratic conditions
  \begin{equation}
    \Pi_{\mathrm{ng}}a_1=0,
    \qquad
    \mathfrak O_2(a_1,a_2)=0
    \label{eq:two-jet-conditions}
  \end{equation}
  both hold.
\end{enumerate}
When these conditions hold, the unique iterated transverse-traceless coefficients are
\begin{equation}
  h_1=\mathscr R_2a_1,
  \qquad
  h_2=\mathscr R_2
  \left[
  a_2-D^2\mathcal B_{g_{\mathbb H}}[h_1,h_1]
  \right].
  \label{eq:second-metric-reconstruction}
\end{equation}
Equivalently, the normal acceleration is fixed by
\begin{equation}
  \Pi_{\mathrm{ng}}a_2
  =\mathrm{II}_{g_{\mathbb H}}([h_1],[h_1]).
  \label{eq:normal-acceleration-condition}
\end{equation}
If, in addition, the boundary-length map is twice differentiable into
$\mathcal Y^\tau$ at the origin along the polynomial metric path
\begin{equation}
 g_t=g_{\mathbb H}+t h_1+\frac{t^2}{2}h_2,
 \label{eq:two-jet-polynomial-metric-path}
\end{equation}
with derivatives $A$ and $D^2\mathcal B$, then
\begin{equation}
 \mathcal B(g_t)=t a_1+\frac{t^2}{2}a_2+o_{\mathcal Y^\tau}(t^2).
 \label{eq:two-jet-strong-realization}
\end{equation}
\end{theorem}

\begin{revision}
\begin{proof}

Assume first that a realizing metric two-jet exists, and denote its
coefficients before gauge fixing by $\widetilde h_1,\widetilde h_2$.
The first coefficient gives $a_1=A\widetilde h_1$, hence
$\Pi_{\mathrm{ng}}a_1=0$. A parameter-dependent boundary-fixing
diffeomorphism fixes the ITT representative through second order. Its
action on the metric coefficients is
\begin{align}
 h_1&=\widetilde h_1+\mathcal L_Xg_{\mathbb H},\nonumber\\
 h_2&=\widetilde h_2+2\mathcal L_X\widetilde h_1
       +\mathcal L_X^2g_{\mathbb H}+\mathcal L_Yg_{\mathbb H}.
 \label{eq:metric-second-order-gauge-transformation}
\end{align}
Here $X$ and $Y$ are chosen with the boundary decay required for the
gauge action. First choose $X$ to impose
$h_1=\mathscr R_2a_1$, then choose $Y$ to remove the potential part of
the transformed second coefficient. Thus changing the first
representative also changes the second coefficient. Diffeomorphism
invariance and Equation~\eqref{eq:boundary-map-second-chain-rule}
give, in this gauge,
\begin{equation}
  a_2=Ah_2+D^2\mathcal B[h_1,h_1].
  \label{eq:two-jet-second-identity}
\end{equation}
Applying $\Pi_{\mathrm{ng}}$ removes $Ah_2$ and proves
$\mathfrak O_2(a_1,a_2)=0$. Equivalently, the same normal identity
follows directly from the gauge covariance of the Hessian.

Conversely, assume Equation~\eqref{eq:two-jet-conditions}. Set $h_1=\mathscr R_2a_1$ and define
\begin{equation}
  r_2=a_2-D^2\mathcal B[h_1,h_1].
  \label{eq:second-order-corrected-data}
\end{equation}
The vanishing of $\mathfrak O_2$ says that $r_2\in\mathcal Y_{\mathrm{geo}}$, so $h_2=\mathscr R_2r_2$ satisfies $Ah_2=r_2$. The regularity assumption gives coefficients of positive boundary order.  The two coefficient identities in Equation~\eqref{eq:two-jet-realization} follow by construction.  For small $t$, the polynomial path in Equation~\eqref{eq:two-jet-polynomial-metric-path} is positive definite, has the fixed conformal infinity, and remains simple.  Under the additional pathwise differentiability hypothesis, Taylor's theorem in $\mathcal Y^\tau$ gives Equation~\eqref{eq:two-jet-strong-realization}. Uniqueness of the ITT coefficients follows from injectivity of $A$ on the gauge slice. Equation~\eqref{eq:normal-acceleration-condition} is the normal projection of Equation~\eqref{eq:two-jet-second-identity}.
\end{proof}
\end{revision}

Theorem~\ref{thm:two-jet-geometrizability} proves that linear geometrizability is not sufficient. For any linearly geometric $a_1$ and any nonzero $n\in\mathcal Y_{\mathrm{ng}}$, the data path
\begin{equation}
  a_{\mathrm{test}}(t)
  =t a_1+\frac{t^2}{2}
  \left[
  D^2\mathcal B[h_1,h_1]+n
  \right]
  \label{eq:quadratically-obstructed-data-path}
\end{equation}
passes every first-order test and has
\begin{equation}
  \mathfrak O_2(a_1,a_2)=n.
  \label{eq:prescribed-quadratic-obstruction}
\end{equation}
No adjustment of $h_2$ can remove this discrepancy because $Ah_2$ is tangential to the geometric data locus.

\begin{revision}
The relative normalization of the first and second coefficients
matters. Let $s_A(t)$ denote the dimensionless entropy difference
used to define the proto-area, after the reference value is subtracted,
and write $s(t)=t s_1+t^2s_2/2+o(t^2)$. If lengths are measured in
units of a fixed reference radius $\ell$, set
\begin{equation}
 \kappa=\frac{4G_{\mathrm{eff}}}{\ell},\qquad
 a_j=\kappa s_j,\qquad H=\mathscr R_2s_1.
 \label{eq:fixed-normalization-two-jet}
\end{equation}
The unit-radius convention elsewhere in this section is $\ell=1$.
Here $s_1$ is assumed to pass the linear range test. Linearity of
reconstruction gives $h_1=\kappa H$, whereas the normal Hessian is
quadratic. Consequently,
\begin{equation}
 \mathfrak O_2(\kappa s_1,\kappa s_2)
 =\kappa\Pi_{\mathrm{ng}}s_2
  -\kappa^2\mathrm{II}_{g_{\mathbb H}}([H],[H]).
 \label{eq:scale-two-jet-obstruction}
\end{equation}
The constant $\kappa$ is fixed by the reference length calibration,
before the perturbative test as discussed in
Subsection~\ref{subsec:reference-normalization}, and is common to all
intervals and states. Varying it changes the reference length calibration and cannot
be absorbed into the second metric coefficient. In an abstract code without an independent
length calibration, the criterion is a criterion at the specified
$\kappa$, and it does not determine $G_{\mathrm{eff}}$ by itself. When
$s_1=0$, the geometric Hessian term vanishes and a nonzero normal
component of $s_2$ obstructs realization for every positive $\kappa$.
\end{revision}

The quadratic obstruction has a complete boundary-witness expansion. For $p\geq0$ and $q\geq2$, define
\begin{equation}
\begin{split}
  W_{p,q}^{(2),+}(a_1,a_2)
  &=W_{p,q}^{+}
  \left(a_2-D^2\mathcal B[h_1,h_1]\right),\\
  W_{p,q}^{(2),-}(a_1,a_2)
  &=W_{p,q}^{-}
  \left(a_2-D^2\mathcal B[h_1,h_1]\right).
  \label{eq:quadratic-moment-witnesses}
\end{split}
\end{equation}
Then
\begin{equation}
  \mathfrak O_2(a_1,a_2)=0
  \quad\Longleftrightarrow\quad
  W_{p,q}^{(2),\pm}(a_1,a_2)=0
  \text{ for all }p\geq0,\ q\geq2.
  \label{eq:quadratic-witness-completeness}
\end{equation}
The conditioned second-order defect is
\begin{equation}
  D_{\mathrm{geo}}^{(2)}(a_1,a_2)
  :=\norm{\mathfrak O_2(a_1,a_2)}_{\mathcal Y^\tau},
  \label{eq:quadratic-geometric-defect}
\end{equation}
and has the exact moment expansion
\begin{equation}
\begin{split}
  \bigl(D_{\mathrm{geo}}^{(2)}(a_1,a_2)\bigr)^2
  ={}&\sum_{q=2}^{\infty}(1+q)^{2\tau}
  \sum_{p=0}^{\infty}(p+1)\\
  &\times\left(
  \abs{W_{p,q}^{(2),+}}^2
  +\abs{W_{p,q}^{(2),-}}^2
  \right).
  \label{eq:quadratic-defect-moment-expansion}
\end{split}
\end{equation}
A finite witness window is obtained by restricting the sums to $0\leq p\leq P$ and $2\leq q\leq Q$. As in Equation~\eqref{eq:truncated-defect-convergence}, these truncated defects increase monotonically to Equation~\eqref{eq:quadratic-geometric-defect}.

\begin{revision}
The nonlinear witness is stable under errors in both derivatives in
norms that control the reconstruction and the Hessian simultaneously.
Let $\mathscr X$ be a normed space of regular gauge-fixed source tensors
and let $\mathscr Y$ be a data space on which $\Pi_{\mathrm{ng}}$ is
bounded. On a specified linear data subspace
$\mathscr D\subset\mathscr Y\cap\mathcal Y_{\mathrm{geo}}$, assume that
$\mathscr R_2:\mathscr D\to\mathscr X$ is bounded. Write
\begin{align}
  C_P&=\norm{\Pi_{\mathrm{ng}}}_{\mathscr Y\to\mathscr Y},
  \label{eq:normal-projection-bound}\\
  \norm{\mathrm{II}_{g_{\mathbb H}}(h,k)}_{\mathscr Y}
  &\leq C_2\norm{h}_{\mathscr X}\norm{k}_{\mathscr X},
  \label{eq:second-form-bilinear-bound}\\
  C_R&=\norm{\mathscr R_2}_{(\mathscr D,\norm{\cdot}_{\mathscr Y})
                           \to\mathscr X}.
  \label{eq:reconstruction-operator-bound}
\end{align}
Here the Hessian estimate is required on the reconstructed source
space. The Hilbert estimate of
Section~\ref{sec:linearized-geometry} controls reconstruction in
$\mathcal X$. It gives no estimate in a stronger source norm without a
corresponding stronger data estimate.
\end{revision}

\begin{revision}
\begin{proposition}[Robustness of the quadratic obstruction]
\label{prop:quadratic-obstruction-robustness}
Let $a_1,\widetilde a_1\in\mathscr D$ and
$a_2,\widetilde a_2\in\mathscr Y$, under the preceding boundedness
hypotheses, and set
\begin{equation}
  h_1=\mathscr R_2a_1,
  \qquad
  \widetilde h_1=\mathscr R_2\widetilde a_1.
\end{equation}
Then
\begin{equation}
\begin{split}
  &\norm{\mathfrak O_2(a_1,a_2)
  -\mathfrak O_2(\widetilde a_1,\widetilde a_2)}_{\mathscr Y}\\
  &\qquad\leq
  C_P\norm{a_2-\widetilde a_2}_{\mathscr Y}
  +C_2
  \bigl(\norm{h_1}_{\mathscr X}
  +\norm{\widetilde h_1}_{\mathscr X}\bigr)
  \norm{h_1-\widetilde h_1}_{\mathscr X},
  \label{eq:quadratic-obstruction-lipschitz}
\end{split}
\end{equation}
and
\begin{equation}
  \norm{h_1-\widetilde h_1}_{\mathscr X}
  \leq C_R
  \norm{a_1-\widetilde a_1}_{\mathscr Y}.
  \label{eq:quadratic-reconstruction-error}
\end{equation}
Thus, after inserting the stated error bounds, an observed defect larger than the right-hand side of Equation~\eqref{eq:quadratic-obstruction-lipschitz} certifies a nongeometric two-jet.
\end{proposition}
\end{revision}

\begin{proof}
By bilinearity and symmetry,
\begin{equation}
\begin{split}
  \mathrm{II}(h_1,h_1)-\mathrm{II}(\widetilde h_1,\widetilde h_1)
  ={}&\mathrm{II}(h_1-\widetilde h_1,h_1)\\
  &+\mathrm{II}(\widetilde h_1,h_1-\widetilde h_1).
\end{split}
\end{equation}
Apply Equation~\eqref{eq:second-form-bilinear-bound}, the triangle inequality, and the bound in Equation~\eqref{eq:normal-projection-bound}. Equation~\eqref{eq:quadratic-reconstruction-error} follows from the definition of $C_R$.
\end{proof}

\begin{revision}
\begin{corollary}[Finite-dimensional regular sources]
\label{cor:finite-source-quadratic-stability}
Let $E\subset\mathcal X_{\mathrm{itt}}$ be finite dimensional and spanned
by regular tensors $e_1,\ldots,e_d$ that are orthonormal in $\mathcal X$.
Assume that every $D^2\mathcal B[e_i,e_j]$ belongs to $\mathcal Y^\tau$.
For first data coefficients in $A(E)$, the preceding proposition applies
with $\mathscr X=E$, $\mathscr Y=\mathcal Y^\tau$, $C_P=1$,
$C_R\leq c^{-1}$, and
\begin{equation}
 C_2=C_E:=\left(\sum_{i,j=1}^d
 \norm{\Pi_{\mathrm{ng}}D^2\mathcal B[e_i,e_j]}_{\mathcal Y^\tau}^2
 \right)^{1/2}.
 \label{eq:finite-source-hessian-constant}
\end{equation}
The constant $c$ is the lower bound in
Equation~\eqref{eq:continuum-stability-two-sided}.
\end{corollary}
\end{revision}

\begin{revision}
\begin{proof}
For $h=\sum_i h_i e_i$ and $k=\sum_j k_j e_j$, bilinearity and
Cauchy--Schwarz give
\begin{equation}
 \norm{\mathrm{II}(h,k)}_{\mathcal Y^\tau}
 \leq C_E\left(\sum_{i,j}|h_i k_j|^2\right)^{1/2}
 =C_E\norm h_{\mathcal X}\norm k_{\mathcal X}.
\end{equation}
Orthogonality gives $C_P=1$, and the lower bound for $A$ gives
$\norm{h-\widetilde h}_{\mathcal X}
\leq c^{-1}\norm{Ah-A\widetilde h}_{\mathcal Y^\tau}$ on $E$.
All constants therefore refer to the same source and data norms.
\end{proof}
\end{revision}

\begin{revision}
\subsection{A nonzero normal Hessian on the disk}
\label{subsec:explicit-normal-hessian}
\end{revision}

\begin{revision}
The normal Hessian can be computed in closed form for a family of
metrics that are all isometric to the hyperbolic disk.  The family is
generated by a diffeomorphism that moves the boundary points.  Such a
diffeomorphism is not a gauge transformation in the sense of
Section~\ref{subsec:gauge-reduction}, because the interval labels are
held fixed while the geometry is moved underneath them, and it changes
the measured lengths even though the curvature never changes.  Applied
to both light-cone coordinates of the boundary cylinder, which for the
odd function $f_t$ below means $x^\pm\mapsto f_t(x^\pm)$, such a
reparametrization of the boundary circle is a boundary conformal
transformation, an asymptotic symmetry
in the sense of Brown and Henneaux~\cite{BrownHenneaux1986} rather than
an isometry of the disk, and the example shows that already the
boundary graviton sector of $\mathrm{AdS}_3$ makes the image of the
length map curved.
\end{revision}

\begin{revision}
For $|t|<1/2$, let
\begin{equation}
 f_t(\theta)=\theta+t\sin2\theta.
 \label{eq:boundary-reparametrization-family}
\end{equation}
This is an orientation-preserving diffeomorphism of the circle.  For
sufficiently small $|t|$ it extends to a smooth diffeomorphism $F_t$ of
the closed disk with $F_0=\operatorname{Id}$, which in the geodesic
collar of Equation~\eqref{eq:hyperbolic-geodesic-normal-form} can be
chosen as
\begin{equation}
 \rho'=\rho f_t'(\theta),
 \qquad
 \theta'=f_t(\theta)-\frac{\rho^2}{2}f_t''(\theta),
 \label{eq:boundary-reparametrization-extension}
\end{equation}
with a smooth cutoff to an extension close to the identity away from the
collar.  The metrics $g_t=F_t^*g_{\mathbb H}$ are complete, simple, and
of curvature $-1$.  They have the same boundary representative
$d\theta^2$, because the two $O(\rho)$ mixed terms in the pullback of
$d\rho'^2+d\theta'^2$ cancel.  In the zero-cotangent frame,
\begin{equation}
 g_t-g_{\mathbb H}=O(\rho^2),
 \qquad
 \nabla_{g_{\mathbb H}}^j\partial_t^k g_t=O(\rho^2),
 \qquad j\geq0,\quad k\geq1,
 \label{eq:explicit-metric-boundary-decay}
\end{equation}
locally uniformly in $t$.  The first two metric coefficients therefore
lie in the weighted source space and satisfy the boundary decay required
for the length variations.
\end{revision}

\begin{revision}
The endpoint normalization in
Equation~\eqref{eq:boundary-reparametrization-extension} gives the exact
renormalized length
\begin{equation}
 \mathcal L_{g_t}^x(u,v)
 =2\log\left|2\sin\frac{f_t(v)-f_t(u)}2\right|
  -\log f_t'(u)-\log f_t'(v).
 \label{eq:explicit-reparametrized-length}
\end{equation}
The signs of the endpoint terms come from the cutoff.  The surface
$\rho=\varepsilon$ is mapped to $\rho'=\varepsilon f_t'(u)+o(\varepsilon)$
at the first endpoint, and similarly at the second, so subtracting the
two logarithmic divergences leaves $-\log f_t'(u)-\log f_t'(v)$.  The
defining functions $x$ and $\rho$ have ratio one at the boundary and
give the same finite part.  Up to a factor $-\Delta$,
Equation~\eqref{eq:explicit-reparametrized-length} is the transformation
law of the logarithm of the vacuum two-point function of a boundary
primary of dimension $\Delta$ under the conformal transformation induced
by $f_t$, with the endpoint terms playing the role of the conformal
factors.  These endpoint terms
are data of the metric $g_t$ measured with the fixed defining function
$x$.  They have the endpoint form $\omega_t(u)+\omega_t(v)$ for each
$t$, but with a $t$-dependent $\omega_t$, so they are not a change of
renormalization convention in the sense of
Proposition~\ref{prop:endpoint-gauge} and they survive background
subtraction.
\end{revision}

\begin{revision}
Write
\begin{equation}
 m=\frac{u+v}{2},\qquad d=\frac{v-u}{2}\in(0,\pi),
 \qquad s_d=\sin d.
 \label{eq:midpoint-opening-coordinates}
\end{equation}
Differentiating Equation~\eqref{eq:explicit-reparametrized-length} twice
at $t=0$ gives
\begin{align}
 a_1&=4s_d^2\cos2m,
 \label{eq:explicit-length-first-coefficient}\\
 a_2&=4s_d^2+4s_d^2(8s_d^2-7)\cos4m.
 \label{eq:explicit-length-second-coefficient}
\end{align}
Before simplification the first derivative reads
\begin{equation}
 \cot d\,[\sin2v-\sin2u]-2\cos2u-2\cos2v,
\end{equation}
and the second
\begin{equation}
 -\frac{[\sin2v-\sin2u]^2}{2\sin^2d}
 +4\cos^22u+4\cos^22v,
\end{equation}
and the product-to-sum identities reduce them to
Equations~\eqref{eq:explicit-length-first-coefficient}
and~\eqref{eq:explicit-length-second-coefficient}.
\end{revision}

\begin{revision}
The tensor modes of Equation~\eqref{eq:distinguished-tensor-data} take
a simple form in these coordinates,
\begin{align}
 I_{0,2}&=-\frac23 e^{2im}s_d^2,
 \label{eq:explicit-mode-02}\\
 I_{2,2}&=e^{4im}\left(-\frac23s_d^2+\frac45s_d^4\right),
 \label{eq:explicit-mode-22}\\
 I_{0,4}&=\frac4{35}e^{4im}s_d^4.
 \label{eq:explicit-mode-04}
\end{align}
To derive them, rotate the endpoints to $e^{-id},e^{id}$ and set
$\zeta=e^{-im}z$.  With the orientation from $e^{-id}$ to $e^{id}$, the
unit-speed geodesic satisfies
\begin{equation}
 \dot\zeta=\frac{i}{2\sin d}
 (\zeta^2-2\cos d\,\zeta+1),
\end{equation}
which is the hyperbolic translation along the geodesic, a one-parameter
subgroup of disk automorphisms fixing its two ideal endpoints.
Consequently
\begin{equation}
 I_{p,q}=e^{i(p+q)m}
 \left(\frac{i}{2\sin d}\right)^{q-1}
 \int_{e^{-id}}^{e^{id}}
 \zeta^p(\zeta^2-2\cos d\,\zeta+1)^{q-1}\,d\zeta.
 \label{eq:explicit-holomorphic-mode-integral}
\end{equation}
The integrand is a polynomial, so the path may be replaced by the
straight segment $\zeta=\cos d+i\sin d\,y$, $-1\leq y\leq1$, on which
$\zeta^2-2\cos d\,\zeta+1=s_d^2(1-y^2)$.  The integrals of $1-y^2$,
$y^2(1-y^2)$, and $(1-y^2)^3$ over $[-1,1]$ are $4/3$, $4/15$, and
$32/35$, and they give
Equations~\eqref{eq:explicit-mode-02}--\eqref{eq:explicit-mode-04}.
\end{revision}

\begin{revision}
The scalar transform is equally explicit.  In the Klein coordinate
$\xi=2z/(1+|z|^2)$ the defining function satisfies $x^2=1-|\xi|^2$, and
the geodesic is the straight chord $\xi=e^{im}(\cos d+iy)$,
$|y|<\sin d$, on which
\begin{equation}
 d\sigma=\frac{\sin d}{\sin^2d-y^2}\,dy,
 \qquad y=\sin d\tanh\sigma,
 \qquad x=\sin d\,\operatorname{sech}\sigma.
\end{equation}
Hence
\begin{equation}
 I_0(2x^2)=4\sin^2d.
\end{equation}
Comparing with Equations~\eqref{eq:explicit-length-first-coefficient}
and~\eqref{eq:explicit-length-second-coefficient}, the two length
coefficients decompose as
\begin{align}
 a_1&=-6\operatorname{Re}I_{0,2},
 \label{eq:explicit-first-sector-decomposition}\\
 a_2&=I_0(2x^2)+42\operatorname{Re}I_{2,2}
       -14\operatorname{Re}I_{0,4}.
 \label{eq:explicit-second-sector-decomposition}
\end{align}
Each term on the right lies in a single sector of the decomposition in
Equation~\eqref{eq:full-even-data-decomposition}.  The function $2x^2$
belongs to $x^{1/2}L^2(M)$, so the first term lies in
$\mathcal H_0$, while the remaining terms are single modes in
$\mathcal H_2$ and $\mathcal H_4$.  Both coefficients therefore belong
to $\mathcal Y^\tau$ for every finite $\tau$, and $a_1$ satisfies the
scalar and rank-two range conditions.
\end{revision}

\begin{revision}
The first reconstructed coefficient is
\begin{equation}
 h_1=\mathscr R_2a_1=-12\operatorname{Re}[(dz)^2].
 \label{eq:explicit-tt-first-coefficient}
\end{equation}
It is a holomorphic quadratic differential, hence transverse-traceless,
and it is $O(x^2)$ in the zero-cotangent frame.  It is also the actual
first metric coefficient of a suitable extension.  The smooth real
vector field with complex component
\begin{equation}
 V(z)=\frac12\left(z^3-3\bar z+3z\bar z^2-z^2\bar z^3\right)
 \label{eq:explicit-disk-vector-field}
\end{equation}
is tangent to the boundary, restricts there to
$\sin2\theta\,\partial_\theta$, and obeys
$\mathcal L_Vg_{\mathbb H}=-12\operatorname{Re}[(dz)^2]$.  Its polar
components are
\begin{equation}
 V^r=-\frac r2(1-r^2)(3-r^2)\cos2\theta,
 \qquad
 V^\theta=\frac12(3-2r^2+r^4)\sin2\theta,
\end{equation}
and in collar coordinates they agree with the $t$-derivative of
Equation~\eqref{eq:boundary-reparametrization-extension} up to
$O(\rho^3)$.  Composing that extension with the flow of the difference,
a boundary-fixing field of positive boundary order, sets
$\partial_tF_t|_0=V$ without changing the boundary map or its normal
derivative, and $h_2=\partial_t^2g_t|_0$ still satisfies
Equation~\eqref{eq:explicit-metric-boundary-decay}.  The chain rule then
gives
\begin{equation}
 a_2=Ah_2+D^2\mathcal B_{g_{\mathbb H}}[h_1,h_1],
\end{equation}
in which the first term is geometric.  Together with
Equation~\eqref{eq:explicit-second-sector-decomposition} this
determines the normal Hessian exactly,
\begin{equation}
 \mathrm{II}_{g_{\mathbb H}}(h_1,h_1)
 =\Pi_{\mathrm{ng}}a_2
 =-14\operatorname{Re}I_{0,4}\neq0.
 \label{eq:explicit-nonzero-normal-hessian}
\end{equation}
The same decomposition shows that the Hessian datum lies in
$\mathcal Y^\tau$, since both $a_2$ and $Ah_2$ do.  The rank-four
component is not an obstruction.  It is the normal acceleration that
every metric family with this first coefficient must have.
\end{revision}

\begin{revision}
The normalization can be checked in the invariant measure.
Equation~\eqref{eq:endpoint-geodesic-coordinates} gives $b=-\cot d$,
$m=\beta+d$, and
\begin{equation}
 d\beta\,db=\frac{dm\,dd}{\sin^2d}.
\end{equation}
Using $\int_0^\pi\sin^4d\,dd=3\pi/8$ and
$\int_0^\pi\sin^6d\,dd=5\pi/16$,
\begin{align}
 \langle a_2,I_{0,4}\rangle_{L^2}
 &=\frac{4\pi}{35}\int_0^\pi
       (32\sin^6d-28\sin^4d)\,dd
   =-\frac{2\pi^2}{35},
 \label{eq:explicit-normal-moment}\\
 \|I_{0,4}\|_{L^2}^2
 &=\frac{32\pi}{1225}\int_0^\pi\sin^6d\,dd
   =\frac{2\pi^2}{245}.
 \label{eq:explicit-mode-04-norm}
\end{align}
The coefficient of $I_{0,4}$ is therefore $-7$, and the same holds for
its conjugate.  These are the $p=0$, $q=2$ modes of
Equation~\eqref{eq:adapted-data-norm}, so
\begin{equation}
 \|\Pi_{\mathrm{ng}}a_2\|_{\mathcal Y^\tau}
 =3^\tau\frac{2\pi}{\sqrt5}.
 \label{eq:explicit-normal-hessian-norm}
\end{equation}
\end{revision}

\begin{revision}
Now consider the straight interval-data path
\begin{equation}
 a^{\mathrm{lin}}(t)=t a_1.
 \label{eq:explicit-straight-data-path}
\end{equation}
Its reference value is the hyperbolic background, and its first
derivative has the regular metric reconstruction
\eqref{eq:explicit-tt-first-coefficient}.  Its second derivative
vanishes, so
\begin{equation}
 \mathfrak O_2(a_1,0)
 =14\operatorname{Re}I_{0,4},
 \qquad
 D_{\mathrm{geo}}^{(2)}(a_1,0)
 =3^\tau\frac{2\pi}{\sqrt5}.
 \label{eq:explicit-straight-path-obstruction}
\end{equation}
No regular metric two-jet realizes this path.  This is an explicit
interval-data obstruction with nonzero first response.  The path
$a^{\mathrm{lin}}$ is an interval-data construction rather than the
entropy data of a code, and its role is to show what the geometric
path \eqref{eq:explicit-reparametrized-length} forces.  Once the first shear
is fixed, the length data must acquire the rank-four component in
Equation~\eqref{eq:explicit-nonzero-normal-hessian}, and data that stay
in the rank-two sector at second order are not the data of any metric.
\end{revision}

\subsection{State-space factorization}
\label{subsec:state-space-factorization}

A single curve in state space probes only diagonal second derivatives. A geometrizable family must also satisfy mixed consistency conditions between independent state directions. Let $\Sigma$ be a finite-dimensional smooth manifold of recovery-regular full-rank logical states, let $\rho_0\in\Sigma$, and let
\begin{equation}
  \mathcal A:\Sigma\longrightarrow\mathcal Y^\tau_{\mathbb R}
  \label{eq:proto-area-state-map}
\end{equation}
be the background-subtracted proto-area map, with $\mathcal A(\rho_0)=0$. We assume that the first derivatives and all corrected second derivatives used below satisfy the regularity hypothesis of Theorem~\ref{thm:two-jet-geometrizability}, so that their ITT reconstructions \rev{are} asymptotically hyperbolic metric coefficients. Choose any torsion-free connection $\nabla^\Sigma$ on $\Sigma$. Its Hessian is the $\mathcal Y^\tau$-valued symmetric tensor
\begin{equation}
  (\nabla d\mathcal A)_{\rho_0}(X,Y)
  =X\bigl(d\mathcal A(Y)\bigr)
  -d\mathcal A(\nabla_X^\Sigma Y)
  \label{eq:state-map-hessian}
\end{equation}
for $X,Y\in T_{\rho_0}\Sigma$.

Assume the first-order condition
\begin{equation}
  \Pi_{\mathrm{ng}}d\mathcal A_{\rho_0}=0.
  \label{eq:state-map-first-order-condition}
\end{equation}
For each tangent vector define
\begin{equation}
  h_X=\mathscr R_2\bigl(d\mathcal A_{\rho_0}X\bigr).
  \label{eq:state-direction-metric}
\end{equation}
The state-space quadratic obstruction is
\begin{equation}
\begin{split}
  \mathfrak O^{(2)}_{\rho_0}(X,Y)
  :=\Pi_{\mathrm{ng}}
  \Bigl[
  &(\nabla d\mathcal A)_{\rho_0}(X,Y)\\
  &-D^2\mathcal B_{g_{\mathbb H}}[h_X,h_Y]
  \Bigr].
  \label{eq:state-space-obstruction-tensor}
\end{split}
\end{equation}

\begin{proposition}[Intrinsic character]
\label{prop:state-obstruction-intrinsic}
Under Equation~\eqref{eq:state-map-first-order-condition}, the tensor $\mathfrak O^{(2)}_{\rho_0}$ is independent of the torsion-free connection used in Equation~\eqref{eq:state-map-hessian}. It is a symmetric bilinear map
\begin{equation}
  \mathfrak O^{(2)}_{\rho_0}:
  \operatorname{Sym}^2T_{\rho_0}\Sigma
  \longrightarrow\mathcal Y_{\mathrm{ng}}.
  \label{eq:obstruction-tensor-type}
\end{equation}
It is also independent of the representatives chosen for the metric deformation classes $[h_X]$.
\end{proposition}

\begin{proof}
Let $\widehat\nabla^\Sigma$ be another torsion-free connection. Their difference is a symmetric $(1,2)$ tensor $S$, and
\begin{equation}
  (\widehat\nabla d\mathcal A)(X,Y)
  -(\nabla d\mathcal A)(X,Y)
  =-d\mathcal A(S(X,Y)).
  \label{eq:hessian-connection-change}
\end{equation}
Equation~\eqref{eq:state-map-first-order-condition} places the right-hand side in $\mathcal Y_{\mathrm{geo}}$, so its normal projection vanishes. Symmetry follows from torsion-freeness and the symmetry of $D^2\mathcal B$. Gauge independence follows from Proposition~\ref{prop:hessian-gauge-covariance}.
\end{proof}

\begin{revision}
We say that $\mathcal A$ factors through local metrics coefficientwise
to second order at $\rho_0$ if there is a regular metric-valued
two-jet $j^2_{\rho_0}\mathcal G$, with $\mathcal G(\rho_0)=g_{\mathbb H}$ and fixed
conformal infinity, such that the first and second chain-rule
coefficients obey
\begin{equation}
  j^2_{\rho_0}(\mathcal B\circ\mathcal G)
  =j^2_{\rho_0}\mathcal A.
  \label{eq:second-order-factorization-definition}
\end{equation}
Equality is understood modulo boundary-fixing diffeomorphisms on the
metric side.  If the composition of $\mathcal B$ with a representative
polynomial metric map is jointly twice differentiable at the origin
into $\mathcal Y^\tau$, this coefficientwise definition agrees with
equality of ordinary $\mathcal Y^\tau$-valued two-jets.
\end{revision}

\begin{theorem}[Multiparameter factorization criterion]
\label{thm:multiparameter-factorization}
The proto-area map $\mathcal A$ factors through local metrics
coefficientwise to second order at $\rho_0$ if and only if
\begin{equation}
  \Pi_{\mathrm{ng}}d\mathcal A_{\rho_0}=0
  \label{eq:multiparameter-linear-condition}
\end{equation}
and
\begin{equation}
  \mathfrak O^{(2)}_{\rho_0}=0.
  \label{eq:multiparameter-quadratic-condition}
\end{equation}
When these conditions hold, the first and second derivatives of a gauge-fixed realizing metric map are
\begin{align}
  d\mathcal G_{\rho_0}(X)
  &=h_X,
  \label{eq:metric-map-first-derivative}\\
  (\nabla d\mathcal G)_{\rho_0}(X,Y)
  &=\mathscr R_2
  \Bigl[
  (\nabla d\mathcal A)_{\rho_0}(X,Y)
  -D^2\mathcal B[h_X,h_Y]
  \Bigr].
  \label{eq:metric-map-second-derivative}
\end{align}
\end{theorem}

\begin{revision}
\begin{proof}
If $\mathcal A$ and $\mathcal B\circ\mathcal G$ have the same coefficient two-jet, the first-order chain rule gives
\begin{equation}
  d\mathcal A(X)=A\,d\mathcal G(X),
\end{equation}
which proves Equation~\eqref{eq:multiparameter-linear-condition} and identifies the gauge-fixed derivative with Equation~\eqref{eq:metric-map-first-derivative}. The covariant second-order chain rule gives
\begin{equation}
\begin{split}
  (\nabla d\mathcal A)(X,Y)
  ={}&A(\nabla d\mathcal G)(X,Y)\\
  &+D^2\mathcal B[h_X,h_Y].
  \label{eq:multiparameter-second-chain-rule}
\end{split}
\end{equation}
Normal projection proves Equation~\eqref{eq:multiparameter-quadratic-condition}.

Conversely, choose $\nabla^\Sigma$-normal coordinates $y^1,\ldots,y^d$ centered at $\rho_0$ and let $e_i=\partial_{y^i}|_{\rho_0}$. Put $h_i=h_{e_i}$ and define $h_{ij}$ by the right-hand side of Equation~\eqref{eq:metric-map-second-derivative}. The vanishing of the obstruction ensures that the bracket is geometric and therefore lies in the domain of $\mathscr R_2$. The local metric map
\begin{equation}
  \mathcal G(y)=g_{\mathbb H}
  +\sum_i y^ih_i
  +\frac12\sum_{i,j}y^iy^jh_{ij}
  \label{eq:constructed-metric-state-map}
\end{equation}
has the prescribed first and covariant second derivatives at the origin. For sufficiently small $y$, it is positive, simple, asymptotically hyperbolic, and has fixed conformal infinity. Equations~\eqref{eq:metric-map-first-derivative}, \eqref{eq:metric-map-second-derivative}, and the chain rule show that $\mathcal B\circ\mathcal G$ and $\mathcal A$ have the same coefficient two-jet. Joint twice differentiability of $\mathcal B\circ\mathcal G$ into $\mathcal Y^\tau$ gives the ordinary two-jet, with a Taylor remainder uniform in the finite-dimensional parameter directions. Separate differentiability along each line is not used for this upgrade.
\end{proof}
\end{revision}

\begin{revision}
\begin{theorem}[BKM--Jacobi matching]
\label{thm:bkm-jacobi-matching}
Fix the entropy-to-length conversion independently as in Section~\ref{subsec:reference-normalization}, and work on the smooth single-interval geodesic branch specified above.
Let $U$ be an affine neighborhood of a faithful logical state $\rho_0$.
For every boundary interval $A\subset S^1$, let $\mathcal N_A$ be the
boundary marginal channel, let $\mathcal D_A$ be a fixed
entropy-preserving decoding isometry, set
$\widehat{\mathcal N}_A=\mathcal D_A\circ\mathcal N_A$, and let
$\Phi_A$ be a completely positive trace-preserving (CPTP) recovered-matter
channel on the decoded output.  Put
\begin{equation}
 \widehat\omega_A(\rho)=\widehat{\mathcal N}_A(\rho),
 \qquad
 \tau_A(\rho)=\Phi_A\widehat{\mathcal N}_A(\rho),
 \label{eq:bkm-jacobi-decoded-family}
\end{equation}
and assume that $\widehat\omega_{A,0}$ and $\tau_{A,0}$ are faithful on
fixed supports.  Define
\begin{equation}
 \begin{split}
 \Delta_A(X,Y)={}&
 \mathfrak g_{\widehat\omega_{A,0}}
 \bigl(\widehat{\mathcal N}_AX,
       \widehat{\mathcal N}_AY\bigr)\\
 &-\mathfrak g_{\tau_{A,0}}
 \bigl(\Phi_A\widehat{\mathcal N}_AX,
       \Phi_A\widehat{\mathcal N}_AY\bigr),
 \end{split}
 \label{eq:bkm-jacobi-sufficiency-form}
\end{equation}
Assume that, for every $X,Y$, the interval function
\begin{equation}
 \boldsymbol\Delta(X,Y):
 A\longmapsto\Delta_A(X,Y)
 \label{eq:bkm-jacobi-defect-datum}
\end{equation}
belongs to $\mathcal Y^\tau_{\mathbb R}$, and that the first derivative
of the proto-area family below is an element of the same data space.
Let
\begin{equation}
 \begin{split}
 \mathcal A_{\rho_0}(\rho)(A)=4G_{\mathrm{eff}}\Bigl\{&
 S(\widehat\omega_A(\rho))-S(\tau_A(\rho))\\
 &-S(\widehat\omega_{A,0})+S(\tau_{A,0})\Bigr\}.
 \end{split}
 \label{eq:bkm-jacobi-proto-area}
\end{equation}
Because $\mathcal D_A$ is isometric,
$S(\widehat\omega_A)=S(\mathcal N_A(\rho))$, so that
Equation~\eqref{eq:bkm-jacobi-proto-area} is the branch proto-area of
Section~\ref{sec:proto-area}, now indexed by the full interval family.
Assume that this first derivative is linearly geometrizable and that the
data have the regularity required in
Theorem~\ref{thm:multiparameter-factorization}.  Set
\begin{equation}
 h_X=\mathscr R_2\bigl((d\mathcal A_{\rho_0})_{\rho_0}X\bigr).
 \label{eq:bkm-jacobi-metric-tangent}
\end{equation}
Then $\mathcal A_{\rho_0}$ factors through one local metric coefficientwise to second order at
$\rho_0$ if and only if, for all logical tangents $X,Y$,
\begin{equation}
 \boxed{
 \Pi_{\mathrm{ng}}
 \left[
 -4G_{\mathrm{eff}}\,\boldsymbol\Delta(X,Y)
 -D^2\mathcal B_{g_{\mathbb H}}[h_X,h_Y]
 \right]=0.}
 \label{eq:bkm-jacobi-matching}
\end{equation}
If $\widehat\omega_{A,0}=\tau_{A,0}\otimes\chi_A$ and
$\Phi_A=\tr_{A_2}$, then
\begin{equation}
 \Delta_A(X,Y)=
 \mathfrak g_{\widehat\omega_{A,0}}
 \bigl(Q_A\widehat{\mathcal N}_AX,
       Q_A\widehat{\mathcal N}_AY\bigr),
 \label{eq:bkm-jacobi-product-square}
\end{equation}
with $Q_A=I-R_{\chi_A}\Phi_A$.  The statement extends to a fixed
central fiber by taking the weighted direct sum over central blocks.
\end{theorem}
\end{revision}

\begin{proof}
The two output states in Equation~\eqref{eq:bkm-jacobi-proto-area}
depend affinely on the logical density matrix.  The mixed entropy
Hessian therefore gives
\begin{equation}
 (\nabla d\mathcal A_{\rho_0})_{\rho_0}(X,Y)(A)
 =-4G_{\mathrm{eff}}\Delta_A(X,Y)
 \label{eq:bkm-jacobi-area-hessian}
\end{equation}
in affine coordinates.  Substitution into the intrinsic obstruction
tensor in Equation~\eqref{eq:state-space-obstruction-tensor} yields
Equation~\eqref{eq:bkm-jacobi-matching}.  Its vanishing is necessary
and sufficient by
Theorem~\ref{thm:multiparameter-factorization}.  In the product case,
Theorem~\ref{thm:section2-bkm-petz-decomposition} proves
Equation~\eqref{eq:bkm-jacobi-product-square}.  Additivity of the BKM
form over fixed-weight central blocks proves the last assertion.
\end{proof}

\begin{revision}
Equation~\eqref{eq:bkm-jacobi-matching} matches the proto-area Hessian
generated by negative BKM loss with the Jacobi normal acceleration.
At a product reference, $\Delta_A(X,X)$ equals the squared norm
discarded by the Petz-sufficient lift.  The geometric term contains both the direct second variation of the line
element and the displacement of the bulk geodesic.  A local
metric exists coefficientwise to second order only when
\begin{equation}
 \Pi_{\mathrm{ng}}\bigl(-4G_{\mathrm{eff}}
 \boldsymbol\Delta(X,Y)\bigr)
 =\Pi_{\mathrm{ng}}D^2\mathcal B[h_X,h_Y].
 \label{eq:bkm-jacobi-signed-equality}
\end{equation}
In the stationary case
$d\mathcal A_{\rho_0}=0$, one has $h_X=0$ and the condition reduces to
\begin{equation}
 \Pi_{\mathrm{ng}}\boldsymbol\Delta(X,Y)=0.
 \label{eq:stationary-bkm-geometrizability}
\end{equation}
Thus a nonzero normal component of the sufficiency defect is an intrinsic
quadratic obstruction in the stationary continuum setting and does not
invoke a discrete approximation.
\end{revision}

The tensor in Equation~\eqref{eq:state-space-obstruction-tensor} is the natural nonlinear analogue of the linear witness map. It tests mixed state directions as well as individual paths. In geometric terms, the normal Hessian of the proto-area map must equal the pullback of the second fundamental form of the boundary-length data locus.

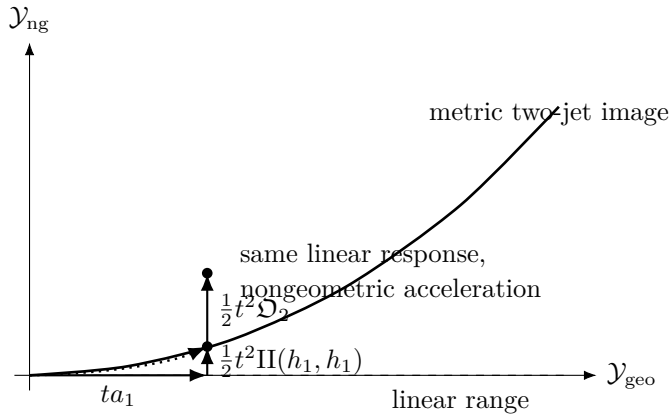
\begin{figure}[t]
\centering
\begin{tikzpicture}[font=\small,>=Latex,x=1cm,y=1cm]
  \draw[->] (-0.2,0)--(7.5,0) node[right] {$\mathcal Y_{\mathrm{geo}}$};
  \draw[->] (0,-0.2)--(0,4.4) node[above] {$\mathcal Y_{\mathrm{ng}}$};
  \draw[line width=1pt]
    plot[smooth] coordinates {(0,0) (1.3,0.12) (2.7,0.48) (4.2,1.18) (5.7,2.25) (7,3.55)};
  \node[anchor=west] at (5.15,3.45) {metric two-jet image};
  \draw[dashed] (0,0)--(7.1,0);
  \node[below] at (5.7,-0.05) {linear range};
  \draw[->,line width=0.8pt] (0,0)--(2.35,0) node[midway,below] {$t a_1$};
  \draw[->,line width=0.8pt] (2.35,0)--(2.35,0.38)
    node[midway,right] {$\frac12t^2\mathrm{II}(h_1,h_1)$};
  \fill (2.35,0.38) circle (2pt);
  \draw[->,dotted,line width=0.9pt] (0,0).. controls (1.1,0.02) and (1.9,0.18)..(2.35,0.38);
  \draw[->,line width=0.8pt] (2.35,0.38)--(2.35,1.35)
    node[midway,right] {$\frac12t^2\mathfrak O_2$};
  \fill (2.35,1.35) circle (2pt);
  \node[align=left,anchor=west] at (2.65,1.35)
    {same linear response,\\nongeometric acceleration};
\end{tikzpicture}
\caption{Coefficientwise two-jet geometry near the hyperbolic background. The tangent space is the rank-two X-ray range $\mathcal Y_{\mathrm{geo}}$. A regular metric two-jet has normal acceleration fixed by $\mathrm{II}_{g_{\mathbb H}}$. A proto-area two-jet can share the same first coefficient while differing by a quadratic obstruction $\mathfrak O_2$, which no second metric coefficient can remove. The diagram represents the regular two-jet image and does not assume a full nonlinear Banach submanifold.}
\label{fig:quadratic-geometric-locus}
\end{figure}

Figure~\ref{fig:quadratic-geometric-locus} summarizes the coefficientwise test\rev{. A} linearly admissible velocity fixes the tangential direction, while the normal Hessian fixes the acceleration of every realizing metric two-jet.

At first nonlinear order, the normal acceleration is fixed and gauge
covariant, giving a necessary and sufficient criterion for regular
coefficient two-jets. The code constructions of
Section~\ref{sec:skewed-codes-revised} distinguish local geometric state
dependence from connected nongeometric response.

\section{Code realizations of geometric and nongeometric response}
\label{sec:skewed-codes-revised}

Approximate recovery may make the proto-area depend on the logical state, while geometrizability requires the full interval family to arise from one local set of geometric degrees of freedom. Overlapping cuts reveal the difference. A local change of a bond Schmidt spectrum enters every interval through the incidence of that bond with its minimum cut and is represented by an edge coordinate. A coupling between separated bond modes produces a connected response depending on two incidences and may point outside the local edge-weight image. The four-terminal model of Section~\ref{subsec:star-example} makes this separation explicit.

\begin{revision}
Tensor-network codes realize Ryu--Takayanagi-type entropy formulae and operator-algebra quantum error correction \cite{PastawskiEtAl2015,AlmheiriDongHarlow2015,sec6BenyKempfKribs2007}. Approximate holographic codes obtained by skewing an exact code subspace were introduced in Ref.~\cite{CaoLackey2021} and developed further in Ref.~\cite{CaoEtAl2026}. The two constructions below isolate complementary sewing witnesses. The controlled-sector family has exact central area operators for every region. Its local controlled bond skews are geometrizable within a fixed cut chamber, whereas a cross-cell skew produces an exact left-null witness. The fixed-central-fiber code follows an affine path inside a fixed central fiber, starting from a reference state whose proto-areas match a positive-weight star geometry for every bipartition, and its BKM area Hessian produces a witness of its own. The microscopic dilation and the finite-dimensional calculations are given in Appendix~\ref{app:code-validation}.
\end{revision}

\subsection{Controlled bond skews}
\label{subsec:controlled-bond-skews}

\begin{revision}
Let $G=(V,E)$ be a finite tensor-network graph with boundary regions
$\calI$, and fix a regular minimum-cut chamber $\mathcal C$ with
incidence matrix $M_{\mathcal C}$.  We assume a simultaneous fixed-cut
factorization for the monitored regions.  This means that for each
$A\in\calI$ a state-independent regional decoder exposes the bond
factors crossing the cut $C_A$ together with the logical algebra
assigned to the entanglement wedge of $A$, and that these factorizations
remain valid on the bond subspace explored by the controlled
perturbation introduced below.  The property holds for the fixed-point
tensor-network codes from which graph entropy formulae are derived.  It
is an explicit hypothesis on the encoding, and graph incidence alone does
not imply it.  All entropy identities in this section are exact under
this hypothesis.
\end{revision}

\begin{revision}
The source Hilbert space is $\mathcal H_X\cong\mathbb C^m$ with
orthonormal basis $\{|x\rangle\}_{x\in\mathsf X}$.  It is part of the
logical input rather than an auxiliary boundary flag.  We consider the
exponential family
\begin{equation}
 \rho_\lambda
 =\sum_{x\in\mathsf X}p_x(\lambda)|x\rangle\!\langle x|\otimes\sigma_x,
 \qquad
 p_x(\lambda)
 =\frac{p_x^{0}e^{\lambda q_x}}{\sum_zp_z^{0}e^{\lambda q_z}},
 \label{eq:sec6-gibbs-family}
\end{equation}
a Gibbs distribution over the sectors with charge $q_x$ and inverse
temperature $\lambda$, positive on its fixed block-diagonal support.
Here $p_x^0>0$, each $\sigma_x$ has fixed support, and the charge is
centered,
\begin{equation}
 \langle q\rangle_0:=\sum_xp_x^0q_x=0,
 \qquad
 \sigma_q^2:=\langle q^2\rangle_0.
 \label{eq:sec6-centered-charge}
\end{equation}
The source observable is the classical sector label $x$, and the family
lies in the block-diagonal logical algebra.  The matter state recovered
from region $A$ is denoted $\sigma_{A,x}$.  It is the restriction of
$\sigma_x$ to the wedge algebra assigned to $A$, and the matter
subsystems assigned to different wedges need not coincide.  The
derivatives below change the probabilities $p_x$ at fixed encoding and
fixed within-sector matter states.
\end{revision}

\begin{revision}
The projectors $P_x$ generate the common center of the monitored wedge
algebras.  Their regional representatives are orthogonal sector records,
one on either side of each monitored bipartition, so the same classical
label is recoverable from complementary regions, as operator-algebra
quantum error correction permits.  A consequence is that a single
regional marginal contains no coherence between different sectors, since
the record on the other side distinguishes them.  Exact recovery in this
construction therefore refers to the block-diagonal wedge algebra
including its center, and not to the full matrix algebra
$\mathcal B(\mathcal H_X)$.  We write
\begin{equation}
 \mathcal Z_X=\operatorname{span}\{P_x:x\in\mathsf X\}
 \cong\bigoplus_{x\in\mathsf X}\mathbb C P_x
 \label{eq:sec6-common-center-algebra}
\end{equation}
for the common center.
\end{revision}

\begin{revision}
Take an isometric dilation of the fixed-cut code in which the bond
pairs and the sector-copying isometry are explicit, and in which the
vertex map is isometric on the logical and bond subspace explored by
the perturbation.  This dilation assumption is used throughout, because
a general tensor contraction need not stay isometric after a virtual
perturbation.  Let $|\Omega_e\rangle$ be the unperturbed state of bond
$e$.  For $P_x=|x\rangle\!\langle x|$ and Hermitian $K_j$ acting on bond
modes, the controlled unitary is
\begin{equation}
 U_\varepsilon^{\mathrm{virt}}
 =\exp\!\left[-i\varepsilon\sum_j\sum_xc_j(x)P_x\otimes K_j\right].
 \label{eq:sec6-controlled-isometry}
\end{equation}
If $\mathcal T=\bigotimes_vT_v$ denotes the vertex contraction, the
perturbed encoding is
\begin{equation}
 V_\varepsilon
 =\mathcal T U_\varepsilon^{\mathrm{virt}}
 \left(I_L\otimes\bigotimes_e|\Omega_e\rangle\right)
 =:\widetilde U_\varepsilon V_0,
 \label{eq:sec6-controlled-encoding-contraction}
\end{equation}
where $\widetilde U_\varepsilon$ is a physical unitary extension of the
action pushed through the isometric dilation.  The terms commute because
the supports of the $K_j$ are disjoint.  The local and cross-cell
descriptions below refer to the support on virtual bond factors, which
the physical extension need not preserve.  The sector-copying map
continues to record $x$ on both sides of each cut, and the controlled
Hamiltonian changes the conditional bond state without mixing these
records or the recoverable matter.  A local bond skew has bounded
network diameter, whereas a cross-cell skew couples modes in separated
cells.
\end{revision}

\begin{revision}
For a diagonal source state the fixed-cut decoder produces a
classical--quantum state,
\begin{equation}
 \omega_A(\lambda,\varepsilon)
 =\bigoplus_x p_x(\lambda)
   \sigma_{A,x}\otimes\tau_{A,x}(\varepsilon),
 \label{eq:sec6-decoded-cq-state}
\end{equation}
where the direct sum is over the orthogonal regional sector record,
$\tau_{A,x}$ is the conditional bond marginal after fixed-cut decoding,
and $\sigma_{A,x}$ is the recovered within-sector matter state.  The
canonical branch reads the record, retains the matter, and discards the
bond factor, so its recovered algebraic state is
$\bigoplus_xp_x(\lambda)\sigma_{A,x}$.  The operational certificate for
this fixed branch is exact recovery of the specified wedge algebra, and
unitaries acting only on the discarded bond factors do not change the
recovered entropy.  The background-subtracted branch proto-area, in
entropy units, is therefore
\begin{equation}
 a_A(\lambda,\varepsilon)
 =\sum_x\bigl[p_x(\lambda)-p_x^0\bigr]
 S\bigl(\tau_{A,x}(\varepsilon)\bigr).
 \label{eq:sec6-proto-area-cq}
\end{equation}
The Shannon entropy of the sector probabilities and the within-sector
matter entropy cancel identically between the boundary and recovered
states.  Equation~\eqref{eq:sec6-proto-area-cq} is the central affine
benchmark of Proposition~\ref{prop:section2-central-area-operator},
realized by the fixed algebraic recovery branch of
Definition~\ref{def:section2-certified-recovery-branch}, and it involves
no optimization over recoveries of quantum source coherence.
\end{revision}

For fixed $\varepsilon$, define
\begin{equation}
 L_A(\varepsilon)=\sum_{x\in\mathsf X}
 S\bigl(\tau_{A,x}(\varepsilon)\bigr)P_x.
 \label{eq:sec6-central-area-operator}
\end{equation}
In particular, $L_A(\varepsilon)\in(\mathcal Z_X)_{\mathrm{sa}}$.
Then
\begin{equation}
 a_A(\rho,\varepsilon)-a_A(\rho_0,\varepsilon)
 =\tr\!\left[(\rho-\rho_0)L_A(\varepsilon)\right].
 \label{eq:sec6-central-affine-expectation}
\end{equation}
This controlled-sector model therefore has a state-independent central
area operator and zero Hessian on the block-diagonal algebraic affine state space. A
nonzero second derivative along the Gibbs coordinate is
\begin{equation}
 \left.\frac{d^2}{d\lambda^2}\right|_0a_A(\rho_\lambda,\varepsilon)
 =\tr\!\left[\rho_\lambda''(0)L_A(\varepsilon)\right],
 \label{eq:sec6-gibbs-curve-acceleration}
\end{equation}
which is acceleration of the chosen state curve rather than an
intrinsic within-sector area Hessian.  Its role is operator-valued
sewing, as in Theorem~\ref{thm:operator-valued-central-sewing}.

\begin{revision}
\begin{proposition}[Exact algebraic recovery and loss of source coherence]
\label{prop:sec6-approximate-recovery}
The canonical branch recovers the common center and the assigned
within-sector matter exactly.  With the other logical factors fixed,
the induced source channel on either side of a monitored bipartition is
\begin{equation}
 \mathcal D_X(|x\rangle\!\langle y|)
 =\delta_{xy}|x\rangle\!\langle x|.
 \label{eq:sec6-schur-channel}
\end{equation}
For $m\geq2$, every regional source recovery $\mathcal R_A$ obeys
\begin{equation}
 \bigl\|\mathcal R_A\circ\mathcal N_A^X
       -\operatorname{id}_X\bigr\|_\diamond\geq1,
 \label{eq:sec6-diamond-bound}
\end{equation}
where $\mathcal N_A^X$ is the regional encoding channel restricted
to the source input.  Thus exact recovery of the center does not give
small quantum-channel error on $\mathcal B(\mathcal H_X)$.
\end{proposition}
\end{revision}

\begin{revision}
\begin{proof}
For $x\ne y$, tracing the complementary record gives
$\operatorname{tr}(|x\rangle\!\langle y|)=0$, which removes the
off-diagonal source block before any regional recovery is applied.
Reading the accessible record therefore implements $\mathcal D_X$,
and preserves every diagonal source state, including the Gibbs family.
For two distinct sectors, the source states
$|\pm\rangle=(|x\rangle\pm|y\rangle)/\sqrt2$ have identical
regional marginals.  Any recovered outputs for these two inputs are
therefore the same state $\zeta$.  Since their ideal outputs have
trace-norm distance $2$, the triangle inequality implies that at least
one of $\|\zeta-|+\rangle\!\langle+|\|_1$ and
$\|\zeta-|-\rangle\!\langle-|\|_1$ is at least $1$.
This proves the diamond-norm lower bound.  The conclusion holds for
every controlled bond strength because the Hamiltonian preserves the
sector records.
\end{proof}
\end{revision}

The derivatives of the state family enter through a simple centered-moment identity. For any function $f$ on $\mathsf X$,
\begin{equation}
 \left.\frac{d}{d\lambda}\langle f\rangle_\lambda\right|_{0}
 =\langle qf\rangle_0,
 \qquad
 \left.\frac{d^2}{d\lambda^2}\langle f\rangle_\lambda\right|_{0}
 =\left\langle(q^2-\sigma_q^2)f\right\rangle_0.
 \label{eq:sec6-gibbs-jet}
\end{equation}
These formulae follow by differentiating the normalized exponential family.

\subsection{Local skews and exact geometrizability}
\label{subsec:local-skews-exact}

Each local bond is taken to contain a two-dimensional Schmidt sector
\begin{equation}
 |\Omega(\vartheta)\rangle
 =\cos\vartheta\,|00\rangle+\sin\vartheta\,|11\rangle,
 \qquad 0<\vartheta<\frac{\pi}{4}.
 \label{eq:sec6-schmidt-state}
\end{equation}
Let
\begin{equation}
 K=i\bigl(|11\rangle\!\langle00|-|00\rangle\!\langle11|\bigr).
 \label{eq:sec6-schmidt-generator}
\end{equation}
On the Schmidt sector,
\begin{equation}
 e^{-i\alpha K}|\Omega(\vartheta)\rangle
 =|\Omega(\vartheta+\alpha)\rangle.
 \label{eq:sec6-angle-shift}
\end{equation}
The one-sided entropy is
\begin{equation}
 s(\vartheta)
 =h_2(\sin^2\vartheta),
 \qquad
 s'(\vartheta)
 =\sin(2\vartheta)\log\!\bigl(\cot^2\vartheta\bigr)>0,
 \label{eq:sec6-bond-entropy}
\end{equation}
where logarithms are natural. The strict derivative is important. Near a nonmaximal bond, the Schmidt angle is a \rev{local} coordinate on the space of edge entropies.

For edge $e$, choose a source profile $c_e(x)$ and insert the controlled generator $P_x\otimes K_e$. If $e\notin C_A$, both halves of the bond lie on the same side of the cut, and the rotation does not change the entropy of $A$. If $e\in C_A$, it contributes $s(\vartheta_e+\varepsilon c_e(x))$.

\begin{revision}
\begin{theorem}[All-order geometrizability of local controlled skews]
\label{thm:sec6-local-all-order}
Suppose that the unperturbed absolute regional remainders are
$M_{\mathcal C}\boldsymbol\ell^0$, where
$\boldsymbol\ell^0$ is strictly positive and lies in the interior of
$\mathcal C$.  Include any state-independent spectator contribution
in this reference matching.  Define the absolute edge weights by
\begin{equation}
 \ell_e(\lambda,\varepsilon)
 =\ell_e^0+\sum_xp_x(\lambda)
 \left[s(\vartheta_e+\varepsilon c_e(x))-s(\vartheta_e)\right],
 \label{eq:sec6-absolute-local-edge-weight}
\end{equation}
and their changes relative to the reference state at the same encoding by
\begin{equation}
 \begin{split}
 w_e(\lambda,\varepsilon)
 &:={\ell_e(\lambda,\varepsilon)-\ell_e(0,\varepsilon)}\\
 &=\sum_x[p_x(\lambda)-p_x^0]
 s(\vartheta_e+\varepsilon c_e(x)).
 \end{split}
 \label{eq:sec6-local-edge-weight}
\end{equation}
For sufficiently small $|\varepsilon|$, these absolute weights remain
positive and in $\mathcal C$.  The absolute proto-area vector is
$\bfa_G(\boldsymbol\ell(\lambda,\varepsilon))$, and its centered
response is exactly
\begin{equation}
 \mathbf a^{\mathrm{loc}}(\lambda,\varepsilon)
 =M_{\mathcal C}\mathbf w(\lambda,\varepsilon).
 \label{eq:sec6-local-factorization}
\end{equation}
Consequently
\begin{equation}
 y^{\mathsf T}\mathbf a^{\mathrm{loc}}(\lambda,\varepsilon)=0
 \qquad(y\in\ker M_{\mathcal C}^{\mathsf T})
 \label{eq:sec6-local-witness-zero}
\end{equation}
holds to all orders in the state parameter.  The quantities $w_e$
are signed changes of edge weights, and positivity applies to $\ell_e$.
Without the absolute reference-matching hypothesis, the same
factorization proves compatibility of the centered response with
$M_{\mathcal C}$.
\end{theorem}
\end{revision}

\begin{revision}
\begin{proof}
At fixed $x$, entropy additivity over the independent cut bonds gives
the change
$\sum_{e\in C_A}[s(\vartheta_e+\varepsilon c_e(x))-s(\vartheta_e)]$.
Averaging over $p_x(\lambda)$ and adding the matched unperturbed
remainder gives $(M_{\mathcal C}\boldsymbol\ell)_A$.
Since $\mathsf X$ and $E$ are finite, the absolute value of each
bracket is bounded by $C|\varepsilon|$, uniformly in $x$.
The same bound holds after averaging over any probability vector.
Strict positivity and the finitely many strict cut inequalities at
$\boldsymbol\ell^0$ therefore persist for sufficiently small
$|\varepsilon|$.  On this chamber the minimum-cut map equals
$M_{\mathcal C}\boldsymbol\ell$.  Subtracting its value at $\lambda=0$
proves Equation~\eqref{eq:sec6-local-factorization}, and multiplication
by a left-null vector proves Equation~\eqref{eq:sec6-local-witness-zero}.
\end{proof}
\end{revision}

\begin{revision}
Because $s'(\vartheta_e)\neq0$, the same construction locally realizes every sufficiently small collection of conditional edge entropies. The inverse function theorem gives a neighborhood $U_e$ of $s(\vartheta_e)$ on which $s^{-1}$ is smooth. For a sufficiently small target entropy change $f_e(x)$ with
$s(\vartheta_e)+f_e(x)\in U_e$, the angle displacement
\begin{equation}
 \alpha_e(x)
 =s^{-1}\!\bigl(s(\vartheta_e)+f_e(x)\bigr)-\vartheta_e
 \label{eq:sec6-inverse-entropy-coordinate}
\end{equation}
produces the prescribed conditional edge-entropy change in every
sector.  These finitely many conditional changes can be chosen
independently near zero.  At fixed source probabilities their
expectations determine the absolute edge weights, and varying those
probabilities gives the centered profiles in
Equation~\eqref{eq:sec6-local-edge-weight}.  Every region uses this
same family of edge weights.
\end{revision}

\subsection{Cross-cell skews and connected cut response}
\label{subsec:cross-cell-skews}

A nonlocal skew is detected by its connected response across two bond incidences. Consider two unperturbed bonds
\begin{equation}
 |\Psi_r\rangle
 =|\Omega_r\rangle_{u\bar u}\otimes|\Omega_r\rangle_{v\bar v},
 \qquad
 |\Omega_r\rangle
 =\sqrt r\,|00\rangle+\sqrt{1-r}\,|11\rangle,
 \qquad \frac12<r<1.
 \label{eq:sec6-two-bond-state}
\end{equation}
The modes $u$ and $v$ belong to different network cells. We couple them by
\begin{equation}
 U_\times(\alpha)=e^{-i\alpha X_uX_v}.
 \label{eq:sec6-cross-gate}
\end{equation}
If $u$ and $v$ lie on the same side of a bipartition, this is a local unitary and the entropy is unchanged. The nontrivial configuration is the crossed cut
\begin{equation}
 u,\bar v\in A,
 \qquad
 \bar u,v\in\bar A.
 \label{eq:sec6-crossed-configuration}
\end{equation}
Both original bonds cross the cut, so the reduced density matrix is full rank and its entropy is analytic in $\alpha$.

\begin{revision}
Let $F_r(\alpha)$ denote the entropy increment of $A$ in the configuration \eqref{eq:sec6-crossed-configuration}. In the ordered basis
$\{|00\rangle,|01\rangle,|10\rangle,|11\rangle\}$ of $u\bar v$, the reduced state is
\begin{equation}
 \rho_A(\alpha)=
 \begin{pmatrix}
 rA&0&0&i\zeta_r\\
 0&(1-r)A&i\zeta_r&0\\
 0&-i\zeta_r&rB&0\\
 -i\zeta_r&0&0&(1-r)B
 \end{pmatrix},
 \label{eq:sec6-cross-reduced-matrix}
\end{equation}
where
\begin{equation}
 A=rc^2+(1-r)s^2,
 \qquad
 B=(1-r)c^2+rs^2,
 \qquad
 \zeta_r=(2r-1)\sqrt{r(1-r)}\,cs,
 \label{eq:sec6-cross-matrix-abbreviations}
\end{equation}
with $c=\cos\alpha$ and $s=\sin\alpha$. Its four eigenvalues are obtained from the two displayed $2\times2$ blocks, and
\begin{equation}
 F_r(\alpha)=S\bigl(\rho_A(\alpha)\bigr)-2h_2(r).
 \label{eq:sec6-exact-Fr}
\end{equation}
This is the exact entropy increment when both supporting bonds cross
with opposite orientations.  A second configuration is essential for
extracting the connected part.  If only the $u\bar u$ bond crosses and
$u$ and $v$ are on opposite sides, the reduced state of $u$ has
eigenvalues $A$ and $1-A$, and the entropy increment is
\begin{equation}
 f_r(\alpha)=h_2\!\left(r\cos^2\alpha+(1-r)\sin^2\alpha\right)-h_2(r).
 \label{eq:sec6-single-cross-response}
\end{equation}
The response not representable by a one-edge weight is
\begin{equation}
 H_r(\alpha)=F_r(\alpha)-f_r(\alpha).
 \label{eq:sec6-connected-response}
\end{equation}
Its meaning follows directly from the two one-mode marginals of
Equation~\eqref{eq:sec6-cross-reduced-matrix}.  They have entropies
$h_2(A)$ and $h_2(r)$, respectively.  Hence
\begin{equation}
 \begin{split}
 H_r(\alpha)
 &=S(\rho_{u\bar v}(\alpha))-h_2(A)-h_2(r)\\
 &=-I(u:\bar v)_{\rho_{u\bar v}(\alpha)}\leq0,
 \end{split}
 \label{eq:sec6-connected-mutual-information}
\end{equation}
where $I(u:\bar v)=S(\rho_u)+S(\rho_{\bar v})-S(\rho_{u\bar v})$.
The crossed gate creates correlations between the two modes retained
on the same side of the cut.  Their mutual information is exactly
the entropy contribution left after subtracting the one-edge change.
\end{revision}

\begin{proposition}[Raw and connected cross-cell susceptibilities]
\label{prop:sec6-cross-susceptibility}
The functions $F_r$, $f_r$, and $H_r$ are even and analytic near the
origin.  The raw crossed response satisfies
\begin{equation}
 F_r(0)=F_r'(0)=0,
 \qquad
 F_r''(0)=\chi(r),
 \label{eq:sec6-Fr-jet}
\end{equation}
where
\begin{equation}
 \chi(r)
 =2\delta(1+\delta^2)\operatorname{arctanh}\delta-2\delta^2>0,
 \qquad \delta=2r-1.
 \label{eq:sec6-chi-r}
\end{equation}
Equivalently,
\begin{equation}
 F_r(\alpha)=\frac{\chi(r)}{2}\alpha^2+O(\alpha^4).
 \label{eq:sec6-Fr-expansion}
\end{equation}
For $r=4/5$,
\begin{equation}
 \chi(4/5)=\frac{204}{125}\log2-\frac{18}{25}>0.
 \label{eq:sec6-chi-example}
\end{equation}
The connected response instead obeys
\begin{align}
 H_r(0)&=H_r'(0)=0,
 \qquad H_r''(0)=\widetilde\chi(r),
 \label{eq:sec6-connected-jet}\\
 \widetilde\chi(r)
 &=-2\delta\left[(1-\delta^2)\operatorname{arctanh}\delta
 +\delta\right]<0,
 \label{eq:sec6-connected-susceptibility}\\
 H_r(\alpha)&=\frac{\widetilde\chi(r)}2\alpha^2+O(\alpha^4).
 \label{eq:sec6-connected-expansion}
\end{align}
For $r=4/5$,
\begin{equation}
 \widetilde\chi(4/5)
 =-\frac{96}{125}\log2-\frac{18}{25}
 \simeq-1.252337.
 \label{eq:sec6-connected-example}
\end{equation}
\end{proposition}

\begin{proof}
Appendix~\ref{appD:cross-cell-calculation} derives the reduced parity blocks and evaluates the entropy two-jet. Equations~\eqref{eq:sec6-chi-intermediate} and~\eqref{eq:sec6-single-cross-susceptibility} give the raw and one-edge responses\rev{, and} their difference is Equation~\eqref{eq:sec6-connected-susceptibility}, whose sign is strict for $0<2r-1<1$.
\end{proof}

The complete cut pattern must be recorded before applying a witness.
In the four-terminal chamber of
Equation~\eqref{eq:star-cut-matrix}, place $u$ on the boundary side of
$e_1$ and $v$ on the central side of $e_2$.  In the region ordering
$(1,2,3,4,12,23)$, the six entropy increments are
\begin{equation}
 \mathbf q_\times(\alpha)
 =\bigl(f_r(\alpha),0,0,0,F_r(\alpha),0\bigr)^{\mathsf T}.
 \label{eq:sec6-complete-six-region-response}
\end{equation}
The first entry is a local change of edge $e_1$, so
\begin{equation}
 \mathbf q_\times(\alpha)
 =M\bigl(f_r(\alpha),0,0,0\bigr)^{\mathsf T}
 +H_r(\alpha)b^\times,
 \qquad
 b^\times=(0,0,0,0,1,0)^{\mathsf T}.
 \label{eq:sec6-exact-connected-decomposition}
\end{equation}
Reversing both bond orientations moves the single-cross entries but
leaves the connected coefficient $H_r=F_r-f_r$ unchanged after local
edge subtraction.  Equation~\eqref{eq:sec6-exact-connected-decomposition}
fixes the convention for $b^\times$ below.

For the remainder of this subsection, $\mathcal C$ denotes this
four-terminal chamber, $M_{\mathcal C}=M$ is the matrix in
Equation~\eqref{eq:star-cut-matrix}, and $b^\times\in\mathbb R^6$ has
the ordering displayed above.  The additional local controlled
rotations are assumed to act on bond subfactors disjoint from
$u,\bar u,v,\bar v$.  Their generators therefore commute with the
cross gate and their decoded bond states factor from the crossed
two-bond subsystem.

\begin{revision}
The controlled cross-cell skew uses
\begin{equation}
 U_{\times,\varepsilon}
 =\sum_xP_x\otimes e^{-i\varepsilon c_\times(x)X_uX_v}.
 \label{eq:sec6-controlled-cross-gate}
\end{equation}
Combining it with arbitrary local skews satisfying this disjoint-support
condition gives the exact data decomposition
\begin{equation}
 \mathbf a(\lambda,\varepsilon)
 =M_{\mathcal C}\mathbf w(\lambda,\varepsilon)
 +\beta(\lambda,\varepsilon)b^\times,
 \label{eq:sec6-local-plus-cross}
\end{equation}
where
\begin{equation}
 \beta(\lambda,\varepsilon)
 =\sum_x[p_x(\lambda)-p_x^0]
 H_r\bigl(\varepsilon c_\times(x)\bigr).
 \label{eq:sec6-beta-exact}
\end{equation}
Equation~\eqref{eq:sec6-local-plus-cross} is an exact entropy
identity for the assumed fixed-cut factorization and the oriented
incidence pattern in
Equation~\eqref{eq:sec6-exact-connected-decomposition}.  The chamber
$\mathcal C$ is the one selected by the unperturbed code.  At
$\varepsilon=0$ the bonds are product states independent of the source,
and the regional remainders are the star cuts with the bond entropies
as weights, which lie in the regular chamber by the standing
assumption on $\mathcal C$.  The controlled skew is a deformation of
this matched encoding, and since the cut gaps are strict the same
chamber remains the relevant one for all sufficiently small
$|\varepsilon|$.
\end{revision}

At the operator level, collect the local one-edge contributions into
$W_e(\varepsilon)\in(\mathcal Z_X)_{\mathrm{sa}}$ and define
\begin{equation}
 K_\times(\varepsilon)=\sum_x
 H_r\bigl(\varepsilon c_\times(x)\bigr)P_x.
 \label{eq:sec6-connected-central-operator}
\end{equation}
Thus $K_\times(\varepsilon)\in(\mathcal Z_X)_{\mathrm{sa}}$, and all
$L_A(\varepsilon)$ are compared through the fixed embeddings into
$\mathcal Z_X$.
The regional central area operators obey the exact decomposition
\begin{equation}
 (L_A(\varepsilon))_{A\in\calI}
 =M_{\mathcal C}(W_e(\varepsilon))_{e\in E}
 +b^\times K_\times(\varepsilon).
 \label{eq:sec6-operator-valued-cross-decomposition}
\end{equation}
Consequently every $y\in\ker M_{\mathcal C}^{\mathsf T}$ satisfies
\begin{equation}
 \sum_Ay_AL_A(\varepsilon)
 =(y^{\mathsf T}b^\times)K_\times(\varepsilon).
 \label{eq:sec6-operator-valued-cross-witness}
\end{equation}
Whenever the right-hand side is nonzero, the exact regional central
operators fail the sewing criterion in
Theorem~\ref{thm:operator-valued-central-sewing}.  This conclusion is
state independent and is stronger than detecting the violation on one
chosen Gibbs path.

\begin{revision}
The deformed code leaves the cut image already at the reference state,
before any state derivative is taken.  Write $\mathbf p_0(\varepsilon)$
for the vector with entries $\sum_xp_x^0 S(\tau_{A,x}(\varepsilon))$,
the absolute remainders at $\lambda=0$.  Then
\begin{equation}
 y^{\mathsf T}\mathbf p_0(\varepsilon)
 =(y^{\mathsf T}b^\times)
 \sum_xp_x^0H_r(\varepsilon c_\times(x)),
 \label{eq:sec6-cross-reference-witness}
\end{equation}
up to a state-independent spectator contribution, which is matched by
the star and annihilated by $y^{\mathsf T}$.  For a nonzero coupling
the sum is strictly negative for all sufficiently small nonzero
$\varepsilon$, because $H_r$ is minus a mutual information.  The skewed
encoding therefore fails the zeroth-order sewing condition by an amount
of order $\varepsilon^2$, its central area operators fail the
operator-valued condition by
Equation~\eqref{eq:sec6-operator-valued-cross-witness}, and the state
derivatives inherit the failure through the coefficient $\beta$.  The
three statements are the same connected term seen at the level of the
reference, of the operators, and of the response.  The example in
Section~\ref{subsec:fixed-center-erasure-code} is complementary.  There
the reference is matched exactly and the first response vanishes, and
the failure appears only at second order.
\end{revision}

\begin{revision}
\begin{theorem}[Exact first-order witness of the cross-cell skew]
\label{thm:sec6-first-order-nongeometric}
Suppose $b^\times\notin\operatorname{im}M_{\mathcal C}$, and choose
$y\in\ker M_{\mathcal C}^{\mathsf T}$ with $y^{\mathsf T}b^\times\neq0$. Then
\begin{equation}
 y^{\mathsf T}\mathbf a(\lambda,\varepsilon)
 =\beta(\lambda,\varepsilon)y^{\mathsf T}b^\times.
 \label{eq:sec6-exact-cross-witness}
\end{equation}
For the binary Gibbs source $\mathsf X=\{-1,+1\}$ with $p_x^0=1/2$ and $q_x=x$,
\begin{equation}
 \left.\frac{d\beta}{d\lambda}\right|_0
 =\frac12\left[
 H_r\bigl(\varepsilon c_\times(+1)\bigr)
 -H_r\bigl(\varepsilon c_\times(-1)\bigr)
 \right].
 \label{eq:sec6-binary-beta-prime}
\end{equation}
If $c_\times(+1)^2\neq c_\times(-1)^2$, this derivative is nonzero for all sufficiently small nonzero $\varepsilon$. More specifically,
\begin{equation}
 \left.\frac{d\beta}{d\lambda}\right|_0
 =\frac{\widetilde\chi(r)\varepsilon^2}{4}
 \left[c_\times(+1)^2-c_\times(-1)^2\right]+O(\varepsilon^4).
 \label{eq:sec6-binary-beta-expansion}
\end{equation}
The state-dependent response is therefore nongeometric already at first order in the state parameter.
\end{theorem}
\end{revision}

\begin{proof}
Multiplying \eqref{eq:sec6-local-plus-cross} by $y^{\mathsf T}$ removes the local term and gives \eqref{eq:sec6-exact-cross-witness}. Equation~\eqref{eq:sec6-binary-beta-prime} is the first identity in \eqref{eq:sec6-gibbs-jet}. The expansion follows from \eqref{eq:sec6-connected-expansion}. Since $\widetilde\chi(r)<0$, the leading coefficient is nonzero whenever the squared couplings differ.
\end{proof}

Figure~\ref{fig:sec6-controlled-bond-skews} summarizes the mechanism. A local controlled rotation changes one Schmidt spectrum and hence one edge coordinate. A cross-cell gate also has a single-cross contribution, which must first be absorbed into the local edge profile.  The remainder $H_r$ is a connected two-edge observable and has the strictly negative susceptibility in Equation~\eqref{eq:sec6-connected-susceptibility}.

\begin{figure}[t]
 \centering
 \includegraphics[width=\textwidth]{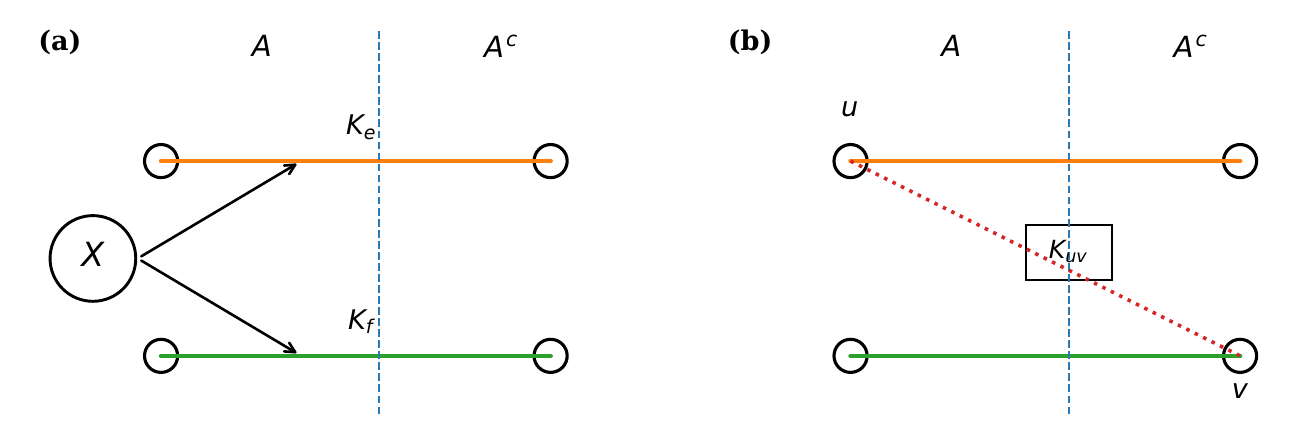}
 \caption{Controlled bond skews. Local Schmidt rotations alter the entropy of each crossed bond independently and produce data in $\operatorname{im}M_{\mathcal C}$. A cross-cell gate produces a one-edge response and the connected component $H_r b^\times$ in Equation~\eqref{eq:sec6-exact-connected-decomposition}.}
 \label{fig:sec6-controlled-bond-skews}
\end{figure}

\begin{revision}
\subsection{A fixed-central-fiber code with a geometric reference state}
\label{subsec:fixed-center-erasure-code}
\end{revision}

\begin{revision}
The preceding family changes the central probabilities.  We now keep
those probabilities fixed and vary a correlation between two logical
qubits instead.  The construction is arranged so that the zeroth-order
question of Section~\ref{subsec:reference-normalization} has a clean
answer.  The reference proto-areas agree with the minimum cuts of one
positive-weight four-terminal star for every boundary bipartition, and
all competing cuts have strictly positive gaps.  All first derivatives
vanish.  The second derivatives nevertheless violate a cut relation of
that same star.  Each
logical qubit is split between two sites, and a region that holds both
halves of both qubits sees the full two-qubit correlation, whereas a
region that holds only one half of each sees it attenuated.  The
proto-area therefore responds to the correlation only through the joint
presence of both qubits.  No assignment of weights to single edges can
reproduce a response that depends on two incidences at once, which is
the same connected structure that produced the cross-cell witness, now
generated by information loss inside a fixed central fiber.
\end{revision}

\begin{revision}
Let $\pi_2=I_2/2$, $\pi_4=I_4/4$, and
\begin{equation}
 \mathcal M_L=\mathcal B(\mathbb C^2\otimes\mathbb C^2)
 \oplus\mathcal B(\mathbb C^2\otimes\mathbb C^2),
 \qquad Z(\mathcal M_L)=\operatorname{span}\{P_0,P_1\}.
 \label{eq:fixed-fiber-logical-algebra}
\end{equation}
For $0<q<1$ consider the faithful affine path
\begin{equation}
 \rho_t=q\varrho_t\oplus(1-q)\pi_4,
 \qquad \varrho_t=\frac14(I_4+tZ\otimes Z),
 \qquad |t|<1.
 \label{eq:fixed-fiber-affine-path}
\end{equation}
Both one-qubit marginals are $\pi_2$.  The central probabilities
\begin{equation}
 \tr(P_0\rho_t)=q,\qquad \tr(P_1\rho_t)=1-q
 \label{eq:fixed-fiber-central-probabilities}
\end{equation}
are constant, whereas the correlation in the first matrix block varies.
\end{revision}

\begin{revision}
Each logical qubit is encoded by the balanced erasure isometry
\begin{equation}
 E|\psi\rangle
 =\frac{|\psi\rangle_A|e\rangle_E
              +|e\rangle_A|\psi\rangle_E}{\sqrt2},
 \qquad \mathbb C^3=\mathbb C^2\oplus\mathbb C|e\rangle.
 \label{eq:fixed-fiber-erasure-isometry}
\end{equation}
The qubit and erasure subspaces are orthogonal, so $E^\dagger E=I_2$.
The isometry puts the qubit into a symmetric superposition of being
transmitted to the receiving system $A$ and being transmitted to the
environment $E$, with the orthogonal flag $|e\rangle$ marking the empty
side.  The receiving systems $A_1,A_2$ of the two qubits occupy sites
$1$ and $2$, and both environment systems $E_1,E_2$ occupy site $4$.
To give the reference state the entanglement pattern of a star we also
place one qutrit $T_i$ at each site and a Bell pair on extra qubits
$D_3,D_4$, in the state
\begin{equation}
 |\Omega\rangle_{T_1T_2T_3T_4}
 =\frac13\sum_{i,j\in\mathbb Z_3}|i,j,i+j,i+2j\rangle,
 \qquad
 |\Phi\rangle_{D_3D_4}=\frac{|00\rangle+|11\rangle}{\sqrt2}.
 \label{eq:fixed-fiber-resource-state}
\end{equation}
Every two-qutrit marginal of $|\Omega\rangle$ is maximally mixed, so
$|\Omega\rangle$ is an absolutely maximally entangled state in the
terminology of Refs.~\cite{HelwigEtAl2012,HelwigCui2013}.  The marginal
property follows directly from the displayed sum, as shown in
Appendix~\ref{appD:fixed-center-calculation}.  These resource states
carry no logical information.  Their only role is to make the reference
entropies those of a star with positive weights.
\end{revision}

\begin{revision}
Each site carries a copied center record $C_i\simeq\mathbb C^2$.
Suppressing tensor-factor reorderings, the single global encoding is
\begin{equation}
 V(\psi_0\oplus\psi_1)
 =\sum_{\alpha=0}^1|\alpha\rangle_{C_1}\cdots
 |\alpha\rangle_{C_4}\otimes
 (E\otimes E)|\psi_\alpha\rangle
 \otimes|\Omega\rangle\otimes|\Phi\rangle.
 \label{eq:fixed-fiber-global-encoding}
\end{equation}
The orthogonal center records and $E^\dagger E=I_2$ imply
$V^\dagger V=I$.  Thus the regional states below are marginals of one
encoding, rather than separately prescribed density matrices.
\end{revision}

\begin{revision}
The recovery prescription is fixed for every nonempty proper region
$A\subset\{1,2,3,4\}$.  Define
\begin{equation}
 J_A=\{i\in\{1,2\}:A\cap\{i,4\}\ne\varnothing\},
 \quad m_A=|J_A|,
 \quad r_A=\#\{i\in\{1,2\}:|A\cap\{i,4\}|=1\}.
 \label{eq:fixed-fiber-regional-target}
\end{equation}
The target algebra consists of the copied center and the logical
qubits in $J_A$.  Its two blocks are
$\mathcal B((\mathbb C^2)^{\otimes m_A})$, and when $m_A=0$ only the
center remains.  The decoder discards the resource state, reads the
center, and inverts $E$ whenever both halves of an encoded qubit are
present.  On a single half it keeps the qubit when received and
prepares $\pi_2$ on erasure.  Its recovered channel on that qubit is
\begin{equation}
 \mathcal D(X)=\frac12X+\frac12\pi_2\tr X.
 \label{eq:fixed-fiber-one-half-recovery}
\end{equation}
The regional decoder is the product of these operations, with any
absent logical qubits traced out.  Among all regional decoders,
including joint decoders acting on the full regional output, it
minimizes the diamond distance to the target marginal channel.  The
diamond norm is the completely bounded trace norm, the worst-case
distinguishability of two channels over all inputs including entangled
reference systems, and the unrestricted minimax proof is given in
Appendix~\ref{appD:fixed-center-calculation}.
\end{revision}

\begin{revision}
We use the fixed-recovery proto-area
$p_A(t)=S(\omega_A(t))-S(\tau_A(t))$ in entropy units and put
$a_A(t)=\kappa[p_A(t)-p_A(0)]$ with $\kappa=4G_{\mathrm{eff}}$.  The
target algebras just specified need not be complementary logical
factors, and two opposite regions may each attempt an approximate
recovery of the same qubit.  This is a fixed-recovery example with the
channels specified above.  The diamond-minimax calibration is its
operational certificate, and we do not identify it with the jointly
coherent-information-optimal prescription of Ref.~\cite{CaoEtAl2026}.
\end{revision}

\begin{revision}
Write $K=\log3$, $L=\log2$, and let $\epsilon_A$ equal one when
exactly one of sites $3,4$ belongs to $A$, and zero otherwise.  The
reference data are
\begin{equation}
 p_A(0)=K\min\{|A|,4-|A|\}
            +L\epsilon_A+\frac{L}{2}r_A.
 \label{eq:fixed-fiber-reference-formula}
\end{equation}
The center entropy cancels between the boundary and recovered states.
The three surviving terms are respectively the qutrit resource entropy,
the cut Bell-pair entropy, and the erasure remainder.  They are actual
entropies of the encoding, without region-dependent subtractions.
\end{revision}

\begin{revision}
\begin{proposition}[Matched reference and exact BKM defect]
\label{prop:fixed-fiber-exact-defect}
For every nonempty proper region, $p_A(0)$ is the minimum-cut value of
the four-terminal star with positive edge weights
\begin{equation}
 \bfw_0=(K+L/2,\ K+L/2,\ K+L,\ K+2L)^{\mathsf T}.
 \label{eq:fixed-fiber-star-reference}
\end{equation}
All cuts are unique.  In length units the edge weights are
$\kappa\bfw_0$.  Every first derivative vanishes.  In the ordering
$\calI=(1,2,3,4,12,23)$,
\begin{equation}
 \mathbf a_1=0,\qquad
 \boldsymbol\Delta=\frac{3q}{16}(0,0,0,1,1,0)^{\mathsf T},
 \label{eq:fixed-fiber-defect-vector}
\end{equation}
and hence
\begin{equation}
 \mathbf a_2=-\frac{3\kappa q}{16}
                   (0,0,0,1,1,0)^{\mathsf T}.
 \label{eq:fixed-fiber-second-jet}
\end{equation}
Here $\Delta_A$ is the loss of BKM norm under the fixed regional
recovery.  The complete set of bipartitions and defects is given in
Table~\ref{tab:fixed-fiber-all-regions}.
\end{proposition}
\end{revision}

\begin{revision}
\begin{table}[t]
\ifshowchanges\color{blue}\else\color{black}\fi
 \centering
 \begin{tabular}{ccccc}
 \toprule
 $A$ & $\bar A$ & $p_A(0)=p_{\bar A}(0)$
       & $\Delta_A/q$ & $\Delta_{\bar A}/q$\\
 \midrule
 $1$ & $234$ & $K+L/2$ & $0$ & $1/4$\\
 $2$ & $134$ & $K+L/2$ & $0$ & $1/4$\\
 $3$ & $124$ & $K+L$ & $0$ & $0$\\
 $4$ & $123$ & $K+2L$ & $3/16$ & $3/16$\\
 $12$ & $34$ & $2K+L$ & $3/16$ & $3/16$\\
 $13$ & $24$ & $2K+3L/2$ & $0$ & $1/4$\\
 $23$ & $14$ & $2K+3L/2$ & $0$ & $1/4$\\
 \bottomrule
 \end{tabular}
 \caption{Reference matching and second-order BKM defects for every
 nontrivial bipartition of the four boundary sites.  All first
 derivatives vanish and $a_A''(0)=-\kappa\Delta_A$.}
 \label{tab:fixed-fiber-all-regions}
\end{table}
\end{revision}

\begin{revision}
\begin{proof}
Equation~\eqref{eq:fixed-fiber-reference-formula} gives the third column
of the table.  The gaps between the larger and smaller star cuts,
ordered by its seven rows, are
\begin{equation}
 (2K+3L,\ 2K+3L,\ 2K+2L,\ 2K,\ 2L,\ L,\ L).
 \label{eq:fixed-fiber-cut-gaps}
\end{equation}
They are strictly positive, and the selected cuts on $\calI$ have
exactly the matrix in Equation~\eqref{eq:star-cut-matrix}.
For the derivatives define
\begin{equation}
 f(t)=h_2\!\left(\frac{1+t}{2}\right)-L
     =-\frac{t^2}{2}+O(t^4).
 \label{eq:fixed-fiber-binary-expansion}
\end{equation}
If $m_A<2$, the regional state has no dependence on $t$.  If $m_A=2$,
the correlation survives in the boundary state only when all $r_A$
split qubits are received, with probability $\lambda_A=2^{-r_A}$.
The decoder attenuates its amplitude by the same factor.  Thus
\begin{equation}
 p_A(t)-p_A(0)
 =q\{\lambda_A f(t)-f(\lambda_A t)\},\qquad m_A=2.
 \label{eq:fixed-fiber-all-region-response}
\end{equation}
Consequently $p_A'(0)=0$ and
$\Delta_A=q\lambda_A(1-\lambda_A)$, while $\Delta_A=0$ for
$m_A<2$.  This gives the remaining columns and the displayed vectors.
Appendix~\ref{appD:fixed-center-calculation} derives the entropy formula
and the corresponding BKM norms on their fixed supports.
\end{proof}
\end{revision}

\begin{revision}
\begin{theorem}[Quadratic obstruction from a matched reference]
\label{thm:fixed-fiber-quadratic-witness}
The reference state in Equation~\eqref{eq:fixed-fiber-affine-path} has
the common star geometry in Equation~\eqref{eq:fixed-fiber-star-reference},
and its first response is geometrizable.  Its second response obeys
\begin{equation}
 (y^{12})^{\mathsf T}\mathbf a_2=-\frac{3\kappa q}{16}\ne0,
 \qquad
 \operatorname{dist}_2(\mathbf a_2,\operatorname{im}M)
 =\frac{3\kappa q\sqrt6}{64}.
 \label{eq:fixed-fiber-witness}
 \relax
\end{equation}
There is no $C^2$ edge-weight path through $\kappa\bfw_0$ that
reproduces the proto-areas to second order.  The functions for regions
$12$ and $4$ cannot be expectations of fixed area operators on a
convex neighborhood of this state in the fixed central fiber.
\end{theorem}
\end{revision}

\begin{revision}
\begin{proof}
The strict cut gaps fix the chamber near the reference.  Since $M$ is
injective, $M\bfw_1=\mathbf a_1=0$ implies $\bfw_1=0$, and a realizing
second derivative would have to satisfy $M\bfw_2=\mathbf a_2$.
The vector $y^{12}=(-1,-1,0,0,1,0)^{\mathsf T}$ annihilates $M$ and
gives the stated contradiction.  To obtain the exact distance, put
$v=(0,0,0,1,1,0)^{\mathsf T}$.  Direct projection gives
\begin{equation}
 (I-MM^+)v=\frac18(-3,-2,1,0,3,-1)^{\mathsf T},
 \qquad \|(I-MM^+)v\|_2=\frac{\sqrt6}{4}.
 \label{eq:fixed-fiber-exact-residual}
\end{equation}
Multiplying by $3\kappa q/16$ proves the distance formula.  A fixed
operator has an affine expectation on the affine path
\eqref{eq:fixed-fiber-affine-path}, whereas the two indicated
proto-areas have second derivative $-3q/16$.
\end{proof}
\end{revision}

\begin{revision}
Table~\ref{tab:fixed-fiber-all-regions} shows one more feature.  At
second order the data are not complement symmetric, since for instance
$\Delta_1=0$ while $\Delta_{234}=q/4$.  Region $1$ holds a single half
of one qubit and its marginal is maximally mixed for every $t$, whereas
its complement holds the whole of the other qubit and sees the
correlation attenuated once.  By Proposition~\ref{prop:section2-complement-parity}
this asymmetry is already incompatible with a single-valued disk
geometry, and it reflects the fact that the logical state is mixed and
the regional recovery targets are not complementary.  The six-region
witness is a separate failure of a local cut relation, and it does not
rely on discarding the complement-odd components.  This finite
construction therefore establishes a second-order failure of star
geometry on a fixed-recovery branch with a correctly matched reference
state.  It does not by itself provide an entropy-derived obstruction in
the complement-even continuum X-ray sectors.  Such a
limit requires the transfer hypotheses of
Appendix~\ref{app:finite-geometry}.
\end{revision}

\subsection{Robustness and finite-size persistence}
\label{subsec:sec6-robustness}

The preceding witnesses are exact. Their stability under imperfect recovery
and entropy estimation, and their persistence along a refining sequence,
follow from the same dual estimates.

\begin{revision}
Let $\widehat{\mathbf a}=\mathbf a+\delta\mathbf a$ be measured data. For any $y\in\ker M_{\mathcal C}^{\mathsf T}$,
\begin{equation}
 \left|y^{\mathsf T}\widehat{\mathbf a}-y^{\mathsf T}\mathbf a\right|
 \leq\|y\|_2\|\delta\mathbf a\|_2,
 \label{eq:sec6-witness-noise}
\end{equation}
while
\begin{equation}
 \operatorname{dist}_2(\widehat{\mathbf a},\operatorname{im}M_{\mathcal C})
 \geq\frac{|y^{\mathsf T}\widehat{\mathbf a}|}{\|y\|_2}.
 \label{eq:sec6-distance-lower-bound}
\end{equation}
If an approximate recovery returns a logical state at trace distance at most $\delta_A$ from the calibrated state, the Audenaert continuity bound \cite{Audenaert2007} gives
\begin{equation}
 |\Delta S_{\mathrm{rec}}(A)|
 \leq F_{d_L}(\delta_A),
 \label{eq:sec6-recovery-entropy-error}
\end{equation}
\end{revision}

\begin{revision}
where $\delta_A$ denotes one half of the trace norm and $d_L$ is
the recovered algebraic-state dimension.  The piecewise function
$F_{d_L}$ is defined in Equation~\eqref{eq:section2-audenaert-function}.
Equations~\eqref{eq:sec6-witness-noise} and
\eqref{eq:sec6-recovery-entropy-error} propagate a specified recovery
error to each sewing witness, and the channel-error estimate is proved in
Appendix~\ref{appA:recovery-errors}.  The canonical controlled-center
branch has $\delta_A=0$ on the block-diagonal Gibbs family.
Proposition~\ref{prop:sec6-approximate-recovery} distinguishes this
exact algebraic recovery from recovery of quantum source coherences.
\end{revision}

\begin{proposition}[Uniform persistence under refinement]
\label{prop:sec6-uniform-persistence}
Let $(G_N,\mathcal C_N)$ be a sequence of refined graph codes with incidence matrices $M_N$. Suppose there are cross-cell vectors $b_N^\times$ and witnesses
$y_N\in\ker M_N^{\mathsf T}$ such that
\begin{equation}
 |y_N^{\mathsf T}b_N^\times|\geq c>0,
 \qquad
 \|y_N\|_2\leq C<\infty
 \label{eq:sec6-uniform-motif}
\end{equation}
uniformly in $N$. If the nonlocal coefficient is $\beta_N$, then
\begin{equation}
 \operatorname{dist}_2(\mathbf a_N,\operatorname{im}M_N)
 \geq\frac{c}{C}|\beta_N|.
 \label{eq:sec6-uniform-distance}
\end{equation}
For noisy data with $\|\delta\mathbf a_N\|_2\leq\eta_N$,
\begin{equation}
 \operatorname{dist}_2(\widehat{\mathbf a}_N,\operatorname{im}M_N)
 \geq\frac{c}{C}|\beta_N|-\eta_N.
 \label{eq:sec6-uniform-noisy-distance}
\end{equation}
A fixed cross-cell motif remains detectable in the thermodynamic limit whenever $\inf_N|\beta_N|>0$ and \rev{$\limsup_{N\to\infty}\eta_N<c\inf_N|\beta_N|/C$, since the lower bound in Equation~\eqref{eq:sec6-uniform-noisy-distance} then stays bounded away from zero for all large $N$}.
\end{proposition}

\begin{proof}
The exact data have the form $\mathbf a_N=M_Nw_N+\beta_Nb_N^\times$. By duality,
\begin{equation}
 \operatorname{dist}_2(\mathbf a_N,\operatorname{im}M_N)
 \geq\frac{|y_N^{\mathsf T}\mathbf a_N|}{\|y_N\|_2}
 =\frac{|\beta_N|\,|y_N^{\mathsf T}b_N^\times|}{\|y_N\|_2},
 \label{eq:sec6-uniform-proof}
\end{equation}
which gives \eqref{eq:sec6-uniform-distance}. The distance to a closed subspace is one-Lipschitz, so adding an error of norm at most $\eta_N$ gives \eqref{eq:sec6-uniform-noisy-distance}.
\end{proof}

A uniform family satisfying \eqref{eq:sec6-uniform-motif} can be produced by embedding one fixed four-terminal witness motif and refining the graph outside its support. Extend the primitive witness by zero on the added interval data. Local columns generated entirely outside the motif remain annihilated, whereas the cross-cell vector has the same pairing with the witness. The dual lower bound is unchanged by refinements outside the witness motif.

\begin{revision}
The three mechanisms are summarized in Table~\ref{tab:sec6-summary}.
\begin{table}[t]
\ifshowchanges\color{blue}\else\color{black}\fi
 \centering
 \small
 \begin{tabular}{p{0.22\textwidth}p{0.37\textwidth}p{0.31\textwidth}}
 \toprule
 construction & calibrated datum & geometric status\\
 \midrule
 local controlled bonds
 & $M_{\mathcal C}w(\lambda,\varepsilon)$
 & geometric to all orders with matched reference weights\\
 controlled cross-cell gate
 & $M_{\mathcal C}w+\beta(\lambda,\varepsilon)b^\times$ and
   $(L_A)_A=M_{\mathcal C}(W_e)_e+b^\times K_\times$
 & first-order scalar witness and operator-valued sewing obstruction\\
 fixed-central-fiber erasure code
 & $\mathbf a_1=0$ and
   $\mathbf a_2=-\frac{3\kappa q}{16}(0,0,0,1,1,0)^{\mathsf T}$
 & matched reference and quadratic obstruction\\
 \bottomrule
 \end{tabular}
 \caption{Three code mechanisms. The controlled-sector rows have central affine area operators, whereas the fixed-central-fiber row has a noncentral BKM Hessian.}
 \label{tab:sec6-summary}
\end{table}
Geometrizability is controlled by the spatial organization of the response. A one-bond contribution is represented by an edge coordinate. A connected response of separated bond modes produces a normal component in data space. In the fixed-central-fiber code the same normal direction is selected by a BKM loss rather than a change of central probabilities.
\end{revision}

\begin{revision}
The controlled-center example tests the spatial sewing of
conditional bond entropies.  Its nonzero witness coexists with exact
recovery of the specified block-diagonal algebra.  Source coherence
is already absent from every single regional marginal because of the
complementary sector record.  The fixed-central-fiber erasure code
tests the separate mechanism of a noncentral area Hessian.
\end{revision}

\section{Physical interpretation and outlook}
\label{sec:physical-interpretation}

\begin{revision}
A regional proto-area, an area operator, and a local metric impose different
conditions on the same entropy data.  Affineness determines whether the
remainder is an operator expectation.  Sewing determines whether regional
expectations or derivatives arise from shared local sources.  At second
order, the length Hessian fixes a further relation between those derivatives.
The BKM--Jacobi condition equates this metric response with the information
lost under recovery, across the full interval family.
\end{revision}

\subsection{Locality of the response}
\label{subsec:kinematic-meaning}

The factorization
\begin{equation}
  \mathcal A_{\rho_0}=\mathcal B\circ\mathcal G
  \label{eq:section7-kinematic-factorization}
\end{equation}
is stronger than assigning a separate geometry to each interval.  It requires
one metric-valued map $\mathcal G$ whose first derivative supplies all line
integrals \rev{together with the Jacobi displacement of every geodesic and
the resulting normal second variation, and whose second derivative supplies
the remaining linear term $Ah_2$ of the second-order data}.  Failure of the linear or
quadratic condition excludes every regular path of simple asymptotically
hyperbolic metrics in the stated neighborhood of $g_{\mathbb H}$, independently of the choice of
boundary-fixing gauge.  Passing the conditions reconstructs the metric
coefficients in the iterated transverse-traceless gauge\rev{, but} it does not exclude
nonmetric descriptions or metrics outside this local class.

\begin{revision}
The logical structure is therefore
\begin{equation}
  \begin{aligned}
    \text{dynamical admissibility}
    &\Longrightarrow \text{metric geometrizability},\\
    \text{state dependence}
    &\not\Longrightarrow \text{metric geometrizability}.
  \end{aligned}
  \label{eq:section7-implication-hierarchy}
\end{equation}
\end{revision}

\begin{revision}
The examples separate the stages at which geometric compatibility can fail.
The cross-cell code has exact central area operators, but the mutual
information generated between bond modes prevents their additive sewing.
The fixed-central-probability code starts from a reference that matches a
positive-weight star for every bipartition.  Its first response is zero and
its second response leaves the same cut range.  This second example therefore
isolates an obstruction that is absent at the two preceding orders.
In the continuum, the hyperbolic family of
Section~\ref{subsec:explicit-normal-hessian} has a nonzero first response
and an explicit rank-four normal acceleration.  The corresponding straight
data path has the same first derivative and fails the quadratic condition.
These continuum interval data and the finite code are separate constructions.
Transferring a finite-code obstruction to the full interval family requires
the source-density and operator-consistency hypotheses of
Theorem~\ref{thm:finite-code-continuum-limit}, together with the second-order
hypotheses of Theorem~\ref{thm:second-order-discrete-continuum}.
\end{revision}

Non-Clifford structure alone does not determine the sewing obstruction\rev{, since a} local controlled Schmidt
rotation can be non-Clifford while its response remains in the image of the cut
map.  The quantities
that diagnose locality are instead the normal response components
\begin{equation}
  a^{\mathrm{ng}}_1=\Pi_{\mathrm{ng}}a_1,
  \qquad
  a^{\mathrm{ng}}_{2,\mathrm{corr}}
  =\Pi_{\mathrm{ng}}\!\left[
     a_2-D^2\mathcal B[h_1,h_1]
   \right].
  \label{eq:section7-resource-normal-components}
\end{equation}
Tripartite nonlocal magic can generate matter--proto-area coupling
\cite{CaoEtAl2026,Cao2024}, while BKM information loss controls its
fixed-reference Hessian.  Geometry imposes the additional requirement that
the regional pattern of this loss cancel the Jacobi normal acceleration.  The
relevant code-design conditions are
\begin{equation}
  \Pi_{\mathrm{ng}}a_1=0,
  \qquad
  \Pi_{\mathrm{ng}}\!\left[
    a_2-D^2\mathcal B[h_1,h_1]
  \right]=0.
  \label{eq:section7-code-design-conditions}
\end{equation}

\subsection{Dynamics and quantum extremality}
\label{subsec:dynamical-admissibility}

\begin{revision}
Metric geometrizability is kinematic.  A quantum extremal surface is instead
defined by the generalized entropy
\begin{equation}
  S_{\mathrm{gen}}[X;\rho]
  =\frac{\operatorname{Length}_{g_\rho}(X)}{4G_N}
   +S_{\mathrm{bulk}}[\Sigma_X;\rho]+S_{\mathrm{ct}}[X],
  \qquad \delta_XS_{\mathrm{gen}}=0.
  \label{eq:section7-generalized-entropy}
\end{equation}
Equating the recovered entropy used here with the renormalized bulk entropy is
additional holographic input.  Moreover, the shape Hessian of bulk entropy
changes the Jacobi problem.  If $\mathfrak f^{\mathrm q}$ is the quantum shape
force, $\mathcal K^{\mathrm q}_\gamma=D_X\mathfrak f^{\mathrm q}$ its shape
linearization, and $\dot{\mathfrak f}^{\mathrm q}_\gamma$ its state-and-metric
variation, quantum extremal deviation obeys
\begin{equation}
  \left(\mathcal J_\gamma+4G_N\mathcal K^{\mathrm q}_\gamma\right)J
  =\mathcal F_\gamma(h)-4G_N\dot{\mathfrak f}^{\mathrm q}_\gamma.
  \label{eq:section7-qes-jacobi}
\end{equation}
This is the weak-stability operator of Ref.~\cite{EngelhardtFischetti2019}.
The boundary-length theory is recovered when the quantum force and its relevant shape and state variations vanish, or when their combined effect is reproduced by an independently justified effective metric.  Allowing an
arbitrary state-dependent force without a defined bulk entropy functional
would remove the content of the common-source test.
\end{revision}

Dynamics adds another independent condition.  Given a bulk matter model with
normal energy density $\varepsilon(\rho)$, a time-symmetric slice must satisfy
\begin{equation}
  R(g)-2\Lambda-16\pi G\varepsilon=0.
  \label{eq:section7-hamiltonian-residual}
\end{equation}
For $h_X=\phi_Xg_{\mathbb H}+h_X^{\mathrm{tt}}$, the linearized scalar
constraint reduces to
\begin{equation}
  (-\Delta+2)\phi_X-16\pi G\,d\varepsilon_{\rho_0}(X)=0.
  \label{eq:section7-scalar-dynamical-residual}
\end{equation}
The remaining spacetime equations constrain the transverse-traceless sector
and the Lorentzian development.  Without independently supplied matter data,
one could define an effective energy density from the reconstructed scalar
curvature and satisfy the Hamiltonian constraint tautologically.  Entanglement
derivations of the linearized Einstein equation require further input, notably
the boundary first law and the modular-energy dictionary
\cite{FaulknerEtAl2014,JafferisEtAl2016}.

\begin{revision}
\subsection{Extremal-surface sectors and RT transitions}
\label{subsec:extremal-branches}
\end{revision}

\begin{revision}
The continuum obstruction is a statement within one extremal-surface
sector.  In the simple disk class of the theorems each pair of endpoints
selects one geodesic, so the sector is unambiguous there.  The choice is
distinct from the choice of a calibrated recovery branch.  A fixed
recovery, on a sector of states with fixed supports, makes the entropy
function differentiable, while a fixed extremal branch selects the
smooth geometric map with which its derivatives are compared.  Both must stay fixed while the two-jet
criterion is applied, and this is the exact continuum counterpart of
staying inside one minimum-cut chamber in the graph problem, where the
incidence matrix $M_\sigma$ is held fixed.
\end{revision}

\begin{revision}
For a classical RT problem with competing regular extremal
configurations, write $L_\sigma(g;A)$ for the length of configuration
$\sigma$, summed over its connected components and subject to the
homology condition.  The RT length is the lower envelope
\begin{equation}
 L_{\mathrm{RT}}(g;A)=\min_\sigma L_\sigma(g;A).
 \label{eq:rt-branch-envelope}
\end{equation}
Where one configuration is strictly minimal and remains regular, its
length is differentiable as a single branch.  For a finite family of
regions with finitely many competing regular configurations, a strict
gap keeps that branch dominant under small perturbations, and a
continuous family of regions needs corresponding uniform control.
\end{revision}

\begin{revision}
At a transverse crossing the envelope is generally not differentiable,
even though the metric is smooth.  If two active branches obey
\begin{equation}
 L_1(t)=L_0+p\,t+O(t^2),\qquad
 L_2(t)=L_0+q\,t+O(t^2),\qquad p\ne q,
 \label{eq:rt-transverse-branch-crossing}
\end{equation}
along one smooth metric path, the right derivative of their minimum is
$\min(p,q)$ and the left derivative is $\max(p,q)$.  There is then no
ordinary two-jet to which the smooth criterion could apply.  Such a
kink is not a nongeometric witness, because it arises from the
competing lengths of one and the same metric.
\end{revision}

\begin{revision}
The vacuum two-interval problem is the standard
example~\cite{Headrick2010}.  In the Poincar\'e upper half-plane with
boundary points $u_1<u_2<u_3<u_4$ and a fixed boundary representative,
the renormalized length of the geodesic joining $u_i$ and $u_j$ is
$L_{ij}=2\log|u_i-u_j|$ in unit-radius coordinates.  For
$A=[u_1,u_2]\cup[u_3,u_4]$ the two configurations are
\begin{equation}
 L_{\mathrm{dis}}=L_{12}+L_{34},\qquad
 L_{\mathrm{con}}=L_{14}+L_{23}.
 \label{eq:rt-two-interval-configurations}
\end{equation}
Both contain the same four ideal endpoints, so every endpoint
counterterm cancels in their difference.  With the cross ratio
\begin{equation}
 \eta=\frac{(u_2-u_1)(u_4-u_3)}
 {(u_3-u_1)(u_4-u_2)},\qquad 0<\eta<1,
 \label{eq:rt-two-interval-cross-ratio}
\end{equation}
and the identity
\begin{equation*}
 (u_3-u_1)(u_4-u_2)
 =(u_2-u_1)(u_4-u_3)+(u_4-u_1)(u_3-u_2),
\end{equation*}
one finds
\begin{equation}
 L_{\mathrm{dis}}-L_{\mathrm{con}}
 =2\log\frac{\eta}{1-\eta}.
 \label{eq:rt-two-interval-transition}
\end{equation}
The disconnected configuration is selected for $\eta<1/2$, the
connected one for $\eta>1/2$, and the two lengths agree at $\eta=1/2$.
Moving the intervals across this value produces the familiar classical
RT kink in a fixed hyperbolic geometry, and a metric perturbation can
move the wall as well.  The example lies outside the single-interval
data family of the main theorems.  It shows why a branch crossing has
to be treated separately from a smooth-range obstruction.
\end{revision}

\begin{revision}
What a transition-aware extension would have to establish can be
stated.  The candidate branch data
must be produced by one common metric, their length inequalities must
select the observed branch, and their realizations must remain
compatible at the wall.  Provided a branch has an invertible Jacobi
operator and a range and reconstruction theory of its own, the same
chain rule yields its tangent-range and normal-Hessian conditions.  The
graph model makes the whole structure finite through its union of
chambers, with the Farkas certificates of
Theorem~\ref{thm:farkas-finite} recording the distance to each wall.  A
continuum version would in addition have to control the competing
extremal configurations and their homology sectors.  The disk range
theorem used here establishes none of this for a general competing
sector and gives no global criterion across transition walls.  We
expect the answer to be chamberwise in the continuum as well, with
one-sided two-jets on either side of a wall and a matching condition at
the wall, but this remains a conjecture.
\end{revision}

\begin{revision}
For quantum extremal surfaces the selection is made among stationary
branches of the generalized entropy in
Equation~\eqref{eq:section7-generalized-entropy}, which includes the
branch-dependent bulk entropy.  Minimizing proto-area or length alone
would in general select the wrong branch.  Recovery-branch crossings
and extremal-surface crossings are therefore two separate sources of
nonsmoothness.  The present results apply on a sector where the
recovery is regular and the selected geometric length map is smooth,
and quantum shape forces and cross-sector selection require the
additional input described above.
\end{revision}

\subsection{Extensions}
\label{subsec:physical-extensions}

The covariant problem replaces geodesics by HRT extremal curves and the
Riemannian X-ray transform by a Lorentzian tensor transform
\cite{HubenyRangamaniTakayanagi2007}.  Caustics, conjugate points, and competing
branches make the data locus stratified\rev{, but} inside one branch the same tangent
range and normal-curvature logic applies.  Minimum-cut chambers provide a
finite-dimensional model of the resulting branch transitions, while
symmetry-reduced covariant inversion already exhibits related integrability
conditions \cite{Chae2026}.

In higher dimensions, boundary lengths are replaced by codimension-two
extremal areas.  Their first variation is a tensor Radon transform over the
reference surfaces, and their normal Hessian contains the inverse
normal-bundle Jacobi operator.  A higher-dimensional BKM--Jacobi test therefore
requires both a range theorem for that surface transform and control of Jacobi
zero modes.  For infinite-dimensional operator algebras, the corresponding
problem requires modular tangent spaces and a von Neumann algebraic
information metric.  The sharp angle result of Section~\ref{sec:proto-area}
already identifies the necessary stability hypothesis\rev{, that} the admissible tangent
space must remain uniformly separated from the channel kernel.

State dependence can encode backreaction, but locality resides in its
cross-region organization.  At linear order that organization is the X-ray
range\rev{, and} at quadratic order it is the BKM--Jacobi cancellation.  Only after both
conditions hold is there a common metric two-jet on which quantum extremality
and gravitational dynamics can be tested to the same order.

\appendix
\section{Recovery regularity and entropy jets}
\label{app:recovery-entropy}

A differentiable proto-area map can fail when the maximizing recovery bifurcates or when an output density matrix reaches the boundary of state space. Recovery-regular sectors exclude both phenomena at the reference point. A nondegenerate calibrated optimum on a local gauge slice yields a smooth recovery branch, while fixed-support entropy calculus supplies the proto-area two-jet. Continuity estimates then transfer recovery errors to the geometrizability witnesses without introducing zero modes from redundant Stinespring coordinates.

\subsection{Calibrated recovery branches}
\label{appA:calibrated-recovery}

\begin{revision}
For fixed ancilla dimensions, an admissible local Stinespring isometry
belongs to a compact complex Stiefel manifold.  The output-unitary
group in Equation~\eqref{eq:recovery-full-output-gauge} is also compact.
Discarded-output unitaries leave $\mathcal Q_{A,\mathbf R}$ unchanged,
whereas retained-output unitaries conjugate it by
$U_{A_1}\otimes U_{\bar A_1}$.  Both terms in the coherent information
are therefore invariant.  In particular, both kinds of unitary produce
zero modes before taking the quotient.  The smooth-quotient hypothesis
in Proposition~\ref{prop:section2-optimal-recovery-persistence} is a local
regularity assumption on the orbit space, and at a singular orbit it
must be checked separately.  A local section removes these symmetry
directions without removing changes of the recovered entropy.
\end{revision}

For a state manifold $\Sigma$, define an equivalence relation on optimal recovery pairs by
\begin{equation}
  \mathbf R\sim_{A,\Sigma}\mathbf R'
  \quad\Longleftrightarrow\quad
  S\!\left(\tau_{A,\mathbf R}(\rho)\right)
  =S\!\left(\tau_{A,\mathbf R'}(\rho)\right)
  \quad\text{for every }\rho\in\Sigma.
  \label{eq:appA-recovery-equivalence}
\end{equation}
The set in Equation~\eqref{eq:section2-set-valued-proto-area} is a singleton on $\Sigma$ if and only if $\mathsf{Opt}_A(V)$ is contained in one equivalence class. This criterion is weaker than uniqueness of a recovery map and stronger than equality at one calibration state. It removes output-isometry gauge while retaining every ambiguity that can change the entropy function.

\begin{proof}[Proof of Proposition~\ref{prop:section2-unitary-equivalent-recoveries}]
Equation~\eqref{eq:section2-unitary-equivalent-recovered-states} implies
\begin{equation}
  S\!\left(\tau_{A,\mathbf R}(\rho)\right)
  =S\!\left(\tau_{A,\mathbf R_A^\star}(\rho)\right)
\end{equation}
for all $\rho$ because von Neumann entropy is invariant under unitary conjugation. Thus all optimal recovery pairs lie in one class for the relation in Equation~\eqref{eq:appA-recovery-equivalence}. A fixed unitary maps the support projection of the reference marginal to a projection of the same rank. Constant support of the reference recovered branch therefore gives constant support for every equivalent recovered branch. Together with the stated constant-support hypothesis for $\omega_A(\rho)$, this is recovery regularity.
\end{proof}

A nondegenerate critical branch need not remain the global recovery optimum. Compactness of the admissible recovery space supplies the required global separation. Write
\begin{equation}
  F_A(s,\mathbf R)
  :=I_c\!\left(\mathcal Q^{(s)}_{A,\mathbf R}\right)
  \label{eq:appA-code-dependent-objective}
\end{equation}
for the coherent-information objective associated with the encoding $V_s$.

\begin{revision}
\begin{proof}[Proof of Proposition~\ref{prop:section2-optimal-recovery-persistence}]
Choose local coordinates $r\in\mathbb R^m$ on the smooth quotient
near $[\mathbf R_0]$, and let $\mathbf R(r)$ be its smooth local
section.  The objective is constant on every orbit, so it induces
the well-defined function
\begin{equation}
  f(s,r)=F_A(s,\mathbf R(r)).
\end{equation}
The constant-rank hypothesis separates each positive spectral cluster
from zero.  Its spectral projection and the entropy of its restriction
are $C^k$ by functional calculus.  Thus $f$ is $C^k$.  The
stationarity equation is
\begin{equation}
  D_r f(s,r)=0.
  \label{eq:appA-stationarity-equation}
\end{equation}
At $(0,r_0)$ its derivative with respect to $r$ is the negative-definite Hessian
\begin{equation}
  H_0=D_r^2f(0,r_0).
\end{equation}
The implicit function theorem gives an $\varepsilon>0$ and a unique $C^{k-1}$ solution $r_\star(s)$ of Equation~\eqref{eq:appA-stationarity-equation} in a neighborhood of $r_0$. Negative definiteness is open, so the branch consists of strict local maxima after decreasing $\varepsilon$.

It remains to exclude a competing global maximum outside the slice neighborhood. Let $U$ be a gauge-saturated neighborhood on which the critical orbit is unique. Compactness of the full Stiefel parameter space and uniqueness of the global maximizing orbit at $s=0$ give a number $\delta>0$ such that
\begin{equation}
  F_A(0,\mathbf R_0)
  -\sup_{\mathbf R\notin U}F_A(0,\mathbf R)
  \geq 3\delta.
  \label{eq:appA-global-optimizer-gap}
\end{equation}
Uniform continuity in $(s,\mathbf R)$ implies, after decreasing $\varepsilon$, that every value of $F_A(s,\cdot)$ differs from its value at $s=0$ by less than $\delta$. Hence every recovery outside $U$ remains at least $\delta$ below the local branch. The branch is therefore the unique global optimum modulo gauge.

In the chosen local section the recovered channel is $C^{k-1}$ in
$s$.  The additional constant-rank hypothesis for the boundary and
recovered marginals gives the same spectral separation locally in
$(s,\rho)$.  Their entropies, and hence their difference, are
$C^{k-1}$.  A different local section changes the recovered marginal
by retained-output conjugation and leaves this scalar function
unchanged.
\end{proof}
\end{revision}

When the code varies, standard sensitivity theory~\cite{BonnansShapiro2000}
determines the velocity of the maximizing branch by differentiating its
stationarity equation.  The envelope cancellation applies to the optimized
coherent information itself.  A distinct proto-area function generally
contains the resulting optimizer-velocity term.  This distinction is absent
for logical-state derivatives at fixed code, which are the derivatives used
in the geometrizability theorems.

\subsection{Fixed-support entropy calculus}
\label{appA:fixed-support-entropy}

Let $P$ be a fixed support projection and regard $P\mathcal H$ as the ambient Hilbert space. A density matrix is faithful on this support when
\begin{equation}
  \sigma\geq\mu P
  \label{eq:appA-spectral-gap}
\end{equation}
for some $\mu>0$. The Fr\'echet calculus for matrix functions on this open cone is standard~\cite{Higham2008}. The resolvent formula in Equation~\eqref{eq:section2-frechet-logarithm} is particularly useful because it preserves positivity and gives explicit norm estimates.

\begin{lemma}[Resolvent identities for the logarithm]
\label{lem:appA-resolvent-logarithm}
Let $\sigma>0$ on $P\mathcal H$. Then the map $\mathcal T_\sigma=D\log_\sigma$ is self-adjoint with respect to the Hilbert--Schmidt inner product, completely positive, and satisfies
\begin{equation}
  \tr\!\left[\sigma\mathcal T_\sigma(X)\right]=\tr X.
  \label{eq:appA-logarithm-trace-identity}
\end{equation}
If $\sigma=\sum_i p_i|i\rangle\langle i|$, then
\begin{equation}
  \bigl(\mathcal T_\sigma(X)\bigr)_{ij}
  =\ell(p_i,p_j)X_{ij},
  \qquad
  \ell(x,y)=\int_0^\infty\frac{ds}{(x+s)(y+s)}.
  \label{eq:appA-divided-difference-integral}
\end{equation}
Thus
\begin{equation}
  \ell(x,y)=
  \begin{cases}
  \dfrac{\log x-\log y}{x-y},&x\neq y,\\[2mm]
  \dfrac1x,&x=y,
  \end{cases}
  \label{eq:appA-divided-difference-log}
\end{equation}
and, for Hermitian $X$,
\begin{equation}
  \mathfrak g_\sigma(X,X)
  =\sum_{i,j}\ell(p_i,p_j)|X_{ij}|^2.
  \label{eq:appA-bkm-spectral-sum}
\end{equation}
In particular,
\begin{equation}
  \frac{1}{\lambda_{\max}(\sigma)}\norm{X}_2^2
  \leq\mathfrak g_\sigma(X,X)
  \leq\frac{1}{\lambda_{\min}(\sigma)}\norm{X}_2^2.
  \label{eq:appA-bkm-spectral-bounds}
\end{equation}
\end{lemma}

\begin{proof}
Each integrand in Equation~\eqref{eq:section2-frechet-logarithm} has the form $X\mapsto R_sXR_s$ with $R_s=(\sigma+sP)^{-1}$. It is completely positive and Hilbert--Schmidt self-adjoint, and these properties are preserved by the Bochner integral. In the eigenbasis of $\sigma$, direct evaluation gives Equation~\eqref{eq:appA-divided-difference-integral}. Partial fractions give Equation~\eqref{eq:appA-divided-difference-log}. The diagonal entries of $\sigma\mathcal T_\sigma(X)$ are $X_{ii}$, which proves Equation~\eqref{eq:appA-logarithm-trace-identity}. Equation~\eqref{eq:appA-bkm-spectral-sum} follows by taking the trace. The scalar mean-value theorem places $\ell(x,y)$ between $1/\max\{x,y\}$ and $1/\min\{x,y\}$, proving Equation~\eqref{eq:appA-bkm-spectral-bounds}. This also proves the positivity and symmetry asserted in Section~\ref{subsec:proto-area-differential-structure}. The resulting metric is the Bogoliubov--Kubo--Mori member of the monotone quantum information metrics~\cite{Petz1996,LesniewskiRuskai1999}.
\end{proof}

\begin{proof}[Proof of Theorem~\ref{thm:section2-entropy-two-jet}]
For a traceless Hermitian variation $X$, the differential of
$S(\sigma)=-\tr(\sigma\log\sigma)$ is
\begin{equation}
  dS_\sigma[X]
  =-\tr(X\log\sigma)
   -\tr\!\left[\sigma\mathcal T_\sigma(X)\right].
  \label{eq:appA-entropy-first-differential}
\end{equation}
Equation~\eqref{eq:appA-logarithm-trace-identity} reduces the second term to $-\tr X=0$. This proves Equation~\eqref{eq:section2-entropy-first-derivative}. Differentiating it along the curve gives
\begin{equation}
  \frac{d^2}{dt^2}S(\sigma(t))\bigg|_0
  =-\tr(\ddot\sigma\log\sigma)
   -\tr\!\left[\dot\sigma\mathcal T_\sigma(\dot\sigma)\right],
\end{equation}
which is Equation~\eqref{eq:section2-entropy-second-derivative}. Applying the same calculation to a two-parameter family and polarizing the quadratic term gives Equation~\eqref{eq:section2-entropy-mixed-hessian}. Symmetry follows from the Hilbert--Schmidt self-adjointness of $\mathcal T_\sigma$.
\end{proof}

The ordinary $o(t^2)$ remainder in Equation~\eqref{eq:section2-proto-area-jet-expansion} follows from twice Fr\'echet differentiability on the open fixed-support cone. On a set satisfying Equation~\eqref{eq:appA-spectral-gap} with a uniform $\mu$, uniform continuity of the derivatives makes this remainder uniform for state curves whose first two derivatives remain in a bounded set.

\subsection{Proto-area differentials}
\label{appA:proto-area-differentials}

\begin{proof}[Proof of Proposition~\ref{prop:section2-exact-code-state-independence}]
Taking the $A$ marginal of Equation~\eqref{eq:section2-exact-recovery-factorization} gives
\begin{equation}
  R_A\omega_A(\rho)R_A^\dagger
  =\rho_{a_A}\otimes\chi_{A_2}.
\end{equation}
Isometric conjugation preserves the nonzero spectrum, so
\begin{equation}
  S(\omega_A(\rho))
  =S(\rho_{a_A})+S(\chi_{A_2}).
\end{equation}
The recovered matter marginal is $\tau_A^\star(\rho)=\rho_{a_A}$. Their difference is therefore the constant in Equation~\eqref{eq:section2-exact-code-proto-area}\rev{, and} background subtraction annihilates it together with all state derivatives.
\end{proof}

\begin{proof}[Proof of Proposition~\ref{prop:section2-area-function-operator}]
An expression of the form $\tr(O\rho)+c$ is affine. Conversely, choose $\rho_\ast\in U$. The differential of an affine function is a constant real linear functional on the vector space of traceless Hermitian matrices. Nondegeneracy of the Hilbert--Schmidt pairing produces a traceless Hermitian $O_0$ such that
\begin{equation}
  dF[X]=\tr(O_0X).
\end{equation}
Affineness then gives
\begin{equation}
  F(\rho)-F(\rho_\ast)
  =\tr\!\left[O_0(\rho-\rho_\ast)\right].
\end{equation}
Taking $c=F(\rho_\ast)-\tr(O_0\rho_\ast)$ proves the representation. Every affine function has zero Hessian, proving the final statement.
\end{proof}

\begin{proof}[Proof of Theorem~\ref{thm:section2-proto-area-two-jet}]
Once the recovery is calibrated, the two output maps in Equation~\eqref{eq:section2-linear-output-maps} are linear. Apply Theorem~\ref{thm:section2-entropy-two-jet} to $\omega_A(t)$ and $\tau_A(t)$, subtract the two entropy expansions, and multiply by $4G_{\mathrm{eff}}$. This gives Equations~\eqref{eq:section2-proto-area-first-jet} and~\eqref{eq:section2-proto-area-second-jet}. For an affine logical path, linearity gives
$\ddot\omega_A=\ddot\tau_A=0$, and Equation~\eqref{eq:section2-affine-proto-area-curvature} follows.

For a multidimensional state manifold, choose normal coordinates for the fixed torsion-free connection at $\rho_0$. Applying Equation~\eqref{eq:section2-entropy-mixed-hessian} to the two channel outputs gives Equation~\eqref{eq:section2-proto-area-mixed-hessian}. Since the expression is the Hessian of a scalar function with respect to the chosen connection, the coordinate calculation determines the intrinsic bilinear form at the base point.
\end{proof}

\begin{proof}[Proof of Corollary~\ref{cor:section2-proto-area-continuity}]
Both $\mathcal N_A$ and $\mathcal M_A^\star$ are quantum channels, so trace distance contracts:
\begin{equation}
  \frac12\norm{\omega_A(\rho)-\omega_A(\sigma)}_1
  \leq\varepsilon,
  \qquad
  \frac12\norm{\tau_A^\star(\rho)-\tau_A^\star(\sigma)}_1
  \leq\varepsilon.
\end{equation}
Apply the sharp entropy continuity estimate of Ref.~\cite{Audenaert2007} to the two output pairs and use the triangle inequality in Equation~\eqref{eq:section2-proto-area-entropy}. Multiplication by $4G_{\mathrm{eff}}$ yields Equation~\eqref{eq:section2-proto-area-continuity-bound}.
\end{proof}

\begin{revision}
The fixed calibration is needed in these formulas.  A recovery chosen
separately for each logical state contributes its own first and second
derivatives, and the selected branch can fail to be differentiable
when optimizers cross.  With Equation~\eqref{eq:section2-optimal-recovery-set}
calibrated once, both output maps are linear in the logical state.
\end{revision}

\subsection{Constant rank and rank-changing paths}
\label{appA:rank-changing}

A fixed support is convenient but slightly stronger than constant rank. The following reduction shows that no entropy regularity is lost by using it locally.

\begin{proposition}[Smooth trivialization of a constant-rank sector]
\label{prop:appA-constant-rank-trivialization}
Let $t\mapsto\sigma(t)$ be a $C^k$ curve of density matrices, $k\geq1$, with constant rank near $t=0$. Then, after shrinking the parameter interval, there is a $C^k$ unitary curve $U(t)$ such that
\begin{equation}
  U(t)^\dagger\Supp\sigma(t)\,U(t)=P_0:=\Supp\sigma(0).
  \label{eq:appA-support-trivialization}
\end{equation}
The curve $\widehat\sigma(t)=U(t)^\dagger\sigma(t)U(t)$ has fixed support $P_0$ and
\begin{equation}
  S(\widehat\sigma(t))=S(\sigma(t)).
\end{equation}
Thus the fixed-support formulas apply to every constant-rank curve after a smooth unitary trivialization.
\end{proposition}

\begin{proof}
Constant rank separates the nonzero spectrum from zero by a positive local gap. The Riesz projection onto the nonzero spectral cluster is therefore a $C^k$ family of projections. Standard perturbation theory provides a local $C^k$ unitary intertwining this projection with $P_0$~\cite{Kato1995}. Unitary invariance of entropy gives the final statement.
\end{proof}

When rank changes, the logarithmic terms are intrinsic and cannot be removed by a choice of coordinates.

\begin{revision}
\begin{proposition}[One-sided entropy asymptotics at a rank change]
\label{prop:appA-rank-changing-asymptotics}
Let $\sigma(t)$ be a density-matrix curve for $0\leq t<\varepsilon$.
Assume that its nonzero eigenvalues at $t=0$ continue as smooth positive
branches, while the eigenvalues emerging from the kernel satisfy
\begin{equation}
 \lambda_j(t)=c_jt^{m_j}(1+r_j(t)),
 \qquad c_j>0,\quad m_j\in\mathbb N,
 \qquad
 r_j^{(k)}(t)=O(t^{\eta_j-k})\quad(0\leq k\leq2),
 \label{eq:appA-emerging-eigenvalues}
\end{equation}
for some $\eta_j>0$.  Identically zero branches are omitted.  Then
\begin{equation}
 \begin{split}
 S(\sigma(t))={}&S_{\mathrm{reg}}(t)
 -\sum_jc_jt^{m_j}(m_j\log t+\log c_j)+E(t),\\
 E^{(k)}(t)={}&O\!\left(
 \sum_jt^{m_j+\eta_j-k}(1+|\log t|)\right),
 \qquad 0\leq k\leq2,
 \end{split}
 \label{eq:appA-rank-changing-entropy}
\end{equation}
where $S_{\mathrm{reg}}$ is smooth.  Write
$m_*:=\min_jm_j$ when an emerging branch is present.  If $m_*=1$,
the right first derivative is not finite.  If $m_*=2$, the entropy
has a finite right first derivative but no finite right second
derivative.  If $m_*>2$, it extends to a $C^2$ function at the
one-sided endpoint.  These assumptions hold, in particular, for
real-analytic Hermitian families that are positive for $t\geq0$.
\end{proposition}
\end{revision}

\begin{revision}
\begin{proof}
The positive eigenvalues present at $t=0$ remain separated from zero
and contribute $S_{\mathrm{reg}}$.  For an emerging branch,
\begin{equation}
 \log\lambda_j(t)=\log c_j+m_j\log t+\log(1+r_j(t)).
\end{equation}
The hypotheses imply
$\partial_t^k\log(1+r_j(t))=O(t^{\eta_j-k})$ for $0\leq k\leq2$.
Expanding $-\lambda_j\log\lambda_j$, the terms not displayed
explicitly are products of $t^{m_j}r_j$ with a constant or $\log t$,
and of $t^{m_j}(1+r_j)$ with $\log(1+r_j)$.  The product rule gives
the asserted bounds for their first two derivatives.  For the leading
singular term,
\begin{equation}
 \frac{d^2}{dt^2}(t^m\log t)
 =m(m-1)t^{m-2}\log t+(2m-1)t^{m-2}.
\end{equation}
The branches of minimal order have positive coefficients, so their
leading logarithms do not cancel.  This proves the three regularity
statements.  Rellich--Kato perturbation theory provides analytic
real eigenvalue branches for an analytic Hermitian family, and each
nonzero emerging branch has the stated expansion with $\eta_j=1$
and the required derivative bounds~\cite{Kato1995}.
\end{proof}
\end{revision}

\begin{revision}
Rank creation at order two therefore prevents an ordinary entropy two-jet. For the proto-area, however, the singular terms of the two entropies must be compared before drawing this conclusion. They can cancel, as an identity recovery shows, since it gives a zero remainder even along a rank-changing path. The full-rank or constant-support hypothesis in Section~\ref{subsec:proto-area-differential-structure} guarantees regularity without requiring such a cancellation.
\end{revision}

\subsection{Recovery-error estimates}
\label{appA:recovery-errors}

The first estimate is nonperturbative and requires no spectral gap. The next estimate controls derivatives and therefore necessarily assumes a fixed gap from zero.

\begin{proposition}[Proto-area stability under channel error]
\label{prop:appA-channel-error}
Let $\mathcal M$ and $\widetilde{\mathcal M}$ be channels with $d$-dimensional outputs and
\begin{equation}
  \norm{\widetilde{\mathcal M}-\mathcal M}_\diamond\leq\epsilon.
  \label{eq:appA-diamond-channel-error}
\end{equation}
Then for every input state $\rho$,
\begin{equation}
  \abs{S(\widetilde{\mathcal M}(\rho))-S(\mathcal M(\rho))}
  \leq F_d\!\left(\min\{\epsilon/2,1\}\right),
  \label{eq:appA-channel-entropy-error}
\end{equation}
with $F_d$ defined in Equation~\eqref{eq:section2-audenaert-function}. If both the boundary channel and recovered-matter channel are perturbed, the error in $a_\rho(A)$ is at most $4G_{\mathrm{eff}}$ times the sum of the two corresponding right-hand sides.
\end{proposition}

\begin{proof}
The diamond bound implies
\begin{equation}
  \frac12\norm{
  \widetilde{\mathcal M}(\rho)-\mathcal M(\rho)
  }_1\leq\epsilon/2.
\end{equation}
Equation~\eqref{eq:appA-channel-entropy-error} is the Audenaert continuity estimate. The proto-area statement follows from the triangle inequality.
\end{proof}

The value estimate is all that is required for the finite witness gaps in
Section~\ref{subsec:sec6-robustness}.  Derivative-level perturbation bounds can
be obtained on a fixed faithful support by combining the resolvent identity
with the entropy two-jet, but they require an explicit spectral gap and are not
used in the geometrizability criteria.

\section{Discrete-to-continuum approximation}
\label{app:finite-geometry}

The finite witnesses become continuum statements only when the entire
cut-response operator, rather than a selected list of interval values,
converges to the rank-two geodesic X-ray transform.  This appendix gives
checkable sufficient conditions for that operator limit and proves the
first- and second-order transfer theorems used in the main text.

\subsection{Constructive continuum approximation}
\label{appB:continuum-approximation}

\begin{revision}
The assumptions in Theorem~\ref{thm:finite-code-continuum-limit} concern
operators on the full resolved source space. A finite-dimensional
quadrature gives one way to verify them without assuming pointwise
sampling is bounded on all of $\mathcal Y^\tau$.
\end{revision}

\begin{revision}
Choose nested finite-dimensional regular source spaces
$X_N\subset\mathcal X$ with dense union, let $E_N:X_N\hookrightarrow
\mathcal X$ be inclusion, and let $P_N=E_NE_N^*$. Then
\begin{equation}
 \norm{(I-P_N)x}_{\mathcal X}\longrightarrow0
 \qquad(x\in\mathcal X).
 \label{eq:appB-source-approximation-rate}
\end{equation}
For a graph, the visible edge space is identified isometrically with
$X_N$, which is part of the choice of source embedding. Dense regular
spaces can be constructed from the iterated-TT representatives of a
dense sequence of smooth source tensors. Their transforms are smooth
off the endpoint diagonal. Choose finite-dimensional spaces
$Y_{L_N}\subset\mathcal Y^\tau$ of such smooth data with
\begin{equation}
 AE_NX_N\subset Y_{L_N}.
 \label{eq:appB-finite-band-inclusion}
\end{equation}
Additional higher-sector modes may be included to test
nongeometric responses. The inclusion is the finite-band condition
used below, and no rate of approximation for arbitrary data is required.
\end{revision}

\begin{revision}
For a fixed $L$, choose an $L^2(\mathcal G_{\mathbb H},d\mu)$-orthonormal
basis $e_1,\ldots,e_{d_L}$ of $Y_L$. Each product
$e_\nu\overline{e_\mu}$ is integrable by Cauchy--Schwarz. The real
finite-dimensional space
\begin{equation}
 \mathscr F_L=\operatorname{span}_{\mathbb R}
 \{\operatorname{Re}(e_\nu\overline{e_\mu}),
   \operatorname{Im}(e_\nu\overline{e_\mu}):
           1\leq\nu,\mu\leq d_L\}
 \label{eq:appB-quadrature-space}
\end{equation}
therefore admits a positive quadrature. The
Richter--Tchakaloff theorem in its finite-dimensional $L^1$ form
\cite{BerschneiderSasvari2012} gives nodes $\xi_{L,j}$ in interval space
and positive weights $\omega_{L,j}$ such that
\begin{equation}
 \int_{\mathcal G_{\mathbb H}}F\,d\mu
 =\sum_j\omega_{L,j}F(\xi_{L,j})
 \qquad(F\in\mathscr F_L),
 \label{eq:appB-exact-quadrature}
\end{equation}
with at most $\dim_{\mathbb R}\mathscr F_L$ nodes after zero weights are
removed. The theorem applies to these integrable products even though
the measure of the full geodesic space is infinite.
\end{revision}

\begin{revision}
Define sampling on $Y_L$ and reconstruction from finite data by
\begin{align}
 (S_La)_j&=a(\xi_{L,j}),
 \label{eq:appB-sampling-map}\\
 J_Lz&=\sum_{\nu=1}^{d_L}
 \left(\sum_j\omega_{L,j}z_j\overline{e_\nu(\xi_{L,j})}\right)e_\nu.
 \label{eq:appB-moment-reconstruction}
\end{align}
Exact quadrature of the basis products gives
\begin{equation}
 \ip{S_Le_\nu}{S_Le_\mu}_{\ell^2(\omega_L)}=\delta_{\nu\mu},
 \qquad J_LS_La=a\quad(a\in Y_L).
 \label{eq:appB-bandlimited-exactness}
\end{equation}
The reconstruction norm is also explicit. Let
\begin{equation}
 (G_L)_{\nu\mu}=\ip{e_\nu}{e_\mu}_{\mathcal Y^\tau},
 \qquad C_L=\lambda_{\max}(G_L)^{1/2}.
 \label{eq:appB-quadrature-stability-factor}
\end{equation}
Since $z\mapsto(\ip{z}{S_Le_\nu}_{\ell^2(\omega_L)})_\nu$ is a
contraction, synthesis in the basis $(e_\nu)$ gives
\begin{equation}
 \norm{J_Lz}_{\mathcal Y^\tau}
 \leq C_L\norm z_{\ell^2(\omega_L)}.
 \label{eq:appB-interpolation-error}
\end{equation}
Thus $C_L$ measures the relation between the sampling norm and the
source-adapted data norm on the resolved space. It need not remain
bounded as the resolution grows.
\end{revision}

\begin{revision}
The remaining condition compares the graph with the continuum. Set
\begin{equation}
 \delta_N=
 \norm{\overline M_N-S_{L_N}AE_N}_{X_N\to\ell^2(\omega_{L_N})}.
 \label{eq:appB-column-calibration}
\end{equation}
The common source identification makes this a comparison of operators,
rather than an independent fit for each region.
\end{revision}

\begin{revision}
\begin{proposition}[Finite-band operator consistency]
\label{prop:appB-operator-consistency-sufficient}
Assume Equations~\eqref{eq:appB-source-approximation-rate} and
\eqref{eq:appB-finite-band-inclusion}, and use the exact quadrature above.
For $J_N=J_{L_N}$,
\begin{equation}
 \varepsilon_N
 :=\norm{J_N\overline M_N-AE_N}_{X_N\to\mathcal Y^\tau}
 \leq C_{L_N}\delta_N.
 \label{eq:appB-operator-consistency-bound}
\end{equation}
In particular, $C_{L_N}\delta_N\to0$ implies the operator-consistency
hypothesis of Theorem~\ref{thm:finite-code-continuum-limit}.
\end{proposition}
\end{revision}

\begin{revision}
\begin{proof}
By finite-band inclusion and Equation~\eqref{eq:appB-bandlimited-exactness},
$J_NS_{L_N}AE_N=AE_N$. Consequently,
\begin{equation}
 J_N\overline M_N-AE_N
 =J_N(\overline M_N-S_{L_N}AE_N).
\end{equation}
Taking the operator norm and using
Equation~\eqref{eq:appB-interpolation-error} proves the claim.
\end{proof}
\end{revision}

\begin{revision}
The quadrature resolves every chosen continuum source response exactly.
It does not ensure that cut-incidence columns approximate those
responses. That geometric assertion is the condition
$C_{L_N}\delta_N\to0$, which must be checked for the graph sequence in
question. Increasing $\tau$ in Equation~\eqref{eq:adapted-data-norm}
only changes tensor-order weights and provides neither a finite-rank
tail bound within $\mathcal H_0$ nor a point-evaluation estimate on the
whole data space.
\end{revision}

\begin{proof}[Proof of Theorem~\ref{thm:finite-code-continuum-limit}]
The continuum estimate in Equation~\eqref{eq:continuum-stability-two-sided} and the isometry of $E_N$ give
\begin{equation}
  \norm{B_Nx}
  \geq\norm{AE_Nx}-\norm{(B_N-AE_N)x}
  \geq(c-\varepsilon_N)\norm{x},
  \label{eq:discrete-lower-bound-proof}
\end{equation}
which proves \textnormal{(i)}. If $z=B_Nx\in\mathcal R_N$, then $AE_Nx\in\mathcal Y_{\mathrm{geo}}$ and hence
\begin{equation}
  \norm{\Pi_{\mathrm{ng}}z}
  \leq\varepsilon_N\norm{x}
  \leq\frac{\varepsilon_N}{c-\varepsilon_N}\norm{z}.
  \label{eq:finite-leakage-proof}
\end{equation}
This proves Equation~\eqref{eq:finite-range-leakage}.

To prove strong convergence, first take $n\in\mathcal Y_{\mathrm{ng}}$. Since $n\perp\mathcal Y_{\mathrm{geo}}$,
\begin{equation}
  \norm{\Pi_Nn}
  =\sup_{\substack{z\in\mathcal R_N\\\norm{z}=1}}
  \abs{\ip{n}{z}}
  =\sup_{\substack{z\in\mathcal R_N\\\norm{z}=1}}
  \abs{\ip{n}{\Pi_{\mathrm{ng}}z}}
  \leq\frac{\varepsilon_N}{c-\varepsilon_N}\norm{n}.
  \label{eq:nongeometric-projector-bound}
\end{equation}
Next let $r=Ax\in\mathcal Y_{\mathrm{geo}}$. With $x_N=E_N^*x$,
\begin{equation}
\begin{split}
  \operatorname{dist}(r,\mathcal R_N)
  &\leq\norm{Ax-B_Nx_N}\\
  &\leq\norm{A(\operatorname{Id}-P_N)x}
  +\varepsilon_N\norm{x_N}
  \longrightarrow0.
  \label{eq:geometric-range-approximation}
\end{split}
\end{equation}
Thus $\Pi_Nr\to r$. Decomposing $a=r+n$ proves Equation~\eqref{eq:finite-projector-convergence}.

Under Equation~\eqref{eq:code-data-convergence},
\begin{equation}
  \Pi_NJ_Nq_N-\Pi_{\mathrm{geo}}a_1
  =\Pi_N(J_Nq_N-a_1)
  +(\Pi_Na_1-\Pi_{\mathrm{geo}}a_1)
  \longrightarrow0.
  \label{eq:projected-data-convergence-proof}
\end{equation}
The residual and distance limits follow immediately. Put $x_N=B_N^\dagger J_Nq_N$. Equation~\eqref{eq:uniform-discrete-inf-sup} gives a uniform bound on $\norm{x_N}$, and
\begin{equation}
\begin{split}
  \norm{AE_Nx_N-\Pi_{\mathrm{geo}}a_1}
  &\leq\varepsilon_N\norm{x_N}
  +\norm{\Pi_NJ_Nq_N-\Pi_{\mathrm{geo}}a_1}
  \longrightarrow0.
  \label{eq:reconstruction-limit-proof}
\end{split}
\end{equation}
Because $A$ is bounded below by $c$ and $A\mathscr R_2a_1=\Pi_{\mathrm{geo}}a_1$, Equation~\eqref{eq:reconstruction-continuum-convergence} follows.
\end{proof}

\subsection{Second-order consistency of nonlinear discretizations}
\label{appB:second-order-consistency}

First-order operator consistency does not control displacement of the
discrete extremizer.  Approximating the Jacobi contribution requires a
nonlinear discrete extremal-length map and convergence of its second
derivative.  Second-order transfer requires the two operator limits below.

\begin{revision}
Let $\mathscr X$ and
$\mathscr Y=\mathscr Y_{\mathrm{geo}}\widehat\oplus
\mathscr Y_{\mathrm{ng}}$ be Hilbert source and data spaces on which
the orthogonal range splitting below is bounded. Assume that the
boundary-length map is defined and twice Fr\'echet differentiable at the
origin in these norms, with bounded first and second derivatives. Write
\begin{equation}
 \mathbf B(h)=\mathcal B(g_{\mathbb H}+h)-\mathcal B(g_{\mathbb H}),
 \qquad A=D\mathbf B(0),
 \qquad H=D^2\mathbf B(0).
 \label{eq:appB-continuum-nonlinear-operators}
\end{equation}
Let $\mathscr X_N$ be finite-dimensional Hilbert spaces with isometric
embeddings $E_N:\mathscr X_N\to\mathscr X$ such that
$E_NE_N^*\to I_{\mathscr X}$ strongly.  Let
\begin{equation}
 \mathbf B_N:U_N\subset\mathscr X_N\longrightarrow\mathscr Y,
 \qquad \mathbf B_N(0)=0,
 \label{eq:appB-nonlinear-discrete-map}
\end{equation}
be $C^2$ maps.  In an application, $\mathbf B_N$ is a discrete
extremal-length map followed by interpolation into $\mathscr Y$.
\end{revision}

\begin{revision}
\begin{theorem}[Second-order discrete-to-continuum transfer]
\label{thm:second-order-discrete-continuum}
Assume that $A:\mathscr X\to\mathscr Y_{\mathrm{geo}}$ is onto and
bounded below, so that
\begin{equation}
 \norm{Ax}_{\mathscr Y}\geq c\norm{x}_{\mathscr X}
 \qquad (x\in\mathscr X)
 \label{eq:appB-continuum-lower-bound}
\end{equation}
for some $c>0$.  Assume also
\begin{align}
 \epsilon_N^{(1)}
 &:=\norm{D\mathbf B_N(0)-AE_N}_{\mathscr X_N\to\mathscr Y}
 \longrightarrow0,
 \label{eq:appB-first-derivative-consistency}\\
 \epsilon_N^{(2)}
 &:=\sup_{u,v\neq0}
 \frac{\norm{D^2\mathbf B_N(0)[u,v]-H[E_Nu,E_Nv]}_{\mathscr Y}}
 {\norm{u}_{\mathscr X_N}\norm{v}_{\mathscr X_N}}
 \longrightarrow0.
 \label{eq:appB-second-derivative-consistency}
\end{align}
Put $A_N=D\mathbf B_N(0)$, let
$\mathscr R_N^{(2)}=\operatorname{Ran}A_N$, and let
$\Pi_N^{(2)}$ be the orthogonal projection onto this range.  Suppose
$d_{1,N},d_{2,N}\in\mathscr Y$ satisfy
\begin{equation}
 d_{1,N}\longrightarrow a_1,
 \qquad d_{2,N}\longrightarrow a_2
 \quad\text{in }\mathscr Y.
 \label{eq:appB-discrete-two-jet-convergence}
\end{equation}
Assume $a_1\in\mathscr Y_{\mathrm{geo}}$ and define
$h_{1,N}=A_N^\dagger d_{1,N}$. If each $d_{1,N}$ already belongs to
$\mathscr R_N^{(2)}$, its geometric limit is automatic. With
$h_1=A^{-1}a_1$,
\begin{equation}
 E_Nh_{1,N}\longrightarrow h_1.
 \label{eq:appB-first-reconstruction-convergence}
\end{equation}
Moreover,
\begin{equation}
 \begin{split}
 &(I-\Pi_N^{(2)})
 \left[d_{2,N}-D^2\mathbf B_N(0)[h_{1,N},h_{1,N}]\right]\\
 &\qquad\longrightarrow
 \Pi_{\mathrm{ng}}
 \left[a_2-H[h_1,h_1]\right]
 \quad\text{in }\mathscr Y.
 \end{split}
 \label{eq:appB-second-obstruction-convergence}
\end{equation}
Thus a nonzero continuum quadratic obstruction is eventually detected
by the corrected discrete residual.
\end{theorem}
\end{revision}

\begin{revision}
\begin{proof}
For all large $N$,
Equation~\eqref{eq:appB-first-derivative-consistency} gives
\begin{equation}
 \norm{A_Nu}\geq(c-\epsilon_N^{(1)})\norm{u}
 \geq\frac c2\norm{u}.
 \label{eq:appB-second-order-uniform-lower-bound}
\end{equation}
We first prove $\Pi_N^{(2)}\to\Pi_{\mathrm{geo}}$ strongly.  If
$y=Ax\in\mathscr Y_{\mathrm{geo}}$, set $u_N=E_N^*x$.  Source density
and first-derivative consistency give $A_Nu_N\to y$, so
$\operatorname{dist}(y,\mathscr R_N^{(2)})\to0$.  If
$z=A_Nu\in\mathscr R_N^{(2)}$ has unit norm, then
\begin{equation}
 \norm{\Pi_{\mathrm{ng}}z}
 \leq\epsilon_N^{(1)}\norm{u}
 \leq\frac{2\epsilon_N^{(1)}}c.
\end{equation}
Consequently, for $n\in\mathscr Y_{\mathrm{ng}}$,
\begin{equation}
 \norm{\Pi_N^{(2)}n}
 =\sup_{\substack{z\in\mathscr R_N^{(2)}\\\norm z=1}}
 \abs{\ip{n}{z}}
 \leq\frac{2\epsilon_N^{(1)}}c\norm n\longrightarrow0.
\end{equation}
The orthogonal decomposition proves strong projector convergence.

Set $\widehat d_{1,N}=\Pi_N^{(2)}d_{1,N}=A_Nh_{1,N}$.
Strong projector convergence and $a_1\in\mathscr Y_{\mathrm{geo}}$ give
$\widehat d_{1,N}\to a_1$. The lower bound makes $(h_{1,N})$ bounded.
If the unprojected first coefficients are already in the finite ranges,
the range-angle estimate also proves
$a_1\in\mathscr Y_{\mathrm{geo}}$ without assuming it. Furthermore,
\begin{equation}
 \norm{A(E_Nh_{1,N}-h_1)}
 \leq\norm{(AE_N-A_N)h_{1,N}}+\norm{\widehat d_{1,N}-a_1}
 \longrightarrow0,
\end{equation}
so bounded-below invertibility of $A$ proves
Equation~\eqref{eq:appB-first-reconstruction-convergence}.  Bilinearity,
continuity of $H$, and
Equation~\eqref{eq:appB-second-derivative-consistency} give
\begin{equation}
 D^2\mathbf B_N(0)[h_{1,N},h_{1,N}]
 \longrightarrow H[h_1,h_1].
\end{equation}
The corrected discrete data therefore converge to
$a_2-H[h_1,h_1]$, and strong projector convergence proves
Equation~\eqref{eq:appB-second-obstruction-convergence}.
\end{proof}
\end{revision}

\begin{corollary}[Fixed chambers do not approximate Jacobi bending]
\label{cor:fixed-chamber-no-jacobi-limit}
Suppose every $\mathbf B_N$ is linear in a neighborhood of the origin,
as in a fixed minimum-cut chamber.  If
Equation~\eqref{eq:appB-second-derivative-consistency} holds on a
source-dense sequence, then $H[h,k]=0$ for all
$h,k\in\mathscr X$.  Hence a fixed-chamber family cannot approximate
the nonlinear boundary-length map to second order when its Jacobi
Hessian is nonzero.
\end{corollary}

\begin{proof}
Linearity gives $D^2\mathbf B_N(0)=0$.  Put $P_N=E_NE_N^*$.  For
$h,k\in\mathscr X$, Equation~\eqref{eq:appB-second-derivative-consistency}
with $u=E_N^*h$ and $v=E_N^*k$ gives
\begin{equation}
 \norm{H[P_Nh,P_Nk]}_{\mathscr Y}
 \leq\epsilon_N^{(2)}\norm{h}_{\mathscr X}\norm{k}_{\mathscr X}
 \longrightarrow0.
\end{equation}
Since $P_Nh\to h$ and $P_Nk\to k$, continuity of $H$ yields
$H[h,k]=0$.
\end{proof}

The corollary shows why a fixed chamber is insufficient.  A
nonlinear discretization must allow its minimizing path to move with
the discrete metric.  An incidence matrix in one chamber supplies the
linearized operator but cannot generate the forced-Jacobi term.

\begin{theorem}[Stationary quadratic departure under refinement]
\label{thm:stationary-quadratic-refinement}
Let $\mathscr R_N\subset\mathscr Y$ be finite-dimensional geometric
response spaces with orthogonal projectors
$\Pi_N\to\Pi_{\mathrm{geo}}$ strongly.  Let $a_N(t)$ be $C^2$ data
paths such that
\begin{equation}
 a_N(0)=0,
 \qquad a_N'(0)\longrightarrow0,
 \qquad a_N''(0)=z_N+\beta_Nb_N,
 \label{eq:appB-stationary-discrete-path}
\end{equation}
where $z_N\in\mathscr R_N$, $z_N\to z$,
$\beta_N\to\beta$, and $b_N\to b$ in $\mathscr Y$.  If the limiting
path has two-jet
$a(t)=t^2(z+\beta b)/2+o(t^2)$, then
\begin{equation}
 \mathfrak O_2(0,a_2)=\Pi_{\mathrm{ng}}a_2
 =\beta\Pi_{\mathrm{ng}}b.
 \label{eq:appB-stationary-continuum-obstruction}
\end{equation}
If the right-hand side is nonzero, no regular metric two-jet realizes
the limiting path.
\end{theorem}

\begin{proof}
Since $z_N=\Pi_Nz_N$, strong projector convergence and $z_N\to z$
give $z=\Pi_{\mathrm{geo}}z$.  A realizing metric two-jet would obey
$Ah_1=a_1=0$, hence $h_1=0$ in the gauge slice.  The second-order
chain rule would reduce to $a_2=Ah_2$, so its normal projection would
vanish.  Projecting $a_2=z+\beta b$ proves the stated obstruction.
\end{proof}

This stationary theorem concerns a parametrized data path.  It does
not by itself make a central-affine proto-area into a non-affine
function on state space.  The fixed-fiber construction in
Section~\ref{subsec:fixed-center-erasure-code} supplies the stronger
ambient state-space Hessian, while
Theorem~\ref{thm:stationary-quadratic-refinement} supplies a
conditional continuum transfer under the stated strong-limit
hypotheses.

\section{Analytic geometry of the boundary-length map}
\label{app:analytic-geometry}

The linear and quadratic geometrizability tests require compatible function spaces for metric perturbations and interval data, controlled endpoint limits, gauge cancellation, and differentiability of the boundary-length map. The Hilbert range theorem and the quadratic mapping problem are kept separate below. The former is established at the tomography weight used in Section~\ref{sec:linearized-geometry}\rev{, while} the latter is stated on a higher-regularity, positive-boundary-order class.

\subsection{Asymptotically hyperbolic spaces and gauge fixing}
\label{appC:ah-gauge}

\begin{revision}
Let $\overline M$ be the closed disk and let $M$ be its interior. A boundary defining function is a smooth nonnegative function $\rho$ on $\overline M$ with $\rho^{-1}(0)=\partial M$ and $d\rho\neq0$ on $\partial M$. A metric $g$ is asymptotically hyperbolic with the conformal representative $d\theta^2$ fixed at infinity if
\begin{equation}
  g=\rho^{-2}\overline g,
  \qquad
  \overline g\in C^\infty(\overline M;S^2T^*\overline M),
  \qquad
  |d\rho|_{\overline g}=1\quad\text{at }\partial M,
  \qquad
  \overline g|_{T\partial M}=d\theta^2.
  \label{eq:appC-ah-metric}
\end{equation}
\end{revision}

\begin{revision}
The reference metric is $g_{\mathbb H}$. Equation~\eqref{eq:appC-ah-metric}
describes the smooth compactification class. We also allow a
polyhomogeneous compactification, whose subleading coefficients admit
asymptotic expansions in powers of $\rho$ and $\log\rho$, with smooth
coefficients along the boundary and the same leading boundary metric.
For the variation formulas, all coefficients satisfy the positive-order
decay and covariant differentiability in
Equation~\eqref{eq:quadratic-decay-class}. Thus the use of
polyhomogeneous metric coefficients does not require smoothness of
every subleading term, but it retains the endpoint estimates used
below. Metric perturbations are measured as zero-cotensors.
 For $m\geq6$, $0<\alpha<1$, and $\eta>0$, set
\begin{equation}
  \mathscr X^{m,\alpha}_\eta
  =\rho^{1+\eta} C^{m,\alpha}_0
  (\overline M;S^2({}^0T^*M)),
  \label{eq:appC-source-space}
\end{equation}
where $C^{m,\alpha}_0$ is the H\"older space defined by a smooth zero-frame $\rho\partial_\rho,\rho\partial_\theta$. The positive weight fixes the conformal infinity and supplies the endpoint decay used in the Hessian calculation. This high-regularity source is not identified with the $\rho^{1/2}L^2$ completion used for the linear range theorem.
\end{revision}

The geometric data are functions of ordered ideal endpoints. Write
\begin{equation}
  \partial_2M=(S^1\times S^1)\setminus\operatorname{diag},
  \qquad
  a(u,v)=a(v,u).
  \label{eq:appC-boundary-pairs}
\end{equation}
Changing the defining function from $\rho$ to $e^\omega\rho$ changes
the renormalized length by the endpoint term $\omega(u)+\omega(v)$.
Both $\rho$ and the conformal representative are
fixed, so no endpoint-Weyl quotient is taken.  The topology near the
endpoint diagonal is the weighted topology incorporated in the chosen
data space\rev{, and} for the two-jet calculation it is the weighted topology in
which $A$, $\mathscr R_2$, and the Hessian terms below are defined.

\begin{revision}
Let $\mathfrak X^{m+1,\alpha}_{\eta}$ be the space of
$C^{m+1,\alpha}_0$ zero-vector fields $X$ with
$X=O(\rho^{1+\eta})$ in a unit zero-frame. Its elements satisfy
$X(\rho)/\rho\to0$ and generate flows in
$\operatorname{Diff}_0^x(M;\partial M)$, since $x/\rho\to1$ for the
fixed boundary representative. Equivalently, $X^\flat$ has positive
order as a zero-cotensor. These conditions imply the vanishing of the
endpoint pairing $g_{\mathbb H}(X,T)$ used below. The infinitesimal
gauge action is
\begin{equation}
  X\longmapsto\mathcal L_Xg_{\mathbb H}=2\operatorname{sym}\nabla X.
  \label{eq:appC-gauge-action}
\end{equation}
The following identity gives its action on the linearized data.
\end{revision}

\begin{lemma}[Boundary-fixing gauge annihilation]
\label{lem:appC-gauge-annihilation}
For $X\in\mathfrak X^{m+1,\alpha}_{\eta}$,
\begin{equation}
  I_2(\mathcal L_Xg_{\mathbb H})(u,v)=0
  \label{eq:appC-potential-zero}
\end{equation}
for every ordered pair of ideal endpoints.
\end{lemma}

\begin{proof}
Let $\gamma$ be the hyperbolic geodesic from $u$ to $v$, parametrized by arclength $s$ and tangent $T$. For the truncated segment $\gamma_{[s_-,s_+]}$,
\begin{equation}
  \int_{s_-}^{s_+}(\mathcal L_Xg_{\mathbb H})(T,T)\,ds
  =2\int_{s_-}^{s_+}g_{\mathbb H}(\nabla_TX,T)\,ds
  =2g_{\mathbb H}(X,T)\big|_{s_-}^{s_+}.
  \label{eq:appC-potential-integration}
\end{equation}
The positive boundary order makes the endpoint contribution vanish as
$s_\pm\to\pm\infty$.
\end{proof}

A vector field with nonzero tangential boundary value is not part of
this gauge group.  If it generates $Y$ on $S^1$, the induced endpoint
relabeling contains
\begin{equation}
 Y(u)\partial_u\mathcal L_g(u,v)
 +Y(v)\partial_v\mathcal L_g(u,v),
 \label{eq:appC-boundary-relabeling}
\end{equation}
which is not generally of the endpoint-Weyl form
$\omega(u)+\omega(v)$.  It changes the labels of measured boundary
regions and is not quotiented out when the endpoints and the conformal
representative are fixed.

The gauge slice used in the main text is the iterated transverse-traceless representative. The solenoidal projection below is a distinct elliptic decomposition and supplies the comparison estimate.

\begin{lemma}[Solenoidal gauge projection]
\label{lem:appC-elliptic-gauge-projection}
Assume that the weight $\eta$ is chosen so that
\begin{equation}
 P=\delta_{g_{\mathbb H}}\mathcal L_{(\cdot)}g_{\mathbb H}:
 \rho^{1+\eta}C_0^{m+1,\alpha}\longrightarrow
 \rho^{1+\eta}C_0^{m-1,\alpha}
 \label{eq:appC-vector-isomorphism}
\end{equation}
is an isomorphism.  Every $h\in\mathscr X^{m,\alpha}_\eta$ then has a unique decomposition
\begin{equation}
  h=\mathcal L_Xg_{\mathbb H}+h^{\mathrm{sol}},
  \qquad
  \delta_{g_{\mathbb H}}h^{\mathrm{sol}}=0,
  \qquad
  X\in\mathfrak X^{m+1,\alpha}_{\eta},
  \label{eq:appC-solenoidal-decomposition}
\end{equation}
with
\begin{equation}
  \|X\|_{\rho^{1+\eta}C^{m+1,\alpha}_0}
  +\|h^{\mathrm{sol}}\|_{\rho^{1+\eta}C^{m,\alpha}_0}
  \leq C\|h\|_{\rho^{1+\eta}C^{m,\alpha}_0}.
  \label{eq:appC-solenoidal-estimate}
\end{equation}
Independently, applying the iterated decomposition of
Theorem~\ref{thm:itt-decomposition} to $h$ produces a representative
$h^{\mathrm{itt}}$ satisfying
\begin{equation}
  I_2h=I_2h^{\mathrm{itt}}
  \label{eq:appC-transform-gauge-slice}
\end{equation}
and obeys the estimate \eqref{eq:itt-decomposition-estimate}.
\end{lemma}

\begin{revision}
\begin{proof}
Apply the divergence operator to the desired decomposition. The vector field $X$ must solve
\begin{equation}
  P X:=\delta_{g_{\mathbb H}}\mathcal L_Xg_{\mathbb H}=\delta_{g_{\mathbb H}}h.
  \label{eq:appC-vector-elliptic-equation}
\end{equation}
The assumed isomorphism gives $X=P^{-1}\delta h$ with the stated
weighted Schauder estimate\rev{, the elliptic regularity estimate in
weighted H\"older norms}, proving
Equations~\eqref{eq:appC-solenoidal-decomposition}
and~\eqref{eq:appC-solenoidal-estimate}.  The iterated-TT
representative is obtained from its own decomposition theorem, not by
identifying it with the solenoidal representative. The potential term
in that decomposition belongs to the weighted one-form class in
Corollary~\ref{cor:diffeomorphism-gauge}, which proves
Equation~\eqref{eq:appC-transform-gauge-slice}. The displayed Schauder
estimate concerns the solenoidal decomposition, and it does not by itself
supply a Schauder estimate or preserve the same boundary weight for
the iterated-TT representative.
\end{proof}
\end{revision}

\subsection{Linearized boundary length and tensor tomography}
\label{appC:linearized-length}

For a metric perturbation $h$ of positive boundary order, the derivative of renormalized length is obtained by first differentiating a finite segment and only then taking the endpoint limit. This order avoids spurious endpoint terms.

Let $\gamma_{u,v}$ be parametrized by hyperbolic arclength and write $p_\epsilon(u,v)$ and $q_\epsilon(u,v)$ for its intersections with $\{\rho=\epsilon\}$, as in Equation~\eqref{eq:truncated-endpoints}. Let $\gamma_{t,\epsilon}$ be the $g_t=g_{\mathbb H}+th$ geodesic joining those two finite points. Its length is
\begin{equation}
  L_\epsilon(t)=\operatorname{Length}_{g_t}(\gamma_{t,\epsilon}).
  \label{eq:appC-truncated-length}
\end{equation}
Stationarity of $\gamma_{0,\epsilon}$ gives
\begin{equation}
  \dot L_\epsilon(0)
  =\frac12\int_{\gamma_{0,\epsilon}}h(T,T)\,ds.
  \label{eq:appC-truncated-first-variation}
\end{equation}
If one instead keeps the ideal endpoints fixed while moving the cutoff points with the metric, the additional finite-endpoint variation \rev{tends to zero as $\epsilon\to0$ under the decay assumption in Equation~\eqref{eq:quadratic-decay-class}.  The displacement of a cutoff point has a component tangent to the level set $\rho=\epsilon$, which pairs with the geodesic tangent only through the small angle between the geodesic and the normal to the level set, and a normal component controlled by the boundary decay of the perturbation, and both pairings vanish in the limit}. Therefore
\begin{equation}
  \left.\frac{d}{dt}\right|_{0}
  \mathcal L_{g_t}(u,v)
  =\frac12\int_{\gamma_{u,v}}h(T,T)\,ds.
  \label{eq:appC-renormalized-first-variation}
\end{equation}
The integral is absolutely convergent under Equation~\eqref{eq:quadratic-decay-class}. For a merely weighted tensor it is defined by density from compact support and continuity in the tomography norm.

\begin{revision}
The transform in \eqref{eq:appC-renormalized-first-variation} is the rank-two geodesic X-ray transform. Let $\mathcal X_{\mathrm{itt}}$ be the Hilbert space of iterated representatives with the source norm used in Section~\ref{sec:linearized-geometry}, and let $\mathcal Y^\tau$ be the even data Hilbert space defined in Equation~\eqref{eq:adapted-data-space}. The tensor-tomography range theorem gives a homeomorphism
\begin{equation}
  I_2:\mathcal X_{\mathrm{itt}}
  \longrightarrow \mathcal Y_{\mathrm{geo}}
  \subset\mathcal Y^\tau,
  \label{eq:appC-transform-homeomorphism}
\end{equation}
where
\begin{equation}
  \mathcal Y_{\mathrm{geo}}
  =\mathcal H_0\widehat\oplus\mathcal H_2,
  \qquad
  \mathcal Y_{\mathrm{ng}}
  =\widehat\bigoplus_{q\geq2}\mathcal H_{2q}.
  \label{eq:appC-data-sector-splitting}
\end{equation}
Here the scalar sector records the trace part and the second even sector records the transverse-traceless part, while the higher even sectors are the linear nongeometric data. The closed-range projector is
\begin{equation}
  \Pi_{\mathrm{geo}}=\frac12I_2\mathscr R_2,
  \qquad
  \Pi_{\mathrm{ng}}=I-\Pi_{\mathrm{geo}},
  \qquad
  \mathscr R_2=2(I_2|_{\mathcal X_{\mathrm{itt}}})^{-1}\Pi_{\mathrm{geo}}.
  \label{eq:appC-projectors}
\end{equation}
The factor $2$ is required because $D\mathcal B=I_2/2$.
\end{revision}

A finite spectral truncation gives the same structure without any functional-analytic subtlety. Let $E_L\subset\mathcal X_{\mathrm{itt}}$ be a finite-dimensional space of smooth tensors and let $R_L=I_2(E_L)$. Since $I_2$ is injective on $\mathcal X_{\mathrm{itt}}$, the Gram matrix
\begin{equation}
  G^{(L)}_{ij}=\langle I_2e_i,I_2e_j\rangle_{\mathcal Y^\tau}
  \label{eq:appC-finite-gram}
\end{equation}
is positive definite for every basis $\{e_i\}$ of $E_L$. Hence
\begin{equation}
  \Pi_L a=\sum_{i,j}I_2e_i\,(G^{(L)})^{-1}_{ij}
  \langle a,I_2e_j\rangle_{\mathcal Y^\tau}
  \label{eq:appC-finite-projector}
\end{equation}
is the exact orthogonal projection onto $R_L$. If $E_L$ is source-dense and the sampling maps satisfy the operator consistency assumptions of Theorem~\ref{thm:finite-code-continuum-limit}, then $\Pi_L\to\Pi_{\mathrm{geo}}$ strongly. The strong convergence of these projectors underlies the convergence of finite cut witnesses to continuum nongeometric witnesses.

\subsection{The renormalized Hessian}
\label{appC:hessian}

The second derivative of the boundary-length map has two origins. The line element changes along the reference geodesic, and the extremal curve itself moves. The second contribution produces the nonlinear normal acceleration used in Section~\ref{sec:nonlinear-consistency}.

For a symmetric tensor $h$, the variation of the Levi-Civita connection is the tensor $C_h$ defined by
\begin{equation}
  2g_{\mathbb H}(C_h(X,Y),Z)
  = (\nabla_Xh)(Y,Z)+(\nabla_Yh)(X,Z)-(\nabla_Zh)(X,Y).
  \label{eq:appC-connection-variation}
\end{equation}
Along a reference geodesic $\gamma$ with unit tangent $T$ and parallel unit normal $N$, write
\begin{equation}
  \mathcal F_\gamma(h)=[C_h(T,T)]^\perp=f_hN,
  \qquad
  f_h=(\nabla_Th)(T,N)-\frac12(\nabla_Nh)(T,T).
  \label{eq:appC-forcing}
\end{equation}
If $g_t=g_{\mathbb H}+th+O(t^2)$ and $\gamma_t$ is the $g_t$-geodesic with the same ideal endpoints, its normal displacement $J_h$ solves
\begin{equation}
  \mathcal J_\gamma J_h=\mathcal F_\gamma(h),
  \qquad
  \mathcal J_\gamma=-D_s^2-R(\,\cdot\,,T)T.
  \label{eq:appC-forced-jacobi}
\end{equation}
On the Poincar\'e disk, $\mathcal J_\gamma=-\partial_s^2+1$, and the decaying Green kernel is $\frac12e^{-|s-r|}$.

The mixed curve derivative used in the main text follows from one integration by parts. If $V=vN$ is a normal variation of $\gamma$ with sufficient decay, then
\begin{equation}
\begin{split}
  \left.\frac{d}{dr}\right|_0
  \frac12\int_{\eta_r}h(T_r,T_r)\,ds_r
  ={}&\int_\gamma
  \left[\frac12(\nabla_Vh)(T,T)-(\nabla_Th)(V,T)\right]ds \\
  ={}&-\int_\gamma g_{\mathbb H}(\mathcal F_\gamma(h),V)\,ds.
  \label{eq:appC-mixed-curve-derivative}
\end{split}
\end{equation}
The endpoint terms vanish because $V$ and $h$ have positive boundary order. The sign and the inverse Jacobi operator in the Hessian follow from this identity.

Consider an affine two-parameter family $g_{t,u}=g_{\mathbb H}+th+uk$. On a finite segment with endpoints fixed, differentiating the first-variation formula in the $u$ direction gives
\begin{equation}
  \partial_u\partial_t L(0,0)
  =-\frac14\int h(T,T)k(T,T)\,ds
   -\int g_{\mathbb H}(\mathcal F_\gamma(h),J_k)\,ds,
  \label{eq:appC-finite-hessian}
\end{equation}
where $J_k$ is the Dirichlet solution of $\mathcal J_\gamma J_k=\mathcal F_\gamma(k)$. Passing to ideal endpoints requires only one estimate. In Fermi coordinates near an end of $\gamma$\rev{, the arclength along $\gamma$ together with the signed normal distance from it}, Equation~\eqref{eq:quadratic-decay-class} gives
\begin{equation}
  |h(T,T)|+|k(T,T)|+|f_h|+|f_k|
  \leq C e^{-\mu |s|}
  \label{eq:appC-end-decay}
\end{equation}
for some $\mu>0$. The finite Dirichlet kernels are bounded by $Ce^{-|s-r|}$ uniformly on compactly exhausted intervals. Thus all terms in \eqref{eq:appC-finite-hessian} are dominated by an integrable function independent of the cutoff. The limit is
\begin{equation}
\begin{split}
  D^2\mathcal B_{g_{\mathbb H}}[h,k](\gamma)
  ={}&-\frac14\int_\gamma h(T,T)k(T,T)\,ds \\
  &-\int_\gamma g_{\mathbb H}\left(
  \mathcal F_\gamma(h),
  \mathcal J_\gamma^{-1}\mathcal F_\gamma(k)
  \right)ds.
  \label{eq:appC-hessian-operator}
\end{split}
\end{equation}
On the hyperbolic disk this is the double-integral formula
\begin{equation}
\begin{split}
  D^2\mathcal B_{g_{\mathbb H}}[h,k](\gamma)
  ={}&-\frac14\int_{-\infty}^{\infty}h_{TT}(s)k_{TT}(s)\,ds \\
  &-\frac12\int_{-\infty}^{\infty}\int_{-\infty}^{\infty}
  e^{-|s-r|}f_h(s)f_k(r)\,dr\,ds.
  \label{eq:appC-hessian-green}
\end{split}
\end{equation}
The formula is symmetric because $\mathcal J_\gamma^{-1}$ is self-adjoint. It is also gauge covariant. If $h=\mathcal L_Xg_{\mathbb H}$, the flow of $X$ gives
\begin{equation}
  D^2\mathcal B_{g_{\mathbb H}}[\mathcal L_Xg_{\mathbb H},k]
  =-D\mathcal B_{g_{\mathbb H}}[\mathcal L_Xk].
  \label{eq:appC-hessian-gauge}
\end{equation}
The right-hand side lies in $\mathcal Y_{\mathrm{geo}}$. Therefore the normal bilinear form
\begin{equation}
  \mathrm{II}_{g_{\mathbb H}}([h],[k])
  =\Pi_{\mathrm{ng}}D^2\mathcal B_{g_{\mathbb H}}[h,k]
  \label{eq:appC-second-fundamental-form}
\end{equation}
only depends on gauge classes. This gauge-class bilinear form is the second fundamental form entering the quadratic obstruction \eqref{eq:quadratic-obstruction-definition}.

\begin{revision}
The formula alone does not imply a bilinear estimate in the
$\rho^{1/2}L^2$ source norm, because $f_h$ contains one derivative and the
direct term contains the product $h\odot k$. Quantitative estimates require a source and data norm
pair for which
\begin{equation}
  \|D^2\mathcal B[h,k]\|_{\mathscr Y}
  \leq C_H\|h\|_{\mathscr X}\|k\|_{\mathscr X}.
  \label{eq:appC-hessian-bound}
\end{equation}
On a finite-dimensional regular source, membership of all pairwise
Hessian data in $\mathscr Y$ gives this estimate by bilinearity, and the
corresponding constant for the normal projection is exhibited in
Corollary~\ref{cor:finite-source-quadratic-stability}. For a chosen
infinite-dimensional norm pair it is a separate boundedness hypothesis.
The coefficientwise criterion only uses the individual Hessian data
and regular reconstructed coefficients. A strong Taylor expansion
additionally requires the pathwise differentiability in
Theorem~\ref{thm:two-jet-geometrizability}. Neither the source-adapted
linear inverse estimate nor the one-dimensional Green-kernel bound
alone supplies these stronger mapping statements.
\end{revision}

\subsection{Coefficientwise two-jet regularity}
\label{appC:two-jet-regularity}

The coefficientwise criterion does not require a nonlinear inverse-function theorem, but both reconstructed metric coefficients must remain in the positive-boundary-order class used in the Hessian calculation. Suppose
\begin{equation}
  a(t)=t a_1+\frac{t^2}{2}a_2+o(t^2)
  \label{eq:appC-data-two-jet}
\end{equation}
and $\Pi_{\mathrm{ng}}a_1=0$. Let $h_1=\mathscr R_2a_1$ and assume that $h_1$ is smooth or polyhomogeneous with the positive boundary order required in Equations~\eqref{eq:appC-end-decay}--\eqref{eq:appC-hessian-operator}. A realizing metric two-jet necessarily satisfies
\begin{equation}
  a_2=Ah_2+D^2\mathcal B_{g_{\mathbb H}}[h_1,h_1].
  \label{eq:appC-two-jet-equation}
\end{equation}
Normal projection gives the quadratic obstruction in Equation~\eqref{eq:quadratic-obstruction-definition}.

Conversely, if that obstruction vanishes, the corrected datum
\begin{equation}
  r_2=a_2-D^2\mathcal B_{g_{\mathbb H}}[h_1,h_1]
  \label{eq:appC-corrected-second-data}
\end{equation}
lies in $\mathcal Y_{\mathrm{geo}}$. Assume that $h_2=\mathscr R_2r_2$ has the same positive boundary order. Then the polynomial metric path $g_t=g_{\mathbb H}+th_1+t^2h_2/2$ has the prescribed first two boundary-length coefficients by the chain rule and the Hessian formula. If $\mathcal B$ is twice differentiable into $\mathcal Y^\tau$ along this path, Taylor's theorem gives the strong $o_{\mathcal Y^\tau}(t^2)$ remainder. This proves the sufficiency direction of Theorem~\ref{thm:two-jet-geometrizability} without asserting a full nonlinear local image.

\begin{revision}
For a multiparameter state manifold, replace $a_1$ by $d\mathcal A(X)$ and $a_2$ by $\nabla^2\mathcal A(X,Y)$. The coefficient construction is finite dimensional in the state parameters, and a strong Taylor expansion requires joint twice differentiability of the composition with the resulting polynomial metric map. Changing the torsion-free connection adds $d\mathcal A(Z)$ to the Hessian. Under the first-order condition this term belongs to $\mathcal Y_{\mathrm{geo}}$, so its normal projection vanishes. This gives the connection independence used in Proposition~\ref{prop:state-obstruction-intrinsic} and Theorem~\ref{thm:multiparameter-factorization}.
\end{revision}

\section{Approximate code realizations and finite-dimensional analysis}
\label{app:code-validation}

\begin{revision}
For the fixed-cut models of
Section~\ref{sec:skewed-codes-revised}, an isometric dilation makes
the sector records and bond factors explicit.  The controlled
Hamiltonian changes the conditional bonds, while exact recovery of
the block-diagonal wedge algebra fixes the Gibbs-sector proto-area.
The calculations below evaluate the cross-cell entropy response and
the corresponding sewing witness.  The fixed-central-fiber example
has a different, diamond-minimax calibrated recovery branch.
\end{revision}

\subsection{Fixed-cut encoding and controlled Hamiltonian dilation}
\label{appD:fixed-cut-encoding}

\begin{revision}
Let $G=(V,E)$ be a tensor-network graph in a regular minimum-cut chamber. For an oriented edge $e=(v,w)$, write $\mathcal H_{e,v}$ and $\mathcal H_{e,w}$ for the two bond factors and choose a normalized state
\begin{equation}
  |\Omega_e\rangle\in
  \mathcal H_{e,v}\otimes\mathcal H_{e,w}.
  \label{eq:appD-edge-state}
\end{equation}
\end{revision}

\begin{revision}
Write $\mathcal T$ for the fixed isometric dilation from the
logical input and bond factors to the boundary.  It includes the
sector-copying map and the vertex isometries, and is isometric on the
subspace occupied by all conditional bond states considered here.
With the vertex maps denoted by $T_v$, the unperturbed encoding is
\end{revision}

\begin{revision}
\begin{equation}
  V_0
  =\left(\bigotimes_{v\in V}T_v\right)
   \left(I_L\otimes\bigotimes_{e\in E}|\Omega_e\rangle\right).
  \label{eq:appD-base-encoding}
\end{equation}
\end{revision}

\begin{revision}
For each monitored cut, the fixed-cut decoders identify the
encoded support with a subspace of
\begin{equation}
 \bigoplus_{x\in\mathsf X}
 \Bigl(
 \operatorname{span}\{|x\rangle_{C_A}|x\rangle_{C_{\bar A}}\}
 \otimes\mathcal H_{a_A,x}\otimes\mathcal H_{\bar a_A,x}
 \otimes\!\bigotimes_{e\in C_A}
       (\mathcal H_{e,A}\otimes\mathcal H_{e,\bar A})
 \otimes\mathcal H_{\mathrm{sp},A}
 \Bigr).
 \label{eq:appD-fixed-cut-factorization}
\end{equation}
The two records $C_A$ and $C_{\bar A}$ hold the same classical
sector label.  The spectator state is independent of the logical
input and of $x$, so its entropy cancels after background subtraction.
At the unperturbed fixed-cut point the cut bonds are
independent of the logical state.  The controlled Hamiltonian changes
these conditional bonds while leaving the sector records and the
assigned matter algebras unchanged.
\end{revision}

\begin{revision}
The source register $X$ carries the orthogonal projectors $P_x=|x\rangle\langle x|$. Let $K_j$ be Hermitian operators on pairwise disjoint collections of bond modes. The controlled perturbation in Equation~\eqref{eq:sec6-controlled-isometry} is generated by
\begin{equation}
  H_{\mathrm{skew}}
  =\sum_{j}\sum_{x\in\mathsf X}c_j(x)P_x\otimes K_j,
  \qquad
  U_\varepsilon=e^{-i\varepsilon H_{\mathrm{skew}}}.
  \label{eq:appD-skew-hamiltonian}
\end{equation}
Because the supports of the $K_j$ are disjoint, the summands commute and
\begin{equation}
  U_\varepsilon
  =\sum_xP_x\otimes
  \bigotimes_j e^{-i\varepsilon c_j(x)K_j}.
  \label{eq:appD-controlled-product}
\end{equation}
\end{revision}

\begin{revision}
The perturbed encoding in this isometric dilation is
\end{revision}

\begin{revision}
\begin{equation}
 V_\varepsilon
 =\left(\bigotimes_{v\in V}T_v\right)
 U_\varepsilon
 \left(I_L\otimes\bigotimes_{e\in E}|\Omega_e\rangle\right).
 \label{eq:appD-perturbed-encoding}
\end{equation}
\end{revision}

\begin{revision}
It is an isometry whenever the controlled bond states remain
in the isometric dilation domain.  Pushing the controlled unitary
through that dilation gives a physical unitary extension.  For a
fixed cut, the conditional bond state is
\begin{equation}
 |\Xi_{A,x}(\varepsilon)\rangle
 =\bigotimes_j e^{-i\varepsilon c_j(x)K_j}|\Xi_{A,j}\rangle.
 \label{eq:appD-conditioned-bond-state}
\end{equation}
The regional channel is block diagonal in $x$ because the opposite
record is traced out.  Discarding the accessible bond marginal and
retaining the record and the assigned matter implements the exact
algebraic recovery used in
Equation~\eqref{eq:sec6-proto-area-cq}.  There is no additional
small-skew assumption in this recovery statement.  The distinction
between the algebra and its full Hilbert-space extension is proved in
Proposition~\ref{prop:sec6-approximate-recovery}.
\end{revision}

\begin{revision}
\begin{proof}[Proof of Theorem~\ref{thm:sec6-local-all-order}]
For a fixed source sector $x$, the decoded bond state factorizes over cut edges. Entropy additivity and \eqref{eq:sec6-proto-area-cq} give
\begin{align}
 a_A^{\mathrm{loc}}(\lambda,\varepsilon)
 &=\sum_x[p_x(\lambda)-p_x^0]
   \sum_{e\in C_A}s\bigl(\vartheta_e+\varepsilon c_e(x)\bigr)\nonumber\\
 &=\sum_{e\in E}(M_{\mathcal C})_{Ae}
   w_e(\lambda,\varepsilon).
 \label{eq:sec6-local-proof-factorization}
\end{align}
This proves the centered factorization
\eqref{eq:sec6-local-factorization}.  Adding the matched absolute
reference gives the positive edge weights in
Equation~\eqref{eq:sec6-absolute-local-edge-weight}, and the strict cut
inequalities persist by the uniform continuity argument in the
main proof.  Multiplication by a vector in the left kernel of
$M_{\mathcal C}$ gives \eqref{eq:sec6-local-witness-zero}.
\end{proof}
\end{revision}

\begin{revision}
The cross-cell realization uses two bond pairs. With the ordering adopted in Section~\ref{subsec:cross-cell-skews}, define the amplitude matrix
\begin{equation}
  \Psi(\alpha)_{b,c}
  =\langle b|_B\langle c|_C
  U_\times(\alpha)
  |\Omega_r\rangle_{u\bar u}|\Omega_r\rangle_{v\bar v},
  \label{eq:appD-cross-amplitude}
\end{equation}
where $B=(u,\bar v)$ and $C=(\bar u,v)$. The reduced states are
\end{revision}

\begin{revision}
\begin{equation}
  \rho_B(\alpha)=\Psi(\alpha)\Psi(\alpha)^\dagger,
  \qquad
  \rho_C(\alpha)=\Psi(\alpha)^{\mathsf T}\overline{\Psi(\alpha)}.
  \label{eq:appD-reduced-amplitudes}
\end{equation}
The transpose in the second formula follows from the stated
$B\times C$ amplitude convention.  Both marginals have the same
nonzero eigenvalues, namely the squared singular values of $\Psi$.
\end{revision}

\begin{revision}
Direct multiplication gives the matrix in
Equation~\eqref{eq:sec6-cross-reduced-matrix}.  The two parity blocks
determine the exact conditional bond entropy.
\end{revision}

\begin{revision}
\subsection{Cross-cell susceptibility and algebraic recovery}
\label{appD:recovery-alignment}
\label{appD:cross-cell-calculation}
\end{revision}

\begin{proof}[Proof of Proposition~\ref{prop:sec6-cross-susceptibility}]
Equation~\eqref{eq:sec6-cross-reduced-matrix} follows by writing
\begin{equation}
 U_\times(\alpha)|ij\rangle_{uv}
 =c|ij\rangle_{uv}-is|\bar i\,\bar j\rangle_{uv}
 \label{eq:sec6-cross-gate-action}
\end{equation}
inside the product of the two Schmidt decompositions. The matrix depends on $\alpha$ through $c^2$, $s^2$, and $cs$\rev{, and} conjugation by a diagonal unitary changes the sign of $cs$, so its spectrum is even. At $\alpha=0$,
\begin{equation}
 \rho_A(0)
 =\operatorname{diag}\bigl(r^2,r(1-r),r(1-r),(1-r)^2\bigr).
 \label{eq:sec6-cross-rho-zero}
\end{equation}
The first derivative is purely off diagonal,
\begin{align}
 \dot\rho_A(0)
 =i\delta\sqrt{r(1-r)}\bigl(&|00\rangle\!\langle11|-|11\rangle\!\langle00|\nonumber\\
 &+|01\rangle\!\langle10|-|10\rangle\!\langle01|\bigr).
 \label{eq:sec6-cross-rho-dot}
\end{align}
It follows that $-\operatorname{Tr}[\dot\rho_A(0)\log\rho_A(0)]=0$. The second entropy derivative is
\begin{equation}
 \frac{d^2}{d\alpha^2}S(\rho_A(\alpha))\bigg|_{0}
 =-\operatorname{Tr}\bigl[\ddot\rho_A(0)\log\rho_A(0)\bigr]
 -\operatorname{Tr}\bigl[\dot\rho_A(0)D\log_{\rho_A(0)}(\dot\rho_A(0))\bigr].
 \label{eq:sec6-entropy-hessian-cross}
\end{equation}
For a diagonal positive matrix with eigenvalues $\lambda_a$,
\begin{equation}
 \bigl[D\log_\rho(Z)\bigr]_{ab}
 =Z_{ab}\frac{\log\lambda_a-\log\lambda_b}{\lambda_a-\lambda_b},
 \label{eq:sec6-divided-log}
\end{equation}
with the continuous value $1/\lambda_a$ when the eigenvalues coincide. Substitution of \eqref{eq:sec6-cross-rho-zero} and \eqref{eq:sec6-cross-rho-dot}, together with the diagonal second derivative read from \eqref{eq:sec6-cross-reduced-matrix}, gives
\begin{equation}
 F_r''(0)
 =2\delta\bigl[r^2+(1-r)^2\bigr]\log\frac{r}{1-r}-2\delta^2.
 \label{eq:sec6-chi-intermediate}
\end{equation}
Since $r^2+(1-r)^2=(1+\delta^2)/2$ and
$\log[r/(1-r)]=2\operatorname{arctanh}\delta$, this is \eqref{eq:sec6-chi-r}. Positivity follows from $\operatorname{arctanh}\delta>\delta$ for $0<\delta<1$:
\begin{equation}
 (1+\delta^2)\operatorname{arctanh}\delta-\delta
 >\delta^3>0.
 \label{eq:sec6-chi-positive}
\end{equation}
The value \eqref{eq:sec6-chi-example} is obtained by setting $\delta=3/5$ and $\operatorname{arctanh}(3/5)=\log2$.  For the single-cross
function, differentiation of
$A=r-(2r-1)\sin^2\alpha$ gives
\begin{equation}
 f_r''(0)=4\delta\operatorname{arctanh}\delta.
 \label{eq:sec6-single-cross-susceptibility}
\end{equation}
Subtracting this from Equation~\eqref{eq:sec6-chi-r} yields
Equation~\eqref{eq:sec6-connected-susceptibility}.  Its sign is
strictly negative because $0<\delta<1$ and every term in brackets is
positive.  Equation~\eqref{eq:sec6-connected-example} follows from
$\delta=3/5$.
\end{proof}

\begin{revision}
The recovered algebraic entropy can be checked directly, with
no reference to a quantum-source coherence alignment.  Let $B_A$ be
the accessible bond factor and let $E_A$ be an isomorphic discarded
factor.  A fixed-cut Stinespring branch may act as
\begin{equation}
 W_A=\sum_x |x\rangle_{Y_A}\langle x|_{C_A}
       \otimes I_{a_A,x}\otimes U_{A,x}^{B_A\to E_A},
 \qquad U_{A,x}\in U(d_{B_A}).
 \label{eq:appD-center-preserving-recovery}
\end{equation}
It retains the center record and matter, and moves only the bond
factor into the discarded output.  For every such branch,
\begin{equation}
 \bigoplus_xp_x\sigma_{A,x}\otimes\tau_{A,x}
 \longmapsto\bigoplus_xp_x\sigma_{A,x},
 \label{eq:appD-diagonal-sector-fixed}
\end{equation}
after tracing $E_A$.  The recovered entropy is therefore
$H(p)+\sum_xp_xS(\sigma_{A,x})$.  The boundary entropy is this
same quantity plus $\sum_xp_xS(\tau_{A,x})$, proving
Equation~\eqref{eq:sec6-proto-area-cq}.  The unitary freedom in
$U_{A,x}$ changes neither entropy.
\end{revision}

\begin{revision}
For completeness, fixing the matter input and retaining a coherent
source input in the global encoding does not change this regional
conclusion.  The off-diagonal blocks contain the opposite record
$|x\rangle_{C_{\bar A}}\langle y|$, whose trace vanishes for
$x\ne y$.  The induced source channel is thus entanglement breaking, that is,
its output is separable from any reference system.
For the maximally mixed source input its coherent information is zero
before recovery, and data processing gives
\begin{equation}
 I_c(\mathcal R_A\circ\mathcal N_A^X)
 \leq I_c(\mathcal N_A^X)=0.
 \label{eq:appD-raw-channel-upper-bound}
\end{equation}
Both complementary outputs contain the orthogonal sector record, so
the Shannon entropy appears on both sides of the pure-state dilation
and cancels in this equality.  This source-channel statement is
separate from the joint two-sided recovered channel of
Equation~\eqref{eq:section2-recovered-channel}, and from recovery of
the noncentral matter assigned to a particular wedge.
\end{revision}

\subsection{Entropy jets and closed-form evaluation}
\label{appD:entropy-jets}

The source parameter $\lambda$ and the Hamiltonian strength $\varepsilon$ play different roles. The geometry tests differentiate with respect to $\lambda$ at fixed encoding, while the weak-skew expansion controls the size of the resulting coefficient. Let
\begin{equation}
  s_{A,x}(\varepsilon)
  :=S\bigl(\tau_{A,x}(\varepsilon)\bigr).
  \label{eq:appD-sector-entropy-profile}
\end{equation}
Equation~\eqref{eq:sec6-proto-area-cq} and differentiation of the normalized Gibbs weights give the closed-form formulas
\begin{align}
  \left.\partial_\lambda a_A\right|_0
  &=\sum_xp_x^0q_xs_{A,x}(\varepsilon),
  \label{eq:appD-first-state-jet}\\
  \left.\partial_\lambda^2a_A\right|_0
  &=\sum_xp_x^0\bigl(q_x^2-\sigma_q^2\bigr)
  s_{A,x}(\varepsilon).
  \label{eq:appD-second-state-jet}
\end{align}
For a local Schmidt rotation, the one-bond entropy and its first two angle derivatives are
\begin{align}
  s(\vartheta)
  &=h_2(\sin^2\vartheta),
  \label{eq:appD-local-entropy}\\
  s'(\vartheta)
  &=\sin(2\vartheta)\log\cot^2\vartheta,
  \label{eq:appD-local-entropy-first}\\
  s''(\vartheta)
  &=2\cos(2\vartheta)\log\cot^2\vartheta-4.
  \label{eq:appD-local-entropy-second}
\end{align}
These expressions generate the exact edge-weight jets used in Theorem~\ref{thm:sec6-local-all-order}. The strict positivity of $s'$ on $0<\vartheta<\pi/4$ supplies the local inverse entropy coordinate in Equation~\eqref{eq:sec6-inverse-entropy-coordinate}.

For the cross-cell gate, the entropy is evaluated from the two Hermitian $2\times2$ blocks of Equation~\eqref{eq:sec6-cross-reduced-matrix}. If a block is
\begin{equation}
  B=\begin{pmatrix}a&z\\ \bar z&d\end{pmatrix},
  \label{eq:appD-two-by-two-block}
\end{equation}
its eigenvalues are
\begin{equation}
  \lambda_\pm(B)
  =\frac{a+d\pm\sqrt{(a-d)^2+4|z|^2}}{2}.
  \label{eq:appD-block-eigenvalues}
\end{equation}
Applying Equation~\eqref{eq:appD-block-eigenvalues} to both parity blocks gives the four eigenvalues of $\rho_A(\alpha)$ and hence its entropy in closed form.

\begin{revision}
The raw crossed entropy must be combined with the single-cross entropy
before applying a geometric witness.  
\subsection{Erasure correlations and unrestricted minimax recovery}
\label{appD:fixed-center-calculation}
\end{revision}

\begin{revision}
The reference matching, entropy derivatives, and recovery calibration
in Section~\ref{subsec:fixed-center-erasure-code} can all be evaluated
without an ansatz for the regional density matrices.
\end{revision}

\begin{revision}
First consider the resource state in
Equation~\eqref{eq:fixed-fiber-resource-state}.  The four coordinate
functions on $\mathbb Z_3^2$ have coefficient vectors
$(1,0),(0,1),(1,1),(1,2)$.  Every pair has nonzero determinant modulo
$3$.  Thus any two coordinates determine $(i,j)$ uniquely, and so do
the complementary two coordinates.  Tracing the latter from
$|\Omega\rangle\langle\Omega|$ removes all off-diagonal terms and
leaves each two-qutrit basis vector with probability $1/9$.
Every two-site reduction is consequently $I_9/9$, its single-site
marginals are $I_3/3$, and purity fixes the three-site entropies.
The resource entropy in any proper nonempty region is exactly
\begin{equation}
 S(\Omega_A\otimes\Phi_A)
 =K\min\{|A|,4-|A|\}+L\epsilon_A.
 \label{eq:fixed-fiber-resource-entropy}
\end{equation}
\end{revision}

\begin{revision}
Put $k_A=m_A-r_A$.  These are the numbers of encoded qubits whose
two halves and whose single half, respectively, are available in
$A$.  On the support of its output, the noncentral regional encoding
channel is unitarily equivalent to
\begin{equation}
 \mathcal N_{k_A,r_A}
 =\operatorname{id}^{\otimes k_A}\otimes
       \mathcal E^{\otimes r_A},
 \qquad
 \mathcal E(X)=\frac12X\oplus\frac12\tr(X)|e\rangle\langle e|,
 \label{eq:fixed-fiber-regional-channel}
\end{equation}
acting on the target marginal, tensored with the fixed resource
marginal.  If two halves are present, their apparent superposition of
transmission and erasure is retained coherently and inverted by
$E^\dagger$ on $\operatorname{im}E$.  Only a split pair produces the
orthogonal erasure sectors in \eqref{eq:fixed-fiber-regional-channel}.
The inverse on $\operatorname{im}E$ extends to a channel by assigning
an arbitrary fixed output on its orthogonal complement.
\end{revision}

\begin{revision}
\begin{proof}[Entropy and BKM calculation for
Proposition~\ref{prop:fixed-fiber-exact-defect}]
For $m_A=2$, label an erasure sector by the subset $S$ of the $r_A$
split qubits that were received.  All $2^{r_A}$ sectors have weight
$2^{-r_A}$.  If $S$ contains every split qubit, the conditional state
is isometric to $\varrho_t$ and has entropy $2L+f(t)$.  Otherwise at
least one qubit in the $Z\otimes Z$ correlation is traced out and
the remaining conditional state is maximally mixed on its
$k_A+|S|$ qubits.  The noncentral boundary entropy is therefore
\begin{equation}
 \begin{split}
 S(\omega^{(0)}_{A,\mathrm{core}}(t))
 &=r_A L+2^{-r_A}\sum_{S\subset\{1,\ldots,r_A\}}
             (k_A+|S|)L+2^{-r_A}f(t)\\
 &=(k_A+3r_A/2)L+\lambda_A f(t).
 \end{split}
 \label{eq:fixed-fiber-core-entropy}
\end{equation}
This identity includes $r_A=0$, where the single conditional state
is the full logical state.  When $m_A<2$ the target marginal is
maximally mixed and the same expression holds with the $f$ term
absent.

The recovery applies $\mathcal D$ to each split qubit.  Since
$\mathcal D(Z)=Z/2$ and $\mathcal D(I_2)=I_2$, its output for
$m_A=2$ is
\begin{equation}
 \tau_A^{(0)}(t)=\frac14(I_4+\lambda_A tZ\otimes Z),
 \qquad S(\tau_A^{(0)}(t))=2L+f(\lambda_A t).
 \label{eq:fixed-fiber-recovered-correlation}
\end{equation}
For $m_A<2$ it is $\pi_2^{\otimes m_A}$.
The second central block is obtained by setting $t=0$.
For any regional boundary or recovered state the direct-sum entropy
identity reads
\begin{equation}
 S(q\sigma_0\oplus(1-q)\sigma_1)
 =h_2(q)+qS(\sigma_0)+(1-q)S(\sigma_1).
 \label{eq:fixed-fiber-center-cancellation}
\end{equation}
The Shannon term cancels in the proto-area.  Subtracting the
recovered entropy from \eqref{eq:fixed-fiber-core-entropy} and adding
\eqref{eq:fixed-fiber-resource-entropy} proves
\eqref{eq:fixed-fiber-reference-formula} and
\eqref{eq:fixed-fiber-all-region-response}.  In particular,
\begin{equation}
 p_{12}(t)-p_{12}(0)=p_4(t)-p_4(0)
 =q\left\{\frac14 f(t)-f(t/4)\right\}
 =-\frac{3q}{32}t^2+O(t^4).
 \label{eq:fixed-fiber-exact-area12}
\end{equation}

All these states are positive definite on supports independent of
$t$ for $|t|<1$.  At $t=0$, the BKM norm squared of
$\dot\varrho_0=Z\otimes Z/4$ is
$\tr(\dot\varrho_0\pi_4^{-1}\dot\varrho_0)=1$.
Only the fully received sector has a nonzero tangent in the boundary
state, so its BKM norm squared is $q\lambda_A$.
The recovered tangent has norm squared $q\lambda_A^2$.
Tensoring with a fixed state leaves these norms unchanged, and the
factor $q$ follows from
$g_{q\sigma}(qX,qX)=qg_\sigma(X,X)$ in a direct-sum block.
Their difference is $\Delta_A=q\lambda_A(1-\lambda_A)$.
If $m_A<2$, both tangents vanish.  This verifies the defects directly
from the information metric, independently of the entropy expansion.
\end{proof}
\end{revision}

\begin{revision}
\begin{proposition}[Unrestricted diamond-minimax calibration]
\label{prop:fixed-fiber-minimax}
For the channel \eqref{eq:fixed-fiber-regional-channel}, let the
infimum range over all channels from its entire output to the
$k+r$ target qubits.  Then
\begin{equation}
 \inf_{\mathcal R}
 \|\mathcal R\mathcal N_{k,r}-\operatorname{id}_{2^{k+r}}\|_\diamond
 =2\left[1-\left(\frac58\right)^r\right].
 \label{eq:fixed-fiber-minimax-objective}
\end{equation}
The product decoder that exactly recovers intact pairs and replaces
an erased qubit by $\pi_2$ attains the infimum.  The same result
holds for the center-preserving direct-sum code and its specified
regional target marginal channels.
\end{proposition}
\end{revision}

\begin{revision}
\begin{proof}
For a lower bound, input a tensor product of $k+r$ maximally
entangled qubit pairs, keeping one half of each as a reference.
In a sector where $l$ split qubits were erased, their reference
halves are $\pi_2^{\otimes l}$ and independent of everything
accessible to the decoder.  They remain so under an arbitrary
joint recovery.  Contracting these reference factors with the
projector onto the desired maximally entangled output gives
$4^{-l}$ times a density-matrix overlap on the surviving pairs,
which is at most one.  Thus the overlap with the desired pure
Choi state, the output of the ideal channel on half of the maximally
entangled input, is at most $4^{-l}$ in that sector.
The sector weights are $2^{-r}$, independently of the input.  The
total overlap is bounded by
\begin{equation}
 F\le2^{-r}\sum_{l=0}^r\binom rl4^{-l}
       =\left(\frac58\right)^r.
 \label{eq:fixed-fiber-minimax-fidelity}
\end{equation}
Measuring the projector onto the desired pure state and its
complement gives trace distance at least $2(1-F)$.
This is a lower bound on the diamond norm for every recovery,
without a covariance or product assumption.

For the stated decoder, each split qubit has recovered channel
\begin{equation}
 \mathcal D(X)=\frac58X+\frac18(X_1XX_1+X_2XX_2+X_3XX_3),
 \label{eq:fixed-fiber-pauli-channel}
\end{equation}
where $X_1,X_2,X_3$ are the three Pauli matrices.  Hence the
$r$-qubit recovered channel is a Pauli channel whose identity
probability is $(5/8)^r$.  For any input, the triangle inequality
bounds its distance from the identity channel by twice the total
nonidentity probability.  On a maximally entangled input all
distinct Pauli errors give orthogonal states, so this bound is
attained.  The untouched $k$ qubits do not change the diamond norm.
This proves \eqref{eq:fixed-fiber-minimax-objective}.

The fixed resource marginal is independent of the logical input
and can be absorbed into an arbitrary decoder, so it cannot weaken
the lower bound.  For a target marginal with absent logical
qubits, the regional channel factors through their partial trace.
Composition with that partial trace gives the same upper bound, and
a test input supported on the target factors gives the same lower
bound.  Finally, a test input within either central block supplies
the direct-sum lower bound, and the block-preserving decoder
attains the bound in both blocks.  The copied center itself is
recovered exactly.
\end{proof}
\end{revision}

\begin{revision}
The calibration criterion is consequently operational and does not
depend on the chosen correlation path.  It is nevertheless distinct
from optimization of the joint coherent information for complementary
logical factors.  Nothing in the preceding proof asserts that the
two optimization problems have the same solution.
\end{revision}

\begin{revision}
\begin{proposition}[Separation of four error quantities]
\label{prop:fixed-fiber-error-separation}
For region $12$, let $\omega_{12}(t)$ and $\tau_{12}(t)$ include both
central blocks and let $\rho_t$ be the full logical state in
\eqref{eq:fixed-fiber-affine-path}.  Then
\begin{align}
 D(\omega_{12}(t)\Vert\omega_{12}(0))
 &=\frac{q}{8}t^2+O(t^4),
 \label{eq:fixed-fiber-full-relative-entropy}\\
 D(\tau_{12}(t)\Vert\tau_{12}(0))
 &=\frac{q}{32}t^2+O(t^4),
 \label{eq:fixed-fiber-reduced-relative-entropy}\\
 D(\omega_{12}(t)\Vert\omega_{12}(0))
 -D(\tau_{12}(t)\Vert\tau_{12}(0))
 &=\frac{3q}{32}t^2+O(t^4),
 \label{eq:fixed-fiber-relative-entropy-loss}\\
 D(\tau_{12}(t)\Vert\rho_t)
 &=\frac{9q}{32}t^2+O(t^4),
 \label{eq:fixed-fiber-state-recovery-error}\\
 \|\mathcal D\otimes\mathcal D-\operatorname{id}_4\|_\diamond
 &=\frac{39}{32}.
 \label{eq:fixed-fiber-diamond-error}
\end{align}
Thus retained distinguishability, data-processing loss, state-specific
recovery deviation, and worst-case channel error are different
quantities.
\end{proposition}
\end{revision}

\begin{revision}
\begin{proof}
The first two expressions are exactly $-qf(t)/4$ and $-qf(t/4)$,
respectively, since only the fully received boundary sector carries
the correlation.  Their difference gives the third expression.
For two commuting correlation states,
\begin{equation}
 D(\varrho_{at}\Vert\varrho_{bt})
 =\frac{1+at}{2}\log\frac{1+at}{1+bt}
  +\frac{1-at}{2}\log\frac{1-at}{1-bt}
 =\frac{(a-b)^2}{2}t^2+O(t^4).
 \label{eq:fixed-fiber-correlation-relative-entropy}
\end{equation}
Taking $a=1/4$, $b=1$, and multiplying by $q$ gives the fourth
expression.  The final identity is
\eqref{eq:fixed-fiber-minimax-objective} with $r=2$.
\end{proof}
\end{revision}

\subsection{Witness reconstruction and finite-size certificates}
\label{appD:witness-certificates}

The cross-cell contribution is tested against the fixed geometric range defined by the local bond response. In the four-terminal chamber of Equation~\eqref{eq:star-cut-matrix}, choose the unit crossed-cut vector
\begin{equation}
  b^\times=(0,0,0,0,1,0)^{\mathsf T}.
  \label{eq:appD-cross-vector}
\end{equation}
The exact least-squares edge reconstruction and residual are
\begin{align}
  M^+b^\times
  &=\left(\frac38,\frac14,-\frac18,0\right)^{\mathsf T},
  \label{eq:appD-cross-pseudoinverse}\\
  (I-MM^+)b^\times
  &=\left(-\frac38,-\frac14,\frac18,0,
  \frac38,-\frac18\right)^{\mathsf T}.
  \label{eq:appD-cross-residual}
\end{align}
Hence
\begin{equation}
  \operatorname{dist}_2(b^\times,\operatorname{im}M)
  =\frac{\sqrt6}{4},
  \qquad
  y_\star
  =\frac{1}{2\sqrt6}(-3,-2,1,0,3,-1)^{\mathsf T}.
  \label{eq:appD-cross-distance}
\end{equation}
The vector $y_\star$ obeys $M^{\mathsf T}y_\star=0$ and
$y_\star^{\mathsf T}b^\times=\sqrt6/4$. Therefore a code coefficient $\beta b^\times$ has the exact nongeometric distance
\begin{equation}
  D_{\mathrm{geo}}=\frac{\sqrt6}{4}|\beta|.
  \label{eq:appD-beta-distance}
\end{equation}

Refining the
graph outside the four-terminal support appends rows and columns
\rev{to $M$.  The zero-extended witness still annihilates the enlarged
incidence matrix and has the same pairing with the extended cross-cell
vector, so the dual lower bound in
Equation~\eqref{eq:appD-beta-distance} persists, which is all that
Proposition~\ref{prop:sec6-uniform-persistence} uses.  The exact
reconstruction and residual in
Equations~\eqref{eq:appD-cross-pseudoinverse}--\eqref{eq:appD-cross-distance}
refer to the six-region problem.  A refinement whose new rows also
involve the original columns adds constraints to the least-squares
problem, so it cannot decrease the distance but in general changes both
the reconstruction and its exact value}.

\begin{revision}
The sector records in the controlled-center construction are
already classical from the viewpoint of every single region.  The
Hamiltonian dilation retains quantum bond states within each sector,
and their conditional entropies produce the exact cross-cell witness.
Equation~\eqref{eq:appD-raw-channel-upper-bound} concerns quantum
source recovery and supplies no substitute for the conditional
entropy calculation.
\end{revision}

\clearpage

\providecommand{\href}[2]{#2}\begingroup\raggedright\endgroup

\end{document}